\documentclass[12pt]{article}

\usepackage{natbib}
\usepackage[nottoc,notlof,notlot]{tocbibind}

\bibpunct{(}{)}{;}{a}{}{;}
\usepackage{breakcites}
\usepackage[colorlinks,linkcolor=black, urlcolor  = black, anchorcolor=black, citecolor=black]{hyperref}
\usepackage{graphicx}
\graphicspath{ {images/} }

\usepackage{algorithm}
\usepackage{algpseudocode}
\usepackage{amsmath}
\usepackage{thmtools}
\usepackage{mathtools,amsthm,amssymb}
\usepackage{bbm}
\usepackage{amsfonts}%
\usepackage{mathrsfs}

\usepackage{geometry} 
\usepackage{framed}
\usepackage{subfigure}
\usepackage{xcolor}
\usepackage{cleveref}

\usepackage{booktabs} 
\usepackage{array} 
\usepackage{paralist} 
\usepackage{verbatim} 
\usepackage{float}

\usepackage[justification=centering]{caption}

\let\oldbibliography\thebibliography
\renewcommand{\thebibliography}[1]{\oldbibliography{#1}\setstretch{1}\setlength{\itemsep}{0pt plus 0.3ex}}

\usepackage{setspace}

\numberwithin{equation}{section}

\usepackage{lipsum}
\usepackage{tabularx}

\usepackage{multicol}
\usepackage{multirow}
\usepackage[title]{appendix}

\usepackage{microtype}

\newtheoremstyle{myremark}
  {\topsep}   
  {\topsep}   
  {\normalfont}  
  {}          
  {\bfseries} 
  {.}         
  { }         
  {}          

\newtheorem{lemma}{Lemma}
\newtheorem{corollary}{Corollary}
\newtheorem{proposition}{Proposition}
\newtheorem{assumption}{Assumption}

\theoremstyle{definition}

\newtheorem{remark}{Remark}

\title{Post-selection inference for network structure}
\author{Eric Auerbach\footnote{Department of Economics, Northwestern University. E-mail: eric.auerbach@northwestern.edu.} \and Jonathan Auerbach\footnote{Department of Statistics, George Mason University. Email: jauerba@gmu.edu } \and Sidonia McKenzie\footnote{Department of Economics, University of Colorado Boulder. Email: Sidonia.McKenzie@colorado.edu \newline We thank Richard Blundell, Stephane Bonhomme, Fede Bugni, Ivan Canay, Ben Golub, Joel Horowitz, Sid Kankanala, Chuck Manski, Adam McCloskey, Stephen Nei, Edoardo Rainone, Andres Santos, Azeem Shaikh, Alex Torgovitsky, Zhen Xie, and seminar/conference participants at Emory, the University of Chicago, the University of Iowa, Northwestern University, the University of Toronto, and Washington University for helpful discussions.}}

\begin{document}
\maketitle

\begin{abstract}\setstretch{1}\noindent
Researchers often use the density of connections between groups of agents, such as communities, blocs, or markets, to characterize the structure of a social or economic network. In many cases, these groups are selected using the network data, making conventional fixed-group inference procedures potentially invalid. To address this issue, we develop two new confidence intervals that are universally valid post-selection in the sense that they guarantee simultaneous coverage asymptotically over all pairs of groups whose relative sizes do not vanish. Our first interval builds on a strategy of \cite{berk2013valid}. Our second interval is based on a Talagrand-type concentration inequality for empirical processes. Both intervals are simple to compute and scalable to large networks, but a key technical contribution of our paper is to show that the second interval is rate-optimal over a broader class of intervals. Three empirical illustrations show that accounting for selection can matter in practice. Some evidence for homophily in a social network and a hub-and-spoke structure in a trade network survives our correction, while evidence for a segmented market structure in a worker transition network does not. \looseness=-1 
\end{abstract}


\section{Introduction}\label{sec:introduction}
Network structure, as measured by the expected fraction or \emph{density} of connections between two groups of agents, is used to explain a wide variety of economic phenomena. For example, \cite{elliott2014financial} use transactions between financial institutions to characterize a core-periphery structure that influences a market's susceptibility to contagion. \cite{chetty2022socialI} use Facebook friendships between households in different socioeconomic groups to measure cross-class social connectedness which predicts economic mobility. \cite{jarosch2024granular} use worker transitions between Austrian firms to define market segments that determine firm market power. \looseness=-1 

In practice, the groups used to define network structure are rarely fixed in advance. They are instead typically selected by agents, institutions, or the researcher in a data-dependent way. For example, \cite{elliott2014financial} define the core to be the most densely connected institutions in the financial network.\footnote{\cite{elliott2014financial} refer to \cite{soramaki2007topology}, who use such a definition in Figure 2 of Section 4.} \cite{chetty2022socialI} define groups based on neighborhoods, which households can choose using their social ties. \cite{jarosch2024granular} apply a clustering algorithm to the network of worker transitions to define market segments.\looseness=-1 

When the groups are selected, conventional inference procedures that ignore selection can lead researchers to overstate the evidence for a particular network structure and, in extreme cases, even hallucinate structures that do not exist at all. To see what can go wrong, consider Figure~\ref{fig:apparent-core-periphery}. This figure depicts an undirected and unweighted network with $97$ edges connecting $84$ nodes. A visual inspection of the network suggests a weak core-periphery structure with a densely-connected core of twelve blue square nodes and a sparsely-connected periphery of seventy-two orange circle nodes. For this network, the difference in the observed fraction of connections for the two groups is large: approximately $0.21$ for the core and $0.03$ for the periphery. Using a stochastic blockmodel for the distribution of network connections,\footnote{This model is used by \cite{elliott2014financial} in their Section 4. See also Chapter 13.2.3 of \cite{jackson2008social}.} a conventional 95\% confidence interval for the density of connections between the blue squares is $[0.114,0.311]$. For the orange circles it is $[0.021,0.033]$. One might infer from these intervals that the density of connections between the blue squares is likely between $3$ to $15$ times that of the orange circles, and conclude that the network has a statistically significant core-periphery structure. \looseness=-1 

\begin{figure}[h!]
 \centering
  \includegraphics[width = 8cm,trim={3cm 3cm 3cm 3cm},clip]{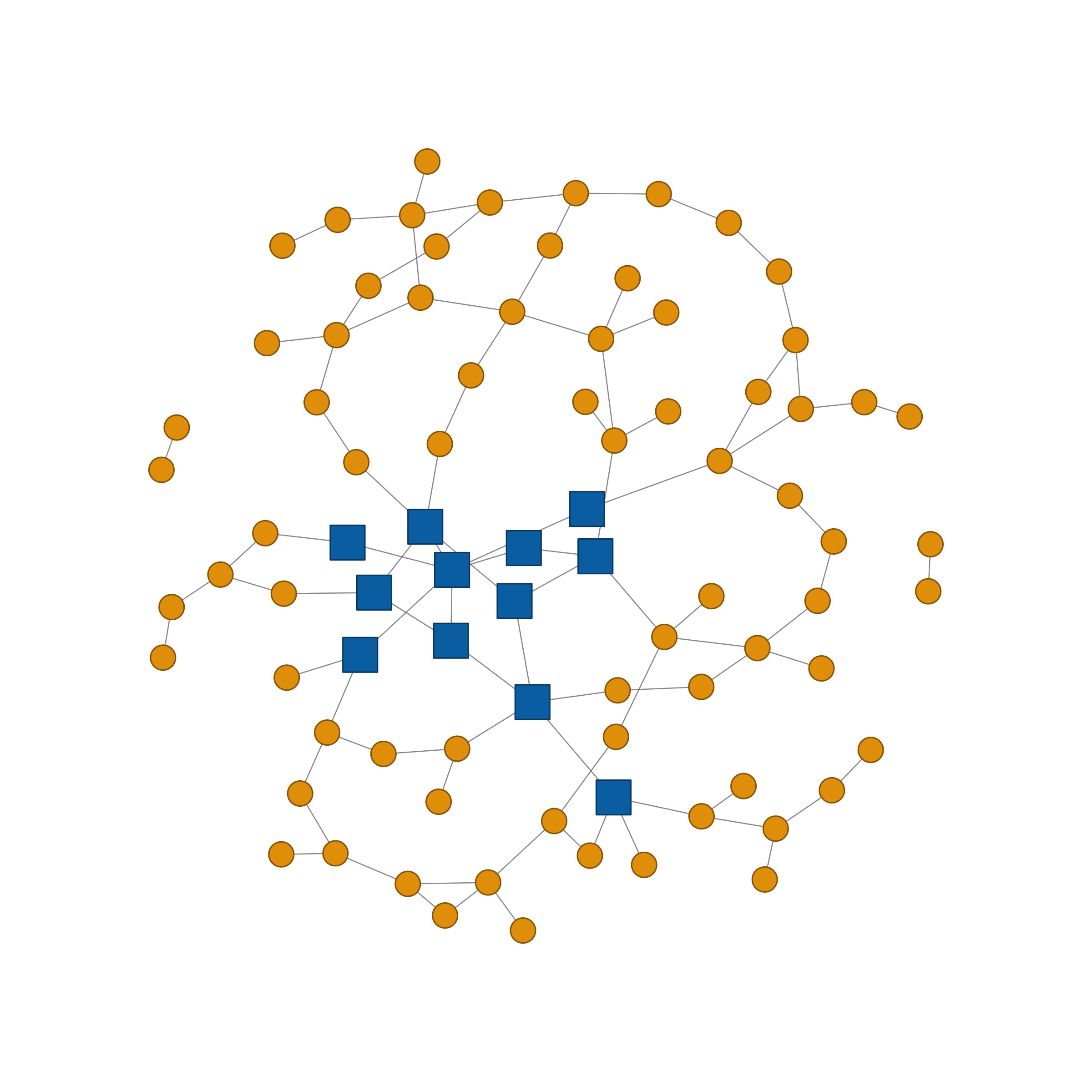}
 \caption{A network with an apparent core-periphery structure. The selected core nodes are square and colored blue. The selected periphery nodes are circles and colored orange. The node layout was determined using the Fruchterman-Reingold algorithm.}
 \label{fig:apparent-core-periphery}
 \end{figure}

The problem with this conclusion is that the statistical model that generated Figure~\ref{fig:apparent-core-periphery} does not have a core-periphery structure. In fact, the network is a draw from an Erdős--Rényi model where every pair of agents is connected with the same probability. The densities for the blue squares and orange circles are the same. The colors were selected by a spectral clustering algorithm \cite[see Section 1.1 of][]{rohe2011spectral} and isolated nodes were not depicted. The apparent gap in densities between the two groups is simply a coincidence of idiosyncratic error accentuated by a graph drawing and network clustering algorithms.\footnote{We emphasize that this is not a model-misspecification issue. The
Erdős--Rényi model is a special case of the stochastic blockmodel where all connection probabilities are equal. The problem comes from treating the data-selected groups as fixed. That selection can make conventional inferential approximations invalid is known in the post-selection inference literature, see for instance \cite{potscher1991effects,leeb2003finite,leeb2005model}.\looseness=-1} \looseness=-1 

\subsection{Organization of our paper}\label{sec:organization}
In this paper, we develop two confidence intervals for the density of connections between groups that are valid post-selection. We introduce the model and inference problem in Section~\ref{sec:model-inference}. The model nests a large class of dyadic regression models popular in the economics literature. The inference problem is to use network data drawn from the model to infer the density of connections between two groups. The network data may be used to determine which agents belong to which group. As motivated above, and formalized by our Corollary~\ref{cor:ci0-invalid} below, conventional inference procedures that do not account for group selection may be invalid. \looseness=-1 

We describe our approach to post-selection inference in Section~\ref{sec:post-selection}. Inference in our setting is complicated by the fact that, in practice, the economics literature often takes a relatively exploratory approach to network structure, where groups may be based on demographics, choice data, a community detection or clustering algorithm, or a visual inspection of a graph drawing. In some extreme cases, the researchers are themselves unable to formally articulate exactly how the groups were chosen, which \cite{jackson2008social} characterizes as an ``I will know a community [structure] when I see it'' approach in his Section 13.2.2. To accommodate such a wide variety of procedures for selecting groups that may be unstructured, ambiguously-defined, or aggressively data-mined, we recommend confidence intervals that simultaneously cover all nonvanishing group assignments. See Section~\ref{sec:simultaneous-inference} for a formal definition. Simultaneous coverage is a stringent requirement when compared to other approaches to post-selection inference (see our literature review in Section~\ref{sec:related-work}). However, as we argue in Section~\ref{sec:universal-validity}, it is the only way to ensure post-selection validity for arbitrary group selection rules.\looseness=-1 

We derive two post-selection confidence intervals in Section~\ref{sec:main-results}. Our first interval follows a strategy of \cite{berk2013valid}, which is to start with a conventional interval that does not account for selection and inflate its width until simultaneous coverage is achieved. The logic behind our second interval is to our knowledge new to the post-selection inference literature. It considers the empirical process defined by the estimation error indexed by the set of possible group assignments, and bounds the deviations of this process using a Talagrand-type concentration inequality due to \cite{klein2005concentration}. \looseness=-1 

While both intervals are easy to compute and scalable to large networks, their widths can contract at very different rates. In particular, a key technical contribution of our paper is to derive an asymptotic lower bound for the width of any collection of intervals that is both simultaneously valid and contains the estimated density. We then find that, among our two intervals, only the width of the second generally attains this bound asymptotically up to a constant factor. This is our Proposition~\ref{prop:propfour} in Section~\ref{sec:optimality-results}. By contrast, the width of the first interval can be made arbitrarily large relative to this lower bound.\footnote{The second interval does not dominate the first, however, so we recommend that researchers use the intersection of the two intervals in practice. See Section 4.1.3 below.  \looseness=-1 } We show that the two intervals have the same asymptotic order of width when linking behavior is homogeneous, as in an Erdős–Rényi model, but the first interval can be substantially wider for the kinds of sparse and degree-heterogeneous networks common in economic research. Corroborating simulation evidence can be found in Appendix Section~\ref{sec:simulations}. \looseness=-1 

We demonstrate our two intervals with three empirical illustrations in Section~\ref{sec:empirical}. The first illustration studies the homophily structure of 50 collegiate Facebook networks and finds evidence for homophily on graduation year and student/faculty status, but not on gender, choice of major, or residence. The second illustration finds evidence of a hub-and-spoke structure of a trade network. The third illustration finds relatively little evidence for a segmented market structure in a pseudo-employer worker-transition network originally constructed by \cite{schmutte2014free}. From this last illustration, we recommend that researchers exhibit some caution when using clustering algorithms to infer market segments from network data in practice.\looseness=-1 

Section~\ref{sec:conclusion} concludes. Proof of claims and other details can be found in the appendix.\looseness=-1 

\subsection{Related work}\label{sec:related-work}
While, to our knowledge, our paper is the first to provide simultaneous coverage guarantees for the problem of inferring network structure, there is a large and active econometrics and statistics literature on simultaneous inference in other settings. Examples include post-model selection inference and uniform confidence bands for structural, density, or impulse response functions. See Chapter 9 of \cite{lehmann2006testing}, reviews by \cite{chernozhukov2015valid,KuchibhotlaKolassaKuffner2022}, and, for specific examples, \cite{working1929applications,scheffe1953method,tukey1953problem,gine2010confidence,hardle2010confidence,liu2010simultaneous,horowitz2011applied,lounici2011global,horowitz2012uniform,horowitz2017nonparametric,hall2013simple,berk2013valid,chernozhukov2014anti,belloni2015some,lee2017doubly,zhang2017simultaneous,chen2018optimal,freyberger2018uniform,kato2018uniform,bachoc2019valid,kato2019uniform,montiel2019simultaneous,bachoc2020uniformly,davezies2021empirical,frandsen2021partial,chiang2023inference,cattaneo2024uniform,mccloskey2024hybrid,chen2025adaptive,frandsen2025simultaneous}. \looseness=-1 

A technical complication that distinguishes our paper from this literature has to do with the size of the index set. In our setting, the number of possible group pairs grows exponentially with the number of agents. One consequence of this regime is that it is not computationally feasible to compute objects like the largest t-statistic or smallest p-value over the index set. Another consequence is that we do not know of any natural way to justify a confidence interval based on a Gaussian approximation or a bootstrap procedure.\footnote{See Footnote~\ref{fn:gaussian-approximation} of Section 4.3.1.}  This means that many popular algorithms for constructing simultaneous confidence intervals in the literature \cite[for instance, Algorithm 1 of][]{KuchibhotlaKolassaKuffner2022} are not justified in our setting.\looseness=-1 

To deal with this complication, our intervals build on the finite sample concentration inequality literature. In particular, our second interval uses a Talagrand-type result due to \cite{klein2005concentration}, combined with a novel bound building on \cite{alon2006approximating,gittens2009error}, which does not restrict the cardinality of models being compared. In the literature, our use of Talagrand's inequality has some conceptual precedence in the work of  \cite{lounici2011global}, who consider a different problem of deconvolution density estimation. See Section 1.1 of \cite{chernozhukov2014anti} for a discussion.  \looseness=-1 

Alternatives to simultaneous coverage trade universality for
guarantees tailored to a specified selection procedure. Conditional or
selective methods provide coverage conditional on a selection event or for
the target produced by a particular selection rule. Examples include
inference on ranks
\citep[e.g.][]{andrews2024inference,mogstad2024inference,petrou2024inference}
and inference after model selection
\citep[e.g.][]{fan2001variable,PoetscherLeeb2009,belloni2012sparse,
BelloniChernozhukovHansen2014,Efron2014,Farrell2015,lee2016exact,
markovic2017unifying,chen2023selective,gao2024selective,
kelekidou2025high}. Related approaches use sample splitting, data fission,
or data thinning
\citep[e.g.][]{Moran1973,Cox1975,FanLv2008,
MeinshausenMeierBuehlmann2009,RinaldoWassermanGSell2019,
DiCiccioDiCiccioRomano2020,RitzwollerRomano2023,
neufeld2024thinning,fava2025training,leiner2025fission}.
Hybrid procedures are considered by \cite{mccloskey2024hybrid} and local simultaneous inference by \cite{ZrnicFithian2024}.  Average-coverage methods guarantee coverage after averaging over a specified
distribution of selection events
\citep[e.g.][]{ArmstrongKolesarPlagborgMoller2022,
LiuMoonSchorfheide2023}. \looseness=-1

These alternatives can be less conservative than global simultaneous
inference, but they require structure on how selection occurs. Many
tractable conditional and hybrid procedures rely on a joint Gaussian
approximation for the candidate estimators and an explicit characterization
of the selection event. We do not establish such a joint approximation over our exponentially large collection
of group pairs, and the procedures motivating our applications, including
community detection, centrality rankings, tuning-parameter searches, and
visual inspection, may generate complicated or incompletely specified
selection events. \looseness=-1

Our paper shares some similarities with, but is fundamentally different from, the literature that infers the parameters of a model with a latent group structure, sometimes called discrete heterogeneity. See, for instance, \cite{BonhommeManresa2015,SuShiPhillips2016,LuSu2017,BonhommeLamadonManresa2022,chetverikov2022spectral,MeleHaoCapePriebe2023,Jochmans2024,KitamuraLaage2024}. What distinguishes our problem of inferring network structure from this literature is that, in our setting, we do not assume that the groups selected by the researcher are (or are approximations of) some latent heterogeneity that determines (and is identified from) the distribution of network connections. Another strand tests for the presence or number of communities in a network. See, for instance, \cite{bickel2016hypothesis,lei2016goodness,Auerbach2022testing}. These papers do not consider our problem of post-selection inference for the density of connections between groups, however.  \looseness=-1

\subsection{Motivating examples}\label{sec:motivating-examples}
We describe three settings in the literature where researchers conduct inference about network structure. These examples motivate our model and assumptions in Section~\ref{sec:model-inference}, our focus on simultaneous coverage in Section~\ref{sec:post-selection}, and our empirical illustrations in Section~\ref{sec:empirical} below. \looseness=-1

\subsubsection{Social network}\label{sec:motivation-social-capital}
Social networks are often described as having a \emph{homophily structure} where groups of agents with similar socioeconomic characteristics are more densely connected than groups of agents with different characteristics. This network structure has been used to explain the diffusion of information, adoption of social norms, social segregation, and economic mobility. See, for instance, \cite{marmaros2006friendships,currarini2009economic,goeree20101,golub2012homophily,zeltzer2020gender,chetty2022socialI,chetty2022socialII,michelman2022old}. Because researchers often consider many different groups based on demographic, socioeconomic, geographic, or other characteristics, these density measures are natural candidates for post-selection inference.\looseness=-1

\subsubsection{Trade network}\label{sec:motivation-shock-propagation}
Trade, production, and financial networks are often described as having a \emph{core-periphery} or \emph{hub-and-spoke structure} where a small group of influential agents are more densely connected in the network. These network structures have been used to explain aggregate fluctuations in economic output in production networks as well as a financial market's susceptibility to a contagion. See, for instance, \cite{acemoglu2012network,acemoglu2015systemic,carvalho2014micro,elliott2014financial,elliott2022supply,jackson2024credit,PietrosantiRainone2023,buccheri2025realized}. Because the core or hub groups are often identified using network statistics such as degree or Katz-Bonacich centrality, the resulting density measures for these groups are natural candidates for post-selection inference.\looseness=-1

\subsubsection{Worker-flow network}
\label{sec:motivation-market-competition}
Worker-flow networks are often described as having a \emph{segmented market structure}, in which worker transitions are concentrated within particular groups of firms or jobs. This structure has been used to study worker mobility and wage differences. See, for instance, \cite{schmutte2014free,Nimczik2017,sorkin2018ranking,AbowdMcKinneySchmutte2019,berger2022labor,lamadon2022imperfect,jarosch2024granular,kline2024firm}. When the groups are selected using the worker-flow network, the resulting within and across-segment density measures are natural candidates for post-selection inference.\looseness=-1

\section{Model and inference problem}\label{sec:model-inference}
Section~\ref{sec:model} describes the model and Section~\ref{sec:inference-problem} the post-selection inference problem. Section~\ref{sec:model-motivating-examples} revisits the three motivating examples of Section~\ref{sec:motivating-examples} above. \looseness=-1

\subsection{Model}\label{sec:model}
\subsubsection{Terminology and notation} \label{sec:model-terminology}
A network is defined on two finite sets of agents. There is no restriction on how the two sets are related: the two sets could be identical, have no agents in common, or have some but not all agents in common. For $t \in \{1,2\}$ the $t$th set has $N_t$ agents indexed by $[N_t] :=  \{1,2,\ldots,N_t\}$. Ordered pairs of agents, containing one agent from each set, are indexed by $ij \in [N_1]\times [N_2]$. Every pair of agents $ij$ is assigned a real-valued random variable $Y_{ij}$ that describes the strength of some social or economic relationship between them. For example, $Y_{ij}$ may describe whether two students are friends or the amount of trade between two regions. The $N_1 \times N_2$ dimensional \emph{adjacency matrix} $Y$ contains $Y_{ij}$ as its $ij$th entry. When the relationship between two agents is not well-defined, we put a $0$ in the relevant entry of $Y$. For example, in the context of a social network, both sets may contain agent $i$, but it often does not make sense for $i$ to be friends with themselves. In this case, we follow the convention that sets $Y_{ii} = 0$. This convention should not be interpreted as treating structurally impossible relationships as observed non-links. When the network includes relationships that are not well-defined, the density of connections should be normalized in a way that does not count them, see Remark~\ref{rem:alternative-normalization} below. \looseness=-1

When the two sets of agents are identical, we say the network is \emph{unipartite}. When they contain no agents in common we say it is \emph{bipartite}. The network may also be \emph{directed} or \emph{undirected}. In a directed network, a pair of agents may be associated with two connections, one describing the connection in the direction from $i$ to $j$ and another describing the direction from $j$ to $i$. A convention in the literature is to represent unipartite directed networks with an asymmetric adjacency matrix where the $ij$th entry corresponds to one direction and the $ji$th entry corresponds to the other. In an undirected network, every pair of agents is associated with at most one connection, and a convention is to describe unipartite undirected networks with a symmetric adjacency matrix where both the $ij$th and $ji$th entries describe the connection between agents $i$ and $j$. In this paper, we follow the convention for directed but not undirected unipartite networks. For undirected unipartite networks with no loops we instead represent the network with an asymmetric upper diagonal adjacency matrix where the $ij$th entry describes the relationship between $i$ and $j$ if and only if $i < j$. The adjacency matrix contains $0$ in every entry where $i \geq j$. We adopt this convention because it helps simplify our notation. We do not believe it to be restrictive in practice. \looseness=-1

Finally, the network may be \emph{weighted}, where the entries of $Y$ take values in $\mathbbm{R}$, or \emph{unweighted}, where they take values in $\{0,1\}$. In the unweighted case, we follow a convention where  a value of $1$ indicates that a relationship exists and a value of $0$ indicates that it does not exist or it is not well-defined. \looseness=-1

\subsubsection{Definition of network structure}\label{sec:network-structure}
While the term \emph{network structure} is ubiquitous in the network economics literature, we do not know of any previous work that gives a formal definition. To provide one, we assume that the entries of $Y$ are real-valued random variables, the variation of which is determined by latent variables such as agent characteristics, link covariates, taste shocks, etc. Some of this variation is \emph{systematic}. It reflects social, economic, geographic, or institutional determinants of link formation that make some pairs of agents more likely to connect than others.  The remaining variation is \emph{idiosyncratic}. It reflects residual sources of variation such as idiosyncratic taste shocks, measurement error, or reporting error, not of interest to the researcher. Broadly speaking, when conducting inference about network structure, the goal of the researcher is to filter out the idiosyncratic and focus on the systematic variation.  \looseness=-1

To formalize this intuition, we represent the systematic variation using the sigma-field $\mathcal{H}$.  For each pair $ij$, we denote the distribution of $Y_{ij}$ conditional on $\mathcal{H}$ with $\mathbf F_{ij}(y) = \mathbbm P(Y_{ij}\leq y\mid\mathcal H)$  and use $\mathbf F=\{\mathbf F_{ij}\}_{i \in [N_1], j \in [N_2]}$ to describe the $N_1 \times N_2$ random array of conditional distribution functions. We use
$F=\{F_{ij}\}_{i\in[N_1],\,j\in[N_2]}$ for a generic realization of $\mathbf F$ and $\mathcal F$ for a generic deterministic class of possible realized arrays. We call a generic element of $\mathcal F$ a \emph{random graph model}. Assumption~\ref{ass:model} below imposes pointwise restrictions on the elements of $\mathcal F$, while the later assumptions impose uniform asymptotic restrictions on the class and on specified estimators. Each result applies to any deterministic class $\mathcal F$ on which the assumptions in its statement hold. The class may therefore differ across results, although we reuse the symbol $\mathcal F$. For a fixed $F$ we define the mean $\mu_{ij}(F) := \int y dF_{ij}(y)$, variance $\sigma^2_{ij}(F) := \int (y - \mu_{ij}(F))^2dF_{ij}(y)$, and error $\epsilon_{ij}(F) = Y_{ij} - \mu_{ij}(F)$. We also use $\mu(F) = \{\mu_{ij}(F)\}_{i \in [N_1], j \in [N_2]}$, $\sigma^2(F) =  \{\sigma_{ij}^2(F)\}_{i \in [N_1], j \in [N_2]}$, and $\epsilon (F) =  \{\epsilon_{ij}(F)\}_{i \in [N_1], j \in [N_2]}$. We suppress the dependence of $\mu$, $\epsilon$, and $\sigma$ on $F$ in our notation going forward.\looseness=-1

We use the term \emph{network structure} to refer to functions of the entries of the conditional mean matrix $\mu$. The problem of \emph{inferring network structure} is that of conducting statistical inference on functions of the entries of $\mu$ using $Y$ as data. The idea behind this definition is that $\mu$ collects the systematic component of link formation after conditioning on $\mathcal H$, i.e. the social, economic, geographic, institutional, or latent forces that make some pairs more likely to connect than others. The entries of $\epsilon$ describe the residual variation of $Y$, such as taste shocks and other residual consequences of human indeterminacy.\looseness=-1

\begin{remark}
\label{rem:H-observability}
The researcher does not need to observe the variables generating $\mathcal H$ or consistently estimate the  conditional law of $Y$ given $\mathcal H$. Our confidence intervals instead only require estimators of certain variance parameters that can, under regularity conditions, be consistently estimated using $Y$. The sigma-field $\mathcal H$ is part of the maintained model specified before the realized network is used for inference. Changing $\mathcal{H}$ generally changes the conditional mean matrix $\mu$ and thus our definition of what is network structure. 
\end{remark}

\begin{remark}\label{rem:structural-terminology}
The word ``structural'' in econometrics is often understood in the context of simultaneous equation modeling to refer to the relationship between two or more endogenous variables. It is as opposed to a ``reduced form'' model that describes the joint distribution of the endogenous variables. See \cite{fisher1966identification} for a textbook definition. This terminology is sometimes used in the literature on strategic network formation in which researchers specify a model where agents choose connections to maximize utility and the utility an agent receives from forming a connection depends on the connections made by the other agents. In this literature, the parameters of the agent utility functions are ``structural'' parameters whereas the $F$ serves as a ``reduced form'' description of equilibrium linking behavior \cite[see, for instance,][]{leung2015two,menzel2015strategic,ridder2015estimation}. Since $\mu$ is a function of $F$, our notion of network structure is, in the context of this literature, a ``reduced form'' parameter. \looseness=-1 
\end{remark}

\begin{remark}\label{rem:neyman-structural}
Another use of the word ``structural'' in econometrics is in the sense of \cite{neyman1948consistent}, who use it to describe a parameter that  ``appears in an infinity of probability laws of the observable random variables.'' Our definition of ``network structure'' is not necessarily structural in the sense of \cite{neyman1948consistent} since, for example, the variable $\mu_{ij}$ is a function of $\mu$ but only appears in the probability law of $Y_{ij}$. However, in Section~\ref{sec:inference-problem} below we focus specifically on network density measures defined on groups of nonvanishing size, which is more in the spirit of this definition. \looseness=-1
\end{remark}

\begin{remark}\label{rem:dyadic-model}
A concrete example of a random graph model is the nonparametric dyadic regression model $Y_{ij}=f(u_i,v_j,X_{ij},\eta_{ij})$, where $u_i$ and $v_j$ are agent-specific heterogeneity such as socioeconomic characteristics, $X_{ij}$ is agent-pair-specific heterogeneity such as physical distance, $\eta_{ij}$ is an idiosyncratic noise term, and $f$ is a measurable function. Popular parametric versions of this model include the gravity model, stochastic blockmodel, latent space model, nonlinear two-way or interactive fixed effects model, random geometric graph model, and random dot product graph model. See Section 3 of \cite{de2020econometric} and Section 4 of \cite{graham2020network} for reviews of this literature. To map this model into our notation, let $\mathcal H$ be the sigma-field generated by $\{u_k,v_l,X_{kl}\}_{k\in[N_1],l\in[N_2]}$, $\mu_{ij} = \mathbbm E\left[f(u_i,v_j,X_{ij},\eta_{ij})\mid \mathcal H \right]$, and $\epsilon_{ij}=Y_{ij}-\mu_{ij}$. In this example, network structure refers to the mean of $Y$ conditional on the agent-specific and pair-specific heterogeneity. It does not include the variation in $Y$ due to the idiosyncratic noise $\eta$.\looseness=-1
\end{remark}

\subsubsection{Two key restrictions}\label{sec:key-assumptions}
We make two key restrictions on the class of random graph models. Our first restriction is that, for every $F \in \mathcal{F}$, the support of $F_{ij}$ is contained in $[-B,B]$ for some finite $B$. Our second restriction is that  the entries of $Y$ are independent conditional on $\mathcal H$. The following Assumption~\ref{ass:model} summarizes our model with these two restrictions. \looseness=-1
\begin{assumption}\label{ass:model}
There exists a sigma-field $\mathcal H$ and an $\mathcal{H}$-measurable random array of conditional distribution functions $\mathbf F$ such that $\mathbf F \in \mathcal{F}$ almost surely and, conditional on $\mathcal{H}$, the entries of $Y$ are independent with marginal distribution functions $\mathbf F$. For every $F\in\mathcal F$ and $(i,j) \in [N_1] \times [N_2]$, the marginal distribution $F_{ij}$ is supported on
$[-B,B]$, where $B\in(0,\infty)$ is fixed over $\mathcal F$.
\end{assumption}

For a fixed $F\in\mathcal F$, we use  $\mathbbm P_F$ to denote the product probability measure on $\mathbbm R^{N_1\times N_2}$ whose $ij$th marginal distribution function is $F_{ij}$ and use $\mathbbm E_F$ and $\operatorname{Var}_F$ for expectation and variance under this law. Under Assumption~\ref{ass:model}, $\mathbbm P_{\mathbf F}$ is a
version of the conditional joint law of $Y$ given $\mathcal H$ and so $\mu_{ij}(\mathbf F) = \mathbbm E[Y_{ij}\mid\mathcal H]$ and $\sigma_{ij}^2(\mathbf F) = \mathbbm E[\epsilon_{ij}^2\mid\mathcal H]$ almost surely. For fixed $F$, when working under $\mathbbm P_F$ we
use the same notation
$\epsilon_{ij}=Y_{ij}-\mu_{ij}(F)$. The entries of $\epsilon$ are
then independent and satisfy
$\mathbbm E_F[\epsilon_{ij}]=0$ and
$\mathbbm E_F[\epsilon_{ij}^2]=\sigma_{ij}^2(F)$.

We consider the bounded support condition to be relatively innocuous, since in many settings the network is unweighted and so the entries of $F$ have support $\{0,1\}$. A consequence of this assumption is that the entries of $\mu$ and $\sigma$ exist and are finite. We conjecture that it is possible to weaken this assumption to a tail bound, but leave such an extension to future work.\looseness=-1

Conditional independence is restrictive, but common in the literature that conducts inference on network structure, see Remark~\ref{rem:conditional-independence-literature} below. The assumption is sometimes controversial because, in some settings, independent network connections are thought to rule out triadic closure, strategic complementarities, latent-space geometry, or other interdependencies that may drive network formation in practice. In our framework, however, the sigma-field $\mathcal{H}$ that defines $\mu$ is left relatively unrestricted. As a result, the assumption does not rule out these sorts of interdependencies so long as the variables that generated them are included in the sigma-field $\mathcal{H}$. We do rule out these interdependencies in $\epsilon$ once $\mathcal{H}$ has been fixed, however. \looseness=-1

Returning to the dyadic model in Remark~\ref{rem:dyadic-model}, if $\{\eta_{ij}\}_{i\in[N_1],j\in[N_2]}$ has entries that are independent conditional on the sigma-field generated by $\{u_k,v_l,X_{kl}\}_{k\in[N_1],l\in[N_2]}$  (a common assumption in the literature), then the residuals $\{\epsilon_{ij}\}_{i\in[N_1],j\in[N_2]}$ are conditionally independent. In this case, the entries of $Y$ may be dependent unconditionally through $\mathcal H$, but our assumption is that the remaining variation around $\mu$ is conditionally independent.\looseness=-1

A similar conditioning argument appears in the exchangeable-network literature, which starts from the assumption that the matrix $Y$ is a finite subarray of an infinitely exchangeable population and then appeals to the Aldous--Hoover--Kallenberg representation theorem to justify the model $Y_{ij}=f(t,u_i,v_j,w_{ij})$,  where $t$, $\{u_i\}_{i\in[N_1]}$, $\{v_j\}_{j\in[N_2]}$, and $\{w_{ij}\}_{i\in[N_1],j\in[N_2]}$ are mutually independent random variables with standard uniform marginal distributions.\footnote{For instance, Theorem 7.22 and Corollary 7.23 of \cite{kallenberg2006probabilistic}. For undirected unipartite networks, the analogous representation replaces $(u_i,v_j)$ by $(u_i,u_j)$ and imposes symmetry. The assumption that the latent variables have uniform marginal distributions is without loss.} Examples include \cite{bickel2009nonparametric,davezies2021empirical,menzel2021bootstrap,chiang2023inference,cattaneo2024uniform,chiang2026gaussian}. Setting $\mathcal{H} = \sigma\left(t, \{u_i\}_{i\in[N_1]}, \{v_j\}_{j\in[N_2]}\right)$, $\mu_{ij} = \mathbbm E[Y_{ij}\mid \mathcal{H}]$, and $\epsilon_{ij}=Y_{ij}-\mu_{ij}$ gives a representation in which the entries of $\epsilon$ are conditionally independent. This literature often frames the resulting conditional independence as relatively unrestrictive because exchangeability is viewed as a mild restriction.\looseness=-1

In many exchangeable-network settings, it is natural to condition only on the global variable $t$ and conduct inference on $\mu_{ij} = \mathbbm E[Y_{ij}\mid t]$. This estimand may represent, for example, the average density of an exchangeable population. We do not consider such an estimand in our paper, however, because it integrates out the node-level heterogeneity that generates heterogeneous network structure. That is, under this definition of $\mathcal{H}$, every entry of $\mu$ is the same, so the network is homogeneous by construction and so cannot exhibit a homophily, core-periphery, etc. structure. \looseness=-1

\begin{remark}\label{rem:conditional-independence-literature}
Our conditional independence assumption is common in the network-structure inference literature. For example, in his review of the problem, \cite{jackson2008social} summarizes two inferential approaches. The first approach, described in Chapter 13.2.3, employs a stochastic blockmodel. The second approach, described in Chapter 13.2.5, employs a latent space model. In both approaches, the connections are independent conditional on latent types or positions. Additional concrete examples from the literature are provided in Section~\ref{sec:model-motivating-examples} below.\looseness=-1
\end{remark}

\begin{remark}\label{rem:weak-dependence}
Our conditional independence assumption is used to apply the concentration inequalities in Appendix Section~\ref{app:concentration-lemmas}. Extending our results to weakly dependent networks would require replacing those inequalities with concentration results appropriate for the relevant dependence structure, for example, dependency-graph or mixing-based concentration inequalities.\footnote{For example, if the dependence between the entries of $\epsilon$ is described by a dependency graph \citep[as in, for instance,][]{fafchamps2007risk,tabord2019inference}, one could replace our independence-based concentration inequalities with dependency-graph analogues such as \citet{janson2004large} or \citet{ralaivola2015entropy}. If the dependence between the entries of $\epsilon$ is described by a $\psi$-mixing condition \citep[as in, for instance,][]{kojevnikov2021limit,leung2022causal}, one could use mixing-based concentration inequalities such as \citet{doukhan1999new} or \citet{amorino2025concentration}.}
We conjecture that our proof strategy could be adapted in this way, but the constants, rates, and variance quantities would likely change. We leave a formal weak-dependence extension to future work.\looseness=-1
\end{remark}

\subsubsection{Asymptotics}\label{sec:asymptotics}
Our post-selection coverage guarantees are asymptotic in that they refer to sequences of random graph models whose dimensions diverge. Formally, for any deterministic model class $\mathcal F$, we define $\mathcal F(n) := \left\{F\in\mathcal F:\min\{N_1(F),N_2(F)\}\geq n\right\}$. For any real-valued functional $a(F)$, we write $\liminf_{F\in\mathcal F:\,N_1,N_2\to\infty}a(F) := \liminf_{n\to\infty} \inf_{F\in\mathcal F(n)}a(F)$ and define the corresponding $\limsup$ analogously. We write $\lim_{F\in\mathcal F:\,N_1,N_2\to\infty}a(F)=a$ when the uniform liminf and limsup are both equal to $a$. We also often use $N_1,N_2 \to \infty$ instead of $n \to \infty$.\looseness=-1

Nearly every quantity we define in our paper is allowed to vary with $n$ along the sequence with three exceptions. The first exception is the uniform absolute bound on the support of the entries of $F$, given by $B$. The second exception is the confidence level, given by $\alpha$. The third exception is a uniform lower bound on the relative group sizes, given by $c$ in Section~\ref{sec:index-set} below. In principle, one could amend our proofs to allow these quantities to also vary with $n$, however this would complicate the arguments and since, to our knowledge, doing so has no clear benefit, we do not pursue this in our paper. \looseness=-1

\subsection{Inference problem}\label{sec:inference-problem}
We focus on the expected fraction or \emph{density} of connections between two groups of agents
\begin{align}\label{poi}
\theta(G_1,G_2) := \frac{1}{M_1M_2}\sum_{i \in [N_1], j \in [N_2]}\mu_{ij}G_{i,1}G_{j,2}
\end{align}
where,  for each  $t \in \{1,2\}$ and $i \in [N_t]$, the variable $G_{i,t} \in \{0,1\}$ indicates whether agent $i$ from set $t$ belongs to group $t$ and $M_t := \sum_{i \in [N_t]}G_{i,t}$ is the number of agents in group $t$. The $N_t$ dimensional vector $G_t$ contains $G_{i,t}$ as its $i$th entry.\footnote{Formally, $G_1$ and $G_2$ are random vectors on the same probability space as $Y$. We permit them to depend on $Y$ and on observed auxiliary information measurable with respect to $\mathcal H$. Independent algorithmic randomization may be absorbed into $\mathcal H$ without changing the conditional mean matrix. Remark~\ref{rem:external-auxiliary-data} extends the framework to arbitrary auxiliary data. We assume $M_1M_2>0$ almost surely.} Many network structures are described using densities of the form of (\ref{poi}). We illustrate this in the context of our three motivating examples in Section~\ref{sec:model-motivating-examples} below.

For undirected unipartite networks with no loops, researchers typically consider densities that are computed by summing over all unordered dyads. Appendix Section~\ref{app:undirected-normalization} extends the results of our paper to this setting. 

\begin{remark}\label{rem:alternative-normalization}
In some settings it may be natural to normalize the sum of connections with a denominator that is different than $M_1M_2$. For instance, if it is not possible for some pairs of agents to form a link, it may be preferable to normalize by the sum of possible connections, rather than the total number of pairs between the two groups. When the parameter of interest differs from \eqref{poi} only by replacing the denominator $M_1M_2$ with a positive measurable denominator $D$, one may first construct a confidence interval for \eqref{poi} and then rescale both endpoints by $M_1M_2/D$. The rescaled interval will inherit the coverage properties of the original. 
\end{remark}

The inference problem is to use the network data $Y$ to construct an asymptotic sequence of confidence intervals (indexed by $n \in \mathbbm{N}$ as in Section~\ref{sec:asymptotics} above) for $\theta(G_1,G_2)$ that is uniformly (over $\mathcal{F}(n)$) asymptotically level $1-\alpha$. That is, for a fixed $\alpha \in (0,1)$, to specify a confidence interval $CI(G_1,G_2;\alpha) = [L(G_1,G_2;\alpha), U(G_1,G_2;\alpha)]$ such that 
\begin{align}\label{postselectioninference}
\liminf_{F \in \mathcal{F}: N_1,N_2 \to \infty} \mathbbm{P}_{F}\left(\theta(G_1,G_2) \in CI(G_1,G_2;\alpha)\right) \geq 1-\alpha.
\end{align}
Condition (\ref{postselectioninference}) follows Equation 11.8 in Definition 11.1.4 of \cite{lehmann2006testing}.

Our only restriction on the groups $G_1$ and $G_2$ (introduced in Section~\ref{sec:index-set} below) is that the relative group sizes $M_1/N_1$ and $M_2/N_2$ are assumed not to vanish as $N_1,N_2 \to \infty$. Aside from this restriction, the groups can be \emph{any} collection of the two sets of agents. For example, the groups could be determined by sociodemographic characteristics of the agents such as age, race, or gender. They could describe agent choices such as place of residence, occupation, or political affiliation. They could also be the result of agent interactions over the network such as the diffusion of information or a strategic game played between neighbors. They could also be a binary treatment that is self-selected or assigned by a central planner. They could also be constructed by the researcher using a graph drawing, clustering, community detection, model selection or similar algorithm. Finally, the groups could be formed by combining two or more of the methods above. Our confidence intervals are designed to be used in all of these settings, as we discuss in Section~\ref{sec:post-selection} below. \looseness=-1

In these examples, the group assignments can depend on the network connections $Y$. A consequence of this is that a conventional confidence interval that ignores this dependence may fail to cover the parameter of interest at the desired level. We show this analytically as our Corollary~\ref{cor:ci0-invalid} in Section~\ref{sec:optimality-results} below. Confidence intervals that maintain coverage at the desired level are said to be valid post-selection.\looseness=-1

\begin{remark}\label{rem:random-estimand}
Conditional on $\mathcal H$, the matrix $\mu$ is fixed. When the groups depend on $Y$, however, the realized group assignment and hence the selected density $\theta(G_1,G_2)$ are random under $\mathbbm P_F$.\looseness=-1
\end{remark}

\begin{remark}\label{rem:selected-groups}
The groups may, but are not required to, be a determinant of link formation in the random graph model $F$. For instance, if the distribution of network connections is determined by a stochastic blockmodel, one could use a clustering algorithm to approximate the latent block assignments and use this output to define the groups of interest. Alternatively, the groups might be determined by agents interacting over the network. In this second scenario, the resulting groups may have little to do with the underlying incentives for the agents to form connections. However, we do not consider the groups in our setting to be noisy estimates of some ``true'' or oracle assignment. That is, if the groups are determined by a clustering algorithm, then our estimand is the density associated with the group assignment actually reported by the researcher, and not at some unobserved oracle assignment that the reported groups may theoretically be approximating. \looseness=-1 \end{remark}

\subsection{Motivating examples, continued}\label{sec:model-motivating-examples}
We revisit our model and assumptions in the context of our three motivating examples. To make our discussion concrete, we focus each example on a single paper from the literature.   \looseness=-1

\subsubsection{Social network}\label{sec:model-social-capital}
\cite{golub2012homophily} study how homophily affects the speed of learning in a social network. They specify a network formation model where each agent is assigned a type and the probability of a connection between two agents depends on their types. Types may combine demographic, socioeconomic, geographic, and behavioral characteristics. For instance, the authors write ``a type might consist of the 18-year-old female African Americans who have completed high school, live in a particular neighborhood, and do not smoke.'' Conditional on the type assignments and type-specific connection probabilities, connections are independent across unordered pairs, consistent with our Assumption~\ref{ass:model}. Taking $\mathcal{H}$ to include the type assignments and connection probabilities, $\mu_{ij}$ is the conditional probability of a connection for $ij$ and $\theta(g_1,g_2)$ is the average conditional probability of a connection between two groups $g_1$ and $g_2$. The groups may, but need not, be indicators for the agent types that determine the connection probabilities. In our empirical illustration in Section~\ref{sec:facebook100}, we use the difference $\Delta = \theta(g_1,g_1) - \theta(g_1,g_2)$ to describe the level of homophily in a social network. This difference is an unnormalized analog of a measure used by  \cite{golub2012homophily} in the equal-sized islands model.  \looseness=-1

\subsubsection{Trade network}\label{sec:model-shock-propagation}
\citet{elliott2022supply} study production in a supply network where each firm sources several essential inputs through potential relationships with other firms and invests to increase the strength of those relationships. Conditional on the potential supply network, the chosen relationship strengths, and any aggregate shock, potential relationships are independently operational or disrupted. In a finite-population version of their model, taking $\mathcal H$ to contain these variables and letting $Y_{ij}$ indicate whether firm $i$'s potential sourcing relationship with supplier $j$ is operational, with $Y_{ij}=0$ when $j$ is not a potential supplier of $i$, gives a network that satisfies Assumption~\ref{ass:model}. The conditional mean $\mu_{ij}$ equals the chosen strength of the relationship when $j$ is a potential supplier of $i$, and zero otherwise, and $\theta(g_1,g_2)$ is the expected fraction of ordered pairs, from firms in $g_1$ to suppliers in $g_2$, that are joined by an operational sourcing relationship, which is an example of our parameter of interest~(\ref{poi}). \looseness=-1

\subsubsection{Worker-flow network}\label{sec:model-market-competition}
\citet{jarosch2024granular} study how a segmented market structure drives firm market power. They specify a degree-corrected stochastic blockmodel where the numbers of transitions between ordered firm pairs are independent Poisson variables conditional on the firms' market assignments, market-pair flow intensities, and firm-specific outflow and inflow propensities. Letting $Y_{ij}$ indicate whether at least one worker transition from firm $i$ to firm $j$ is observed therefore produces conditionally independent Bernoulli connections. Taking
$\mathcal H$ to contain the market assignments, market-pair intensities, and firm-specific propensities, the resulting network satisfies Assumption~\ref{ass:model}. The conditional mean $\mu_{ij}$ is the probability
of observing at least one transition from firm $i$ to firm $j$, and $\theta(g_1,g_2)$ is the average conditional probability of observing at least one transition from a firm in $g_1$ to a firm in $g_2$, which is an example of our parameter of interest~(\ref{poi}).
\looseness=-1

\section{Post-selection inference}\label{sec:post-selection}
In this section, we propose a strategy for constructing confidence intervals for $\theta(G_1,G_2)$ that are valid post-selection in the sense of (\ref{postselectioninference}). Our strategy is to specify a collection of intervals that, under certain regularity conditions, is simultaneously asymptotically level $1-\alpha$ over all group pairs whose relative sizes are bounded below by a fixed positive constant. We then report the interval from this collection corresponding to the selected group $(G_1,G_2)$. We show that the resulting interval will satisfy (\ref{postselectioninference}) so long as the size of the selected groups do not vanish with probability approaching one. A key advantage of this simultaneous coverage approach is that the researcher does not need to specify a model for $(G_1,G_2)$ or even articulate how the groups were chosen. \looseness=-1

Section~\ref{sec:postselection-terminology} defines our simultaneous inference condition. Section~\ref{sec:postselection-motivating-examples} motivates simultaneous inference in the context of our three motivating examples. \looseness=-1

\subsection{Terminology and notation}\label{sec:postselection-terminology}
\subsubsection{Index set}\label{sec:index-set}
We define the index set of permissible group assignments to be 
\begin{align}
\mathcal{G}_c := \left\{(g_1,g_2) \in \{0,1\}^{N_1} \times \{0,1\}^{N_2}: \min_{t \in \{1,2\}}\frac{1}{N_t}\sum_{i \in [N_t]}g_{i,t} \geq c\right\}
\end{align}  
where $g_{i,t}$ denotes the $i$th entry of the vector $g_{t}$ and $c \in (0,1/2]$ is an arbitrary constant. In words, $\mathcal{G}_c$ is the set of all group-pairs $(g_1,g_2)$ such that the count of agents in $g_1$ is not smaller than $cN_1$ and the count of agents in $g_2$ is not smaller than $cN_2$. We emphasize that the variable $c$ is fixed and is not allowed to vary along the asymptotic sequence of random graph models (see Section~\ref{sec:asymptotics}). It is, however, not necessary for the researcher to choose a particular value of $c$ to implement our confidence intervals as this constant does not explicitly appear in any of our constructions.\looseness=-1

We assume that there exists a $c  \in (0,1/2]$ such that $\lim_{F \in \mathcal{F}: N_1, N_2 \to \infty} \mathbbm P_{ F}\left((G_1,G_2) \not\in \mathcal{G}_c \right) = 0$, i.e. the selected groups $(G_1,G_2)$ are contained in $\mathcal{G}_c$ with probability approaching one. The restriction $c>0$ requires both selected groups to contain a nonvanishing fraction of the corresponding node sets. This is a restriction of our uniform theory rather than a claim that smaller groups are economically unimportant. A group may contain a growing number of nodes while its relative size converges to zero, but our results do not provide uniform protection over all such groups when
the researcher may search across an exponentially large collection of assignments. The restriction $c\leq1/2$ is used only in the optimality arguments, where the proof constructs candidate groups containing at least one half of the nodes. It is not required for the coverage arguments in Propositions~\ref{prop:propone} and \ref{prop:proptwo}. We do not believe this second condition to be restrictive in practice.\looseness=-1

\subsubsection{Confidence interval and simultaneous inference condition}\label{sec:simultaneous-inference}
For any fixed $(g_1,g_2) \in \mathcal{G}_c$ and $\alpha \in (0,1)$, we define a \emph{confidence interval} to be an interval $CI(g_1,g_2; \alpha) = \left[L(g_1,g_2;\alpha),U(g_1,g_2;\alpha)\right]$ that is determined by the entries of $Y$ where $L(g_1,g_2;\alpha) \leq  U(g_1,g_2;\alpha)$ for any $(g_1,g_2)$ and $\alpha$. Each $L$ and $U$ are assumed to be measurable functions of $Y$, and we treat the collection of confidence intervals indexed by $\mathcal{G}_c$, $\{CI(g_1,g_2;\alpha)\}_{(g_1,g_2) \in \mathcal{G}_c}$, as a single measurable map from $Y$ to $(\mathbbm{R}^{2})^{\mathcal{G}_c}$. The variable $\alpha$ is fixed and is not allowed to vary along the asymptotic sequence of random graph models (see Section~\ref{sec:asymptotics}). We say that $\{CI(g_1,g_2;\alpha)\}_{(g_1,g_2) \in \mathcal{G}_c}$ is simultaneously (over $\mathcal{G}_c$) uniformly (over $\mathcal{F}$) asymptotically level $1-\alpha$ if 
\begin{align}\label{simul}
\liminf_{F \in \mathcal{F}: N_1,N_2\to\infty} \mathbbm P_{F}\left( \cap_{(g_1,g_2) \in \mathcal{G}_c}\left\{\theta(g_1,g_2) \in CI(g_1,g_2;\alpha)\right\}\right) \geq 1-\alpha.
\end{align}
In words, the collection of confidence intervals $\{CI(g_1,g_2;\alpha)\}_{(g_1,g_2) \in \mathcal{G}_c}$ controls the family-wise error rate over the index set $\mathcal{G}_c$ uniformly over any asymptotic sequence of random graph models in $\mathcal{F}$ with dimensions diverging to infinity. \looseness=-1

We measure the \emph{length} of the individual interval $CI(g_1,g_2;\alpha)$ using the difference 
\begin{align*}
|CI(g_1,g_2;\alpha)| := U(g_1,g_2;\alpha)-L(g_1,g_2;\alpha)
\end{align*}
 and the \emph{width} of the collection $\{CI(g_1,g_2;\alpha)\}_{(g_1,g_2) \in \mathcal{G}_c}$ using the supremum norm, i.e. 
\begin{align*}
\left|\left|\{CI(g_1,g_2;\alpha)\}_{(g_1,g_2) \in \mathcal{G}_c}\right|\right|_{\infty} := \max_{(g_1,g_2) \in \mathcal{G}_c}|CI(g_1,g_2;\alpha)|.
\end{align*}
In words, the width of a collection of confidence intervals indexed by $\mathcal{G}_c$ is the maximum length of the intervals in the collection. The supremum norm is commonly used in the simultaneous inference literature cited in Section~\ref{sec:related-work} above, since a single uncovered group is enough for the simultaneous event to fail. \cite{chen2025adaptive} write that ``the sup-norm provides a stronger, more informative sense in which the estimator is converging as it measures the maximal, rather than average, error over the support.''\looseness=-1

\subsubsection{Universal post-selection validity}\label{sec:universal-validity}
Any collection of intervals satisfying (\ref{simul}) also satisfies (\ref{postselectioninference}) if $\lim_{F \in \mathcal{F}: N_1, N_2 \to \infty} \mathbbm P_{F}\left((G_1,G_2) \not\in \mathcal{G}_c \right) = 0$. This is because
\begin{align*}
 \{\theta(G_1,G_2) \in CI(G_1,G_2;\alpha)\} \supseteq \cap_{(g_1,g_2) \in \mathcal{G}_c}\{\theta(g_1,g_2) \in CI(g_1,g_2;\alpha)\} \cap \{(G_1,G_2) \in \mathcal{G}_c\} 
\end{align*}
implies that, for any $F \in \mathcal{F}$,  
\begin{align*}
 \mathbbm P_{F}\left(\theta(G_1,G_2) \in CI(G_1,G_2;\alpha)\right) \geq  \mathbbm P_{ F}\left(\cap_{(g_1,g_2) \in \mathcal{G}_c}\{\theta(g_1,g_2) \in CI(g_1,g_2;\alpha)\} \cap \{(G_1,G_2) \in \mathcal{G}_c\} \right) \\
\geq  \mathbbm P_{ F}\left(\cap_{(g_1,g_2) \in \mathcal{G}_c}\{\theta(g_1,g_2) \in CI(g_1,g_2;\alpha)\}\right) -  \mathbbm P_{ F}\left(\{(G_1,G_2) \not\in \mathcal{G}_c\} \right).
\end{align*}
Since the argument does not depend on the selection rule, \cite{berk2013valid} call confidence intervals that satisfy (\ref{simul}) ``universally valid post-selection.''\looseness=-1

\begin{remark}\label{rem:universal-necessity}
Following \cite{kuchibhotla2020valid}, we remark that the simultaneous inference condition (\ref{simul}) is also necessary for post-selection validity uniformly over all selection rules with support in $\mathcal G_c$ (see specifically, the discussion after their Remark 3.4). To see this, fix a deterministic ordering of the finite set $\mathcal G_c$. For each realization of the data, define the random set $\mathcal G_c^\dagger(Y) := \{(g_1,g_2)\in\mathcal G_c: \theta(g_1,g_2)\notin CI(g_1,g_2;\alpha)\}$. If $\mathcal G_c^\dagger(Y)$ is nonempty, let $(G_1^\dagger,G_2^\dagger)$ be the first element of $\mathcal G_c^\dagger(Y)$ in the fixed ordering. If $\mathcal G_c^\dagger(Y)$ is empty, let $(G_1^\dagger,G_2^\dagger)$ be the first element of $\mathcal G_c$. Then $\{\theta(G_1^\dagger,G_2^\dagger)\in CI(G_1^\dagger,G_2^\dagger;\alpha)\} = \bigcap_{(g_1,g_2)\in\mathcal G_c} \{\theta(g_1,g_2)\in CI(g_1,g_2;\alpha)\}$. In words, this adversarial rule selects a missed interval whenever one exists.\looseness=-1\end{remark}

\begin{remark}\label{rem:external-auxiliary-data}
For simplicity, we assume above that auxiliary data used to select the groups are measurable with respect to $\mathcal H$. The set-inclusion argument establishing universal post-selection validity extends immediately to arbitrary auxiliary data $Z$. Specifically, under any joint law of $(Y,Z)$ whose marginal law for $Y$ is $\mathbbm P_F$, the simultaneous coverage event has the same probability and covers any groups selected as functions of $(Y,Z)$, provided they belong to $\mathcal G_c$ with probability approaching one. \looseness=-1
\end{remark}

\subsubsection{Simultaneous difference intervals}\label{sec:difference-intervals}
Any collection of intervals satisfying (\ref{simul}) can be used to construct a collection of intervals that simultaneously cover all pairwise differences in densities in the following way. Let $(g_1,g_2)$ and $(g_1',g_2')$ be a fixed pair of groups in $\mathcal{G}_c$ with $CI(g_1,g_2;\alpha)=[L(g_1,g_2;\alpha),U(g_1,g_2;\alpha)]$. Define the difference in densities $\Delta(g_1,g_2) := \theta(g_1,g_2) - \theta(g_1',g_2')$ and the difference interval
\begin{align}\label{diffint}
CI_{\Delta}(g_1,g_2,g_1',g_2';\alpha) := \left[L(g_1,g_2;\alpha)-U(g_1',g_2';\alpha),U(g_1,g_2;\alpha)-L(g_1',g_2';\alpha)\right]. 
\end{align}
Then since 
\[
\bigcap_{(g_1,g_2)\in\mathcal G_c}
\{\theta(g_1,g_2)\in CI(g_1,g_2;\alpha)\} 
\subseteq
\bigcap_{(g_1,g_2),(g_1',g_2')\in\mathcal G_c}
\{
\theta(g_1,g_2)-\theta(g_1',g_2')
\in CI_{\Delta}(g_1,g_2,g_1',g_2';\alpha)
\}
\]
it follows that if $CI$ is simultaneously uniformly asymptotically level $1-\alpha$ then so too is $CI_{\Delta}$. The result also remains valid after data-dependent selection of two group pairs, provided both selected pairs belong to $\mathcal G_c$ with probability approaching one. If $CI(g_1,g_2;\alpha)$ and $CI(g_1',g_2';\alpha)$ are strictly disjoint, then $CI_{\Delta}(g_1,g_2,g_1',g_2';\alpha)$ excludes zero, providing evidence of a nonzero difference in densities. Overlap does not imply that the two densities are equal, however.\looseness=-1

\subsection{Motivating examples, continued}\label{sec:postselection-motivating-examples}
We discuss the simultaneous inference condition in the context of our three motivating examples. \looseness=-1

\subsubsection{Social network}\label{sec:postselection-social-capital}
We highlight two selection concerns in this example. First, researchers may examine many demographic,
socioeconomic, geographic, and behavioral characteristics, together with their intersections, and emphasize those results displaying the clearest patterns of homophily. Second, the group memberships may depend on agent choices, such as place of residence or choice of college major, which may be informed by social ties. For example, to characterize the homophily structure of a social network between Caltech students, \cite{jackson2023dynamics} consider a large class of group definitions based on gender, ethnicity, housing choice, major choice, and pairwise intersections of these characteristics in their Table 2 as well as malleable characteristics such as elicited risk preference, level of altruism, body mass index, academic performance, time spent sleeping, time spent working, or time spent playing video games in their Table 4. In principle, one could conduct post-selection inference in this setting by fixing ex-ante the reported characteristics and a joint model of how the malleable characteristics relate to link formation. Alternatively, one can use our proposed intervals, which are valid post-selection without specifying either selection mechanism. \looseness=-1

\subsubsection{Trade network}\label{sec:postselection-shock-propagation}
We highlight two selection concerns in this example. First, membership in the proposed core or hub group is often determined with the same network connections used to estimate the group densities. Second, researchers may consider several centrality measures, thresholds, and core/hub group sizes, and then emphasize those producing the clearest structure. For example, \citet{soramaki2007topology} retain the Fedwire payment links that account for 75 percent of daily payment value and identify a tightly connected clique in the resulting thresholded network. \citet{carvalho2014micro} characterize influential suppliers using binary outdegree in Figure~2, weighted outdegree in Figure~3, and Katz--Bonacich centrality in Figure~4. In principle, one could conduct post-selection inference in this setting by specifying a network formation model and deriving the conditional distribution of the empirical densities for the given choice of centrality measure, threshold, and core/hub group sizes. Alternatively, one can use our proposed intervals, which are valid post-selection without requiring the researcher to commit to these choices or characterize the resulting selection event.  \looseness=-1

\subsubsection{Worker-flow network}
\label{sec:postselection-market-competition}
We are concerned about selection in this example when the same worker-flow network is used
both to define market segments and to estimate densities within or
across the selected segments. For example, \citet{schmutte2014free} selects segments by modularity maximization, while \citet{Nimczik2017} and \citet{jarosch2024granular} estimate market assignments under a degree-corrected stochastic blockmodel, using a marginal-likelihood criterion and a minimum-description-length criterion, respectively. Our simultaneous intervals instead remain valid for any selected groups in $\mathcal G_c$, allowing the researcher to compare multiple market definitions, numbers of markets, or
clustering algorithms without having to characterize the resulting selection event.
\looseness=-1

\section{Main results}\label{sec:main-results}
Section~\ref{sec:confidence-intervals} describes our two confidence intervals. Section~\ref{sec:assumptions} states two additional assumptions. Our main results are in Section~\ref{sec:results}. \looseness=-1

\subsection{Two confidence intervals}\label{sec:confidence-intervals}
We develop two confidence intervals. Our first interval builds on a proposal of \cite{berk2013valid}, which is to start with a conventional interval (that covers the parameter of interest (\ref{poi}) for a fixed group assignment) and then inflate the length of the interval until condition (\ref{simul}) holds. Our second interval is derived by combining a Talagrand-like concentration inequality for the maximum of an empirical process with a novel bound building on \cite{alon2006approximating,gittens2009error}. \looseness=-1

\subsubsection{First interval}\label{sec:first-interval}
The first interval we consider is
\begin{align}
CI_1(g_1,g_2;\alpha) = \hat{\theta}(g_1,g_2) \pm \left[K_1(\alpha) \times \hat{\sigma}(g_1,g_2)\right]/\sqrt{m_1m_2}
\end{align}
where $\hat{\theta}(g_1,g_2) = \frac{1}{m_1m_2}\sum_{i \in [N_1], j \in [N_2]}Y_{ij}g_{i,1}g_{j,2}$, $m_t=\sum_{i\in[N_t]}g_{i,t}$ for $t\in\{1,2\}$, $\sigma(g_1,g_2) := \sqrt{\frac{1}{m_1m_2}\sum_{i \in [N_1], j \in [N_2]}\sigma_{ij}^2g_{i,1}g_{j,2}}$, $\sigma(g_1,g_2)/\sqrt{m_1m_2}$ is the standard deviation of $\hat{\theta}(g_1,g_2)$ under $\mathbbm P_F$, $\hat{\sigma}(g_1,g_2)$ is a consistent estimator of $\sigma(g_1,g_2) $,\footnote{For now we assume that a consistent estimator exists. See Assumption~\ref{ass:scale-consistency}(i) in Section~\ref{sec:assumptions} below. Appendix~\ref{app:variance-estimation} provides three strategies for constructing $\hat{\sigma}(g_1,g_2)$: deterministic conservative bounds in Section~\ref{sec:deterministic-scale-bounds}, a raw second-moment plug-in in Section~\ref{sec:raw-plugin}, and a residual plug-in based on singular-value thresholding in Section~\ref{sec:usvt}.  Sections~\ref{sec:raw-plugin} and~\ref{sec:usvt} provide sufficient conditions for consistency.  \looseness=-1} $\alpha \in (0,1)$ is a fixed constant, and \looseness=-1
\begin{align*}
K_1(\alpha) &:= \sqrt{1.39\left(N_1 + N_2\right) - 2\ln(\alpha/2)}.
\end{align*} 

The logic behind the interval $CI_1$ follows \cite{berk2013valid} and the formal argument is detailed in the proof of Proposition~\ref{prop:propone} in Appendix Section~\ref{app:propositions} below. It starts with the conventional interval $CI_0(g_1,g_2;\alpha) = \hat{\theta}(g_1,g_2) \pm \left[K_0(\alpha) \times \hat{\sigma}(g_1,g_2)\right]/\sqrt{m_1m_2}$ where $K_0(\alpha)$ is the $1-\alpha/2$ quantile of a standard normal distribution. We show in our Corollary~\ref{cor:ci0-invalid} of Section~\ref{sec:optimality-results} below that the interval $CI_0$ does not satisfy the simultaneous inference condition (\ref{simul}) and so is not necessarily valid post-selection (see Remark \ref{rem:universal-necessity} in Section~\ref{sec:universal-validity} above). To address this issue, the idea is to replace the critical value $K_0(\alpha)$ with one that is large enough to satisfy (\ref{simul}). \looseness=-1

\begin{remark}\label{rem:posi-constant}
\cite{berk2013valid} call the smallest value of $K_1(\alpha)$ that guarantees (\ref{simul}) the ``PoSI constant'' (see the discussion after their Lemma 4.1). The constant may vary with $N_1$, $N_2$, and $\alpha$, but not with any other feature of the random graph model that generated the data. We do not analytically solve for this constant because doing so is, to our knowledge, computationally intractable in our setting. Instead, we use $\sqrt{1.39\left(N_1 + N_2\right) - 2\ln(\alpha/2)}$. While this choice of $K_1(\alpha)$ is conservative, we show in Proposition~\ref{prop:propthree} of Section~\ref{sec:optimality-results} below that it is optimal up to a constant. Specifically, we show that, under certain conditions on class of random graph models $\mathcal{F}$, the PoSI constant cannot be less than $K_1(\alpha)/6$. \looseness=-1
\end{remark}

\begin{remark}\label{rem:ci1-width}
The width of the ``oracle''  $CI_1$ that uses the unknown $\sigma(g_1,g_2)$ instead of the estimator $\hat{\sigma}(g_1,g_2)$ satisfies 
\begin{align*}
\left|\left|\left\{CI_1(g_1,g_2;\alpha)\right\}_{(g_1,g_2) \in \mathcal{G}_c}\right|\right|_{\infty} \asymp \frac{\sqrt{(N_1+N_2)}||\sigma||_F}{N_1N_2}
\end{align*}
where we use the notation $a \asymp b$ to mean that there exists constants $\overline{c} \geq \underline{c} > 0$ such that $\underline{c} a \leq b \leq \overline{c} a$. This is because 
 \begin{align*}
\left|\left|\left\{CI_1(g_1,g_2;\alpha)\right\}_{(g_1,g_2) \in \mathcal{G}_c}\right|\right|_{\infty}  := \max_{(g_1,g_2) \in \mathcal{G}_c}2K_1(\alpha)\sigma(g_1,g_2)/\sqrt{m_1m_2} \leq 2K_1(\alpha)||\sigma||_F/(c^2N_1N_2) 
 \end{align*} 
since $m_1m_2 \geq c^2N_1N_2$ by definition of $\mathcal{G}_c$ and 
 \begin{align*}
\left|\left|\left\{CI_1(g_1,g_2;\alpha)\right\}_{(g_1,g_2) \in \mathcal{G}_c}\right|\right|_{\infty}  := \max_{(g_1,g_2) \in \mathcal{G}_c}2K_1(\alpha)\sigma(g_1,g_2)/\sqrt{m_1m_2} \geq 2K_1(\alpha)||\sigma||_F/(N_1N_2)
 \end{align*} 
 by choosing $(g_1,g_2) = (\iota_{N_1},\iota_{N_2})$ where $\iota_{N_t}$ is an $N_t\times 1$ dimensional vector of $1$s for $t \in \{1,2\}$. It follows that the width of the oracle $CI_1$ converges to $0$ as $N_1, N_2 \to \infty$ since the entries of $\sigma$ are uniformly bounded by Assumption~\ref{ass:model}.  \looseness=-1
\end{remark}

\subsubsection{Second interval}\label{sec:second-interval}
The second interval we consider is
\begin{align}\label{secondinterval}
CI_2(g_1,g_2;\alpha) = \hat{\theta}(g_1,g_2) \pm \left[\hat{\bar{\tau}} + K_2(\alpha)\times \hat{V}\right]/(m_1m_2)
\end{align}
where $\hat{\theta}(g_1,g_2) = \frac{1}{m_1m_2}\sum_{i \in [N_1], j \in [N_2]}Y_{ij}g_{i,1}g_{j,2}$, $m_t = \sum_{i \in [N_t]}g_{i,t}$ for $t \in \{1,2\}$, $\alpha \in (0,1)$ is a fixed constant, $K_2(\alpha) := \sqrt{-2\ln(\alpha)}$,
\begin{align*}
\bar{\tau} &:= 1.01\mathbbm{E}_F\left[\sum_{i \in [N_1]}\sqrt{\sum_{j \in [N_2]}\epsilon_{ij}^2} + \sum_{j \in [N_2]}\sqrt{\sum_{i \in [N_1]}\epsilon_{ij}^2}\right] + 0.25\sqrt{\sum_{i \in [N_1],j\in[N_2]}\sigma_{ij}^2}, \\
V &:= \sqrt{4B\mathbbm{E}_F\left[\sum_{i \in [N_1]}\sqrt{\sum_{j \in [N_2]}\epsilon_{ij}^2} + \sum_{j \in [N_2]}\sqrt{\sum_{i \in [N_1]}\epsilon_{ij}^2}\right] + \sum_{i \in [N_1],j \in [N_2]}\sigma_{ij}^2 + B\sqrt{ \sum_{i \in [N_1],j \in [N_2]}\sigma_{ij}^2}},
\end{align*}
$\hat{\bar{\tau}}$ is a consistent estimator of $\bar{\tau}$ and $\hat{V}$ is a consistent estimator of $V$.\footnote{For now we assume that consistent estimators exist. See Assumption~\ref{ass:scale-consistency}(ii)--(iii) in Section~\ref{sec:assumptions} below. Appendix~\ref{app:variance-estimation} provides three strategies for constructing $\hat{\bar{\tau}}$ and $\hat V$: deterministic conservative bounds in Section~\ref{sec:deterministic-scale-bounds}, raw second-moment plug-ins in Section~\ref{sec:raw-plugin}, and residual plug-ins based on singular-value thresholding in Section~\ref{sec:usvt}. Sections~\ref{sec:raw-plugin} and~\ref{sec:usvt} provide sufficient conditions for consistency. \looseness=-1}  \looseness=-1

The logic behind the interval $CI_2$, detailed in the proof of Proposition~\ref{prop:proptwo} in Appendix Section~\ref{app:propositions} below, follows a combination of a version of Talagrand's inequality due to \cite{klein2005concentration} (Lemma~\ref{lem:klein-rio} in Appendix Section~\ref{app:concentration-lemmas} below), Theorem 3 of \cite{gittens2009error} (Lemma~\ref{lem:dagger-infty} in Appendix Section~\ref{app:matrix-norm-lemmas}), and a refinement of Lemma 3.1 of \cite{alon2006approximating} (Lemmas~\ref{lem:cut-infty} and \ref{lem:restricted-cut} in Appendix Section~\ref{app:matrix-norm-lemmas}). While our use of Talagrand's inequality shares some conceptual similarities with that in the proof of Proposition 1 of \cite{lounici2011global}, the main proof strategy is, to our knowledge, original to our paper.  \looseness=-1

\begin{remark}\label{rem:ci2-width}
The width of the ``oracle'' $CI_2$ that uses the unknown $\bar{\tau}$ and $V$ instead of the estimators $\hat{\bar{\tau}}$ and $\hat{V}$ satisfies 
\begin{align*}
\left| \left| \left\{CI_2(g_1,g_2;\alpha)\right\}_{(g_1,g_2) \in \mathcal{G}_c}\right| \right|_{\infty}  \asymp \left[\bar{\tau} + V\right]/(N_1N_2)
\end{align*}
where we use the notation $a \asymp b$ to mean that there exists constants $\overline{c} \geq \underline{c} > 0$ such that $\underline{c}a \leq b \leq \overline{c}a$. This is because 
\begin{align*}
\left| \left| \left\{CI_2(g_1,g_2;\alpha)\right\}_{(g_1,g_2) \in \mathcal{G}_c}\right| \right|_{\infty}  := \max_{(g_1,g_2) \in \mathcal{G}_c}2\left[\bar{\tau} + K_2(\alpha)\times V\right]/(m_1m_2) \leq 2\left[\bar{\tau} + K_2(\alpha)\times V\right]/(c^2 N_1N_2)
\end{align*}
since $m_1m_2 \geq c^2N_1N_2$ by definition of $\mathcal{G}_c$ and
\begin{align*}
\left| \left| \left\{CI_2(g_1,g_2;\alpha)\right\}_{(g_1,g_2) \in \mathcal{G}_c}\right| \right|_{\infty}  := \max_{(g_1,g_2) \in \mathcal{G}_c}2\left[\bar{\tau} + K_2(\alpha)\times V\right]/(m_1m_2)  \geq 2\left[\bar{\tau} + K_2(\alpha)\times V\right]/(N_1N_2)
\end{align*}
since $m_1m_2 \leq N_1N_2$. It follows that the width of the oracle $CI_2$ converges to $0$ as $N_1,N_2 \to \infty$, since the entries of $\epsilon$ and $\sigma$ are uniformly bounded by Assumption~\ref{ass:model}. Under Assumptions~\ref{ass:model} and \ref{ass:nondegeneracy}(w), $\bar{\tau}$ is not small relative to $V$, and the width of $CI_2$ is $O\left(\frac{\bar{\tau}}{N_1N_2}\right)$. See Lemma~\ref{lem:residual-norm-relations} in Appendix Section~\ref{app:additional-lemmas} below. We show that, under these conditions, this width is optimal up to a constant factor in Proposition~\ref{prop:propfour} of Section~\ref{sec:optimality-results} below. \looseness=-1
\end{remark}

\subsubsection{Combined interval}\label{sec:combined}
We call the confidence interval formed by the intersection of $CI_1$ and $CI_2$ the combined interval. That is,\looseness=-1
\begin{align}\label{combinedinterval}
CI_{\cap}(g_1,g_2;\alpha) = CI_1(g_1,g_2;\alpha/2)\cap CI_2(g_1,g_2;\alpha/2). 
\end{align}
Since both intervals are centered at $\hat{\theta}(g_1,g_2)$ and their widths are relatively insensitive to $\alpha$ (see Remarks \ref{rem:ci1-width} and \ref{rem:ci2-width} above), we find that the combined interval is typically close to the shorter of the two intervals in practice. Our main recommendation for applied researchers is to use $CI_{\cap}$ in practice when the conditions in both Propositions~\ref{prop:propone} and ~\ref{prop:proptwo} are plausible. \looseness=-1

\looseness=-1

\subsection{Assumptions}\label{sec:assumptions}
In addition to Assumption~\ref{ass:model} in Section~\ref{sec:key-assumptions}, we state two additional conditions as our Assumptions~\ref{ass:nondegeneracy} and \ref{ass:scale-consistency} below. Unlike Assumption~\ref{ass:model}, the conditions in this subsection are not pointwise restrictions on individual random graph models. Assumption~\ref{ass:nondegeneracy} restricts the deterministic
class $\mathcal F$ as a whole, while Assumption~\ref{ass:scale-consistency} restricts the estimation error of $\hat{\sigma}(g_1,g_2)$, $\hat{V}$, and $\hat{\bar{\tau}}$ uniformly over that class. We say that such assumptions \emph{hold on $\mathcal F$} when the corresponding uniform limit is satisfied for any asymptotic sequence of models in $\mathcal{F}$.\looseness=-1

Our second assumption is a lower bound on the magnitude of the variation in the network connections. There is a weak version of the assumption, Assumption~\ref{ass:nondegeneracy}(w), and a strong version, Assumption~\ref{ass:nondegeneracy}(s).\looseness=-1
\begin{assumption}\label{ass:nondegeneracy}
\begin{itemize}
\item[w.] $\lim_{F \in \mathcal{F}: N_1,N_2 \to \infty}||\sigma||_F = \infty$ where $||\sigma||_F := \sqrt{\sum_{i \in [N_1], j \in [N_2]}\sigma_{ij}^2}$.
\item[s.] $\lim_{F \in \mathcal{F}: N_1, N_2 \to \infty}\sqrt{\min(N_1,N_2)}\min_{(g_1,g_2) \in \mathcal{G}_c}\sigma(g_1,g_2) = \infty$.
\end{itemize}
\end{assumption}

We consider Assumption~\ref{ass:nondegeneracy}(w), to be a mild lower bound on the total residual variation in the network. It is used to ensure that the coverage of our intervals is uniform in the sense of condition (\ref{simul}) of Section~\ref{sec:simultaneous-inference} below.\footnote{See, for instance, Example 11.2.7 of \cite{lehmann2006testing} for an explanation of how uniform coverage can fail without such a condition.}  Recall from Section~\ref{sec:asymptotics} that our asymptotic arguments refer to a sequence of models, where $N_1$ and $N_2$ are increasing along the sequence, and that the entries of $\sigma$ may change arbitrarily along this sequence. In particular, the entries of $\sigma$ may converge to $0$, which is used in the literature to model sparse networks. See, for example, \cite{bickel2009nonparametric}. What Assumption~\ref{ass:nondegeneracy}(w) says is that, even if the entries of $\sigma$ are converging to $0$, the sum of the squared entries is diverging. Intuitively, the condition says that ``on average'' the entries of $\sigma$ cannot vanish at the rate of $1/\sqrt{N_1N_2}$ or faster. For unweighted networks, this rules out ultra-sparse sequences of networks with an almost surely bounded number of edges (i.e. networks where almost every agent has degree $0$). Such ultra-sparse regimes exist, but we do not believe them to be common in economic research, and so we do not think that this condition is restrictive in practice. \looseness=-1

\begin{remark}\label{rem:tau-diverges}
A consequence of Assumptions~\ref{ass:model} and \ref{ass:nondegeneracy}(w) is that $\lim_{F \in \mathcal{F}: N_1,N_2 \to \infty}\bar{\tau} = \infty$. This follows the third displayed equation in Lemma~\ref{lem:residual-norm-relations} of Appendix Section~\ref{app:additional-lemmas}.  \looseness=-1
\end{remark}

Assumption~\ref{ass:nondegeneracy}(s) places more restrictions on the amount of variation in the matrix $\sigma$. Intuitively, it says that for any groups $(g_1,g_2) \in \mathcal{G}_c$ the average of the entries of $\sigma^2$ over the groups must be large relative to $\min(N_1,N_2)^{-1}$. For unweighted networks, this rules out sequences of networks where the agents have a bounded number of connections (i.e. the agent degrees do not grow with the dimensions of the matrix). Such ``sparse'' regimes exist and are not uncommon in the economics literature. As a result, we recommend that when researchers are working with sparse (but not ultra-sparse) networks, they use our $CI_2$ (whose justification relies only on Assumption~\ref{ass:nondegeneracy}(w)) rather than our $CI_1$ (whose justification relies on Assumption~\ref{ass:nondegeneracy}(s)). See Section~\ref{sec:results} below for a discussion.  \looseness=-1

Our third assumption is a high-level condition that says that $\hat{\sigma}(g_1,g_2)$, $\hat{V}$, and $\hat{\bar{\tau}}$ converge to their analogs $\sigma(g_1,g_2)$, $V$, and $\bar{\tau}$ as $N_1,N_2 \to \infty$. \begin{assumption}\label{ass:scale-consistency}
There exists a deterministic sequence of nonnegative real numbers $r$ with $r \to 0$ as $N_1,N_2 \to \infty$ such that
\begin{itemize}
\item[i.] $\lim_{F \in \mathcal{F}: N_1,N_2 \to \infty}\mathbbm{P}_{F}\left(\max_{(g_1,g_2) \in \mathcal{G}_c}\left|\frac{\sigma(g_1,g_2) - \hat{\sigma}(g_1,g_2)}{\sigma(g_1,g_2)}\right| > r\right) = 0$,
\item[ii.] $\lim_{F \in \mathcal{F}: N_1,N_2 \to \infty}\mathbbm{P}_{F}\left(\left|\frac{\bar{\tau} - \hat{\bar{\tau}}}{\bar{\tau}}\right| > r\right) = 0$.
\item[iii.] $\lim_{F \in \mathcal{F}: N_1,N_2 \to \infty}\mathbbm{P}_{F}\left(\left|\frac{V - \hat{V}}{V}\right| > r\right) = 0$.
\end{itemize}
\end{assumption}

Assumption~\ref{ass:scale-consistency}(ii) says that the estimation error for $\hat{\bar{\tau}}$ vanishes relative to the magnitude of $\bar{\tau}$. Assumption~\ref{ass:scale-consistency}(iii) says the estimation error for $\hat{V}$ vanishes relative to the magnitude of $V$. Assumption~\ref{ass:scale-consistency}(i) says that the estimation error for $\hat{\sigma}(g_1,g_2)$ vanishes relative to the magnitude of $\sigma(g_1,g_2)$, uniformly over $(g_1,g_2) \in \mathcal{G}_c$.  \looseness=-1

\begin{remark}\label{rem:deterministic-bounds}
Assumption~\ref{ass:scale-consistency} is stronger than necessary for our intervals to satisfy the simultaneous inference condition (\ref{simul}), since if the estimators are larger than their respective estimands, the resulting intervals will be conservative. For instance, under Assumptions~\ref{ass:model} and \ref{ass:nondegeneracy}, the choice of $\hat{\sigma}(g_1,g_2) = 2B$, $\hat{V}  = \sqrt{8(N_1\sqrt{N_2}+N_2\sqrt{N_1}) + 4N_1N_2 + 2\sqrt{N_1N_2}}B$, and $\hat{\bar{\tau}} = 2.02(N_1\sqrt{N_2} + N_2\sqrt{N_1})B + 0.5\sqrt{N_1N_2}B$ result in intervals $CI_1$ and $CI_2$ that satisfy (\ref{simul}). See Appendix Section~\ref{sec:deterministic-scale-bounds} for a proof.   \looseness=-1
\end{remark}

\begin{remark}\label{rem:variance-estimation-strategies}
In Appendix~\ref{app:variance-estimation} below we provide three strategies for constructing $\hat{\sigma}(g_1,g_2)$, $\hat{V}$, and $\hat{\bar{\tau}}$. Section~\ref{sec:deterministic-scale-bounds} gives deterministic conservative bounds (see also Remark \ref{rem:deterministic-bounds} above). Section~\ref{sec:raw-plugin} considers estimates based on plugging the observed $Y$ in for the unknown $\epsilon$, and provides one set of conditions under which the resulting estimates are asymptotically conservative and another under which they are consistent. Section~\ref{sec:usvt} considers estimates based on singular value thresholding along the lines of \cite{chatterjee2015matrix}, and provides a set of conditions under which the resulting estimates are consistent. \looseness=-1
\end{remark}

\subsection{Results}\label{sec:results}
We state four propositions in this section. The first two propositions are in Section~\ref{sec:simultaneous-results}. Proposition~\ref{prop:propone} establishes that Assumptions~\ref{ass:model}, \ref{ass:nondegeneracy}(s), and \ref{ass:scale-consistency}(i) are sufficient for $CI_1$ to satisfy (\ref{simul}). Proposition~\ref{prop:proptwo} establishes that Assumptions~\ref{ass:model}, \ref{ass:nondegeneracy}(w), \ref{ass:scale-consistency}(ii) and \ref{ass:scale-consistency}(iii) are sufficient for $CI_2$ to satisfy (\ref{simul}). The second two propositions are in Section~\ref{sec:optimality-results}. They state that, under certain conditions, the widths of the intervals $CI_1$ and $CI_2$ satisfy certain optimality properties. Proofs are in Appendix Section~\ref{app:propositions}.  \looseness=-1

\subsubsection{Simultaneous inference results}\label{sec:simultaneous-results}
Under Assumptions~\ref{ass:model}, \ref{ass:nondegeneracy}(s), and \ref{ass:scale-consistency}(i), the confidence interval $CI_1$ satisfies the simultaneous inference condition (\ref{simul}). That is,  \looseness=-1
\begin{restatable}{proposition}{propone}\label{prop:propone}
Let $\mathcal F$ be a deterministic model class on which Assumptions~\ref{ass:model}, \ref{ass:nondegeneracy}(s), and \ref{ass:scale-consistency}(i) hold. Then for any $\alpha \in (0,1)$, 
\[\liminf_{F \in \mathcal{F}: N_1,N_2 \to \infty}\mathbbm{P}_F\left(\cap_{(g_1,g_2) \in \mathcal{G}_c}\{\theta(g_1,g_2) \in CI_1(g_1,g_2;\alpha)\}\right) \geq 1-\alpha.\]
\end{restatable}

Under Assumptions~\ref{ass:model}, \ref{ass:nondegeneracy}(w), \ref{ass:scale-consistency}(ii), and \ref{ass:scale-consistency}(iii), the confidence interval $CI_2$ satisfies the simultaneous inference condition (\ref{simul}). That is,  \looseness=-1
\begin{restatable}{proposition}{proptwo}\label{prop:proptwo}
Let $\mathcal F$ be a deterministic model class on which Assumptions~\ref{ass:model}, \ref{ass:nondegeneracy}(w), and \ref{ass:scale-consistency}(ii)--(iii) hold. Then for any $\alpha \in (0,1)$, 
\[\liminf_{F \in \mathcal{F}: N_1,N_2 \to \infty}\mathbbm{P}_F\left(\cap_{(g_1,g_2) \in \mathcal{G}_c}\{\theta(g_1,g_2) \in CI_2(g_1,g_2;\alpha)\}\right) \geq 1-\alpha.\]
\end{restatable}

A key difference between the two results is that Proposition~\ref{prop:propone} requires the strong version of Assumption~\ref{ass:nondegeneracy} while Proposition~\ref{prop:proptwo} only requires the weak version. \looseness=-1

The proof of Proposition~\ref{prop:propone} is in Appendix Section~\ref{app:propositions}. The main idea is to start with an interval that is marginally valid in the sense that it covers the parameter of interest at the desired level for a fixed group assignment. We use a version of Bernstein's inequality, see Lemma~\ref{lem:bernstein} in Appendix Section~\ref{app:concentration-lemmas}, rather than a normal approximation to form the interval because we are not aware of any formal justification of the latter in our setting.\footnote{ \label{fn:gaussian-approximation}For example, \citet{chernozhukov2022improved} study Gaussian and bootstrap approximations for the maximum coordinate of a centered sample mean of $n$ independent $p$-dimensional random vectors, obtaining bounds that depend polynomially on $\log p$. To apply their results to our setting, it is natural for a coordinate to be indexed by candidate groups pairs in $\mathcal G_c$. But in this setting $p$ is exponential in the number of nodes, and the bounds do not yield a vanishing approximation error. The authors' Remark 2.2 indicates that this logarithmic dimension dependence is sharp in general.} We then inflate the width of the interval until (\ref{simul}) is satisfied. While the resulting interval is conservative, we show in Proposition~\ref{prop:propthree} that, under an additional restriction on the class of random graph models $\mathcal{F}$, the width of the interval is optimal up to a constant within the class of intervals that satisfy (\ref{simul}) and have the conventional ``point estimate plus or minus constant times standard error'' shape that characterizes the \cite{berk2013valid} approach. \looseness=-1

The proof of Proposition~\ref{prop:proptwo} is also in Appendix Section~\ref{app:propositions}. The main idea is to bound the maximum estimation error $\max_{(g_1,g_2) \in \mathcal{G}_c}\left|\hat{\theta}(g_1,g_2) - \theta(g_1,g_2)\right|$ using a version of Talagrand's inequality due to \cite{klein2005concentration}, see Lemma~\ref{lem:klein-rio} in Appendix Section~\ref{app:concentration-lemmas}. This inequality bounds deviations of the maximum estimation error around its expectation. To bound the expectation, we adapt Lemma 3.1 of \cite{alon2006approximating} and combine it with a bound derived from Theorem 3 of \cite{gittens2009error}. These results are Lemmas~\ref{lem:cut-infty}, \ref{lem:restricted-cut}, and \ref{lem:dagger-infty} in Appendix Section~\ref{app:matrix-norm-lemmas}. While the resulting interval is conservative, we show in Proposition~\ref{prop:propfour} that, under the same assumptions as in Proposition~\ref{prop:proptwo}, the width of the interval is optimal up to a constant within the class of intervals that satisfy (\ref{simul}) and contain the point estimate $\hat{\theta}(g_1,g_2)$.  \looseness=-1

Let $\mathcal F$ be a deterministic model class on which Assumptions~\ref{ass:model}, \ref{ass:nondegeneracy}(s), and \ref{ass:scale-consistency}(i)--(iii) hold. A consequence of Propositions~\ref{prop:propone} and \ref{prop:proptwo} is that the combined interval $CI_{\cap}$ also satisfies the simultaneous inference condition (\ref{simul}). That is, 
\[\liminf_{F \in \mathcal{F}: N_1,N_2 \to \infty}\mathbbm{P}_F\left(\cap_{(g_1,g_2) \in \mathcal{G}_c}\{\theta(g_1,g_2) \in CI_{\cap}(g_1,g_2;\alpha)\}\right) \geq 1-\alpha.\]
Since the widths of $CI_1$ and $CI_2$ are relatively insensitive to $\alpha$, $CI_{\cap}$ is typically indistinguishable from the shorter of the two intervals in practice, and so we recommend using $CI_{\cap}$  when Assumption~\ref{ass:nondegeneracy}(s) is plausible, and $CI_2$ alone when Assumption~\ref{ass:nondegeneracy}(w) but not 2(s) is plausible. \looseness=-1

\subsubsection{Optimality results}\label{sec:optimality-results}
For model classes satisfying Assumptions~\ref{ass:model},
\ref{ass:nondegeneracy}(s), and
\ref{ass:scale-consistency}(i) and containing the i.i.d.\ sequence
specified below, our choice
$K_1(\alpha):=\sqrt{1.39(N_1+N_2)-2\ln(\alpha/2)}$
is optimal up to a constant factor. Specifically, \looseness=-1

\begin{restatable}{proposition}{propthree}\label{prop:propthree}
Let $\mathcal F$ be a deterministic model class on which Assumptions~\ref{ass:model},
\ref{ass:nondegeneracy}(s), and \ref{ass:scale-consistency}(i) hold. Suppose further that there exist a fixed nondegenerate distribution $F_0$ supported on $[-B,B]$ and an asymptotic sequence
$\{F_\ell\}_{\ell\geq1}\subseteq\mathcal F$ such that every entry of
every $F_\ell$ has marginal distribution $F_0$. Fix $\alpha \in (0,1)$ and let $\{I(g_1,g_2;\alpha)\}_{(g_1,g_2)\in\mathcal{G}_c}$ be an arbitrary collection of
confidence intervals of the form $I(g_1,g_2;\alpha) = \hat{\theta}(g_1,g_2) \pm \left[K \times \hat{\sigma}(g_1,g_2)\right]/\sqrt{m_1m_2}$, where $K = K(N_1,N_2,\alpha) \geq 0$ is an arbitrary deterministic sequence with $\limsup_{N_1,N_2\to\infty} \frac{K}{\sqrt{N_1+N_2}} < \frac{1}{\sqrt{8\pi}}$. Then\looseness=-1
\begin{align*}
\liminf_{F\in\mathcal{F}:\,N_1,N_2\to\infty}
\mathbbm{P}_F\left(\cap_{(g_1,g_2)\in\mathcal{G}_c}
\left\{\theta(g_1,g_2)\in I(g_1,g_2;\alpha)\right\}\right) \;< \; 1-\alpha.
\end{align*}
\end{restatable}

The proof of Proposition~\ref{prop:propthree} is in Appendix Section~\ref{app:propositions}. The main idea is to bound the PoSI constant of Remark~\ref{rem:posi-constant} from below using a subclass of homogeneous models in which every entry of every model is drawn from a fixed nondegenerate distribution $F_0$. For this subclass, choosing the groups according to the signs of the row and column sums of $\epsilon
$, and applying a central limit theorem to the positive parts of the standardized sums, shows that the maximum studentized estimation error eventually exceeds $(1-s)\sqrt{N_1+N_2}/\sqrt{8\pi}$ with probability approaching one, for any fixed $s\in(0,1)$. It follows that, for any $\mathcal{F}$ containing such a sequence, an interval of the form in Proposition~\ref{prop:propthree} with $\limsup K/\sqrt{N_1+N_2} < 1/\sqrt{8\pi}$ has a vanishing simultaneous coverage probability. \looseness=-1


A corollary of Proposition~\ref{prop:propthree} is that any interval that has the \cite{berk2013valid} ``point estimate plus or minus critical value times standard error'' shape will not satisfy the simultaneous inference condition (\ref{simul}) for the class of random graph models described in that proposition, if the critical value diverges slower than $\sqrt{N_1+N_2}$. In particular, this applies to the conventional confidence interval $CI_0(g_1,g_2;\alpha) := \hat{\theta}(g_1,g_2) \pm \left[K_0(\alpha)\times \hat{\sigma}(g_1,g_2)\right]/\sqrt{m_1m_2}$ where $K_0(\alpha)$ is the $1-\alpha/2$ quantile (or any other fixed quantile) of a standard normal distribution, i.e.\looseness=-1
\begin{corollary}\label{cor:ci0-invalid}
Let $\mathcal{F}$ be the same as in Proposition~\ref{prop:propthree}. Then for every $\alpha  \in (0,1)$
\[\liminf_{F \in \mathcal{F}: N_1,N_2 \to \infty}\mathbbm{P}_F\left(\cap_{(g_1,g_2) \in \mathcal{G}_c}\{\theta(g_1,g_2) \in CI_0(g_1,g_2;\alpha)\}\right) < 1-\alpha.\] 
\end{corollary}

Proposition~\ref{prop:propthree} establishes that the width of the interval $CI_1$ is optimal up to a constant  within a restricted class of intervals that have the conventional shape. By contrast, under Assumptions~\ref{ass:model}, \ref{ass:nondegeneracy}(w), and \ref{ass:scale-consistency}(ii), the width of the interval $CI_2$ is optimal up to a constant factor (the constant may depend on $\alpha$) only under the weaker restriction that the interval contains the point estimate. That is,  \looseness=-1
\begin{restatable}{proposition}{propfour}\label{prop:propfour}
Let $\mathcal F$ be a deterministic model class on which Assumptions~\ref{ass:model},
\ref{ass:nondegeneracy}(w), and \ref{ass:scale-consistency}(ii) hold. Fix $\alpha \in
(0,1)$ and define
\begin{align*}
c^*(\alpha) :=
\frac{\left(2\ln(1+\sqrt{2})\right)^2}
{9\pi\left(1.01\pi\sqrt{648}\sqrt{-\ln(1-\alpha)} + 4.54\ln(1+\sqrt{2})\right)}.
\end{align*}
Let $\{I(g_1,g_2;\alpha)\}_{(g_1,g_2)\in\mathcal{G}_c}$ be an arbitrary collection of
confidence intervals indexed by $\mathcal{G}_c$ such that
$\hat{\theta}(g_1,g_2) \in I(g_1,g_2;\alpha)$ for every $(g_1,g_2) \in
\mathcal{G}_c$. Suppose there exist $\delta'\in(0,1)$ and $n_0\in\mathbbm N$
such that, for every $F\in\mathcal F$ with
$\min\{N_1(F),N_2(F)\}\geq n_0$,
\[
\max_{(g_1,g_2)\in\mathcal G_c}
m_1m_2|I(g_1,g_2;\alpha)|
\leq
(1-\delta')c^*(\alpha)\widehat{\bar\tau}
\]
holds $\mathbbm P_F$-almost surely. Then
\begin{align*}
\limsup_{F\in\mathcal{F}:\,N_1,N_2\to\infty}
\mathbbm{P}_F\left(\cap_{(g_1,g_2)\in\mathcal{G}_c}
\left\{\theta(g_1,g_2)\in I(g_1,g_2;\alpha)\right\}\right)
\;<\; 1-\alpha.
\end{align*}
\end{restatable}

Two key differences distinguish the content of Proposition~\ref{prop:propfour} from that of Proposition~\ref{prop:propthree}. The first difference is that Proposition~\ref{prop:propfour} applies to any interval that contains $\hat{\theta}(g_1,g_2)$, not just intervals of the form ``point estimate plus or minus a critical value times a standard error.'' The second difference is that the coverage probability of the interval is eventually below $1-\alpha$, not just for a subcollection of models, but for any asymptotic sequence in $\mathcal{F}$ (i.e. the left hand side has a ``$\limsup$'' instead of a ``$\liminf$''). This  conclusion is stronger than in Proposition~\ref{prop:propthree} because it bounds the asymptotic coverage probability away from $1-\alpha$ uniformly over the model class, rather than merely constructing one unfavorable subclass. We say that $CI_2$ is uniformly ``sup-norm rate optimal'' since the width of the interval is measured using the sup-norm.\looseness=-1

Since network economists typically characterize network structure by comparing estimated group densities in practice, we believe that the requirement in Proposition~\ref{prop:propfour} that each interval contain the corresponding estimated density is reasonable. It ensures that interval-based evidence for an ordering, separation, or deviation from a benchmark cannot contradict the corresponding empirical density comparison. In this sense, the confidence intervals refine, rather than replace, the descriptive comparison of estimated
densities.The requirement also cannot be dropped from Proposition~\ref{prop:propfour}.\footnote{For example, fix $\theta_0\in(0,1)$ and consider a sequence of bipartite networks with independent Bernoulli$(\theta_0)$ entries. The collection $I(g_1,g_2;\alpha)=\{\theta_0\}$ has zero width and coverage one
along this sequence, although it fails to cover under models with different densities and therefore does not satisfy \eqref{simul}.} We conjecture that the requirement can be dropped if the $\limsup$ in the conclusion is replaced with a $\liminf$ as in Proposition~\ref{prop:propthree}, but leave this to future work.\looseness=-1

The proof of Proposition~\ref{prop:propfour} is in Appendix Section~\ref{app:propositions} below. The main idea of the proof builds on an argument described in the third paragraph of Section 4.2 of \cite{rudelson2007sampling}. In their paper, the authors focus on the special case where the entries of $Y$ take value $1$ or $-1$ with equal probability and, for that model, derive a lower bound on a quantity that is similar in spirit to our maximum estimation error $\max_{(g_1,g_2) \in \mathcal{G}_c}|\hat{\theta}(g_1,g_2) - \theta(g_1,g_2)|$ using  Khintchine's inequality. Our proof strategy is similar, but we do not require the entries of $Y$ have this specific distribution. To accomplish this, we use a version of Grothendieck's inequality due to \cite{krivine1979constantes} as an alternative to Khintchine's inequality. The results is the lower bounds of our Lemma~\ref{lem:dagger-infty} of Appendix Section~\ref{app:matrix-norm-lemmas}. \looseness=-1

The constant $c^*(\alpha)$ is conservative. It is produced by the lower-tail cut-norm bound, Krivine's upper bound on the Grothendieck constant, and the two-regime argument in the proof. At $\alpha=0.05$, $c^*(0.05)\approx0.00493$. The rate conclusion is the main content of Proposition~\ref{prop:propfour}. We do not claim that this width condition is sharp.\looseness=-1

Since Proposition~\ref{prop:propfour} establishes that, for any sequence of random graph models satisfying Assumptions~\ref{ass:model}, \ref{ass:nondegeneracy}(w), \ref{ass:scale-consistency}(ii), and \ref{ass:scale-consistency}(iii), the width of $CI_2$ is rate-optimal, it follows that, under these conditions, $CI_1$ is rate sub-optimal for any sequence of models satisfying these assumptions, Assumption~\ref{ass:scale-consistency}(i), and $\frac{\sqrt{N_1+N_2}||\sigma||_F}{\mathbbm{E}_F\left[\sum_{i \in [N_1]}\sqrt{\sum_{j \in [N_2]}\epsilon_{ij}^2} + \sum_{j \in [N_2]}\sqrt{\sum_{i \in [N_1]}\epsilon_{ij}^2}\right]}  \to \infty$ as $N_1,N_2 \to \infty$.\footnote{See Remark~\ref{rem:ci1-width} of Section~\ref{sec:first-interval} and Remark~\ref{rem:ci2-width} of Section~\ref{sec:second-interval} above for the widths of these intervals.}\looseness=-1

To understand this last condition, write\ $a_i:=\left(\sum_{j\in[N_2]}\epsilon_{ij}^2\right)^{1/2}$ and  $b_j:=\left(\sum_{i\in[N_1]}\epsilon_{ij}^2\right)^{1/2}$ so that $\|\sigma\|_F^2 = \sum_{i \in [N_1]} \mathbb E_F[a_i^2] = \sum_{j \in [N_2]} \mathbb E_F[b_j^2]$, 
 \begin{align*}
 \mathbbm{E}_F\left[\sum_{i \in [N_1]}\sqrt{\sum_{j \in [N_2]}\epsilon_{ij}^2} + \sum_{j \in [N_2]}\sqrt{\sum_{i \in [N_1]}\epsilon_{ij}^2}\right] = \sum_{i \in [N_1]}\mathbb E_F[a_i]+\sum_{j \in [N_2]}\mathbb E_F[b_j],
 \end{align*}  
and the condition $\frac{\sqrt{N_1+N_2}||\sigma||_F}{\mathbbm{E}_F\left[\sum_{i \in [N_1]}\sqrt{\sum_{j \in [N_2]}\epsilon_{ij}^2} + \sum_{j \in [N_2]}\sqrt{\sum_{i \in [N_1]}\epsilon_{ij}^2}\right]} \to \infty$ is equivalent to 
\[
\frac{\sum_{i \in [N_1]}\mathbb E_F[a_i]+\sum_{j \in [N_2]}\mathbb E_F[b_j]}{\sqrt{\sum_{i \in [N_1]} \mathbb E_F[a_i^2]} + \sqrt{\sum_{j \in [N_2]} \mathbb E_F[b_j^2]}}
=
o\!\left(\sqrt{N_1+N_2}\right).
\]
Suppose the nondegenerate entries of $a_i$ and $b_j$ are all supported in $[\underline{s}, \overline{s}]$. Let $R_1 := \sum_{i \in [N_1]}\mathbbm{1}\{\mathbbm{E}_F\left[a_i^2\right] > 0\}$ and $R_2 := \sum_{j \in [N_2]}\mathbbm{1}\{\mathbbm{E}_F\left[b_j^2\right] > 0\}$. Then $\frac{\sum_{i \in [N_1]}\mathbb E_F[a_i]+\sum_{j \in [N_2]}\mathbb E_F[b_j]}{\sqrt{\sum_{i \in [N_1]} \mathbb E_F[a_i^2]} + \sqrt{\sum_{j \in [N_2]} \mathbb E_F[b_j^2]}} \leq \left(\sqrt{R_1}+\sqrt{R_2}\right)\frac{\overline{s}}{\underline{s}}$ which is $o(\sqrt{N_1+N_2})$ if $\overline{s}/\underline{s}$ is uniformly bounded and  $\left(\sqrt{R_1}+\sqrt{R_2}\right) = o\left(\sqrt{N_1} + \sqrt{N_2}\right)$.\looseness=-1 

In words, $CI_2$ improves on $CI_1$ when the variation in network connections is concentrated in a relatively small number of rows and columns of the adjacency matrix. Intuitively, the first interval $CI_1$ satisfies (\ref{simul}) by scaling the width of the conventional interval by a factor of $\sqrt{N_1+N_2}$. We demonstrated in Proposition~\ref{prop:propthree} that this rate is generally necessary for simultaneous coverage, however, it can be very conservative when many of the rows and columns of the adjacency matrix exhibit relatively little variation in the network connections. The second interval $CI_2$ is based on a direct bound on the maximum post-selection estimation error and therefore adapts to this heterogeneity. This is why the two intervals have  widths that are of a similar order of magnitude in homogeneous networks, but $CI_1$ can be substantially wider in sparse or degree-heterogeneous networks.\looseness=-1

\section{Empirical illustrations}\label{sec:empirical}
We demonstrate our methodology in the context of our three motivating examples. \looseness=-1

\subsection{Social network}\label{sec:facebook100}
We apply our intervals to the collegiate Facebook friendship networks introduced by \citet{traud2012social} and distributed through \cite{nr}. We retain every Facebook100 campus network containing between 500 and 10,000 Facebook users, which yields 50 campuses. For each campus, we observe a collection of attributes (e.g. major), each with several categories (e.g. economics), although the categories are not labeled. For each campus, attribute, and eligible category, let $g_1$ denote the group of users in that campus and category, and $g_2$ contain the users with a nonmissing value of the attribute who are not in that category. We use the density difference  $\Delta(g_1,g_2) =\theta(g_1,g_1)-\theta(g_1,g_2)$. A positive difference means that members of the category are more densely connected to one another, which we interpret as evidence of homophily. We restrict attention to ``eligible'' differences where $g_1$ and $g_2$ contain at least 50 users and are at least 5 percent of the full campus. We estimate the parameters $\sigma(g_1,g_2)$, $\bar{\tau}$ and $V$ using the singular value thresholding algorithm described in Appendix Section ~\ref{sec:usvt}. \looseness=-1

Table~\ref{tab:facebook100-type-rest-summary} counts every eligible difference for five attributes: gender, class year, student/faculty status, major, and residence or dormitory. The confidence intervals in the table are the difference intervals described in Section~\ref{sec:universal-validity} above. Appendix Sections~\ref{app:undirected-normalization} and \ref{app:facebook100-details} describe the density normalization, sample construction, and confidence intervals. We construct each campus-level simultaneous collection using $\alpha_{\mathrm{FB}}=0.05/50$. This Bonferroni correction is meant to control familywise error across the 50 campuses. 
\looseness=-1

\begin{table}[!htbp]
\centering
\scriptsize
\setlength{\tabcolsep}{3pt}
\renewcommand{\arraystretch}{1.10}
\caption{Facebook100 evidence of a homophily structure}
\label{tab:facebook100-type-rest-summary}
\resizebox{\textwidth}{!}{%
\begin{tabular}{lrrrrrrrr}
\toprule
Attribute
& Campuses
& Differences
& Median 
& Est. $+/-$
& $CI_{0,\Delta}$ $+/-$
& $CI_{1,\Delta}$ $+/-$
& $CI_{2,\Delta}$ $+/-$
& $CI_{\cap,\Delta}$ $+/-$ \\
\midrule
Gender
& 47 & 94 & 0.0016
& 75/19 & 68/17 & 0/0 & 0/0 & 0/0 \\

Class year
& 50 & 277 & 0.0481
& 277/0 & 277/0 & 74/0 & 0/0 & 73/0 \\

Student/faculty status
& 50 & 109 & 0.0174
& 109/0 & 109/0 & 28/0 & 0/0 & 28/0 \\

Major
& 50 & 208 & 0.0140
& 207/1 & 207/0 & 0/0 & 0/0 & 0/0 \\

Residence/dormitory
& 41 & 151 & 0.0471
& 151/0 & 151/0 & 0/0 & 0/0 & 0/0 \\
\bottomrule
\end{tabular}%
}
\begin{minipage}{0.98\textwidth}
\footnotesize
\emph{Notes:} For a given category, a difference is the friendship density within the category minus the density between members of the category and all other users with a nonmissing value of the same attribute. ``Campuses'' reports the number of campuses with at least one eligible (see text) difference. ``Differences'' reports the total number of eligible differences. ``Median'' is the median value of the differences. ``Est. $+/-$'' reports the numbers of positive and negative differences across the 50 campuses. ``$CI_{0,\Delta}$ $+/-$,'' ``$CI_{1,\Delta}$ $+/-$,'' ``$CI_{2,\Delta}$ $+/-$,'' and ``$CI_{\cap,\Delta}$ $+/-$'' report the numbers of difference intervals (see Section \ref{sec:difference-intervals} above) that lie completely above and completely below $0$.
\end{minipage}
\end{table}

The point estimates indicate within-attribute homophily by class year, student/faculty status, major, and residence. Every eligible class-year, student/faculty, and residence difference is positive, as are 207 of the 208 major differences. The median difference in densities is 0.0481 for class year, 0.0471 for residence, 0.0174 for student/faculty status, and 0.0140 for major. Gender gaps are much smaller and less uniform. The median is 0.0016, with 75 positive and 19 negative estimates. The benchmark $CI_{0,\Delta}$ corroborates most of the point estimate comparisons, as it  excludes 0 for every class-year, student/faculty, and residence difference, as well as 207 of the 208 major differences, and for 68 of the 94 gender differences.
\looseness=-1

Our simultaneous inference correction materially changes these conclusions. We say that a difference \emph{survives our selection correction} when both the benchmark $CI_{0,\Delta}$ and our proposed $CI_{\cap,\Delta}$ exclude zero in the same direction. Under this definition, 73 class-year differences survive in 35 campuses, and 28 student/faculty-status differences survive in 26 campuses. No gender, major, or residence difference survives. No difference is established by $CI_{2,\Delta}$ alone. \looseness=-1

\subsection{Trade network}
\label{sec:baci}
Motivated by the analysis of \citet{carvalho2014micro}, we apply
our intervals to the 2023 BACI HS22 trade data \citep{gaulier2010baci}. We
construct a bipartite network with supplier countries as rows and
destination--product markets as columns, where a destination--product
market is a pair consisting of an importing country and an HS6 product. Supplier country $i$ is connected to market $jk$ when it exports more than \$50 million of product $k$ to destination
$j$. The resulting network contains 226 supplier countries and 1,266,956
destination--product markets. \looseness=-1

We rank suppliers using binary degree, weighted outdegree, and bipartite Katz
centrality (see Appendix Section~\ref{app:baci-details} for the rankings and
density normalization). For each ranking, we define a high-centrality supplier group $g_h$ containing the top 5, 10, or 20 percent of countries and let $g_h^c$ contain the
remaining suppliers. For a market group $g_m$, we estimate the difference $\Delta(g_h,g_h^c,g_m) = \theta(g_h,g_m)-\theta(g_h^c,g_m)$. A positive difference means that the selected high-centrality suppliers are
more densely connected to the markets in $g_m$ than the remaining suppliers, which we interpret as evidence of a hub-and-spoke structure. We estimate the parameters $\sigma(g_1,g_2)$, $\bar{\tau}$, and $V$ using the raw second-moment plug-in estimator described in Appendix Section~\ref{sec:raw-plugin}.  \looseness=-1

Table~\ref{tab:baci-contrast-family-summary} counts every difference belonging to three market-group families. The confidence intervals are the induced difference intervals described in Section~\ref{sec:universal-validity} above. Because each interval collection simultaneously covers every admissible
supplier and market group, no additional correction is needed for the number
or choice of centrality rankings, supplier tiers, or market groups reported
within a network. \looseness=-1

\begin{table}[!htbp]
\centering
\scriptsize
\setlength{\tabcolsep}{4pt}
\renewcommand{\arraystretch}{1.10}
\caption{BACI evidence of a hub-and-spoke structure}
\label{tab:baci-contrast-family-summary}
\label{tab:baci-core-periphery}
\resizebox{\textwidth}{!}{%
\begin{tabular}{lrrrrrrr}
\toprule
Market group family
& Differences
& Median 
& Est. $+/-$
& $CI_{0,\Delta}$ $+/-$
& $CI_{1,\Delta}$ $+/-$
& $CI_{2,\Delta}$ $+/-$
& $CI_{\cap,\Delta}$ $+/-$ \\
\midrule
All destination--product markets
& 9
& 0.00134
& 9/0
& 9/0
& 0/0
& 3/0
& 3/0 \\

Broad HS product baskets
& 63
& 0.00108
& 63/0
& 63/0
& 0/0
& 0/0
& 0/0 \\

Large or central market groups
& 54
& 0.00390
& 54/0
& 54/0
& 0/0
& 9/0
& 9/0 \\
\bottomrule
\end{tabular}%
}
\begin{minipage}{0.98\textwidth}
\footnotesize
\emph{Notes:} For a given supplier ranking, tier, and market group, a difference is the link density of the high-centrality supplier group to the market minus the corresponding density of the remaining suppliers. ``Differences'' reports the number of centrality measure--threshold--market triplets for a given market group family. ``Median'' is the median value of the differences. ``Est. $+/-$'' reports the numbers of positive and negative differences. ``$CI_{0,\Delta}$ $+/-$,'' ``$CI_{1,\Delta}$ $+/-$,'' ``$CI_{2,\Delta}$ $+/-$,'' and ``$CI_{\cap,\Delta}$ $+/-$'' report the numbers of difference intervals (see Section \ref{sec:difference-intervals} above) that lie completely above and completely below $0$. \looseness=-1
\end{minipage}
\end{table}

The point estimates and benchmark $CI_{0,\Delta}$ indicate a hub-and-spoke structure for all 126 differences across the 3 market group families, with the median difference in densities ranging from $0.00108$ to $0.00390$. Our simultaneous inference correction supports these findings for three differences in the all destination--product markets family associated with a threshold of 20 percent, nine differences in the large or central market groups family, and none for the Broad HS product baskets family. No difference is established by $CI_{1,\Delta}$ alone. 
\looseness=-1

\subsection{Worker-flow network}
\label{sec:psid}
Following \citet{schmutte2014free}, we use PSID worker histories from 1987--1993 to construct a directed pseudo-employer mobility network. A pseudo-employer is the interaction of a three-digit industry code
and a three-digit occupation code. For each worker, we order the person-year observations and record a directed transition from pseudo-employer $i$ to pseudo-employer $j$ when the worker is observed being employed in $i$ at one time period and in $j$ at the next time period for which industry or occupation information is observed. We use $M_{ij}$ to denote the number of observed transitions from $i$ to $j$ in the time frame. These transition counts are used to select the Louvain communities as described in Remark~\ref{rem:external-auxiliary-data}.  \looseness=-1

To select mobility segments, we form the undirected weighted network with edge weight $W_{ij} = M_{ij} + M_{ji}$ and apply a weighted Louvain community detection algorithm to its largest connected component, retaining the four largest communities. Our density measures concern the directed binary mobility network $Y_{ij} = \mathbbm{1}\{M_{ij} > 0\}$. The resulting analysis universe contains 6,039 pseudo-employers and 16,123 distinct directed links.\footnote{The empirical analysis treats this largest connected component as its analysis universe for defining $\mathcal{G}_c$. Our simultaneous theory protects against subsequent selection of groups within this universe, but does not separately account for the data-dependent component-selection step. Correcting for the selection of the largest connected component would lead to even wider intervals and so not change any of the content of Table \ref{tab:psid-mobility-segments}. \looseness=-1} For a segment $g$ and its complement $g^c$, we estimate the density difference $\Delta(g,g^c) := \theta(g,g) - \frac{\theta(g,g^c)+\theta(g^c,g)}{2}$. A positive difference means that the pseudo-employers in $g$ are more densely connected to one another than to pseudo-employers outside the segment, averaging the outward and inward directions. We compare segments selected by the weighted Louvain community detection algorithm to segments defined by major industry, major occupation, and their interaction. Appendix Section~\ref{app:psid-details} gives the complete construction and density normalization. We estimate the parameters $\sigma(g_1,g_2)$, $\bar{\tau}$, and $V$ using the raw second-moment plug-in estimator described in Appendix Section~\ref{sec:raw-plugin}. 
\looseness=-1

\begin{table}[!htbp]
\centering
\scriptsize
\setlength{\tabcolsep}{3pt}
\renewcommand{\arraystretch}{1.10}
\caption{PSID evidence of a segmented market structure}
\label{tab:psid-market-structure-summary}
\label{tab:psid-mobility-segments}
\resizebox{\textwidth}{!}{%
\begin{tabular}{lrrrrrrr}
\toprule
Segment definition
& Differences
& Median
& Est. $+/-$
& $CI_{0,\Delta}$ $+/-$
& $CI_{1,\Delta}$ $+/-$
& $CI_{2,\Delta}$ $+/-$
& $CI_{\cap,\Delta}$ $+/-$ \\
\midrule
Weighted Louvain segments
& 4 & 0.00499 & 4/0 & 4/0 & 0/0 & 0/0 & 0/0 \\

Major industry segments
& 6 & 0.00209 & 6/0 & 6/0 & 0/0 & 0/0 & 0/0 \\

Major occupation segments
& 7 & 0.00138 & 7/0 & 7/0 & 0/0 & 0/0 & 0/0 \\

Major industry $\times$ occupation segments
& 4 & 0.00195 & 4/0 & 4/0 & 0/0 & 0/0 & 0/0 \\
\bottomrule
\end{tabular}%
}
\begin{minipage}{0.98\textwidth}
\footnotesize
\emph{Notes:} For a given market segment, a difference is the empirical density of connections within the segment minus half the empirical density of connections between the segment and its complement. ``Differences'' reports the number of differences for a given segmentation of the labor market. ``Median'' is the median value of the differences. ``Est. $+/-$ reports the number of positive and negative differences. ``$CI_{0,\Delta}$ $+/-$,'' ``$CI_{1,\Delta}$ $+/-$,'' ``$CI_{2,\Delta}$ $+/-$,'' and ``$CI_{\cap,\Delta}$ $+/-$'' report the numbers of difference intervals (see Section \ref{sec:difference-intervals} above) that lie completely above and completely below $0$.\looseness=-1
\end{minipage}
\end{table}

The point estimates and benchmark $CI_{0,\Delta}$ indicate a segmented market structure for all 21 differences across 4 segment definitions, with a median gap of $0.00499$ for the segments selected by the Louvain algorithm and a median gap of $0.00138$--$0.00209$ for segments based on major industry and occupation codes. The simultaneous correction materially changes the conclusion. Specifically, none of the $CI_{1,\Delta}$, $CI_{2,\Delta}$, or $CI_{\cap,\Delta}$ intervals exclude zero. Our simultaneous intervals are wide in this application because the selected communities contain only about 5-10 percent of pseudo-employers. Over our model class, even in the absence of a genuine segmented structure, some data-selected groups of this size can exhibit fluctuations in empirical densities on the order of $0.001-0.005$ purely through idiosyncratic variation. As a result, our simultaneous difference intervals all cross zero.  
\looseness=-1

\section{Conclusion}\label{sec:conclusion}
This paper considers inference for the density of network connections between groups of agents, such as communities, blocs, or markets. Such density measures are widely used to characterize stochastic network structure, but in practice the relevant groups are often selected using the network itself. We develop two confidence intervals that are universally valid post-selection, in the sense that they guarantee simultaneous coverage over all group pairs whose relative sizes are bounded away from zero. The first interval inflates the critical value of a conventional fixed-group interval. The second uses a Talagrand-type concentration inequality to control the maximum post-selection estimation error directly. Both intervals are simple to compute and scalable to large networks. The main theoretical distinction is that the Talagrand interval attains the optimal sup-norm rate, up to constants, while the inflated conventional interval can be substantially wider in sparse or degree-heterogeneous networks. Three empirical illustrations show that post-selection correction can materially alter conclusions in practice.\looseness=-1

\bibliographystyle{aer}
\bibliography{literature}

\clearpage
\appendix

\begin{spacing}{1.5}
\section{Proof of claims}\label{app:proofs}
\subsection{Terminology and notation}\label{app:proof-notation}
For any vector $X$ with real-valued entries of length $N$ and positive integer $p$, we use $X_i$ to refer to the $i$th entry of $X$, $||X||_{p} := \left(\sum_{i \in [N]}|X_i|^p\right)^{1/p}$ to refer to the $p$-norm of $X$ and $||X||_{\infty} := \max_{i \in [N]}|X_i|$ to refer to the maximum norm. 

For any $N_1 \times N_2$ matrix $X$ with real-valued entries, we use $X_{ij}$ to refer to the $ij$th entry of $X$ and $X^{T}$ to refer to the transpose of $X$. We also use the following notation
\begin{align*}
||X||_{F} &:= \sqrt{\sum_{i\in[N_1], j \in[N_2]}X_{ij}^2} \\
||X||_\square &:= \max_{\phi \in \{0,1\}^{N_1}, \psi \in \{0,1\}^{N_2}}\left|\sum_{i \in [N_1],j \in [N_2]}X_{ij}\phi_i\psi_j\right|\\
||X||_{\square;c} &:= \max_{\phi \in \{0,1\}^{N_1}, \psi \in \{0,1\}^{N_2}: (\phi,\psi) \in \mathcal{G}_c}\left|\sum_{i \in [N_1],j \in [N_2]}X_{ij}\phi_i\psi_j\right|\\
 ||X||_{\infty \to 1} &:= \max_{\tilde{\phi} \in \{-1,1\}^{N_1}, \tilde{\psi} \in \{-1,1\}^{N_2}}\left|\sum_{i \in [N_1],j \in [N_2]}X_{ij}\tilde{\phi}_i\tilde{\psi}_j\right| \\
 ||X||_{1,\star} &:= \sum_{i \in [N_1]}\left|\sum_{j \in [N_2]}X_{ij}\right|  \\
 ||X||_{1,2} &:= \sum_{i \in [N_1]}\sqrt{\sum_{j \in [N_2]}X_{ij}^2}  \\
 ||X||_{\spadesuit} &:=  ||X||_{1,\star} + ||X^{T}||_{1,\star}  \\
 ||X||_{\dagger} &:= ||X||_{1,2} + ||X^{T}||_{1,2} 
\end{align*}
where for any $c \in [0,1/2]$
\begin{align*}
\mathcal{G}_c := \{(\phi,\psi) \in \{0,1\}^{N_1} \times \{0,1\}^{N_2}: \sum_{i\in[N_1]}\phi_{i} \geq cN_1 \text{ and } \sum_{j\in[N_2]}\psi_{j} \geq cN_2\}. 
\end{align*}
It follows that $||X||_{\square} = ||X||_{\square;0}$ and for any $c', c'' \in [0,1/2]$ with $c' < c''$ we have $\mathcal{G}_{c''} \subset \mathcal{G}_{c'}$ and so $||X||_{\square;c''} \leq ||X||_{\square;c'} \leq ||X||_{\square}$. When $c > 0$, we refer to $||X||_{\square;c}$ as the restricted cut norm. The case $c=0$ is the conventional (unnormalized) cut norm.

\subsection{Lemmas}\label{app:lemmas}
The following lemmas are used to demonstrate our main results. We sort the lemmas into three sections. Appendix Section~\ref{app:matrix-norm-lemmas} contains results that relate the matrix norms listed in Appendix Section~\ref{app:proof-notation}. Appendix Section~\ref{app:concentration-lemmas} contains finite sample concentration inequalities. Appendix Section~\ref{app:additional-lemmas} contains two additional results. 

\subsubsection{Lemmas relating matrix norms}\label{app:matrix-norm-lemmas}
The following  Lemma~\ref{lem:cut-infty} relates the $||\cdot||_{\infty\to1}$ and $||\cdot||_{\square}$ matrix norms. It refines  Lemma 3.1 of \cite{alon2006approximating}.
\begin{lemma}\label{lem:cut-infty} For any $N_1 \times N_2$ matrix $X$ with real-valued entries, the norms $||X||_\square$ and $||X||_{\infty \to 1}$ are related by 
\begin{align*}
||X||_{\infty\to1} \leq 4||X||_{\square} \leq ||X||_{\infty\to1} + ||X||_{\spadesuit} + \left|\sum_{i \in [N_1], j \in [N_2]}X_{ij}\right| \leq 4||X||_{\infty\to1}. 
\end{align*}
Furthermore, if $\sum_{i \in [N_1]}X_{ij} =  0$ for every $j \in [N_2]$ and $\sum_{j \in [N_2]}X_{ij} = 0$ for every $i \in [N_1]$ then $||X||_{\infty\to1} = 4||X||_{\square}$. 
\end{lemma}

\begin{proof}
We start with the second inequality that $4||X||_{\square} \leq ||X||_{\infty\to1} + ||X||_{\spadesuit} + \left|\sum_{i \in [N_1], j \in [N_2]}X_{ij}\right|$. Let $\phi \in \{0,1\}^{N_1}$ and $\psi \in \{0,1\}^{N_2}$ be arbitrary and set $\tilde{\phi} = 2\phi - 1$ and $\tilde{\psi} = 2\psi -1$ so that  
\begin{align*}
\sum_{i \in [N_1], j \in [N_2]}X_{ij}\tilde{\phi}_i\tilde{\psi}_j &= 4\sum_{i \in [N_1], j \in [N_2]}X_{ij}\phi_i\psi_j - 2\sum_{i \in [N_1], j \in [N_2]}X_{ij}\phi_i -  2\sum_{i \in [N_1], j \in [N_2]}X_{ij}\psi_j +  \sum_{i \in [N_1], j \in [N_2]}X_{ij} \\
 &= 4\sum_{i \in [N_1], j \in [N_2]}X_{ij}\phi_i\psi_j  - \sum_{i \in [N_1], j \in [N_2]}X_{ij}\tilde{\phi}_i -  \sum_{i \in [N_1], j \in [N_2]}X_{ij}\tilde{\psi}_j -   \sum_{i \in [N_1], j \in [N_2]}X_{ij}.
\end{align*}
Rearranging terms to put $4\sum_{i \in [N_1], j \in [N_2]}X_{ij}\phi_i\psi_j$ on the left-hand side gives
\begin{align}\label{quart}
4\sum_{i \in [N_1], j \in [N_2]}X_{ij}\phi_i\psi_j= \sum_{i \in [N_1], j \in [N_2]}X_{ij}\tilde{\phi}_i\tilde{\psi}_j + \sum_{i \in [N_1], j \in [N_2]}X_{ij}\tilde{\phi}_i +  \sum_{i \in [N_1], j \in [N_2]}X_{ij}\tilde{\psi}_j +   \sum_{i \in [N_1], j \in [N_2]}X_{ij}
\end{align}
and taking absolute values gives 
\begin{align*}
4\left|\sum_{i \in [N_1], j \in [N_2]}X_{ij}\phi_i\psi_j\right| \leq \left|\sum_{i \in [N_1], j \in [N_2]}X_{ij}\tilde{\phi}_i\tilde{\psi}_j\right| + \left|\sum_{i \in [N_1], j \in [N_2]}X_{ij}\tilde{\phi}_i\right| +  \left|\sum_{i \in [N_1], j \in [N_2]}X_{ij}\tilde{\psi}_j\right| +   \left|\sum_{i \in [N_1], j \in [N_2]}X_{ij}\right|.
\end{align*}

The first summand on the right-hand side is bounded by $||X||_{\infty\to1}$. The second summand is bounded by $||X||_{1,\star}$. The third summand is bounded by $||X^T||_{1,\star}$. It follows that $\left|4\sum_{i \in [N_1], j \in [N_2]}X_{ij}\phi_i\psi_j\right| \leq ||X||_{\infty\to1} + ||X||_{\spadesuit} +  \left|\sum_{i \in [N_1], j \in [N_2]}X_{ij}\right|$. Since the choice of $\phi$ and $\psi$ was arbitrary, it follows that $4||X||_{\square} \leq ||X||_{\infty\to1} + ||X||_{\spadesuit} +   \left|\sum_{i \in [N_1], j \in [N_2]}X_{ij}\right|$. \newline

Next we show the third inequality that $||X||_{\infty\to1} + ||X||_{\spadesuit} + \left|\sum_{i \in [N_1], j \in [N_2]}X_{ij}\right| \leq 4||X||_{\infty\to1}$. This follows directly from the fact that 
\begin{align*}
||X||_{\spadesuit} \leq 2\max_{\tilde{\phi} \in \{-1,1\}^{N_1}, \tilde{\psi} \in \{-1,1\}^{N_2}}\left|\sum_{i \in [N_1],j \in [N_2]}X_{ij}\tilde{\phi}_i\tilde{\psi}_j\right|  = 2||X||_{\infty\to1} \text{ and } \\
\left|\sum_{i \in [N_1], j \in [N_2]}X_{ij}\right| \leq  \max_{\tilde{\phi} \in \{-1,1\}^{N_1}, \tilde{\psi} \in \{-1,1\}^{N_2}}\left|\sum_{i \in [N_1],j \in [N_2]}X_{ij}\tilde{\phi}_i\tilde{\psi}_j\right|  = ||X||_{\infty\to1}.
\end{align*} \newline

Next we show that the first inequality that $||X||_{\infty\to1} \leq 4||X||_{\square}$.  Let $\tilde{\phi} \in \{-1,1\}^{N_1}$ and $\tilde{\psi} \in \{-1,1\}^{N_2}$ be arbitrary and set $\phi^+ =  \mathbbm{1}\{\tilde{\phi} > 0\}$,  $\phi^- =  \mathbbm{1}\{\tilde{\phi} < 0\}$, $\psi^+ =  \mathbbm{1}\{\tilde{\psi} > 0\}$ and  $\psi^- =  \mathbbm{1}\{\tilde{\psi} < 0\}$ so that $\tilde{\phi} = \left(\phi^+ - \phi^-\right)$, $\tilde{\psi} = \left(\psi^+ - \psi^-\right)$ and
\begin{align*}
\sum_{i \in [N_1], j \in [N_2]}X_{ij}\tilde{\phi}_i\tilde{\psi}_j  = \sum_{i \in [N_1], j \in [N_2]}X_{ij}\left(\phi^+_i - \phi^-_i\right)\left(\psi^+_j - \psi^-_j \right) \\
= \sum_{i \in [N_1], j \in [N_2]}X_{ij}\phi^+_i\psi^+_j - \sum_{i \in [N_1], j \in [N_2]}X_{ij}\phi^-_i\psi^+_j - \sum_{i \in [N_1], j \in [N_2]}X_{ij}\phi^+_i\psi^-_j + \sum_{i \in [N_1], j \in [N_2]}X_{ij}\phi^-_i\psi^-_j
\end{align*}
so that
\begin{align*}
\left|\sum_{i \in [N_1], j \in [N_2]}X_{ij}\tilde{\phi}_i\tilde{\psi}_j \right| \leq \left|\sum_{i \in [N_1], j \in [N_2]}X_{ij}\phi^+_i\psi^+_j\right| + \left|\sum_{i \in [N_1], j \in [N_2]}X_{ij}\phi^-_i\psi^+_j\right| \\+ \left|\sum_{i \in [N_1], j \in [N_2]}X_{ij}\phi^+_i\psi^-_j \right| + \left|\sum_{i \in [N_1], j \in [N_2]}X_{ij}\phi^-_i\psi^-_j\right|.
\end{align*}
Since each of the four summands on the right-hand side are bounded by $||X||_{\square}$, it follows that $\left|\sum_{i \in [N_1], j \in [N_2]}X_{ij}\tilde{\phi}_i\tilde{\psi}_j\right|  \leq 4||X||_{\square}$. Since the choice of $\tilde{\phi}$ and $\tilde{\psi}$ was arbitrary, it follows that $||X||_{\infty\to1} \leq 4||X||_{\square}$.  \newline

Finally, we show that $||X||_{\infty\to1} = 4||X||_{\square}$ when $\sum_{i \in [N_1]}X_{ij} =  0$ for every $j \in [N_2]$ and $\sum_{j \in [N_2]}X_{ij} = 0$ for every $i \in [N_1]$. This result follows from (\ref{quart}), since the condition $\sum_{i \in [N_1]}X_{ij} =  0$ for every $j \in [N_2]$ and $\sum_{j \in [N_2]}X_{ij} = 0$ for every $i \in [N_1]$ implies that $\sum_{i \in [N_1], j \in [N_2]}X_{ij}\tilde{\phi}_i +  \sum_{i \in [N_1], j \in [N_2]}X_{ij}\tilde{\psi}_j +   \sum_{i \in [N_1], j \in [N_2]}X_{ij} = 0$. As a result, (\ref{quart}) reduces to
\begin{align*}
\sum_{i \in [N_1], j \in [N_2]}X_{ij}\phi_i\psi_j= \frac{1}{4}\sum_{i \in [N_1], j \in [N_2]}X_{ij}\tilde{\phi}_i\tilde{\psi}_j
\end{align*}
and so $\left|\sum_{i \in [N_1], j \in [N_2]}X_{ij}\phi_i\psi_j\right| = \frac{1}{4}\left|\sum_{i \in [N_1], j \in [N_2]}X_{ij}\tilde{\phi}_i\tilde{\psi}_j\right|$. Since the right-hand side is bounded by $\frac{1}{4}||X||_{\infty\to1}$ we get $\left|\sum_{i \in [N_1], j \in [N_2]}X_{ij}\phi_i\psi_j\right| \leq \frac{1}{4}||X||_{\infty\to1}$, and since this inequality holds for any choice of $\phi \in \{0,1\}^{N_1}$ and $\psi \in \{0,1\}^{N_2}$, it follows that $||X||_{\square} \leq \frac{1}{4}||X||_{\infty\to1}$. But we demonstrated above that $||X||_{\infty\to1} \leq 4||X||_{\square}$, so it follows that $||X||_{\infty\to1} = 4||X||_{\square}$
\end{proof}

The following Lemma~\ref{lem:restricted-cut} relates $||\cdot||_{\infty\to1}$ and $||\cdot||_{\square;c}$ for $c \in [0,1/2]$. The result is, to our knowledge, original to our paper. 
\begin{lemma}\label{lem:restricted-cut} For any $N_1 \times N_2$ matrix $X$ with real-valued entries and $c \in [0,1/2]$, we have that
\begin{align*}
||X||_{\square;c} \leq ||X||_{\infty\to1} \leq 9||X||_{\square;c}.
\end{align*}
\end{lemma} 

\begin{proof}
The first inequality $||X||_{\square;c} \leq ||X||_{\infty\to1}$ follows from the fact that $||X||_{\square;c} \leq ||X||_{\square;0} = ||X||_{\square}$ for any $c \in [0,1/2]$ and by Lemma~\ref{lem:cut-infty} which implies that $||X||_{\square} \leq  ||X||_{\infty\to1}$. \newline

We now show the second inequality that $||X||_{\infty\to1} \leq 9||X||_{\square;c}$. Let $\tilde{\phi} \in \{-1,1\}^{N_1}$ and $\tilde{\psi} \in \{-1,1\}^{N_2}$ be arbitrary and set $\phi = \frac{1+\tilde{\phi} }{2}\mathbbm{1}\{\sum_{i \in [N_1]}\tilde{\phi}_{i} \geq 0\} + \frac{1-\tilde{\phi}}{2}\mathbbm{1}\{\sum_{i \in [N_1]}\tilde{\phi}_{i} < 0\}$ and $\psi = \frac{1+\tilde{\psi} }{2}\mathbbm{1}\{\sum_{j \in [N_2]}\tilde{\psi}_{j} \geq 0\} + \frac{1-\tilde{\psi}}{2}\mathbbm{1}\{\sum_{j \in [N_2]}\tilde{\psi}_{j} < 0\}$ so that, by construction, $(\phi,\psi) \in \mathcal{G}_{1/2}$ and  
 \begin{align*}
 \left|\sum_{i \in [N_1],j\in[N_2]}X_{ij}\tilde{\phi}_{i}\tilde{\psi}_{j}\right| 
 = \left|\sum_{i \in [N_1], j \in [N_2]} X_{ij}\left(2\phi_i - 1\right)\left(2 \psi_j - 1 \right)\right| \\
 = \left|4\sum_{i \in [N_1], j \in [N_2]}X_{ij}\phi_i\psi_j - 2\sum_{i \in [N_1], j \in [N_2]}X_{ij}\phi_i - 2\sum_{i \in [N_1], j \in [N_2]}X_{ij}\psi_j  + \sum_{i \in [N_1], j \in [N_2]}X_{ij}\right|\\
 \leq  9\max_{(\phi,\psi) \in \mathcal{G}_{1/2}}\left|\sum_{i \in [N_1], j\in[N_2]}X_{ij}\phi_i\psi_j\right|
 \end{align*}
 where the first equality is because $\phi$ and $\psi$ are chosen so that $\left(2\phi_i - 1\right)$ is equal to $\tilde{\phi}$ or $-\tilde{\phi}$ and $\left(2\psi_j - 1\right)$ is equal to $\tilde{\psi}$ or $-\tilde{\psi}$, the inequality follows from the triangle inequality and the fact that $(\phi,\psi)$, $(\phi,\iota_{N_2})$, $(\iota_{N_1},\psi)$, and $(\iota_{N_1},\iota_{N_2})$ are all elements of $\mathcal{G}_{1/2}$, and $\iota_{N_t}$ is a $N_t$ dimensional vector with every entry equal to $1$ for $t \in \{1,2\}$. Since the choice of $\tilde{\phi}$ and $\tilde{\psi}$ was arbitrary, it follows that $||X||_{\infty\to1} \leq 9\max_{(\phi,\psi) \in \mathcal{G}_{1/2}}\left|\sum_{i \in [N_1], j\in[N_2]}X_{ij}\phi_i\psi_j\right|$ and since $\mathcal{G}_{1/2} \subseteq \mathcal{G}_{c}$ for $c \in [0,1/2]$ it follows that $||X||_{\infty\to1} \leq 9\max_{(\phi,\psi) \in \mathcal{G}_{c}}\left|\sum_{i \in [N_1], j\in[N_2]}X_{ij}\phi_i\psi_j\right| := 9||X||_{\square;c}$ for $c \in [0,1/2]$. 
\end{proof}

The following Lemma~\ref{lem:grothendieck} is a version of Grothendieck's inequality due to \cite{krivine1979constantes}.  

\begin{lemma}[Grothendieck]\label{lem:grothendieck}
For any $N_1\times N_2$ matrix $X$ with real-valued entries such that $||X||_{\infty\to1} \leq 1$, we have that 
\begin{align}
\sup_{\phi_i,\psi_j \in \mathbb{H}: ||\phi_i||_{\mathbb{H}}, ||\psi_j||_{\mathbb{H}} \leq 1}\left|\sum_{i \in [N_1], j \in [N_2]}X_{ij}\left<\phi_i,\psi_j\right>_{\mathbb{H}}\right| \leq K_G(d)
\end{align}
where $\mathbb{H}$ is an arbitrary real Hilbert space, $||\cdot||_{\mathbb{H}}$ and $\left<.,.\right>_{\mathbb{H}}$ are the associated norm and inner product, $d = \dim(\mathbb{H}) \in \mathbb{N}$, and $K_G(d)$ is the real Grothendieck constant of order $d$. 
\end{lemma}
\begin{proof}
See \cite{krivine1979constantes}, who demonstrates that $K_G(2) = \sqrt{2} \leq 1.4143$ and $K_G(\infty) \leq \frac{\pi}{2\ln(1+\sqrt{2})} \leq 1.7823$.
\end{proof}

The following Lemma~\ref{lem:dagger-infty} relates $||\cdot||_{\dagger}$ to $||\cdot||_{\infty\to1}$ and $||\cdot||_{\spadesuit}$. The upper bound in (\ref{infonetodagger}) is Theorem 3 of \cite{gittens2009error}. The lower bound in (\ref{infonetodagger}) builds on an argument given in the third paragraph of Section 4.2 of \cite{rudelson2007sampling}. In this result, our contribution is to replace their use of Khintchine's inequality (which applies to Rademacher sums) with an alternative bound derived from Lemma~\ref{lem:grothendieck} above. 

\begin{lemma}\label{lem:dagger-infty} Let $Z$ be a random matrix with independent mean-zero entries. Suppose the expectations $\mathbbm{E}\left[||Z||_{\dagger}\right]$ and $\mathbbm{E}\left[||Z||_{\infty\to1}\right]$ exist. Then
\begin{align}
(2K_G(\infty))^{-1}\mathbbm{E}\left[||Z||_{\dagger}\right] \leq \mathbbm{E}\left[||Z||_{\infty\to1}\right] \leq 2\mathbbm{E}\left[||Z||_{\dagger}\right] \label{infonetodagger} \\
\mathbbm{E}\left[||Z||_{\spadesuit}\right] \leq 2\mathbbm{E}\left[||Z||_{\dagger}\right] \label{spadesuittodagger}
\end{align}
where $K_G(\infty) \leq \frac{\pi}{2\ln(1 + \sqrt{2})} \leq 1.7823$. 
\end{lemma}
\begin{proof}
The inequality $ \mathbbm{E}\left[||Z||_{\infty\to1}\right] \leq 2\mathbbm{E}\left[||Z||_{\dagger}\right]
$ is Theorem 3 of \cite{gittens2009error}. \newline 

To demonstrate inequality $(2K_{G}(\infty))^{-1}\mathbbm{E}\left[||Z||_{\dagger}\right] \leq \mathbbm{E}\left[||Z||_{\infty\to1}\right] $, we actually show the stronger result that for any $N_1 \times N_2$ matrix $X$, $||X||_{\dagger} \leq 2K_{G}(\infty)||X||_{\infty\to1}$. To show the stronger result, we first show that $||X||_{1,2} \leq K_{G}(\infty)||X||_{\infty\to1}$. Lemma~\ref{lem:grothendieck} above implies that  $\sum_{ij}X_{ij}\sum_{t}\phi_{it}\psi_{jt} \leq K_{G}(\infty)||X||_{\infty\to1}$ for any $\phi_i, \psi_j \in \mathbbm{R}^{N_2}$ such that $||\phi_i||_2, ||\psi_j||_2 \leq 1$. Choosing $\phi_{it} = X_{it}/\sqrt{\sum_{t \in [N_2]}X_{it}^2}$ if $\sqrt{\sum_{t \in [N_2]}X_{it}^2} > 0$, $\phi_{it} = 0$ otherwise,   and  $\psi_{jt} = \mathbbm{1}\{j = t\}$ gives $\sum_{ij}X_{ij}\sum_{t}\phi_{it}\psi_{jt} = \sum_{ij}X_{ij}\sum_{t}\frac{X_{it}}{\sqrt{\sum_{t}X_{it}^2}}\mathbbm{1}\left\{\sqrt{\sum_{t \in [N_2]}X_{it}^2} > 0\right\}\mathbbm{1}\{j = t\} = \sum_{i}\sqrt{\sum_j X_{ij}^2} = ||X||_{1,2}$. It follows that $||X||_{1,2} \leq K_{G}(\infty)||X||_{\infty\to1}$. \newline

The claim  $||X||_{\dagger} \leq 2K_{G}(\infty)||X||_{\infty\to1}$ is then because
\begin{align*}
||X||_{\dagger} := ||X||_{1,2} + ||X^T||_{1,2} \leq   K_{G}(\infty)\left( ||X||_{\infty\to1} + ||X^{T}||_{\infty\to1}\right) = 2K_G(\infty)||X||_{\infty\to1}
\end{align*}
where the first equality is the definition of $||\cdot||_{\dagger}$, the inequality is from the above result that $||X||_{1,2} \leq K_{G}(\infty)||X||_{\infty\to1}$, and the second equality is due to the fact that $||X||_{\infty\to1} = ||X^T||_{\infty\to1}$. The inequality $(2K_{G}(\infty))^{-1}\mathbbm{E}\left[||Z||_{\dagger}\right] \leq \mathbbm{E}\left[||Z||_{\infty\to1}\right]$ then follows from taking expectations on both sides and dividing by $2K_{G}(\infty)$. \newline

Finally, the inequality $\mathbbm{E}\left[||Z||_{\spadesuit}\right] \leq 2\mathbbm{E}\left[||Z||_{\dagger}\right]$ follows from a symmetrization argument, Khintchine's inequality, and Jensen's inequality. Let $Z'$ be an independent
copy of $Z$ and set $W:=Z-Z'$. Since $\mathbb E[Z'\mid Z]=0$ entrywise and
$\|\cdot\|_{\spadesuit}$ is convex,
\[
\mathbb E\|Z\|_{\spadesuit}
=
\mathbb E\left\|\mathbb E[Z-Z'\mid Z]\right\|_{\spadesuit}
\leq
\mathbb E\|Z-Z'\|_{\spadesuit}.
\]
The entries of $W$ are independent and symmetric. Conditional on the
magnitudes $\{|W_{ij}|\}_{i,j}$, Khintchine's inequality gives, for each
row $i$,
\[
\mathbb E\left[
\left|\sum_{j\in[N_2]}W_{ij}\right|
\,\middle|\,
\{|W_{ij}|\}_{i,j}
\right]
\leq
\left(\sum_{j\in[N_2]}W_{ij}^2\right)^{1/2}.
\]
Summing over $i$ and taking expectations yields
\[
\mathbb E\sum_{i\in[N_1]}
\left|\sum_{j\in[N_2]}W_{ij}\right|
\leq
\mathbb E\sum_{i\in[N_1]}
\left(\sum_{j\in[N_2]}W_{ij}^2\right)^{1/2}.
\]
Applying the same argument to the columns gives $\mathbb E\|W\|_{\spadesuit} \leq \mathbb E\|W\|_\dagger$ and so $\mathbb E\|Z\|_{\spadesuit} \leq \mathbb E\|Z-Z'\|_\dagger \leq \mathbb E\|Z\|_\dagger+\mathbb E\|Z'\|_\dagger = 2\mathbb E\|Z\|_\dagger.$
\end{proof}

The following Lemma~\ref{lem:frobenius-dagger} relates the $||\cdot||_{F}$ norm to the $||\cdot||_{1,2}$ and $||\cdot||_{\dagger}$ norms. 

\begin{lemma}\label{lem:frobenius-dagger}
For any $N_1 \times N_2$ matrix $X$ with real valued entries, we have that $||X||_F \leq ||X||_{1,2}$ and $2||X||_{F} \leq ||X||_{\dagger}$.  
\end{lemma}

\begin{proof}
The first inequality follows 
\begin{align*}
||X||_F := \sqrt{\sum_{i \in [N_1], j \in [N_2]}X_{ij}^2} \leq \sum_{ i \in [N_1]}\sqrt{\sum_{j \in [N_2]}X_{ij}^2} := ||X||_{1,2}
\end{align*}
where the inequality is because the $\sqrt{\cdot}$ function is subadditive. The second inequality follows
\begin{align*}
2||X||_F =  ||X||_F + ||X^T||_F \leq ||X||_{1,2} + ||X^T||_{1,2} := ||X||_{\dagger} 
\end{align*}
where the first equality is because $||X||_F = ||X^T||_F$ and the inequality is because $||X||_F \leq ||X||_{1,2}$. 
\end{proof}

\subsubsection{Lemmas describing concentration results}\label{app:concentration-lemmas}
The following Lemma~\ref{lem:bernstein} is Bernstein's inequality. The condition that the random variables $X_i$ are uniformly bounded almost surely can be weakened to one that only controls certain moments of $X_i$. See, for example, Theorem 2.10 in Section 2.8 of \cite{boucheron2013concentration}. 
\begin{lemma}[Bernstein]\label{lem:bernstein}
Let $X_1,\ldots, X_n$ be independent real-valued random variables with finite variance such that $X_i \leq B$ for some $B > 0$ almost surely for all $i \leq n$. Let $S = \sum_{i=1}^{n}(X_i - \mathbbm{E}\left[X_i\right])$ and $\nu = \sum_{i=1}^{n} \mathbbm{E}\left[X_i^2\right]$. Then for all $t > 0$, 
\begin{align}
\mathbbm{P}\left(S \geq t\right) \leq \exp\left(-\frac{t^2}{2(\nu +Bt/3)}\right).
\end{align}
\end{lemma}

\begin{proof}
See equation (2.10) on page 36 of \cite{boucheron2013concentration}. 
\end{proof}

%
%

The following Lemma~\ref{lem:klein-rio} is a version of Talagrand's inequality for the maximum of an empirical process due to \cite{klein2005concentration}. 

\begin{lemma}[Klein and Rio]\label{lem:klein-rio}
Let $X_1,\ldots, X_n$ be independent random variables with values in some Polish space $\mathcal{X}$ and let $\mathcal{S}$ be a countable class of measurable functions from $\mathcal{X}$ into $[-1,1]^n$. For $s = \left(s^1,\ldots,s^n\right)$  in $\mathcal{S}$ let $S_n(s) = s^1(X_1)+\ldots+ s^n(X_n)$. Suppose that $\mathbbm{E}\left[s^k(X_k)\right] = 0$ for every $s \in \mathcal{S}$ and $k \in [n]$. Then for any positive $x$
\begin{align}\label{thm1.1c}
\mathbbm{P}\left(Z \geq \mathbbm{E}\left[Z\right] + x \right) \leq \exp\left(-\frac{x^2}{2\nu + 3x}\right)
\end{align}
and 
\begin{align}\label{thm1.2c}
\mathbbm{P}\left(Z \leq \mathbbm{E}\left[Z\right] - x \right) \leq \exp\left(-\frac{x^2}{2\nu + 2x}\right)
\end{align}
where $Z = \sup\{S_n(s): s \in \mathcal{S}\}$, $\nu = V_n + 2\mathbbm{E}\left[Z\right]$ and $V_n = \sup_{s \in \mathcal{S}}Var S_n(s)$. 
\end{lemma}

\begin{proof}
Result (\ref{thm1.1c}) is Theorem 1.1(c) and result (\ref{thm1.2c}) is Theorem  1.2(c) of \cite{klein2005concentration}.
\end{proof}

\subsubsection{Additional lemmas}\label{app:additional-lemmas}

The two results in this subsection are, to our knowledge, original to our paper. They concern the matrix of idiosyncratic-errors $\epsilon$ defined in Section~\ref{sec:network-structure}. Unless stated otherwise, all probabilities, expectations, and variances in this subsection are taken under $\mathbbm P_F$, under which $\mu$ and $\sigma$ are deterministic, the entries of $\epsilon=Y-\mu$ are independent, and they satisfy $\mathbbm E_F[\epsilon_{ij}]=0$ and $\mathbbm E_F[\epsilon_{ij}^2]=\sigma_{ij}^2$. Evaluating the pointwise
results below at the random realized array $\mathbf F$ gives the corresponding
conditional-on-$\mathcal H$ statements almost surely. 

For $\delta\in(0,1)$, write
$a_\delta:=3(1-\delta^2)^2/\{(32-8\delta^2)B^2\}$.

\begin{lemma}\label{lem:residual-norm-relations}
Suppose Assumption~\ref{ass:model}. For every fixed $F\in\mathcal F$ such
that $\lVert\sigma\rVert_F>0$ and every $\delta\in(0,1)$,
\begin{align}
\lVert\sigma\rVert_F
&\geq \mathbbm E_F\lVert\epsilon\rVert_F
\geq \delta\lVert\sigma\rVert_F
\left\{1-e^{-a_\delta\lVert\sigma\rVert_F^2}\right\},
\label{swapexp1}\\
\mathbbm E_F\lVert\epsilon\rVert_F
&\leq \lVert\sigma\rVert_F
\leq \delta^{-1}\mathbbm E_F\lVert\epsilon\rVert_F
\left\{1-e^{-a_\delta(\mathbbm E_F\lVert\epsilon\rVert_F)^2}\right\}^{-1},
\label{swapexp2}\\
\mathbbm E_F\lVert\epsilon\rVert_\dagger
&\geq 2\delta\lVert\sigma\rVert_F
\left\{1-e^{-a_\delta\lVert\sigma\rVert_F^2}\right\}.
\label{swapexp3}
\end{align}
\end{lemma}

\begin{proof}
Jensen's inequality gives
$\mathbbm E_F\lVert\epsilon\rVert_F
\leq\{\mathbbm E_F\lVert\epsilon\rVert_F^2\}^{1/2}
=\lVert\sigma\rVert_F$, proving the first inequalities in
\eqref{swapexp1}--\eqref{swapexp2}. 

To prove the second inequality in \eqref{swapexp1}, we start with 
$\lVert\epsilon\rVert_F\geq
\delta\lVert\sigma\rVert_F
\mathbbm 1\{\lVert\epsilon\rVert_F>
\delta\lVert\sigma\rVert_F\}$ which implies that 
$\mathbbm E_F\lVert\epsilon\rVert_F
\geq \delta\lVert\sigma\rVert_F
\left\{1-\mathbbm P_F(\lVert\epsilon\rVert_F
\leq\delta\lVert\sigma\rVert_F)\right\}$.
Since the random variables $Z_{ij}:=\sigma_{ij}^2-\epsilon_{ij}^2$ are independent and mean-zero under $\mathbbm{P}_F$ with $|Z_{ij}|\leq4B^2$ and $\operatorname{Var}_F(Z_{ij})\leq \mathbbm E_F[\epsilon_{ij}^4]\leq4B^2\sigma_{ij}^2$, Lemma ~\ref{lem:bernstein} implies that $\mathbbm P_F(\lVert\epsilon\rVert_F
\leq\delta\lVert\sigma\rVert_F) \leq \exp\{-a_\delta\lVert\sigma\rVert_F^2\}$. Combining this inequality with the previous one gives the second inequality in \eqref{swapexp1}. 

Because $\mathbbm E_F\lVert\epsilon\rVert_F\leq\lVert\sigma\rVert_F$ and
$u\mapsto1-e^{-a_\delta u^2}$ is nondecreasing, replacing
$\lVert\sigma\rVert_F$ inside the exponential in \eqref{swapexp1} by
$\mathbbm E_F\lVert\epsilon\rVert_F$ and dividing proves
\eqref{swapexp2}. Finally,
$\lVert\epsilon\rVert_\dagger\geq2\lVert\epsilon\rVert_F$ by
Lemma~\ref{lem:frobenius-dagger}. Taking expectations and applying the second inequality in 
\eqref{swapexp1} proves \eqref{swapexp3}.
\end{proof}

The next lemma combines Lemmas~\ref{lem:cut-infty},
\ref{lem:restricted-cut}, and \ref{lem:dagger-infty} with
Lemma~\ref{lem:klein-rio}. For brevity, define
$\Gamma_F:=\lVert\sigma\rVert_F^2
+4B\mathbbm E_F\lVert\epsilon\rVert_\dagger
+B\lVert\sigma\rVert_F$.

\begin{lemma}\label{lem:cut-concentration}
Suppose Assumption~\ref{ass:model}. For every fixed $F\in\mathcal F$,
$x>0$, and $c\in(0,1/2]$,
\begin{align}
\mathbbm P_F\!\left(
\lVert\epsilon\rVert_{\square;c}
\geq \mathbbm E_F\lVert\epsilon\rVert_\dagger
+\tfrac14\lVert\sigma\rVert_F+x\right)
&\leq \exp\!\left\{-\frac{x^2}{2\Gamma_F+6Bx}\right\},
\label{ub}\\
\mathbbm P_F\!\left(
\lVert\epsilon\rVert_{\square;c}
\leq \{18K_G(\infty)\}^{-1}
\mathbbm E_F\lVert\epsilon\rVert_\dagger-x\right)
&\leq \exp\!\left\{-\frac{x^2}{2\Gamma_F+4Bx}\right\},
\label{lb}
\end{align}
where $K_G(\infty)\leq\pi/\{2\ln(1+\sqrt2)\}\leq1.7823$.
\end{lemma}

\begin{proof}
Fix $F\in\mathcal F$, $x>0$, and $c\in(0,1/2]$, and work under
$\mathbbm P_F$. Set $\epsilon^*_{ij}=\epsilon_{ij}/(2B)$ and
$\sigma^*_{ij}=\sigma_{ij}/(2B)$. To match the notation of
Lemma~\ref{lem:klein-rio}, let $n=N_1N_2$, set
$k(i,j)=(i-1)N_2+j$, and let $i(k)=1+\lfloor(k-1)/N_2\rfloor$ and
$j(k)=k-(i(k)-1)N_2$. Take
$\mathcal X=[-1,1]\times[N_1]\times[N_2]$, with the finite sets given
the discrete topology, and define
$X_k=(\epsilon^*_{i(k)j(k)},i(k),j(k))$. Under $\mathbbm{P}_F$, 
$X_1,\ldots,X_n$ are independent random variables taking values in the
Polish space $\mathcal X$.

For $\eta\in\{-1,1\}$ and $(\phi,\psi)\in\mathcal G_c$, let
$s_{\eta,\phi,\psi}:\mathcal X\to[-1,1]^n$ have $k$th coordinate
$s_{\eta,\phi,\psi}^k(z,i,j)
=\eta z\phi_i\psi_j\mathbbm 1\{k=k(i,j)\}$, and let
$\mathcal S$ be the collection of these functions. The class
$\mathcal S$ is finite, hence countable; each coordinate function is
measurable and takes values in $[-1,1]$. Moreover, for every
$s=s_{\eta,\phi,\psi}\in\mathcal S$ and $k\in[n]$,
$\mathbbm E_F[s^k(X_k)]
=\eta\phi_{i(k)}\psi_{j(k)}
\mathbbm E_F[\epsilon^*_{i(k)j(k)}]=0$.
Using exactly the notation of Lemma~\ref{lem:klein-rio},
\begin{align*}
S_n(s_{\eta,\phi,\psi})
&=\sum_{k=1}^n s_{\eta,\phi,\psi}^k(X_k)
 =\eta\sum_{i,j}\epsilon^*_{ij}\phi_i\psi_j,\\
Z:=\sup_{s\in\mathcal S}S_n(s)
&=\lVert\epsilon^*\rVert_{\square;c},\\
V_n:=\sup_{s\in\mathcal S}\operatorname{Var}_F\{S_n(s)\}
&=\max_{(\phi,\psi)\in\mathcal G_c}
  \sum_{i,j}\sigma_{ij}^{*2}\phi_i\psi_j
 =\lVert\sigma^*\rVert_F^2.
\end{align*}
The variance equality uses independence and
$(\eta\phi_i\psi_j)^2=\phi_i\psi_j$. The last maximum is attained by
the two all-ones vectors, which belong to $\mathcal G_c$. Thus all the
hypotheses of Lemma~\ref{lem:klein-rio} are satisfied, and its notation
is matched by $Z=\lVert\epsilon^*\rVert_{\square;c}$,
$V_n=\lVert\sigma^*\rVert_F^2$, and
$\nu=V_n+2\mathbbm E_F[Z]$. Consequently,
\eqref{thm1.1c}--\eqref{thm1.2c} apply under $\mathbbm P_F$.

It remains to bound $\mathbbm E_F[Z]$.
Lemma~\ref{lem:cut-infty}, applied pathwise to
$\epsilon^*$, Lemma~\ref{lem:dagger-infty}, and Jensen's
inequality give
\begin{align*}
4\mathbbm E_F[Z]
&\leq4\mathbbm E_F\lVert\epsilon^*\rVert_\square
\leq\mathbbm E_F\lVert\epsilon^*\rVert_{\infty\to1}
+\mathbbm E_F\lVert\epsilon^*\rVert_\spadesuit
+\mathbbm E_F\left|\sum_{i,j}\epsilon^*_{ij}\right|
\leq 4\mathbbm E_F\lVert\epsilon^*\rVert_\dagger
+\lVert\sigma^*\rVert_F,\\
9\mathbbm E_F[Z]
&\overset{\text{Lemma~\ref{lem:restricted-cut}}}{\geq}
\mathbbm E_F\lVert\epsilon^*\rVert_{\infty\to1}
\overset{\text{Lemma~\ref{lem:dagger-infty}}}{\geq}
\{2K_G(\infty)\}^{-1}
\mathbbm E_F\lVert\epsilon^*\rVert_\dagger.
\end{align*}
Here Lemma~\ref{lem:dagger-infty} gives the two upper bounds
$\mathbbm E_F\lVert\epsilon^*\rVert_{\infty\to1}
\leq2\mathbbm E_F\lVert\epsilon^*\rVert_\dagger$ and
$\mathbbm E_F\lVert\epsilon^*\rVert_\spadesuit
\leq2\mathbbm E_F\lVert\epsilon^*\rVert_\dagger$. Its hypotheses
hold because $\epsilon^*$ has independent mean-zero bounded entries under $\mathbbm{P}_F$.
Also,
$\mathbbm E_F|\sum_{i,j}\epsilon^*_{ij}|
\leq\{\operatorname{Var}_F(\sum_{i,j}\epsilon^*_{ij})\}^{1/2}
=\lVert\sigma^*\rVert_F$.

It follows that, with
$U=\mathbbm E_F\lVert\epsilon^*\rVert_\dagger
+\lVert\sigma^*\rVert_F/4$ and
$L=\{18K_G(\infty)\}^{-1}
\mathbbm E_F\lVert\epsilon^*\rVert_\dagger$, we have that
$L\leq\mathbbm E_F[Z]\leq U$ and
$\nu\leq\lVert\sigma^*\rVert_F^2+2U$. Therefore
\eqref{thm1.1c}--\eqref{thm1.2c} imply
\begin{align*}
\mathbbm P_F\{Z\geq U+x'\}
&\leq\exp\left\{-\frac{x'^2}
{2(\lVert\sigma^*\rVert_F^2+2U)+3x'}\right\},\\
\mathbbm P_F\{Z\leq L-x'\}
&\leq\exp\left\{-\frac{x'^2}
{2(\lVert\sigma^*\rVert_F^2+2U)+2x'}\right\}.
\end{align*}
Finally set $x'=x/(2B)$ and use
$2B\lVert\epsilon^*\rVert_{\square;c}
=\lVert\epsilon\rVert_{\square;c}$,
$2B\mathbbm E_F\lVert\epsilon^*\rVert_\dagger
=\mathbbm E_F\lVert\epsilon\rVert_\dagger$, and
$2B\lVert\sigma^*\rVert_F=\lVert\sigma\rVert_F$. In particular,
for $a\in\{2,3\}$,
$4B^2\{2(\lVert\sigma^*\rVert_F^2+2U)+ax'\}
=2\Gamma_F+2aBx$. The two preceding bounds are therefore exactly
\eqref{ub} and \eqref{lb}.
\end{proof}

\subsection{Propositions}\label{app:propositions}
This section contains proofs of the four propositions stated in Section~\ref{sec:results} of the main text. 

\propone*

\begin{proof}
Since $CI_{1}(g_1,g_2;\alpha) := \hat{\theta}(g_1,g_2) \pm \left[K_1(\alpha) \times \hat{\sigma}(g_1,g_2)\right]/\sqrt{m_1m_2}$ with \\ $K_1(\alpha) := \sqrt{1.39(N_1+N_2) - 2\ln(\alpha/2)}$, it is sufficient to show that Assumptions~\ref{ass:model}, \ref{ass:nondegeneracy}(s), and \ref{ass:scale-consistency}(i)  imply, for any $\alpha \in (0,1)$,
\begin{align*}
\limsup_{F \in \mathcal{F}: N_1,N_2 \to \infty}\mathbbm{P}_{F}\left(\cup_{(g_1,g_2) \in \mathcal{G}_c} \left\{\left|\hat{\theta}(g_1,g_2) - \theta(g_1,g_2)\right| \geq  K_1(\alpha)\hat{\sigma}(g_1,g_2)/\sqrt{m_1m_2}\right\}\right) \leq \alpha.
\end{align*}
 Fix $\alpha \in (0,1)$.
 For each $(g_1,g_2) \in \mathcal{G}_c$ define the three events 
\begin{align*}
A(g_1,g_2) &:= \left\{\left|\hat{\theta}(g_1,g_2) - \theta(g_1,g_2)\right| \geq  K_1(\alpha)\hat{\sigma}(g_1,g_2)/\sqrt{m_1m_2}\right\}, \\
B(g_1,g_2) &:= \left\{\hat{\sigma}(g_1,g_2) \geq (\sigma(g_1,g_2) - r'(g_1,g_2))\right\},  \text{ and  } \\
C(g_1,g_2) &:= \left\{\left|\hat{\theta}(g_1,g_2) - \theta(g_1,g_2)\right| \geq  K_1(\alpha)\left(\sigma(g_1,g_2) - r'(g_1,g_2)\right)/\sqrt{m_1m_2}\right\}
\end{align*}
where $r'(g_1,g_2) = r\sigma(g_1,g_2)$ and $r$ is the sequence of real numbers defined in Assumption~\ref{ass:scale-consistency}. For a fixed $F \in \mathcal{F}$, three events are related by
\begin{align*}
\mathbbm{P}_F\left(\cup_{(g_1,g_2) \in \mathcal{G}_c}A(g_1,g_2)\right) = \mathbbm{P}_F\left(\cup_{(g_1,g_2) \in \mathcal{G}_c}\left\{A(g_1,g_2) \cap \left\{B(g_1,g_2) \cup B(g_1,g_2)^C\right\}\right\}\right) \\
=  \mathbbm{P}_F\left(\cup_{(g_1,g_2) \in \mathcal{G}_c}\left\{\left\{A(g_1,g_2) \cap B(g_1,g_2)\right\} \cup \left\{A(g_1,g_2) \cap B(g_1,g_2)^C\right\}\right\}\right) \\
= \mathbbm{P}_F\left(\left\{\cup_{(g_1,g_2) \in \mathcal{G}_c}\left\{A(g_1,g_2) \cap B(g_1,g_2)\right\}\right\} \cup \left\{\cup_{(g_1,g_2) \in \mathcal{G}_c}\left\{A(g_1,g_2) \cap B(g_1,g_2)^C\right\}\right\}\right)\\
\leq \mathbbm{P}_F\left(\cup_{(g_1,g_2) \in \mathcal{G}_c}\left\{A(g_1,g_2) \cap B(g_1,g_2)\right\}\right) + \mathbbm{P}_F\left(\cup_{(g_1,g_2) \in \mathcal{G}_c}\left\{A(g_1,g_2) \cap B(g_1,g_2)^C\right\}\right)\\
\leq \mathbbm{P}_F\left(\cup_{(g_1,g_2) \in \mathcal{G}_c}C(g_1,g_2)\right) + \mathbbm{P}_F\left(\cup_{(g_1,g_2) \in \mathcal{G}_c}B(g_1,g_2)^C\right).
\end{align*}
where $B(g_1,g_2)^C := \left\{\hat{\sigma}(g_1,g_2) < (\sigma(g_1,g_2) - r'(g_1,g_2))\right\}$ is the complement of the event $B(g_1,g_2)$, the first inequality is the union bound, and the second inequality is because the events $A(g_1,g_2)$ and  $B(g_1,g_2)$ imply the event $C(g_1,g_2)$ and the events $A(g_1,g_2)$ and $B(g_1,g_2)^C$ imply the event $B(g_1,g_2)^C$. In Steps 1-3 below we show that\\ $\limsup_{F \in \mathcal{F}: N_1,N_2\to\infty}\mathbbm{P}_F\left(\cup_{(g_1,g_2) \in \mathcal{G}_c}C(g_1,g_2)\right) \leq \alpha$. In Step 4 below we show that\\ $\lim_{F \in \mathcal{F}: N_1,N_2\to\infty}\mathbbm{P}_F\left(\cup_{(g_1,g_2) \in \mathcal{G}_c}B(g_1,g_2)^C\right) = 0$. It follows that\\ $\limsup_{F \in \mathcal{F}: N_1,N_2\to\infty}\mathbbm{P}_F\left(\cup_{(g_1,g_2) \in \mathcal{G}_c}A(g_1,g_2)\right) \leq \alpha$ which is sufficient for the conclusion of Proposition~\ref{prop:propone}.   \newline

\textbf{Step 1:} In this step we show that $\limsup_{F \in \mathcal{F}: N_1,N_2\to\infty}\mathbbm{P}_F\left(\cup_{(g_1,g_2) \in \mathcal{G}_c}C(g_1,g_2)\right) \leq \alpha$. Recall that for a fixed $(g_1,g_2) \in \mathcal{G}_c$, the estimation error 
\begin{align*}
\hat{\theta}(g_1,g_2) - \theta(g_1,g_2)  = \frac{1}{m_1m_2}\sum_{i\in[N_1],j\in[N_2]}\epsilon_{ij}g_{i,1}g_{j,2}
\end{align*}
where $m_t := \sum_{i \in [N_t]}g_{i,t}$. Since, under $\mathbbm P_F$, the summands on the right-hand side are independent, mean-zero, and uniformly bounded in absolute value by $2B$, Lemma~\ref{lem:bernstein} and the union bound implies that for any $\{x(g_1,g_2)\}_{(g_1,g_2) \in \mathcal{G}_c}$ with $\min_{(g_1,g_2) \in \mathcal{G}_c}x(g_1,g_2) > 0$ and fixed $F \in \mathcal{F}$, 
\begin{align*}
\mathbbm{P}_F\left(\cup_{(g_1,g_2) \in \mathcal{G}_c}\left\{\left|\hat{\theta}(g_1,g_2) - \theta(g_1,g_2)\right| \geq  x(g_1,g_2)\right\}\right)\\ 
\leq 2^{N_1+N_2+1}\max_{(g_1,g_2) \in \mathcal{G}_c}\exp\left(-\frac{ x(g_1,g_2)^2m_1^2m_2^2}{2\left(\sigma(g_1,g_2)^2 + 2B x(g_1,g_2)/3\right)m_1m_2}\right)
\end{align*}
since for any $c > 0$, $|\mathcal{G}_c| \leq |\mathcal{G}_0| = 2^{N_1+N_2}$ where $|\mathcal{G}_c| = \sum_{k \geq cN_1, l \geq cN_2}{N_1 \choose k}{N_2 \choose l}$ is the number of elements in the set $\mathcal{G}_c$. Setting $x(g_1,g_2) = K_1(\alpha)\left(\sigma(g_1,g_2)-r'(g_1,g_2)\right)/\sqrt{m_1m_2}$ so that the event $\left\{\left|\hat{\theta}(g_1,g_2) - \theta(g_1,g_2)\right| \geq  x(g_1,g_2)\right\}$ is equivalent to $C(g_1,g_2) $ gives 
\begin{align*}
 \mathbbm{P}_F\left(\cup_{(g_1,g_2) \in \mathcal{G}_c}C(g_1,g_2)\right) \\
 \leq 2^{N_1+N_2+1}\max_{(g_1,g_2) \in \mathcal{G}_c}\exp\left(-\frac{K_1(\alpha)^2\left(\sigma(g_1,g_2)-r'(g_1,g_2)\right)^2m_1m_2}{2\sigma(g_1,g_2)^2m_1m_2 + 4BK_1(\alpha) \left(\sigma(g_1,g_2)-r'(g_1,g_2)\right)\sqrt{m_1m_2}/3}\right) \\
= 2^{N_1+N_2+1}\max_{(g_1,g_2) \in \mathcal{G}_c}\exp\left(-\frac{K_1(\alpha)^2(1-r)^2\sigma(g_1,g_2)^2m_1m_2}{2D}\right) \\
 = 2^{N_1+N_2+1}\exp\left(-\frac{\tilde{K}_1(\alpha)^2}{2}\right)\max_{(g_1,g_2) \in \mathcal{G}_c}\exp\left(\frac{-K_1(\alpha)^2(1-r)^2\sigma(g_1,g_2)^2m_1m_2 + \tilde{K}_1(\alpha)^2D}{2D}\right)
\end{align*}
where $\tilde{K}_1(\alpha) := \sqrt{2\ln(2)(N_1+N_2) - 2\ln(\alpha/2)}$ and $D := \sigma(g_1,g_2)^2m_1m_2 + 2BK_1(\alpha)(1-r)\sigma(g_1,g_2)\sqrt{m_1m_2}/3$. \newline

We show in Step 2 below that $2^{N_1+N_2 +1}\exp\left(\frac{-\tilde{K}_1(\alpha)^2}{2}\right) = \alpha$ and in Step 3 below that
\begin{align*}
\lim_{F \in \mathcal{F}: N_1,N_2 \to \infty}\max_{(g_1,g_2) \in \mathcal{G}_c}\exp\left(\frac{-K_1(\alpha)^2(1-r)^2\sigma(g_1,g_2)^2m_1m_2 + \tilde{K}_1(\alpha)^2D}{2D}\right) = 0.
\end{align*}
It follows from these two steps that $\limsup_{F \in \mathcal{F}: N_1,N_2 \to \infty}\mathbbm{P}_F\left(\cup_{(g_1,g_2) \in \mathcal{G}_c}C(g_1,g_2)\right) \leq \alpha$. \newline

\textbf{Step 2:} In this step we show that $2^{N_1+N_2 +1}\exp\left(\frac{-\tilde{K}_1(\alpha)^2}{2}\right) = \alpha$. Since $\tilde{K}_1(\alpha)^2 = 2\ln(2)(N_1+N_2)-2\ln(\alpha/2)$ we have 
\begin{align*}
2^{N_1+N_2 +1}\exp\left(\frac{-\tilde{K}_1(\alpha)^2}{2}\right) = 2^{N_1+N_2 +1}\exp\left(\ln(\alpha/2) - \ln(2)(N_1+N_2)\right) \\
= 2^{N_1+N_2+1}\exp(\ln(\alpha) - \ln(2)(N_1+N_2+1)) = 2^{N_1+N_2+1}\alpha 2^{-N_1 - N_2 - 1} =  \alpha.
\end{align*}

\textbf{Step 3:} In this step we show that
\begin{align}\label{step3}
\lim_{F \in \mathcal{F}: N_1,N_2 \to \infty}
\max_{(g_1,g_2) \in \mathcal{G}_c}
\exp\left(\frac{-K_1(\alpha)^2(1-r)^2\sigma(g_1,g_2)^2m_1m_2+\tilde{K}_1(\alpha)^2D}{2D}\right)=0
\end{align}
where $\tilde{K}_1(\alpha) := \sqrt{2\ln(2)(N_1+N_2)-2\ln(\alpha/2)}$ and $D := \sigma(g_1,g_2)^2m_1m_2 + 2BK_1(\alpha)(1-r)\sigma(g_1,g_2)\sqrt{m_1m_2}/3$. Define $\Delta_N := K_1(\alpha)^2(1-r)^2-\tilde{K}_1(\alpha)^2$.  Since $K_1(\alpha)^2=1.39(N_1+N_2)-2\ln(\alpha/2)$ and $\tilde{K}_1(\alpha)^2=2\ln(2)(N_1+N_2)-2\ln(\alpha/2)$, we have that 
\begin{align*}
\Delta_N = \bigl(1.39(1-r)^2-2\ln(2)\bigr)(N_1+N_2)+2\ln(\alpha/2)\bigl(1-(1-r)^2\bigr).
\end{align*}
Because $1.39 - 2\ln(2) > 0.003$ and $r\to0$ under Assumption~\ref{ass:scale-consistency}(i), $\Delta_N \geq 0.003(N_1+N_2)$ for $\min(N_1,N_2)$ sufficiently large. Now fix $(g_1,g_2)\in\mathcal{G}_c$ and define
\begin{align*}
q(g_1,g_2):= \frac{\sigma(g_1,g_2)\sqrt{m_1m_2}}{K_1(\alpha)(1-r)}
\end{align*}
so that the exponent in \eqref{step3} can be written as
\begin{align*}
R(g_1,g_2) := \frac{-K_1(\alpha)^2(1-r)^2\sigma(g_1,g_2)^2m_1m_2+\tilde{K}_1(\alpha)^2D}{2D}  
=\frac{-\Delta_N q(g_1,g_2)+\frac{2B}{3}\tilde{K}_1(\alpha)^2}{2\left(q(g_1,g_2)+\frac{2B}{3}\right)}.
\end{align*}

We show that $q(g_1,g_2)$ diverges uniformly over $\mathcal{G}_c$. Since $(g_1,g_2)\in\mathcal{G}_c$, we have $m_1m_2\geq c^2N_1N_2$ and, since $\alpha\in(0,1)$ is fixed, there exists a constant $C_K<\infty$, depending only on $\alpha$, such that $K_1(\alpha) \leq C_K\sqrt{N_1+N_2}$ for all sufficiently large $N_1+N_2$. Also, since $r\to0$, we have $3/2 \geq 1-r\geq 1/2$ for all sufficiently large $\min(N_1,N_2)$. Therefore, for all sufficiently large $\min(N_1,N_2)$, $q(g_1,g_2)  \geq \frac{2c\,\sigma(g_1,g_2)\sqrt{N_1N_2}}{3C_K\sqrt{N_1+N_2}}$ and so eventually
\begin{align*}
\inf_{(g_1,g_2)\in\mathcal{G}_c}q(g_1,g_2) \geq
\frac{2c}{3C_K\sqrt{2}}\sqrt{\min(N_1,N_2)}\inf_{(g_1,g_2)\in\mathcal{G}_c}\sigma(g_1,g_2)
\end{align*}
which diverges to infinity by Assumption~\ref{ass:nondegeneracy}(s). 

Since $\Delta_N\geq 0.003(N_1+N_2)$ and $\tilde{K}_1(\alpha)^2$ is of order $N_1+N_2$, the ratio $\tilde{K}_1(\alpha)^2/\Delta_N$ is bounded for all sufficiently large $\min(N_1,N_2)$. Since $q(g_1,g_2)\to\infty$ uniformly over $\mathcal{G}_c$, we have, for all sufficiently large $\min(N_1,N_2)$, $\Delta_N q(g_1,g_2) \geq
\frac{4B}{3}\tilde{K}_1(\alpha)^2$ uniformly over $(g_1,g_2)\in\mathcal{G}_c$ and so
\begin{align*}
-\Delta_N q(g_1,g_2)+\frac{2B}{3}\tilde{K}_1(\alpha)^2 \leq -\frac{1}{2}\Delta_N q(g_1,g_2).
\end{align*}
Also, since $q(g_1,g_2)\to\infty$ uniformly over $\mathcal{G}_c$, we have, for all sufficiently large $\min(N_1,N_2)$, $q(g_1,g_2)+\frac{2B}{3} \leq 2q(g_1,g_2)$ uniformly over $(g_1,g_2)\in\mathcal{G}_c$. It follows that, for all sufficiently large $\min(N_1,N_2)$,
\begin{align*}
R(g_1,g_2)
\leq
-\frac{\Delta_N}{8}
\leq
-\frac{0.003}{8}(N_1+N_2)
\end{align*}
uniformly over $(g_1,g_2)\in\mathcal{G}_c$, and so
\begin{align*}
\max_{(g_1,g_2)\in\mathcal{G}_c}
\exp\left(R(g_1,g_2)\right)
\leq
\exp\left(-\frac{0.003}{8}(N_1+N_2)\right)
\to0,
\end{align*}
which demonstrates \eqref{step3}.\newline

%
%

\textbf{Step 4:} The result $\lim_{F \in \mathcal{F}: N_1,N_2\to\infty}\mathbbm{P}_F\left(\cup_{(g_1,g_2) \in \mathcal{G}_c}B(g_1,g_2)^C\right) = 0$ follows directly from Assumption~\ref{ass:scale-consistency}(i) because, by definition of $r'(g_1,g_2) := r \sigma(g_1,g_2)$, the event $\cup _{(g_1,g_2) \in \mathcal{G}_c}B(g_1,g_2)^C$ implies the event $\{\max_{(g_1,g_2) \in \mathcal{G}_c}\frac{\left|\hat{\sigma}(g_1,g_2) - \sigma(g_1,g_2)\right|}{\sigma(g_1,g_2)} > r\}$.
\end{proof}

\proptwo*

\begin{proof}
Since $CI_{2}(g_1,g_2;\alpha) := \hat{\theta}(g_1,g_2) \pm \left[\hat{\bar{\tau}} + K_2(\alpha)\hat{V}\right]/m_1m_2$ with \\ $K_2(\alpha) := \sqrt{- 2\ln(\alpha)}$, it is sufficient to show that Assumptions~\ref{ass:model}, \ref{ass:nondegeneracy}(w), \ref{ass:scale-consistency}(ii) and \ref{ass:scale-consistency}(iii) imply, for any $\alpha \in (0,1)$,
\begin{align*}
\limsup_{F \in \mathcal{F}: N_1,N_2 \to \infty}\mathbbm{P}_F\left(\cup_{(g_1,g_2) \in \mathcal{G}_c} \left\{\left|\hat{\theta}(g_1,g_2) - \theta(g_1,g_2)\right| \geq  \left[\hat{\bar{\tau}} + K_2(\alpha)\hat{V}\right]/m_1m_2\right\}\right) \leq \alpha.
\end{align*}
 Fix $\alpha \in (0,1)$. For each $(g_1,g_2) \in \mathcal{G}_c$ define the three events 
\begin{align*}
A(g_1,g_2) &:= \left\{\left|\hat{\theta}(g_1,g_2) - \theta(g_1,g_2)\right| \geq  \left[\hat{\bar{\tau}} + K_2(\alpha)\hat{V}\right]/m_1m_2\right\}, \\
B(g_1,g_2) &:= \left\{\hat{\bar{\tau}} \geq (\bar{\tau} - r')\right\} \cap \left\{\hat{V} \geq (V - r'')\right\}, \text{ and  } \\
C(g_1,g_2) &:= \left\{\left|\hat{\theta}(g_1,g_2) - \theta(g_1,g_2)\right| \geq   \left[(\bar{\tau} - r') + K_2(\alpha)(V - r'')\right]/m_1m_2\right\}
\end{align*}
where $r' := r\bar{\tau}$, $r'' = rV$, and $r$ is the sequence of real numbers defined in Assumption~\ref{ass:scale-consistency}. For any fixed $F \in \mathcal{F}$, three events are related by
\begin{align*}
\mathbbm{P}_F\left(\cup_{(g_1,g_2) \in \mathcal{G}_c}A(g_1,g_2)\right) = \mathbbm{P}_F\left(\cup_{(g_1,g_2) \in \mathcal{G}_c}\left\{A(g_1,g_2) \cap \left\{B(g_1,g_2) \cup B(g_1,g_2)^C\right\}\right\}\right) \\
=  \mathbbm{P}_F\left(\cup_{(g_1,g_2) \in \mathcal{G}_c}\left\{\left\{A(g_1,g_2) \cap B(g_1,g_2)\right\} \cup \left\{A(g_1,g_2) \cap B(g_1,g_2)^C\right\}\right\}\right) \\
= \mathbbm{P}_F\left(\left\{\cup_{(g_1,g_2) \in \mathcal{G}_c}\left\{A(g_1,g_2) \cap B(g_1,g_2)\right\}\right\} \cup \left\{\cup_{(g_1,g_2) \in \mathcal{G}_c}\left\{A(g_1,g_2) \cap B(g_1,g_2)^C\right\}\right\}\right)\\
\leq \mathbbm{P}_F\left(\cup_{(g_1,g_2) \in \mathcal{G}_c}\left\{A(g_1,g_2) \cap B(g_1,g_2)\right\}\right) + \mathbbm{P}_F\left(\cup_{(g_1,g_2) \in \mathcal{G}_c}\left\{A(g_1,g_2) \cap B(g_1,g_2)^C\right\}\right)\\
\leq \mathbbm{P}_F\left(\cup_{(g_1,g_2) \in \mathcal{G}_c}C(g_1,g_2)\right) + \mathbbm{P}_F\left(\cup_{(g_1,g_2) \in \mathcal{G}_c}B(g_1,g_2)^C\right).
\end{align*}
where $B(g_1,g_2)^C := \left\{\hat{\bar{\tau}} < (\bar{\tau} - r')\right\} \cup \left\{\hat{V} < (V - r'')\right\}$ is the complement of the event $B(g_1,g_2)$, the first inequality is the union bound, and the second inequality is because the events $A(g_1,g_2)$ and  $B(g_1,g_2)$ imply the event $C(g_1,g_2)$ and the events $A(g_1,g_2)$ and $B(g_1,g_2)^C$ imply the event $B(g_1,g_2)^C$. In Steps 1-3 below we show that $\limsup_{F \in \mathcal{F}: N_1,N_2\to\infty}\mathbbm{P}_F\left(\cup_{(g_1,g_2) \in \mathcal{G}_c}C(g_1,g_2)\right) \leq \alpha$ follows from Lemmas~\ref{lem:residual-norm-relations} and \ref{lem:cut-concentration} in Appendix Section~\ref{app:additional-lemmas}. In Step 4 below we show that \\$\lim_{F \in \mathcal{F}: N_1,N_2\to\infty}\mathbbm{P}_F\left(\cup_{(g_1,g_2) \in \mathcal{G}_c}B(g_1,g_2)^C\right) = 0$. It follows that \\$\limsup_{F \in \mathcal{F}: N_1,N_2\to\infty}\mathbbm{P}_F\left(\cup_{(g_1,g_2) \in \mathcal{G}_c}A(g_1,g_2)\right) \leq \alpha$ which is sufficient for the conclusion of Proposition~\ref{prop:proptwo}.   \newline

\textbf{Step 1:} For a fixed $F \in \mathcal{F}$, $(g_1,g_2) \in \mathcal{G}_c$, and $x > 0$, equation (\ref{ub}) of Lemma~\ref{lem:cut-concentration} implies that
\begin{align*}
\mathbbm{P}_F\left(||\epsilon||_{\square;c} \geq \mathbbm{E}_F\left[||\epsilon||_{\dagger}\right] + ||\sigma||_F/4 + x \right) \leq \exp\left(-\frac{x^2}{2(||\sigma||_F^2 + 4B\mathbbm{E}_F\left[||\epsilon||_{\dagger}\right] + B||\sigma||_F) + 6Bx}\right).
\end{align*}
Since $\bar{\tau} := 1.01\mathbbm{E}_F\left[||\epsilon||_{\dagger}\right] + .25||\sigma||_F$ and, for any $(g_1,g_2) \in \mathcal{G}_c$,  the estimation error 
\begin{align*}
\hat{\theta}(g_1,g_2) - \theta(g_1,g_2)  = \frac{1}{m_1m_2}\sum_{i\in[N_1],j\in[N_2]}\epsilon_{ij}g_{i,1}g_{j,2},
\end{align*}
satisfies $\left|\hat{\theta}(g_1,g_2) - \theta(g_1,g_2) \right| \leq ||\epsilon||_{\square;c}/(m_1m_2)$ where $m_t  := \sum_{i \in [N_t]}g_{i,t}$ for $t \in \{1,2\}$, it follows that for any $F \in \mathcal{F}$ and  $x > 0$
\begin{align*}
\mathbbm{P}_F\left(\cup_{(g_1,g_2) \in \mathcal{G}_c}\left\{\left|\hat{\theta}(g_1,g_2) - \theta(g_1,g_2) \right| \geq \left[\bar{\tau} - 0.005\mathbbm{E}_F\left[||\epsilon||_{\dagger}\right] + x\right]/m_1m_2 \right\}\right) \\ 
\leq \exp\left(-\frac{x^2}{2(||\sigma||_F^2 + 4B\mathbbm{E}_F\left[||\epsilon||_{\dagger}\right] + B||\sigma||_F) + 6Bx}\right).
\end{align*}
Set $x' = 0.005\mathbbm{E}_F\left[||\epsilon||_{\dagger}\right] -r' + K_2(\alpha)(V - r'')$ with $V := \sqrt{||\sigma||_F^2 + 4B\mathbbm{E}_F\left[|| \epsilon||_{\dagger}\right] +B||\sigma||_F}$ and $K_2(\alpha) := \sqrt{-2\ln(\alpha)}$. Since $r \to 0$, $x'$ is eventually positive for $\min(N_1,N_2)$ sufficiently large. The event $\left\{\left|\hat{\theta}(g_1,g_2) - \theta(g_1,g_2) \right| \geq \left[\bar{\tau} - 0.005\mathbbm{E}_F\left[||\epsilon||_{\dagger}\right] + x'\right]/m_1m_2 \right\}$ is equivalent to $C(g_1,g_2)$, and so for any fixed $F \in \mathcal{F}$
\begin{align*}
\mathbbm{P}_F\left(\cup_{(g_1,g_2) \in \mathcal{G}_c}C(g_1,g_2)\right)
\leq \exp\left(-\frac{\left(0.005\mathbbm{E}_F\left[||\epsilon||_{\dagger}\right] -r' + K_2(\alpha)(V - r'')\right)^2}{2V^2 + 6B\left( 0.005\mathbbm{E}_F\left[||\epsilon||_{\dagger}\right] -r' + K_2(\alpha)(V - r'')\right)}\right) \\
= \exp\left(-\frac{ K_2(\alpha)^2 V^2 + rem}{2V^2 + 6B\left(0.005\mathbbm{E}_F\left[||\epsilon||_{\dagger}\right] -r' + K_2(\alpha)(V - r'')\right)}\right)\\
=  \exp\left(-\frac{ K_2(\alpha)^2 V^2}{2V^2 + 6B\left(0.005\mathbbm{E}_F\left[||\epsilon||_{\dagger}\right]-r' + K_2(\alpha)(V - r'')\right)}\right)\\
\times\exp\left(-\frac{ rem}{2V^2 + 6B\left(0.005\mathbbm{E}_F\left[||\epsilon||_{\dagger}\right]-r' + K_2(\alpha)(V - r'')\right)}\right) \\
= \exp\left(-\frac{ K_2(\alpha)^2}{2}\right) \exp\left(\frac{ K_2(\alpha)^26B\left(0.005\mathbbm{E}_F\left[||\epsilon||_{\dagger}\right]-r' + K_2(\alpha)(V - r'')\right)}{2\left[2V^2 + 6B\left(0.005\mathbbm{E}_F\left[||\epsilon||_{\dagger}\right]-r' + K_2(\alpha)(V - r'')\right)\right]}\right)\\\times\exp\left(-\frac{ rem}{2V^2 + 6B\left(0.005\mathbbm{E}_F\left[||\epsilon||_{\dagger}\right]-r' + K_2(\alpha)(V - r'')\right)}\right) \\
= \exp\left(-\frac{ K_2(\alpha)^2}{2}\right)\exp\left(\frac{ K_2(\alpha)^26B\left(0.005\mathbbm{E}_F\left[||\epsilon||_{\dagger}\right]-r' + K_2(\alpha)(V - r'')\right) - 2rem}{2\left[2V^2 + 6B\left(0.005\mathbbm{E}_F\left[||\epsilon||_{\dagger}\right]-r' + K_2(\alpha)(V - r'')\right)\right]}\right) 
\end{align*}
where 
\begin{align*}
rem = \left[(0.005\mathbbm{E}_F\left[||\epsilon||_{\dagger}\right] - r')^2 + K_2(\alpha)^2(r'')^2 - 2K_2(\alpha)^2Vr'' + 2(0.005\mathbbm{E}_F\left[||\epsilon||_{\dagger}\right]-r')K_2(\alpha)(V - r'')\right].
\end{align*}
 The first term $\exp\left(-\frac{ K_2(\alpha)^2}{2}\right) = \alpha$. In Steps 2 and 3 below we show that the second term satisfies 
\begin{align*}
\limsup_{F \in \mathcal{F}: N_1,N_2 \to \infty}\exp\left(\frac{ K_2(\alpha)^26B\left(0.005\mathbbm{E}_F\left[||\epsilon||_{\dagger}\right]-r' + K_2(\alpha)(V - r'')\right) - 2rem}{2\left[2V^2 + 6B\left(0.005\mathbbm{E}_F\left[||\epsilon||_{\dagger}\right]-r' + K_2(\alpha)(V - r'')\right)\right]}\right)   \leq 1.
\end{align*}
It follows that $\limsup_{F \in \mathcal{F}: N_1,N_2 \to \infty}\mathbbm{P}_F\left(\cup_{(g_1,g_2) \in \mathcal{G}_c}C(g_1,g_2)\right)\leq \alpha.$ \newline
 
 \textbf{Step 2:} In this step we show that
\begin{align*}
\limsup_{F \in \mathcal{F}: N_1,N_2 \to \infty}\exp\left(\frac{ K_2(\alpha)^26B\left(0.005\mathbbm{E}_F\left[||\epsilon||_{\dagger}\right]-r' + K_2(\alpha)(V - r'')\right) - 2rem}{2\left[2V^2 + 6B\left(0.005\mathbbm{E}_F\left[||\epsilon||_{\dagger}\right]-r' + K_2(\alpha)(V - r'')\right)\right]}\right)   \leq 1.
\end{align*}
To show this, we write
\begin{align}\label{p2s2}
\frac{ K_2(\alpha)^26B\left(0.005\mathbbm{E}_F\left[||\epsilon||_{\dagger}\right]-r' + K_2(\alpha)(V - r'')\right) - 2rem}{2\left[2V^2 + 6B\left(0.005\mathbbm{E}_F\left[||\epsilon||_{\dagger}\right]-r' + K_2(\alpha)(V - r'')\right)\right]} \nonumber\\
= \frac{K_2(\alpha)^26B\left(\mathbbm{E}_F\left[||\epsilon||_{\dagger}\right](0.005-1.01r) -0.25||\sigma||_F r + K_2(\alpha)V(1 - r)\right)/\mathbbm{E}_F\left[||\epsilon||_{\dagger}\right]^2 - 2rem/\mathbbm{E}_F\left[||\epsilon||_{\dagger}\right]^2}{2\left[2V^2 + 6B\left(\mathbbm{E}_F\left[||\epsilon||_{\dagger}\right]\left(0.005-1.01r)-0.25||\sigma||_F r + K_2(\alpha)V(1 - r)\right)\right)\right]/\mathbbm{E}_F\left[||\epsilon||_{\dagger}\right]^2}.
\end{align}
Since $r \to 0$, the denominator 
\begin{align*}
2\left[2V^2 + 6B\left(\mathbbm{E}_F\left[||\epsilon||_{\dagger}\right]\left(0.005-1.01r)-0.25||\sigma||_F r + K_2(\alpha)V(1 - r)\right)\right)\right]/\mathbbm{E}_F\left[||\epsilon||_{\dagger}\right]^2
\end{align*}
 is eventually positive. The first summand in the numerator, 
\begin{align*} 
K_2(\alpha)^26B\left(\mathbbm{E}_F\left[||\epsilon||_{\dagger}\right](0.005-1.01r) - .25||\sigma||_F r + K_2(\alpha)V(1 - r)\right)/\mathbbm{E}_F\left[||\epsilon||_{\dagger}\right]^2
\end{align*}
 converges to $0$ as $N_1,N_2 \to \infty$ for any asymptotic sequence in $\mathcal{F}$, because $\liminf_{F \in \mathcal{F}: N_1,N_2\to \infty}\mathbbm{E}_F\left[||\epsilon||_{\dagger}\right] = \infty$ and $\limsup_{F \in \mathcal{F}: N_1,N_2\to \infty}V/\mathbbm{E}_F\left[||\epsilon||_{\dagger}\right] \leq 1$ (we show this in Step 3 below). The second summand in the numerator, $-2rem/\mathbbm{E}_F\left[||\epsilon||_{\dagger}\right]^2$ is equal to
\begin{align*}
-2\left[(0.005\mathbbm{E}_F\left[||\epsilon||_{\dagger}\right] - r')^2 + K_2(\alpha)^2(r'')^2 - 2K_2(\alpha)^2Vr'' \right. \\ \left.+ 2(0.005\mathbbm{E}_F\left[||\epsilon||_{\dagger}\right]-r')K_2(\alpha)(V - r'')\right]/\mathbbm{E}_F\left[||\epsilon||_{\dagger}\right]^2
\end{align*}
The first summand inside the square bracket $(0.005\mathbbm{E}_F\left[||\epsilon||_{\dagger}\right] - r')^2/\mathbbm{E}_F\left[||\epsilon||_{\dagger}\right]^2$  converges to $0.005^2$ as $N_1,N_2 \to \infty$ for any asymptotic sequence in $\mathcal{F}$. The second summand inside the square bracket $K_2(\alpha)^2r^2V^2/\mathbbm{E}_F\left[||\epsilon||_{\dagger}\right]^2 - 2K_2(\alpha)^2rV^2/\mathbbm{E}_F\left[||\epsilon||_{\dagger}\right]^2$ converges to $0$ as $N_1,N_2 \to \infty$ for any asymptotic sequence in $\mathcal{F}$ since $r \to 0$ and $\limsup_{F \in \mathcal{F}: N_1,N_2\to \infty}V/\mathbbm{E}_F\left[||\epsilon||_{\dagger}\right] \leq 1$ (we show this in Step 3 below). The last summand inside the square bracket $2(0.005\mathbbm{E}_F\left[||\epsilon||_{\dagger}\right]-r')K_2(\alpha)(V - r'')/\mathbbm{E}_F\left[||\epsilon||_{\dagger}\right]^2$ is eventually positive also because $r \to 0$ and \\$\limsup_{F \in \mathcal{F}: N_1,N_2\to \infty}V/\mathbbm{E}_F\left[||\epsilon||_{\dagger}\right] \leq 1$ (we show this in Step 3 below). 

It follows that the numerator of (\ref{p2s2}) is eventually negative and the denominator is eventually positive as $N_1,N_2 \to \infty$ for any asymptotic sequence in $\mathcal{F}$. As a result, (\ref{p2s2}) is eventually negative, and so $\limsup_{F \in \mathcal{F}: N_1,N_2 \to \infty}\exp\left(\frac{ K_2(\alpha)^26B\left(0.005\mathbbm{E}_F\left[||\epsilon||_{\dagger}\right]-r' + K_2(\alpha)(V - r'')\right) - 2rem}{2\left[2V^2 + 6B\left(0.005\mathbbm{E}_F\left[||\epsilon||_{\dagger}\right]-r' + K_2(\alpha)(V - r'')\right)\right]}\right)   \leq 1$.\newline
 
 \textbf{Step 3:} In this step we show that  $\liminf_{F \in \mathcal{F}: N_1,N_2\to \infty}\mathbbm{E}_F\left[||\epsilon||_{\dagger}\right] \to \infty$ and \\$\limsup_{F \in \mathcal{F}: N_1,N_2\to \infty}V/\mathbbm{E}_F\left[||\epsilon||_{\dagger}\right] \leq 1$. To show the first claim, we use Lemma~\ref{lem:residual-norm-relations} of Appendix Section~\ref{app:additional-lemmas} which says that, under Assumption~\ref{ass:model}, for any fixed  $F \in \mathcal{F}$ and $\delta \in (0,1)$
\begin{align*}
\mathbbm{E}_F\left[||\epsilon||_{\dagger}\right] \geq 2\delta ||\sigma||_F\left(1 - \exp\left(-\frac{3(1-\delta^2)^2||\sigma||_F^2}{(32-8\delta^2)B^2}\right)\right).
\end{align*}
Since $\liminf_{F \in \mathcal{F}: N_1,N_2\to \infty}||\sigma||_F \to \infty$ by Assumption~\ref{ass:nondegeneracy}(w), it follows that $\liminf_{F \in \mathcal{F}: N_1,N_2\to \infty}\mathbbm{E}_F\left[||\epsilon||_{\dagger}\right] = \infty$. To show the second claim we write 
\begin{align*}
V/\mathbbm{E}_F\left[||\epsilon||_{\dagger}\right] = \sqrt{||\sigma||_F^2 + 4B\mathbbm{E}_F\left[||\epsilon||_{\dagger}\right] + B||\sigma||_F}/\mathbbm{E}_F\left[||\epsilon||_{\dagger}\right] \\
\leq ||\sigma||_F/\mathbbm{E}_F\left[||\epsilon||_{\dagger}\right] + 2\sqrt{B}/\sqrt{\mathbbm{E}_F\left[||\epsilon||_{\dagger}\right]} + \sqrt{B}\sqrt{||\sigma||_F}/\mathbbm{E}_F\left[||\epsilon||_{\dagger}\right].
\end{align*}
The second summand satisfies $\limsup_{F \in \mathcal{F}: N_1,N_2\to \infty}2\sqrt{B}/\sqrt{\mathbbm{E}_F\left[||\epsilon||_{\dagger}\right]} = 0$ because $B$ is a fixed constant and we demonstrated previously that $\liminf_{F \in \mathcal{F}: N_1,N_2\to \infty}\mathbbm{E}_F\left[||\epsilon||_{\dagger}\right] = \infty$. Applying the inequality from Lemma~\ref{lem:residual-norm-relations} implies that the first summand is bounded by 
\begin{align*}
||\sigma||_F/\mathbbm{E}_F\left[||\epsilon||_{\dagger}\right] \leq \delta^{-1}\left(1 - \exp\left(-\frac{3(1-\delta^2)^2||\sigma||_F^2}{(32-8\delta^2)B^2}\right)\right)^{-1}/2
\end{align*}
Choosing $\delta > 1/2$, and since $\liminf_{F \in \mathcal{F}: N_1,N_2\to \infty}||\sigma||_F = \infty$ by Assumption~\ref{ass:nondegeneracy}(w), it follows that the right-hand side is eventually less than $1$ for any asymptotic sequence in $\mathcal{F}$. The third summand satisfies  $\limsup_{F \in \mathcal{F}: N_1,N_2\to \infty}\sqrt{B}\sqrt{||\sigma||_F}/\mathbbm{E}_F\left[||\epsilon||_{\dagger}\right] = 0$ because we demonstrated before that $||\sigma||_F/\mathbbm{E}_F\left[||\epsilon||_{\dagger}\right]$ is bounded and $\liminf_{F \in \mathcal{F}: N_1,N_2\to \infty}||\sigma||_F = \infty$ by Assumption~\ref{ass:nondegeneracy}(w).  \newline

\textbf{Step 4:} The result $\lim_{F \in \mathcal{F}: N_1,N_2\to\infty}\mathbbm{P}_F\left(\cup_{(g_1,g_2) \in \mathcal{G}_c}B(g_1,g_2)^C\right) = 0$ follows directly from Assumptions~\ref{ass:scale-consistency}(ii) and \ref{ass:scale-consistency}(iii) because the event $\cup _{(g_1,g_2) \in \mathcal{G}_c}B(g_1,g_2)^C$ implies the event $\left\{\frac{\left|\hat{V}- V\right|}{V} > r\right\} \cup \left\{\frac{\left|\hat{\bar{\tau}}- \bar{\tau}\right|}{\bar{\tau}} > r\right\} $ by the definition of $r'$ and $r''$. 
\end{proof}

\propthree*

\begin{proof}
Fix $\alpha \in (0,1)$ and let $K = K(N_1,N_2,\alpha)$ and
$\hat{\sigma}(g_1,g_2)$ be as in the statement of the proposition. Since $K$
is deterministic and the hypothesis
$\limsup_{N_1,N_2\to\infty} K/\sqrt{N_1+N_2} < 1/\sqrt{8\pi}$ is
strict, there exist $s \in (0,1)$ and $n_0 \in \mathbbm{N}$ such that
\begin{align}\label{eq:Kmargin}
K \;\leq\; (1-s)\,\frac{\sqrt{N_1 + N_2}}{\sqrt{8\pi}}
\qquad\text{whenever } \min(N_1,N_2) \geq n_0.
\end{align}

Fix $(g_1^*,g_2^*) \in \text{argmax}_{(g_1,g_2) \in \mathcal{G}_c}\left|\sum_{i \in [N_1], j \in [N_2]}\epsilon_{ij}g_{i,1}g_{j,2}\right|$, using an arbitrary deterministic tie-breaking rule, and let $m_{t}^* = \sum_{i \in [N_t]}g_{i,t}^*$ for $t \in \{1,2\}$. We start with the lower bound
\begin{align}\label{lower}
\left|\hat{\theta}(g_1^*,g_2^*) - \theta(g_1^*,g_2^*)\right| = \left|\frac{1}{m_1^*m_2^*}\sum_{i \in [N_1], j \in [N_2]}\epsilon_{ij}g_{i,1}^*g_{j,2}^*\right| \nonumber\\
\geq \frac{1}{2\sqrt{m_1^*m_2^*}}\left[\sqrt{N_1}\min\left(\frac{1}{N_1}\sum_{i \in [N_1]}(u_{i})_{+},\frac{1}{N_1}\sum_{i \in [N_1]}(u_{i})_{-}\right) \right.\nonumber \\ \left.+ \sqrt{N_2}\min\left(\frac{1}{N_2}\sum_{j \in [N_2]}(v_{j})_{+},\frac{1}{N_2}\sum_{j \in [N_2]}(v_{j})_{-}\right)\right] 
\end{align}
where for any scalar $x$, $(x)_+ := x\mathbbm{1}(x > 0)$, $(x)_- := (-x)\mathbbm{1}(x < 0)$, $u_i := \frac{1}{\sqrt{N_2}}\sum_{j \in [N_2]}\epsilon_{ij}$, $v_j := \frac{1}{\sqrt{N_1}}\sum_{i \in [N_1]}\epsilon_{ij}$, and the second inequality in (\ref{lower}) is derived in Step 1 below.

Let $\mathcal F_0:=\{F_\ell:\ell\geq1\}$ denote the asymptotic
sequence specified in the proposition. Define $\mu_0:=\int y\,dF_0(y)$ and $\sigma_0^2:=\int (y-\mu_0)^2\,dF_0(y)>0$. For every $F_\ell\in\mathcal F_0$, all entries have mean $\mu_0$
and variance $\sigma_0^2$. Hence
$\sigma(g_1,g_2)=\sigma_0$ for every
$(g_1,g_2)\in\mathcal G_c$. Assumption~\ref{ass:scale-consistency}(i)
therefore implies
\[
\frac{\widehat\sigma(g_1^*,g_2^*)}{\sigma_0}\to_p1
\]
along this sequence, even though $(g_1^*,g_2^*)$ is stochastic.

 For any asymptotic sequence of models $F \in \mathcal{F}_0$ with $N_1, N_2 \to \infty$, we show, in Step 2 below, that, for any fixed $t > 0$,
 \begin{align*}
 \mathbbm{P}_F\left(\left|\frac{1}{N_1}\sum_{i \in [N_1]}(u_{i})_{+} - \hat{\sigma}(g_1^*,g_2^*)/\sqrt{2\pi}\right| > t\hat{\sigma}(g_1^*,g_2^*)\right) \to 0,\\ 
 \mathbbm{P}_F\left(\left|\frac{1}{N_1}\sum_{i \in [N_1]}(u_{i})_{-} - \hat{\sigma}(g_1^*,g_2^*)/\sqrt{2\pi}\right| > t\hat{\sigma}(g_1^*,g_2^*)\right) \to 0,\\ 
 \mathbbm{P}_F\left(\left|\frac{1}{N_2}\sum_{j \in [N_2]}(v_{j})_{+} - \hat{\sigma}(g_1^*,g_2^*)/\sqrt{2\pi}\right| > t\hat{\sigma}(g_1^*,g_2^*)\right) \to 0, \text{ and } \\ 
 \mathbbm{P}_F\left(\left|\frac{1}{N_2}\sum_{j \in [N_2]}(v_{j})_{-} - \hat{\sigma}(g_1^*,g_2^*)/\sqrt{2\pi}\right| > t\hat{\sigma}(g_1^*,g_2^*)\right) \to 0.
 \end{align*} 

Applying the four convergence statements above with $t = (s/2)/\sqrt{2\pi}$,
each of the four averages exceeds
$(1-s/2)\,\hat{\sigma}(g_1^*,g_2^*)/\sqrt{2\pi}$ with probability approaching
one. Moreover, $\hat{\sigma}(g_1^*,g_2^*) > 0$ with probability approaching
one, since $\hat{\sigma}(g_1^*,g_2^*)/\sigma_0 \to_p 1$ and $\sigma_0 > 0$.
Combining these facts with \eqref{lower} and noting that
$2\sqrt{2\pi} = \sqrt{8\pi}$, it follows that, with probability approaching
one along any asymptotic sequence of models in $\mathcal{F}_0$ with
$N_1,N_2 \to \infty$,
\begin{align*}
\left|\hat{\theta}(g_1^*,g_2^*) - \theta(g_1^*,g_2^*)\right|
\;\geq\; (1-s/2)\,\frac{\sqrt{N_1}+\sqrt{N_2}}{\sqrt{8\pi}}\,
\frac{\hat{\sigma}(g_1^*,g_2^*)}{\sqrt{m_1^*m_2^*}}
\\\;>\; (1-s)\,\frac{\sqrt{N_1 + N_2}}{\sqrt{8\pi}}\,
\frac{\hat{\sigma}(g_1^*,g_2^*)}{\sqrt{m_1^*m_2^*}}
\;\geq\; K\,\frac{\hat{\sigma}(g_1^*,g_2^*)}{\sqrt{m_1^*m_2^*}},
\end{align*}
where the strict inequality uses $\hat{\sigma}(g_1^*,g_2^*) > 0$ and $\sqrt{N_1}+\sqrt{N_2} > \sqrt{N_1+N_2}$, and the
final inequality is \eqref{eq:Kmargin}. Since
$I(g_1^*,g_2^*;\alpha) = \hat{\theta}(g_1^*,g_2^*) \pm
\left[K \times \hat{\sigma}(g_1^*,g_2^*)\right]/\sqrt{m_1^*m_2^*}$, the event
$\{\theta(g_1^*,g_2^*) \in I(g_1^*,g_2^*;\alpha)\}$ implies
$|\hat{\theta}(g_1^*,g_2^*) - \theta(g_1^*,g_2^*)| \leq
K\hat{\sigma}(g_1^*,g_2^*)/\sqrt{m_1^*m_2^*}$, and therefore
\begin{align*}
\lim_{F \in \mathcal{F}_0:\,N_1,N_2 \to \infty}
\mathbbm{P}_F\left(\theta(g_1^*,g_2^*) \in I(g_1^*,g_2^*;\alpha)\right) = 0.
\end{align*}
Finally, since $(g_1^*,g_2^*) \in \mathcal{G}_c$ implies
$\cap_{(g_1,g_2)\in\mathcal{G}_c}\{\theta(g_1,g_2) \in I(g_1,g_2;\alpha)\}
\subseteq \{\theta(g_1^*,g_2^*) \in I(g_1^*,g_2^*;\alpha)\}$, and since
$\mathcal{F}_0 \subseteq \mathcal{F}$,
\begin{align*}
\liminf_{F\in\mathcal{F}:\,N_1,N_2\to\infty}
\mathbbm{P}_F\left(\cap_{(g_1,g_2)\in\mathcal{G}_c}
\left\{\theta(g_1,g_2)\in I(g_1,g_2;\alpha)\right\}\right) = 0,
\end{align*}
which is the conclusion of the proposition. \newline

\textbf{Step 1:} In this step we show the inequality (\ref{lower}). By definition of $(g_1^*,g_2^*)$, we have that $\left|\sum_{i \in [N_1], j \in [N_2]}\epsilon_{ij}g^*_{i,1}g^*_{j,2}\right| \geq \left|\sum_{i \in [N_1], j \in [N_2]}\epsilon_{ij}g'_{i,1}g'_{j,2}\right|$ for any $(g'_1,g'_2) \in \mathcal{G}_c$. Choosing $g'_{j,2} = 1$ for all $j \in [N_2]$ and $g'_{i,1} = \mathbbm{1}\{\sum_{j \in [N_2]}\epsilon_{ij} \geq 0\}\mathbbm{1}\{\sum_{i \in [N_1]}\mathbbm{1}\{\sum_{j \in [N_2]}\epsilon_{ij} \geq 0\} \geq N_1/2\} + \mathbbm{1}\{\sum_{j \in [N_2]}\epsilon_{ij} \leq 0\}\mathbbm{1}\{\sum_{i \in [N_1]}\mathbbm{1}\{\sum_{j \in [N_2]}\epsilon_{ij} \geq 0\} < N_1/2\}$. By construction, the entries of $g_1'$ and $g_2'$ take values in $\{0,1\}$, $\sum_{j \in [N_2]}g'_{j,2} = N_2$ and $\sum_{i \in [N_1]}g'_{i,1} \geq N_1/2$, and so $(g_1',g_2') \in \mathcal{G}_c$ for any $c \in [0,1/2]$. Under this choice of $(g_1',g_2')$ we have that
\begin{align*}
 \left|\sum_{i \in [N_1], j \in [N_2]}\epsilon_{ij}g'_{i,1}g'_{j,2}\right| = \sum_{i \in [N_1]}\left(\sum_{j \in [N_2]}\epsilon_{ij}\right)_{+}\mathbbm{1}\{\sum_{i \in [N_1]}\mathbbm{1}\{\sum_{j \in [N_2]}\epsilon_{ij} \geq 0\} \geq N_1/2\}  \\
 + \sum_{i \in [N_1]}\left(\sum_{j \in [N_2]}\epsilon_{ij}\right)_{-}\mathbbm{1}\{\sum_{i \in [N_1]}\mathbbm{1}\{\sum_{j \in [N_2]}\epsilon_{ij} \geq 0\} < N_1/2\} \\
 \geq \min\left(\sum_{i \in [N_1]}\left(\sum_{j \in [N_2]}\epsilon_{ij}\right)_{+},\sum_{i \in [N_1]}\left(\sum_{j \in [N_2]}\epsilon_{ij}\right)_{-}\right).
\end{align*}
By the same logic, choosing $g'_{i,1} = 1$ for all $i \in [N_1]$ and $g'_{j,2} = \mathbbm{1}\{\sum_{i \in [N_1]}\epsilon_{ij} \geq 0\}\mathbbm{1}\{\sum_{j\in [N_2]}\mathbbm{1}\{\sum_{i \in [N_1]}\epsilon_{ij} \geq 0\} \geq N_2/2\} + \mathbbm{1}\{\sum_{i \in [N_1]}\epsilon_{ij} \leq 0\}\mathbbm{1}\{\sum_{j \in [N_2]}\mathbbm{1}\{\sum_{i \in [N_1]}\epsilon_{ij} \geq 0\} < N_2/2\}$ gives 
\begin{align*}
 \left|\sum_{i \in [N_1], j \in [N_2]}\epsilon_{ij}g'_{i,1}g'_{j,2}\right| 
 \geq \min\left(\sum_{j \in [N_2]}\left(\sum_{i \in [N_1]}\epsilon_{ij}\right)_{+},\sum_{j \in [N_2]}\left(\sum_{i \in [N_1]}\epsilon_{ij}\right)_{-}\right).
\end{align*}
It follows from these two inequalities that 
\begin{align*}
\left|\frac{1}{m_1^*m_2^*}\sum_{i \in [N_1], j \in [N_2]}\epsilon_{ij}g^*_{i,1}g^*_{j,2}\right|  \geq \frac{1}{2m_1^*m_2^*}\left[\min\left( \sum_{i \in [N_1]}\left(\sum_{j \in [N_2]}\epsilon_{ij}\right)_{+}, \sum_{i \in [N_1]}\left(\sum_{j \in [N_2]}\epsilon_{ij}\right)_{-}\right) \right. \\ \left.+  \min\left(\sum_{j \in [N_2]}\left(\sum_{i \in [N_1]}\epsilon_{ij}\right)_{+},\sum_{j \in [N_2]}\left(\sum_{i \in [N_1]}\epsilon_{ij}\right)_{-}\right) \right] \\
\geq \frac{1}{2\sqrt{m_1^*m_2^*}\sqrt{N_1N_2}}\left[ \min\left(\sum_{i \in [N_1]}\left(\sum_{j \in [N_2]}\epsilon_{ij}\right)_{+},\sum_{i \in [N_1]}\left(\sum_{j \in [N_2]}\epsilon_{ij}\right)_{-}\right) \right. \\ \left.+  \min\left(\sum_{j \in [N_2]}\left(\sum_{i \in [N_1]}\epsilon_{ij}\right)_{+},\sum_{j \in [N_2]}\left(\sum_{i \in [N_1]}\epsilon_{ij}\right)_{-}\right) \right] 
\end{align*}
where the second inequality follows from the fact that $m_1^*m_2^* \leq N_1N_2$. The inequality (\ref{lower}) then follows by distributing the $1/\sqrt{N_1N_2}$ into the square brackets on the right-hand side. \newline

\textbf{Step 2:} In this step we show that, for any fixed $t > 0$ and any asymptotic sequence of models $F \in \mathcal{F}_0$ with $N_1, N_2 \to \infty$
\begin{align}\label{conv}
\mathbbm{P}_F\left(\left|\frac{1}{N_1}\sum_{i \in [N_1]}(u_{i})_{+} - \hat{\sigma}(g_1^*,g_2^*)/\sqrt{2\pi}\right| > t\hat{\sigma}(g_1^*,g_2^*)\right) \to 0.
\end{align}
as $N_1, N_2 \to \infty$. A nearly identical argument gives, mutatis mutandis, that 
\begin{align*}
\mathbbm{P}_F\left(\left|\frac{1}{N_1}\sum_{i \in [N_1]}(u_{i})_{-} - \hat{\sigma}(g_1^*,g_2^*)/\sqrt{2\pi}\right| > t\hat{\sigma}(g_1^*,g_2^*)\right) \to 0,\\ 
\mathbbm{P}_F\left(\left|\frac{1}{N_2}\sum_{j \in [N_2]}(v_{j})_{+} - \hat{\sigma}(g_1^*,g_2^*)/\sqrt{2\pi}\right| > t\hat{\sigma}(g_1^*,g_2^*)\right) \to 0, \text{ and }  \\
\mathbbm{P}_F\left(\left|\frac{1}{N_2}\sum_{j \in [N_2]}(v_{j})_{-} - \hat{\sigma}(g_1^*,g_2^*)/\sqrt{2\pi}\right| > t\hat{\sigma}(g_1^*,g_2^*)\right) \to 0.
\end{align*}  

To show (\ref{conv}), we write 
\begin{align*}
\mathbbm{P}_F\left(\left|\frac{1}{N_1}\sum_{i \in [N_1]}(u_{i})_{+} - \hat{\sigma}(g_1^*,g_2^*)/\sqrt{2\pi}\right| > t\hat{\sigma}(g_1^*,g_2^*)\right) = \mathbbm{P}_F\left(\left|\frac{1}{N_1}\sum_{i \in [N_1]}\left(\frac{u_{i}}{\hat{\sigma}(g_1^*,g_2^*)}\right)_{+} - \frac{1}{\sqrt{2\pi}}\right| > t\right) \\
=  \mathbbm{P}_F\left(\left|\frac{1}{N_1}\sum_{i \in [N_1]}\left(\frac{u_{i}}{\sigma_0}\right)_{+}\left[1 - \frac{\frac{\hat{\sigma}(g_1^*,g_2^*) - \sigma_0}{\sigma_0}}{1 + \frac{\hat{\sigma}(g_1^*,g_2^*) - \sigma_0}{\sigma_0}}\right] - \frac{1}{\sqrt{2\pi}}\right| > t\right) \\
= \mathbbm{P}_F\left(\left|A + B - C\right| > t\right) \leq \mathbbm{1}\left(|A| > t/3\right)+\mathbbm{P}_F\left(|B| > t/3\right)+ \mathbbm{P}_F\left(|C| > t/3\right)
\end{align*}
where $\sigma_0 = \sqrt{\mathbbm{E}_F\left[\epsilon_{ij}^2\right]} > 0$ since $\mathcal{F}_0$ is chosen to have nondegenerate elements, $A := \left(\mathbbm{E}_F\left[\left(\frac{u_i}{\sigma_0}\right)_{+}\right] - \frac{1}{\sqrt{2\pi}}\right)$, $B := \left(\frac{1}{N_1}\sum_{i \in [N_1]}\left(\frac{u_i}{\sigma_0}\right)_{+} - \mathbbm{E}_F\left[\left(\frac{u_i}{\sigma_0}\right)_{+}\right]\right)\left[1 - \frac{\frac{\hat{\sigma}(g_1^*,g_2^*) - \sigma_0}{\sigma_0}}{1 + \frac{\hat{\sigma}(g_1^*,g_2^*) - \sigma_0}{\sigma_0}}\right]$, $C := \mathbbm{E}_F\left[\left(\frac{u_i}{\sigma_0}\right)_{+}\right]\left[\frac{\frac{\hat{\sigma}(g_1^*,g_2^*) - \sigma_0}{\sigma_0}}{1 + \frac{\hat{\sigma}(g_1^*,g_2^*) - \sigma_0}{\sigma_0}}\right]$, and the inequality follows from the union bound. In Step 3 below we show that, for any $t > 0$ and as $N_1,N_2 \to \infty$,  $\mathbbm{P}_F\left(|B| > t/3\right)\to 0$, in Step 4 below we show that $|A| \to 0$, and in Step 5 below we show that $\mathbbm{P}_F\left(|C| > t/3\right) \to 0$. The result (\ref{conv}) follows. \newline

\textbf{Step 3:} In this step we show that, for any fixed $t > 0$ and asymptotic sequence of models $F \in \mathcal{F}_0$ with $N_1, N_2 \to \infty$,  
\begin{align*}
\mathbbm{P}_F\left(\left| \left(\frac{1}{N_1}\sum_{i \in [N_1]}\left(\frac{u_i}{\sigma_0}\right)_{+} - \mathbbm{E}_F\left[\left(\frac{u_i}{\sigma_0}\right)_{+}\right]\right)\left[1 - \frac{\frac{\hat{\sigma}(g_1^*,g_2^*) - \sigma_0}{\sigma_0}}{1 + \frac{\hat{\sigma}(g_1^*,g_2^*) - \sigma_0}{\sigma_0}}\right]\right| > t/3 \right) \to 0
\end{align*}
as $N_1,N_2 \to \infty$.

Let $E_N := \left\{\left|\frac{\hat{\sigma}(g_1^*,g_2^*)-\sigma_0}{\sigma_0}\right| \leq \frac{1}{2}\right\}$. On the event $E_N$, $\left|1 -\frac{\frac{\hat{\sigma}(g_1^*,g_2^*)-\sigma_0}{\sigma_0}}{1+\frac{\hat{\sigma}(g_1^*,g_2^*)-\sigma_0}{\sigma_0}}\right| = \frac{\sigma_0}{\hat{\sigma}(g_1^*,g_2^*)}\leq2$ and so
\begin{align*}
\mathbbm{P}_F\left(\left| \left(\frac{1}{N_1}\sum_{i \in [N_1]}\left(\frac{u_i}{\sigma_0}\right)_{+} - \mathbbm{E}_F\left[\left(\frac{u_i}{\sigma_0}\right)_{+}\right]\right)\left[1 - \frac{\frac{\hat{\sigma}(g_1^*,g_2^*) - \sigma_0}{\sigma_0}}{1 + \frac{\hat{\sigma}(g_1^*,g_2^*) - \sigma_0}{\sigma_0}}\right]\right| > t/3 \right) \\
\leq
\mathbbm{P}_F\left(\left| \frac{1}{N_1}\sum_{i \in [N_1]}\left(\frac{u_i}{\sigma_0}\right)_{+} - \mathbbm{E}\left[\left(\frac{u_i}{\sigma_0}\right)_{+}\right]\right| > t/6 \right)+\mathbbm{P}_F(E_N^C).
\end{align*}
Since $\sigma(g_1^*,g_2^*)=\sigma_0$ under $\mathcal{F}_0$ and Assumption~\ref{ass:scale-consistency}(i) holds uniformly over $\mathcal{G}_c$, we have $\mathbbm{P}_F(E_N^C)\to0$. Since $\mathbbm{E}_F\left[(u_i)_+\right] \geq \mathbbm{E}_F\left[u_i\right] = 0$ and $\mathbbm{E}_F\left[(u_i)_+^2\right]  \leq \mathbbm{E}_F\left[(u_i)^2\right] = \sigma_0^2$, it follows that $\text{Var}_F((u_i)_+) \leq \sigma_0^2$, $\text{Var}_F\left(\left(\frac{u_i}{\sigma_0}\right)_+\right) \leq 1$, and so $\text{Var}_F\left(\frac{1}{N_1}\sum_{i \in [N_1]}\left(\frac{u_i}{\sigma_0}\right)_{+}\right) \leq \frac{1}{N_1}$. It follows from Chebyshev's inequality that for any $t > 0$, 
\begin{align*}
\mathbbm{P}_F\left(\left|\frac{1}{N_1}\sum_{i \in [N_1]}\left(\frac{u_i}{\sigma_0}\right)_{+} - \mathbbm{E}_F\left[\left(\frac{u_i}{\sigma_0}\right)_{+}\right]\right| \geq t/6\right) \leq \frac{36}{t^2N_1}
\end{align*}
which converges to $0$ as $N_1,N_2 \to \infty$. This demonstrates the third step. \newline

\textbf{Step 4:} In this step we show that for any asymptotic sequence of models $F \in \mathcal{F}_0$ with $N_1, N_2 \to \infty$, $\left|\mathbbm{E}_F\left[\left(\frac{u_i}{\sigma_0}\right)_{+}\right]-\frac{1}{\sqrt{2\pi}}\right|\to 0$ as $N_1,N_2 \to \infty$. Since, under $\mathbbm{P}_F$,
\begin{align*}
\frac{u_i}{\sigma_0}
=
\frac{1}{\sqrt{N_2}}\sum_{j \in [N_2]}\frac{\epsilon_{ij}}{\sigma_0},
\end{align*}
$\frac{u_i}{\sigma_0}$ is a standardized sum of $N_2$ independent and identically distributed mean-zero random variables with variance $1$. Therefore, by the Lindeberg-Levy Central Limit Theorem \cite[see, for example, Theorem 1.9.1 A in][]{serfling2009approximation}, $\frac{u_i}{\sigma_0}\overset{d}{\to}Z$, where $Z$ is a standard normal random variable. Since the map $x \mapsto x_+$ is continuous, the continuous mapping theorem gives $\left(\frac{u_i}{\sigma_0}\right)_+\overset{d}{\to}Z_+.$ Moreover, $\mathbbm{E}_F\left[\left(\left(\frac{u_i}{\sigma_0}\right)_+\right)^2\right]\leq\mathbbm{E}_F\left[\left(\frac{u_i}{\sigma_0}\right)^2\right]=1$, so the sequence $\left(\frac{u_i}{\sigma_0}\right)_+$ is uniformly integrable. It follows that $\mathbbm{E}_F\left[\left(\frac{u_i}{\sigma_0}\right)_+\right]\to\mathbbm{E}\left[Z_+\right]$, where 
\begin{align*}
\mathbbm{E}\left[Z_+\right]
=\int_0^\infty x\frac{1}{\sqrt{2\pi}}\exp(-x^2/2)\,dx 
=\frac{1}{\sqrt{2\pi}}.
\end{align*}
It follows that $\left|\mathbbm{E}_F\left[\left(\frac{u_i}{\sigma_0}\right)_+\right]-\frac{1}{\sqrt{2\pi}}\right|\to 0$.\newline

\textbf{Step 5:} In this step we show that, for any fixed $t > 0$ and asymptotic sequence of models $F \in \mathcal{F}_0$ with $N_1, N_2 \to \infty$,
\begin{align*}
\mathbbm{P}_F\left(\left|\mathbbm{E}_F\left[\left(\frac{u_i}{\sigma_0}\right)_{+}\right]\left[\frac{\frac{\hat{\sigma}(g_1^*,g_2^*) - \sigma_0}{\sigma_0}}{1 + \frac{\hat{\sigma}(g_1^*,g_2^*) - \sigma_0}{\sigma_0}}\right]\right| > t/3\right) \to 0
\end{align*}
as $N_1,N_2 \to \infty$. To show this, we define $\Delta_N := \frac{\hat{\sigma}(g_1^*,g_2^*) - \sigma_0}{\sigma_0}$. By Step 4, $\mathbbm{E}_F\left[\left(\frac{u_i}{\sigma_0}\right)_{+}\right] \to 1/\sqrt{2\pi}$, and hence $\mathbbm{E}_F\left[\left(\frac{u_i}{\sigma_0}\right)_{+}\right]$ is bounded for all sufficiently large $N_1,N_2$. Since the entries $\sigma_{ij}$ are all equal to $\sigma_0$ under $\mathcal{F}_0$, we have $\sigma(g_1^*,g_2^*) := \sqrt{\frac{1}{m_1^*m_2^*}\sum_{i \in [N_1], j \in [N_2]}\sigma_{ij}^2g_{i,1}^*g_{j,2}^*}= \sigma_0.$ Because Assumption~\ref{ass:scale-consistency}(i) holds uniformly over $\mathcal{G}_c$, it follows that
\begin{align*}
\Delta_N = \frac{\hat{\sigma}(g_1^*,g_2^*) - \sigma(g_1^*,g_2^*)}{\sigma(g_1^*,g_2^*)} = o_p(1).
\end{align*}

Now define the event $E_N:=\left\{|\Delta_N|\leq\frac{1}{2}\right\}$. Since $\Delta_N=o_p(1)$, we have $\mathbbm{P}_F(E_N^C)\to0$. On the event $E_N$, $\left|\frac{\Delta_N}{1+\Delta_N}\right|\leq2|\Delta_N|$, and so
\begin{align*}
\mathbbm{P}_F\left(\left|\mathbbm{E}_F\left[\left(\frac{u_i}{\sigma_0}\right)_{+}\right]\frac{\Delta_N}{1+\Delta_N}\right| > t/3\right)
\leq\mathbbm{P}_F\left(2\mathbbm{E}_F\left[\left(\frac{u_i}{\sigma_0}\right)_{+}\right]|\Delta_N| > t/3\right)+\mathbbm{P}_F(E_N^C).
\end{align*}
Since $\mathbbm{E}_F\left[\left(\frac{u_i}{\sigma_0}\right)_{+}\right]$ is eventually bounded and $\Delta_N=o_p(1)$, the first term on the right-hand side converges to zero. The second term also converges to zero. This demonstrates Step 5. 
\end{proof}

\propfour*

\begin{proof}
Fix $F \in \mathcal{F}$ and work under $\mathbbm P_F$. Let $(g_1^*,g_2^*) \in \text{argmax}_{(g_1,g_2)\in\mathcal{G}_c}
\left|\sum_{i\in[N_1],j\in[N_2]}\epsilon_{ij}g_{i,1}g_{j,2}\right|$ using an
arbitrary deterministic tie-breaking rule and define
$m_t^* := \sum_{i\in[N_t]}g_{i,t}^*$ for $t\in\{1,2\}$. By construction,
$|\hat{\theta}(g_1^*,g_2^*) - \theta(g_1^*,g_2^*)| =
||\epsilon||_{\square;c}/(m_1^*m_2^*)$. Let $r$ be the sequence of real
numbers in Assumption~\ref{ass:scale-consistency}(ii) and set $r' := r\bar{\tau}$, so that
$\bar{\tau} + r' = \bar{\tau}(1+r)$.

Since $(g_1^*,g_2^*)\in\mathcal{G}_c$,
\begin{align}\label{eq:red1}
\mathbbm{P}_F\left(\cap_{(g_1,g_2)\in\mathcal{G}_c}
\left\{\theta(g_1,g_2)\in I(g_1,g_2;\alpha)\right\}\right)
\leq
\mathbbm{P}_F\left(\theta(g_1^*,g_2^*)\in I(g_1^*,g_2^*;\alpha)\right).
\end{align}
On the event $\{\theta(g_1^*,g_2^*)\in I(g_1^*,g_2^*;\alpha)\}$, both
$\theta(g_1^*,g_2^*)$ and $\hat{\theta}(g_1^*,g_2^*)$ belong to
$I(g_1^*,g_2^*;\alpha)$ (since $\{\hat{\theta}(g_1^*,g_2^*)\in I(g_1^*,g_2^*;\alpha)\}$ by assumption)  and so $|\hat{\theta}(g_1^*,g_2^*)-\theta(g_1^*,g_2^*)| \leq
|I(g_1^*,g_2^*;\alpha)|$. For all $(N_1,N_2)$ with $\min(N_1,N_2) \geq n_0$ it follows
that, $\mathbbm{P}_F$-almost surely,
\begin{align*}
\left\{\theta(g_1^*,g_2^*)\in I(g_1^*,g_2^*;\alpha)\right\}
\subseteq
\left\{|\hat{\theta}(g_1^*,g_2^*)-\theta(g_1^*,g_2^*)|
\leq (1-\delta')\,c^*(\alpha)\,
\frac{\hat{\bar{\tau}}}{m_1^*m_2^*}\right\}
\end{align*}
and so
\begin{align}\label{eq:red2}
\mathbbm{P}_F\left(\theta(g_1^*,g_2^*)\in I(g_1^*,g_2^*;\alpha)\right)
\leq
\mathbbm{P}_F\left(|\hat{\theta}(g_1^*,g_2^*)-\theta(g_1^*,g_2^*)|
\leq (1-\delta')\,c^*(\alpha)\,\frac{\hat{\bar{\tau}}}{m_1^*m_2^*}\right)\\
\leq\;
\mathbbm{P}_F\left(|\hat{\theta}(g_1^*,g_2^*)-\theta(g_1^*,g_2^*)|
\leq (1-\delta')c^*(\alpha)\frac{\bar{\tau}+r'}{m_1^*m_2^*}\right)
+\mathbbm{P}_F\left(\hat{\bar{\tau}} > \bar{\tau}+r'\right).
\end{align}
In Step 6  below we show that
$\lim_{F\in\mathcal{F}:\,N_1,N_2\to\infty}
\mathbbm{P}_F\left(\hat{\bar{\tau}} > \bar{\tau}+r'\right) = 0$. In Steps 1-5 we show that $\limsup_{F\in\mathcal{F}:\,N_1,N_2\to\infty}
\mathbbm{P}_F\left(|\hat{\theta}(g_1^*,g_2^*)-\theta(g_1^*,g_2^*)|
\leq (1-\delta')c^*(\alpha)\frac{\bar{\tau}+r'}{m_1^*m_2^*}\right)
< 1-\alpha$. This demonstrates the claim.  \newline

\textbf{Step 1:}  In this step we show that 
\begin{align*}
\limsup_{F\in\mathcal{F}:\,N_1,N_2\to\infty}
\mathbbm{P}_F\left(|\hat{\theta}(g_1^*,g_2^*)-\theta(g_1^*,g_2^*)|
\leq (1-\delta')c^*(\alpha)\frac{\bar{\tau}+r'}{m_1^*m_2^*}\right)
< 1-\alpha. 
\end{align*}
To do this, we first substitute the Grothendieck
constant $K_G(\infty)$, which is not known, with Krivine's upper bound $\bar{K}_G := \frac{\pi}{2\ln(1+\sqrt{2})} \leq 1.7823$ from Lemma~\ref{lem:grothendieck}. Since $K_G(\infty) \leq \bar{K}_G$, the
lower-tail inequality \eqref{lb} of Lemma~\ref{lem:cut-concentration}
remains valid with $K_G(\infty)$ replaced by $\bar{K}_G$, because, for any
fixed $F \in \mathcal{F}$,
$(18\bar{K}_G)^{-1}\mathbbm{E}_F\left[||\epsilon||_{\dagger}\right] \leq
(18K_G(\infty))^{-1}\mathbbm{E}_F\left[||\epsilon||_{\dagger}\right]$ so the
event on the left-hand side of \eqref{lb} only shrinks. Similarly the
inequality $||X||_{1,2} \leq K_G(\infty)||X||_{\infty\to 1}$ from the proof of
Lemma~\ref{lem:dagger-infty} still holds with $\bar{K}_G$.

We then define the helper function $C_1(\cdot,\cdot)$.  For a target level $\beta \in (0,1)$ and a slack parameter
$\delta'' \in (0,1)$ define
\begin{align*}
A_\beta := 1.01\sqrt{648}\,\bar{K}_G\sqrt{-\ln(\beta)} + 0.25
\quad\text{and}\quad
C_1(\beta,\delta'') :=
\frac{1-\delta''}{9\bar{K}_G\left(A_\beta + 2.02(1-\delta'')\right)}.
\end{align*}
Substituting
$\bar{K}_G = \pi/(2\ln(1+\sqrt{2}))$ into the definition of $c^*(\alpha)$ yields
\begin{align}\label{eq:cstarid}
c^*(\alpha) = \frac{1}{9\bar{K}_G\left(A_{1-\alpha} + 2.02\right)} = C_1(1-\alpha,0). 
\end{align}

In Step 2 below we show that  there exist
$\beta' \in (0,1-\alpha)$ and $\delta'' \in (0,1)$ such that $(1-\delta')\,c^*(\alpha) \;\leq\; C_1(\beta',\delta'')$. In Steps 3--5 below we show that, for every $\beta\in(0,1)$ and
$\delta''\in(0,1)$,
\begin{align}\label{eq:P1}\tag{P1}
\limsup_{F\in\mathcal{F}:\,N_1,N_2\to\infty}
\mathbbm{P}_F\left(|\hat{\theta}(g_1^*,g_2^*)-\theta(g_1^*,g_2^*)|
\leq C_1(\beta,\delta'')\,\frac{\bar{\tau}+r'}{m_1^*m_2^*}\right)
\leq \beta.
\end{align}
Combined, these steps demonstrate the claim that 
\begin{align*}
\limsup_{F\in\mathcal{F}:\,N_1,N_2\to\infty}
\mathbbm{P}_F\left(|\hat{\theta}(g_1^*,g_2^*)-\theta(g_1^*,g_2^*)|
\leq (1-\delta')c^*(\alpha)\frac{\bar{\tau}+r'}{m_1^*m_2^*}\right)
< 1-\alpha.\\
\end{align*} 

\textbf{Step 2:} In this step we show that there exist
$\beta' \in (0,1-\alpha)$ and $\delta'' \in (0,1)$ such that
\begin{align}\label{eq:compare}
(1-\delta')\,c^*(\alpha) \;\leq\; C_1(\beta',\delta'').
\end{align}
To see this, take $\delta'' = \delta'/2$ and observe first that, at
$\beta = 1-\alpha$, the identity \eqref{eq:cstarid} and $1-\delta'/2 \leq 1$,
$A_{1-\alpha} + 2.02(1-\delta'/2) \leq A_{1-\alpha} + 2.02$ give
\begin{align*}
C_1(1-\alpha,\delta'/2)
= \frac{1-\delta'/2}{9\bar{K}_G\left(A_{1-\alpha}+2.02(1-\delta'/2)\right)}
&\;\geq\;
\frac{1-\delta'/2}{9\bar{K}_G\left(A_{1-\alpha}+2.02\right)} \\
&\;=\; (1-\delta'/2)\,c^*(\alpha)
\;>\; (1-\delta')\,c^*(\alpha),
\end{align*}
where the final inequality is strict because $\delta'/2 < \delta'$. Since
$\beta \mapsto A_\beta$ is continuous and decreasing on $(0,1)$, the map
$\beta \mapsto C_1(\beta,\delta'/2)$ is continuous and increasing, and
$C_1(\beta,\delta'/2) \uparrow C_1(1-\alpha,\delta'/2)$ as
$\beta \uparrow 1-\alpha$. The strict inequality in the previous display
therefore guarantees the existence of $\beta' \in (0,1-\alpha)$ satisfying
\eqref{eq:compare}. \newline

\textbf{Step 3:} In this step we show that \eqref{eq:P1} holds for every
$\beta \in (0,1)$ and $\delta'' \in (0,1)$. Throughout, fix $\beta$ and
$\delta''$, and recall that $\bar{\tau} + r' = \bar{\tau}(1+r)$. The scale
$\bar{\tau} = 1.01\mathbbm{E}_F\left[||\epsilon||_{\dagger}\right] +
0.25||\sigma||_F$ combines two terms, and the two lower bounds for
$||\epsilon||_{\square;c}$ used below, the Talagrand-type lower-tail
inequality of Lemma~\ref{lem:cut-concentration} in Step 4 and the
Frobenius-norm bound in Step 5, are effective in different regimes. We
therefore partition the set $\mathcal{F}$ into two parts
$\mathcal{F}_{\beta,\delta}$ and $\mathcal{F}^{c}_{\beta,\delta}$, with
$\mathcal{F} = \mathcal{F}_{\beta,\delta} \cup \mathcal{F}^{c}_{\beta,\delta}$,
\begin{align*}
\mathcal{F}_{\beta,\delta} := \left\{F \in \mathcal{F} :
\mathbbm{E}_F\left[||\epsilon||_{\dagger}\right] \geq
C_{\beta,\delta}||\sigma||_F \right\},
\qquad
\mathcal{F}^c_{\beta,\delta} := \left\{F \in \mathcal{F} :
\mathbbm{E}_F\left[||\epsilon||_{\dagger}\right] <
C_{\beta,\delta}||\sigma||_F \right\},
\end{align*}
where
\begin{align*}
C_{\beta,\delta} :=
\frac{\sqrt{648}\,\bar{K}_G\sqrt{-\ln(\beta)} + 0.25(1-\delta)}
{1.01\delta - 0.01}
\end{align*}
and $\delta \in (0.01/1.01,1)$ is a free parameter, governing the location of
the partition, that is chosen below as a function of $(\beta,\delta'')$. In
Step 4 below we show that, for any $\delta \in (0.01/1.01,1)$,
\begin{align}\label{eq:steptwo}
\limsup_{F \in \mathcal{F}_{\beta,\delta}:\,N_1,N_2 \to \infty}
\mathbbm{P}_F\left(|\hat{\theta}(g_1^*,g_2^*) - \theta(g_1^*,g_2^*)|
\leq \frac{\bar{\tau}(1+r)}{m_1^* m_2^*}\,
\frac{1-\delta}{18 \bar{K}_G} \right) \leq \beta,
\end{align}
and in Step 5 below we show that, for any $\delta \in (0.01/1.01,1)$ and
$\delta'' \in (0,1)$,
\begin{align}\label{eq:stepthree}
\limsup_{F \in \mathcal{F}^c_{\beta,\delta}:\,N_1,N_2 \to \infty}
\mathbbm{P}_F\left(|\hat{\theta}(g_1^*,g_2^*) - \theta(g_1^*,g_2^*)|
\leq \frac{\bar{\tau}(1+r)}{9\bar{K}_G\, m_1^* m_2^*}\,
\frac{(1.01\delta-0.01)(1-\delta'')}{A_\beta}\right) = 0.
\end{align}
Now choose $\delta$ so that the thresholds in \eqref{eq:steptwo} and
\eqref{eq:stepthree} coincide:
\begin{align*}
\delta = \delta_{\beta,\delta''} :=
\frac{A_\beta + 0.02(1-\delta'')}{A_\beta + 2.02(1-\delta'')}.
\end{align*}
This choice is admissible: $\delta_{\beta,\delta''} < 1$ because the
numerator is strictly smaller than the denominator, and
$\delta_{\beta,\delta''} \geq A_\beta/(A_\beta + 2.02) \geq 0.25/2.27 >
0.01/1.01$ because $A_\beta \geq 0.25$. Under this choice, direct calculation
gives
\begin{align*}
1-\delta_{\beta,\delta''} = \frac{2(1-\delta'')}{A_\beta + 2.02(1-\delta'')}
\quad\text{and}\quad
1.01\delta_{\beta,\delta''} - 0.01 =
\frac{A_\beta}{A_\beta + 2.02(1-\delta'')},
\end{align*}
and therefore the constants appearing in \eqref{eq:steptwo} and
\eqref{eq:stepthree} coincide:
\begin{align*}
\frac{1-\delta_{\beta,\delta''}}{18\bar{K}_G}
= \frac{1-\delta''}{9\bar{K}_G\left(A_\beta + 2.02(1-\delta'')\right)}
= \frac{(1.01\delta_{\beta,\delta''}-0.01)(1-\delta'')}{9\bar{K}_G\,A_\beta}
= C_1(\beta,\delta'').
\end{align*}
It follows that, for $\delta = \delta_{\beta,\delta''}$, the events in
\eqref{eq:steptwo} and \eqref{eq:stepthree} are both equal to the event
$\left\{|\hat{\theta}(g_1^*,g_2^*)-\theta(g_1^*,g_2^*)| \leq
C_1(\beta,\delta'')(\bar{\tau}+r')/(m_1^*m_2^*)\right\}$ appearing in
\eqref{eq:P1}. Since, for every $n$,
\begin{align*}
\sup_{F \in \mathcal{F}(n)}\mathbbm{P}_F\left(\cdot\right)
\leq \max\left(
\sup_{F \in \mathcal{F}(n)\cap\mathcal{F}_{\beta,\delta}}
\mathbbm{P}_F\left(\cdot\right),\;
\sup_{F \in \mathcal{F}(n)\cap\mathcal{F}^c_{\beta,\delta}}
\mathbbm{P}_F\left(\cdot\right)\right),
\end{align*}
combining \eqref{eq:steptwo} and \eqref{eq:stepthree} yields
\begin{align*}
\limsup_{F\in\mathcal{F}:\,N_1,N_2\to\infty}
\mathbbm{P}_F\left(|\hat{\theta}(g_1^*,g_2^*)-\theta(g_1^*,g_2^*)|
\leq C_1(\beta,\delta'')\,\frac{\bar{\tau}+r'}{m_1^*m_2^*}\right)
\leq \max(\beta,0) = \beta,
\end{align*}
which is \eqref{eq:P1}. \newline

\textbf{Step 4:} In this step we show \eqref{eq:steptwo} for any
$\delta \in (0.01/1.01,1)$. The strategy is to apply the lower-tail
inequality \eqref{lb} of Lemma~\ref{lem:cut-concentration} at a deviation
$x_{\beta,\delta}$ chosen so that the resulting threshold for
$||\epsilon||_{\square;c}$ matches the one in \eqref{eq:steptwo}. For a fixed
$F \in \mathcal{F}_{\beta,\delta}$, define
\begin{align*}
x_{\beta,\delta} := \frac{\left[\delta - (1-\delta)
\left(r\left(1.01+0.25C_{\beta,\delta}^{-1}\right) +
0.01 + 0.25C_{\beta,\delta}^{-1}\right)\right]
\mathbbm{E}_F\left[||\epsilon||_{\dagger}\right]}{18\bar{K}_G}
\end{align*}
and recall from Section~\ref{sec:second-interval} that
$V := \sqrt{||\sigma||_F^2 + 4B\mathbbm{E}_F\left[||\epsilon||_{\dagger}\right]
+ B||\sigma||_F}$. The quantity $x_{\beta,\delta}$ is eventually positive:
the inequality $\delta > (1-\delta)\left(0.01 +
0.25C_{\beta,\delta}^{-1}\right)$ is equivalent to
$(1.01\delta-0.01)C_{\beta,\delta} > 0.25(1-\delta)$, which holds strictly by
the definition of $C_{\beta,\delta}$, and the remaining term in
$x_{\beta,\delta}$ vanishes because $r \to 0$. The claim then follows from
\begin{align*}
&\mathbbm{P}_F\left(|\hat{\theta}(g_1^*,g_2^*) - \theta(g_1^*,g_2^*)| \leq
\frac{\bar{\tau}(1+r)}{m_1^* m_2^*} \frac{1-\delta}{18 \bar{K}_G} \right) \\
&\qquad= \mathbbm{P}_F\left(\frac{||\epsilon||_{\square;c}}{m_1^*m_2^*} \leq
\frac{\left(1.01\mathbbm{E}_F\left[||\epsilon||_{\dagger}\right] +
0.25||\sigma||_F \right)(1+r)}{m_1^* m_2^*}
\frac{1-\delta}{18 \bar{K}_G} \right)\\
&\qquad\leq \mathbbm{P}_F\left(\frac{||\epsilon||_{\square;c}}{m_1^*m_2^*} \leq
\frac{\mathbbm{E}_F\left[||\epsilon||_{\dagger}\right] (1+r)}{m_1^* m_2^*}
\frac{(1.01+0.25C_{\beta,\delta}^{-1})(1-\delta)}{18\bar{K}_G} \right)\\
&\qquad= \mathbbm{P}_F\left(\frac{||\epsilon||_{\square;c}}{m_1^*m_2^*}  \leq
\frac{\mathbbm{E}_F\left[||\epsilon||_{\dagger}\right]}{18 \bar{K}_G\,
m_1^* m_2^*} - \frac{x_{\beta,\delta}}{m_1^* m_2^*} \right)
\leq \exp\left(-\frac{x_{\beta,\delta}^2}{2V^2 + 4Bx_{\beta,\delta}}\right)
\end{align*}
where the first equality is because $|\hat{\theta}(g_1^*,g_2^*) -
\theta(g_1^*,g_2^*)| = ||\epsilon||_{\square;c}/(m_1^*m_2^*)$, the first
inequality is because $F \in \mathcal{F}_{\beta,\delta}$ implies
$||\sigma||_F \leq C_{\beta,\delta}^{-1}
\mathbbm{E}_F\left[||\epsilon||_{\dagger}\right]$, the second equality
rearranges terms, and the second inequality is the lower-tail inequality
\eqref{lb} of Lemma~\ref{lem:cut-concentration} (with $\bar{K}_G$ in place of
$K_G(\infty)$, and noting that $2V^2 = 2\left(||\sigma||^2_F +
4B\mathbbm{E}_F\left[||\epsilon||_{\dagger}\right] + B||\sigma||_F\right)$)
applied with $x = x_{\beta,\delta}$.

Since $r \to 0$,
$\lim_{F\in\mathcal{F}:\,N_1,N_2\to\infty}||\sigma||_F \to \infty$ by
Assumption~\ref{ass:nondegeneracy}(w),
$\lim_{F\in\mathcal{F}:\,N_1,N_2\to\infty}
\mathbbm{E}_F\left[||\epsilon||_{\dagger}\right] \to \infty$ by Lemma~\ref{lem:residual-norm-relations},
$\bar{K}_G$ and $B$ are fixed, and, by construction,
$\sup_{F\in\mathcal{F}_{\beta,\delta}}
||\sigma||_F/\mathbbm{E}_F\left[||\epsilon||_{\dagger}\right] \leq
C_{\beta,\delta}^{-1}$, we have
$x_{\beta,\delta}/\mathbbm{E}_F\left[||\epsilon||_{\dagger}\right] \to
D_{\beta,\delta}/(18\bar{K}_G)$ with $D_{\beta,\delta} := \delta -
(1-\delta)\left(0.01 + 0.25C_{\beta,\delta}^{-1}\right) > 0$,
$V^2/\mathbbm{E}_F\left[||\epsilon||_{\dagger}\right]^2 \leq
C_{\beta,\delta}^{-2} + o(1)$, and
$Bx_{\beta,\delta}/\mathbbm{E}_F\left[||\epsilon||_{\dagger}\right]^2 \to 0$,
in each case uniformly over $\mathcal{F}_{\beta,\delta}$. It follows that
\begin{align*}
\limsup_{F \in \mathcal{F}_{\beta,\delta}:\,N_1,N_2 \to \infty}
\mathbbm{P}_F\left(|\hat{\theta}(g_1^*,g_2^*) - \theta(g_1^*,g_2^*)|
\leq \frac{\bar{\tau}(1+r)}{m_1^* m_2^*}\,
\frac{1-\delta}{18 \bar{K}_G} \right)
\leq \exp\left(-\frac{D_{\beta,\delta}^2\,C_{\beta,\delta}^2}
{648\bar{K}_G^2}\right) = \beta,
\end{align*}
where the final equality holds because
\begin{align*}
D_{\beta,\delta}\,C_{\beta,\delta}
= (1.01\delta - 0.01)C_{\beta,\delta} - 0.25(1-\delta)
= \sqrt{648}\,\bar{K}_G\sqrt{-\ln(\beta)}
\end{align*}
by the definition of $C_{\beta,\delta}$, so that the exponent equals
$\ln(\beta)$. \newline

\textbf{Step 5:} In this step we show \eqref{eq:stepthree} for any
$\delta \in (0.01/1.01,1)$ and $\delta'' \in (0,1)$. The idea is that, on
$\mathcal{F}^c_{\beta,\delta}$, the scale $\bar{\tau}$ is bounded above by a
constant multiple of $||\sigma||_F$ (this is verified below), so a lower
bound on the estimation error in terms of the Frobenius norm
$||\epsilon||_F$ alone suffices; since $||\epsilon||_F$ concentrates around
$||\sigma||_F$ with probability approaching one, the limit in
\eqref{eq:stepthree} is $0$ rather than merely at most $\beta$. The claim
follows from
\begin{align}\label{secondlowerbound}
|\hat{\theta}(g_1^*,g_2^*) - \theta(g_1^*,g_2^*)|
= \frac{||\epsilon||_{\square;c}}{m_1^*m_2^*} \geq
\frac{||\epsilon||_{\infty\to1}}{9m_1^*m_2^*}
\geq \frac{||\epsilon||_{1,2}}{9\bar{K}_G\,m_1^*m_2^*}
\geq \frac{||\epsilon||_{F}}{9\bar{K}_G\,m_1^*m_2^*}
\end{align}
where the first inequality is from Lemma~\ref{lem:restricted-cut}, the second
inequality is from the second paragraph in the proof of
Lemma~\ref{lem:dagger-infty} (which, as noted in Step 1, holds with
$\bar{K}_G$ in place of $K_G(\infty)$), and the third inequality is because
the $\sqrt{\cdot}$ function is subadditive. Since the entries of $\epsilon$
are uniformly bounded by $2B$ by Assumption~\ref{ass:model},
Lemma~\ref{lem:bernstein} implies that, for any fixed $t > 0$ and
$F \in \mathcal{F}_{\beta,\delta}^{c}$,
\begin{align*}
\mathbbm{P}_F\left(\sum_{i \in [N_1], j \in [N_2]}(\sigma_{ij}^2-\epsilon_{ij}^2)
\geq t\right) \leq \exp\left(-\frac{t^2}{2\nu + 8B^2t/3}\right)
\end{align*}
for $\nu = \sum_{i \in [N_1], j \in [N_2]}
\mathbbm{E}_F\left[(\sigma_{ij}^2-\epsilon_{ij}^2)^2\right]$. Assumption~\ref{ass:model} also
implies that $\nu \leq 4B^2||\sigma||_F^2$, so that
\begin{align*}
\mathbbm{P}_F\left(||\epsilon||_F^2 \leq ||\sigma||_F^2 - t\right) \leq
\exp\left(-\frac{t^2}{8B^2||\sigma||_F^2 + 8B^2t/3}\right).
\end{align*}
Choosing $t = t_{r,\delta''} := \left(1 -
(1+r)^2(1-\delta'')^2\right)||\sigma||_F^2$, which is eventually positive
since $r\to0$, taking square roots on either side of the inequality in the
probability on the left-hand side, and dividing by $m_1^*m_2^*$ yields
\begin{align*}
\mathbbm{P}_F\left(\frac{||\epsilon||_F}{m_1^*m_2^*} \leq
\frac{||\sigma||_F(1+r)(1-\delta'')}{m_1^*m_2^*}\right)
\leq \exp\left(-\frac{t_{r,\delta''}^2}{8B^2||\sigma||_F^2 +
8B^2t_{r,\delta''}/3}\right).
\end{align*}
Combining this inequality with \eqref{secondlowerbound} gives
\begin{align*}
\mathbbm{P}_F\left(|\hat{\theta}(g_1^*,g_2^*) - \theta(g_1^*,g_2^*)| \leq
\frac{||\sigma||_F(1+r)(1-\delta'')}{9\bar{K}_G\,m_1^*m_2^*}\right)
\leq \exp\left(-\frac{t_{r,\delta''}^2}{8B^2||\sigma||_F^2 +
8B^2t_{r,\delta''}/3}\right).
\end{align*}
Now, for every $F \in \mathcal{F}_{\beta,\delta}^{c}$, we have
$\mathbbm{E}_F\left[||\epsilon||_{\dagger}\right] <
C_{\beta,\delta}||\sigma||_F$ and so $\bar{\tau} <
(1.01C_{\beta,\delta}+0.25)||\sigma||_F$. Direct calculation using the
definition of $C_{\beta,\delta}$ gives
\begin{align*}
1.01\,C_{\beta,\delta} + 0.25
&= \frac{1.01\left(\sqrt{648}\,\bar{K}_G\sqrt{-\ln(\beta)} +
0.25(1-\delta)\right) + 0.25\left(1.01\delta - 0.01\right)}{1.01\delta - 0.01}
\\
&= \frac{A_\beta}{1.01\delta - 0.01},
\end{align*}
since $1.01 \times 0.25(1-\delta) + 0.25(1.01\delta - 0.01) = 0.25$.
Therefore $||\sigma||_F > \bar{\tau}(1.01\delta-0.01)/A_\beta$ on
$\mathcal{F}^c_{\beta,\delta}$, and so
\begin{align*}
&\mathbbm{P}_F\left(|\hat{\theta}(g_1^*,g_2^*) - \theta(g_1^*,g_2^*)| \leq
\frac{\bar{\tau}(1+r)}{9\bar{K}_G\,m_1^*m_2^*}
\frac{(1.01\delta-0.01)(1-\delta'')}{A_\beta}\right) \\
&\qquad\leq \exp\left(-\frac{t_{r,\delta''}^2}{8B^2||\sigma||_F^2 +
8B^2t_{r,\delta''}/3}\right).
\end{align*}
Since $t_{r,\delta''}/||\sigma||_F^2 \to 1-(1-\delta'')^2 > 0$ as $r \to 0$
and $\lim_{F \in \mathcal{F}:\,N_1,N_2 \to \infty}||\sigma||_F \to \infty$
by Assumption~\ref{ass:nondegeneracy}(w), the right-hand side converges to
$0$, which gives \eqref{eq:stepthree}. \newline

\textbf{Step 6:} In this step we show that
$\lim_{F\in\mathcal{F}:\,N_1,N_2\to\infty}
\mathbbm{P}_F\left(\hat{\bar{\tau}} > \bar{\tau} + r'\right) = 0$. Under
Assumptions~\ref{ass:model} and \ref{ass:nondegeneracy}(w), $\bar{\tau} \geq 0.25||\sigma||_F > 0$ for all
sufficiently large $\min(N_1,N_2)$, and so, by the definition $r' :=
r\bar{\tau}$, the event $\left\{\hat{\bar{\tau}} > \bar{\tau} + r'\right\}$ is
contained in the event $\left\{\left|\frac{\hat{\bar{\tau}} -
\bar{\tau}}{\bar{\tau}}\right| > r\right\}$. The claim then follows directly
from Assumption~\ref{ass:scale-consistency}(ii).
\end{proof}

\end{spacing}
\clearpage
\section{Estimating the variance parameters}\label{app:variance-estimation}
This appendix provides three strategies for constructing $\hat\sigma(g_1,g_2)$, $\hat{\bar\tau}$, and $\hat V$.  Section~\ref{sec:deterministic-scale-bounds} gives deterministic upper bounds that yield valid, though potentially conservative, confidence intervals. Section~\ref{sec:raw-plugin} proposes estimators that use $Y$ as a plug-in for $\epsilon$. We provide conditions under which the resulting estimators are conservative and also show that, under an additional vanishing signal condition, they are consistent in the sense of Assumption~\ref{ass:scale-consistency}. Section~\ref{sec:usvt} proposes estimators that use $Y - \hat{\mu}$ as a plug-in for $\epsilon$ where $\hat{\mu}$ is estimated via singular value thresholding. We provide conditions under which these estimators are consistent in the sense of Assumption~\ref{ass:scale-consistency}.

\begin{remark}
Throughout this appendix, stochastic statements are established for a
fixed realized conditional law $F\in\mathcal F$. Except in lemmas stated for
generic random variables, probabilities, expectations, and variances are taken
under $\mathbbm P_F$ and denoted by $\mathbbm P_F$, $\mathbbm E_F$, and
$\operatorname{Var}_F$. Under Assumption~\ref{ass:model}, the entries of $Y$
are independent under $\mathbbm P_F$, while $\mu$, $\sigma$, $\bar\tau$,
$V$, and the population quantities introduced below are deterministic functions
of $F$. As in the main text, we suppress the dependence on $F$ in our notation. 
\end{remark}

\subsection{Deterministic bounds}
\label{sec:deterministic-scale-bounds}
Our first strategy uses only the known support bound. Since
$|\epsilon_{ij}|\leq2B$, define
\begin{align*}
\hat\sigma(g_1,g_2)=2B,\\
\hat V
=\left\{8\bigl(N_1\sqrt{N_2}+N_2\sqrt{N_1}\bigr)
       +4N_1N_2+2\sqrt{N_1N_2}\right\}^{1/2}B,\\
\hat{\bar\tau}
=2.02\bigl(N_1\sqrt{N_2}+N_2\sqrt{N_1}\bigr)B
  +0.5\sqrt{N_1N_2}\,B.
\end{align*}

\begin{lemma}
\label{lem:deterministic-scale-bounds}
Suppose Assumption~\ref{ass:model} and fix $c>0$. Then, for every
$(g_1,g_2)\in\mathcal G_c$,
\[
\sigma(g_1,g_2)\leq\hat\sigma(g_1,g_2),
\qquad
\bar\tau\leq\hat{\bar\tau},
\qquad
V\leq\hat V.
\]
\end{lemma}

\begin{proof}
For $(g_1,g_2)\in\mathcal G_c$, let
$m_t=\sum_{i\in[N_t]}g_{i,t}$. Since $|\epsilon_{ij}|\leq2B$,
\begin{align*}
\sigma(g_1,g_2)
&=\left\{
\frac{1}{m_1m_2}
\sum_{i\in[N_1],j\in[N_2]}
\sigma_{ij}^2g_{i,1}g_{j,2}
\right\}^{1/2}
&\leq
\left\{
\frac{1}{m_1m_2}
\sum_{i\in[N_1],j\in[N_2]}
4B^2g_{i,1}g_{j,2}
\right\}^{1/2}
=2B.
\end{align*}
Similarly, $\|\sigma\|_F\leq2\sqrt{N_1N_2}\,B$ and $\mathbbm E_F\|\epsilon\|_\dagger \leq 2\bigl(N_1\sqrt{N_2}+N_2\sqrt{N_1}\bigr)B$  so
\[
\bar\tau =1.01\,\mathbbm E_F\|\epsilon\|_\dagger+0.25\|\sigma\|_F \leq
2.02\bigl(N_1\sqrt{N_2}+N_2\sqrt{N_1}\bigr)B +0.5\sqrt{N_1N_2}\,B
\]
and 
\[
V^2=\|\sigma\|_F^2+B\|\sigma\|_F
+4B\mathbbm E_F\|\epsilon\|_\dagger \leq \left\{8\bigl(N_1\sqrt{N_2}+N_2\sqrt{N_1}\bigr)
       +4N_1N_2+2\sqrt{N_1N_2}\right\}B^2.
\]
\end{proof}

Let $CI_1^\circ(g_1,g_2;\alpha)$ and $CI_2^\circ(g_1,g_2;\alpha)$ denote the oracle intervals obtained with the infeasible choices of $\hat{\sigma}(g_1,g_2) = \sigma(g_1,g_2)$, $\hat{\bar{\tau}} = \bar\tau$, and $\hat{V} = V$. These oracle choices satisfy Assumption~\ref{ass:scale-consistency} trivially. Propositions~\ref{prop:propone} and~\ref{prop:proptwo} imply that $CI_1^\circ$ satisfies \eqref{simul} under Assumptions~\ref{ass:model} and~\ref{ass:nondegeneracy}(s), and $CI_2^\circ$ satisfies \eqref{simul} under Assumptions~\ref{ass:model} and~\ref{ass:nondegeneracy}(w). Since the feasible and oracle intervals have the same center, and their lengths are nondecreasing in $\sigma(g_1,g_2)$, $\bar{\tau}$, and $V$, Lemma~\ref{lem:deterministic-scale-bounds} implies that, for every realization of $Y$, $(g_1,g_2)\in\mathcal G_c$, and $\quad t\in\{1,2\}$
\[
CI_t(g_1,g_2;\alpha)\supseteq CI_t^\circ(g_1,g_2;\alpha).
\]
It follows that the feasible interval $CI_1$ based on the deterministic bounds satisfies \eqref{simul} under Assumptions~\ref{ass:model} and~\ref{ass:nondegeneracy}(s). The feasible interval $CI_2$ based on the deterministic bounds atisfies \eqref{simul} under Assumptions~\ref{ass:model} and~\ref{ass:nondegeneracy}(w).  These intervals may be substantially wider than the oracle intervals, however, and need not attain the rates in Propositions~\ref{prop:propthree} and~\ref{prop:propfour}.

\subsection{Raw second-moment plug-in}\label{sec:raw-plugin}
The second approach uses $Y$ as a plug-in for $\epsilon$ in the formulas for $\sigma(g_1,g_2)$, $\bar{\tau}$, and $V$. Because
\[
\mathbbm E_F[Y_{ij}^2]
=\sigma_{ij}^2+\mu_{ij}^2\geq\sigma_{ij}^2,
\]
the resulting squared plug-ins have nonnegative population bias. We first provide sufficient conditions for the resulting estimates to be conservative. We then show that, under an additional vanishing signal condition, the estimates further satisfy Assumption~\ref{ass:scale-consistency}. We also consider a shifted version of the raw plug-in estimator that is more conservative, but is valid under weaker assumptions. 

\subsubsection{Setting and notation}\label{sec:raw-setting}

For a fixed $F\in\mathcal F$, let
$\mathcal D=\mathcal D(F)\subseteq[N_1]\times[N_2]$ denote the deterministic set
of potentially nonzero network entries. For $(i,j)\notin\mathcal D$, the
marginal law $F_{ij}$ is degenerate at zero, so $Y_{ij}=0$
$\mathbbm P_F$-almost surely. Under $\mathbbm P_F$, the variables
$\{Y_{ij}:(i,j)\in\mathcal D\}$ are independent and satisfy $|Y_{ij}|\leq B$
$\mathbbm P_F$-almost surely, where $B$ is fixed across the model class. This formulation
includes a rectangular network by taking
$\mathcal D=[N_1]\times[N_2]$ and the paper's representation of a loopless
undirected network by taking $N_1=N_2=N$ and
$\mathcal D=\{(i,j):i<j\}$.

Write
\[
\mu_{ij}=\mathbbm E_F[Y_{ij}],\qquad
\epsilon_{ij}=Y_{ij}-\mu_{ij},\qquad
\sigma_{ij}^2=\mathbbm E_F[\epsilon_{ij}^2],
\]
and
\[
\vartheta_{ij}=\mathbbm E_F[Y_{ij}^2]
=\sigma_{ij}^2+\mu_{ij}^2.
\]
Thus $\vartheta_{ij}\geq\sigma_{ij}^2$ entrywise.  Define the row and column
energies and their conditional means by
\[
q_i=\sum_jY_{ij}^2,\qquad
p_j=\sum_iY_{ij}^2,\qquad
\bar q_i=\sum_j\vartheta_{ij},\qquad
\bar p_j=\sum_i\vartheta_{ij},
\]
and define the row and column variance totals
\[
v_i=\sum_j\sigma_{ij}^2\leq\bar q_i,
\qquad
w_j=\sum_i\sigma_{ij}^2\leq\bar p_j.
\]
All sums may be taken over the full rectangle because entries outside $\mathcal D$
are zero.  The relevant norms are
\begin{align*}
\lVert X\rVert_F^2&=\sum_{ij}X_{ij}^2,\\
\lVert X\rVert_{1,2}&=\sum_i\left(\sum_jX_{ij}^2\right)^{1/2},\\
\lVert X\rVert_\dagger&=\lVert X\rVert_{1,2}+\lVert X^\top\rVert_{1,2}.
\end{align*}
Hence
\[
\lVert Y\rVert_\dagger=\sum_i\sqrt{q_i}+\sum_j\sqrt{p_j}.
\]

For $(g_1,g_2)\in\mathcal G_c$, let $m_t$ denote the size of group $g_t$ and write
\[
\sigma(g_1,g_2)^2
=\frac{1}{m_1m_2}\sum_{ij}\sigma_{ij}^2g_{i,1}g_{j,2}.
\]
The oracle global scales are
\[
\bar\tau
=1.01\,\mathbbm E_F[\lVert \epsilon\rVert_\dagger]+0.25\,\lVert \sigma\rVert_F
\]
and
\[
V^2
=\lVert \sigma\rVert_F^2+4B\,\mathbbm E_F[\lVert \epsilon\rVert_\dagger]+B\,\lVert \sigma\rVert_F.
\]
The \emph{raw plug-in} is the estimator that is defined to set
$\widehat\mu^{\,0}\equiv0$ for every sample:
\[
\widehat\sigma^{\,0}(g_1,g_2)^2
=\frac{1}{m_1m_2}\sum_{ij}Y_{ij}^2g_{i,1}g_{j,2},
\]
\[
\widehat{\bar\tau}^{\,0}=1.01\,\lVert Y\rVert_\dagger+0.25\,\lVert Y\rVert_F,
\qquad
\widehat V^{\,0}
=\left(\lVert Y\rVert_F^2+4B\,\lVert Y\rVert_\dagger+B\,\lVert Y\rVert_F\right)^{1/2}.
\]

To allow structural zero rows and columns, define the active sets
\[
\mathcal I_+=\{i:\bar q_i>0\},\qquad
\mathcal J_+=\{j:\bar p_j>0\},\qquad
\mathcal D_+=\{(i,j)\in\mathcal D:\vartheta_{ij}>0\}.
\]
If $\bar q_i=0$, then $q_i=0$ $\mathbbm P_F$-almost surely and $v_i=0$; the analogous
statement holds for columns.  Define
\[
S=\sum_{i\in\mathcal I_+}\sqrt{\bar q_i}
 +\sum_{j\in\mathcal J_+}\sqrt{\bar p_j},
\qquad
H=\sum_{i\in\mathcal I_+}\bar q_i^{-1/2}
 +\sum_{j\in\mathcal J_+}\bar p_j^{-1/2}.
\]
When $S>0$, define the dimensionless energy-dispersion index
\[
\lambda=B^2\frac{H}{S};
\]
set $\lambda=\infty$ when $S=0$.  Define the signal-share index
\[
\chi
=\max_{(i,j):\,\vartheta_{ij}>0}
  \frac{\mu_{ij}^2}{\vartheta_{ij}}\in[0,1],
\]
with $\chi=0$ if every $\vartheta_{ij}$ is zero.  Entrywise,
\[
\sigma_{ij}^2\geq(1-\chi)\vartheta_{ij}.
\]
For binary networks, $\vartheta_{ij}=\mu_{ij}$, so $\chi$ is the largest
connection probability.

Let $\mathcal F(n)$ be the model class defined in Section~\ref{sec:asymptotics}. For each
$n$, define the deterministic uniform envelopes
\[
\lambda_n=\sup_{F\in\mathcal F(n)}\lambda,
\qquad
\chi_n=\sup_{F\in\mathcal F(n)}\chi,
\]
\[
s_n
=\inf_{F\in\mathcal F(n)}\inf_{g\in\mathcal G_c}
 \frac{\nu_g}{N_1+N_2},
\qquad
\nu_g:=\sum_{ij}\sigma_{ij}^2g_{i,1}g_{j,2},
\]
and
\[
a_n=\inf_{F\in\mathcal F(n)}\lVert \sigma\rVert_F^2.
\]
For $F\in\mathcal F(n)$, define the dimension-balance quantity
\[
\kappa(F)
=\frac{\min\{N_1(F),N_2(F)\}^2}
{\max\{N_1(F),N_2(F)\}},
\qquad
\underline{\kappa}_n
=\inf_{F\in\mathcal F(n)}\kappa(F).
\]
In this notation, the condition $a_n\to\infty$ is exactly
Assumption~\ref{ass:nondegeneracy}(w).  Moreover, pointwise in $F$,
\[
\frac{c^2}{2}\min(N_1,N_2)
 \min_{g\in\mathcal G_c}\sigma(g)^2
\leq
\inf_{g\in\mathcal G_c}\frac{\nu_g}{N_1+N_2}
\leq
\min(N_1,N_2)
 \min_{g\in\mathcal G_c}\sigma(g)^2.
\]
Since $c$ is fixed, $s_n\to\infty$ is equivalent  to Assumption~\ref{ass:nondegeneracy}(s).

The pointwise arguments below are carried out under $\mathbbm P_F$. The
displayed envelopes make the resulting bounds uniform over
$F\in\mathcal F(n)$. Conditional and unconditional implications are as
described in the probability convention above.

\subsubsection{One-sided scale adequacy}\label{sec:raw-one-sided-definition}
We say that a plug-in triple satisfies \emph{one-sided scale adequacy} if there
exists a deterministic sequence $r_n\to0$ such that
\begin{align}
\sup_{F\in\mathcal F(n)}\mathbbm P_F\left(
 \sup_{g\in\mathcal G_c}
 \frac{\sigma(g)-\widehat\sigma^{\,0}(g)}{\sigma(g)}>r_n
\right)&\longrightarrow0, \tag{i$^-$}\label{eq:oneside-sigma}\\
\sup_{F\in\mathcal F(n)}\mathbbm P_F\left(
 \frac{\bar\tau-\widehat{\bar\tau}^{\,0}}{\bar\tau}>r_n
\right)&\longrightarrow0, \tag{ii$^-$}\label{eq:oneside-tau}\\
\sup_{F\in\mathcal F(n)}\mathbbm P_F\left(
 \frac{V-\widehat V^{\,0}}{V}>r_n
\right)&\longrightarrow0. \tag{iii$^-$}\label{eq:oneside-V}
\end{align}
The ratios are used only when their denominators are positive, which the
relevant nondegeneracy assumptions ensure eventually.  These are the halves
of the paper's Assumption~\ref{ass:scale-consistency} that rule out material \emph{underestimation}.
This condition is weaker than the paper's Assumption~\ref{ass:scale-consistency}, which is a two-sided condition.

\subsubsection{Preliminary lemmas}\label{sec:raw-preliminary-lemmas}

The next three lemmas are generic probability inequalities and do not use the
fixed-$F$ convention.

\begin{lemma}[Two variance inequalities]\label{lem:raw-var}
\emph{(a)} If $Z\in[0,a]$ almost surely, then
$\operatorname{Var}(Z)\leq a\mathbbm E Z$.

\emph{(b)} If $|Y|\leq B$ almost surely, with mean $\mu$ and variance
$\sigma^2$, then
\[
\operatorname{Var}(Y^2)\leq4B^2\sigma^2.
\]
\end{lemma}

\begin{proof}
For part (a),
$\operatorname{Var}(Z)\leq\mathbbm E Z^2\leq a\mathbbm E Z$.
For part (b), variance minimizes mean squared distance over constants, so
\[
\operatorname{Var}(Y^2)
\leq\mathbbm E(Y^2-\mu^2)^2
=\mathbbm E[(Y-\mu)^2(Y+\mu)^2]
\leq4B^2\mathbbm E(Y-\mu)^2.
\]
\end{proof}

\begin{lemma}[A one-sided second-order bound for the square root]
\label{lem:raw-sqrt}
For every $x\geq0$ and $m>0$,
\[
\sqrt{x}
\geq
\sqrt m+\frac{x-m}{2\sqrt m}
 -\frac{(x-m)^2}{2m^{3/2}}.
\]
Consequently, if $X\geq0$ and $\mathbbm E X=m>0$, then
\[
\mathbbm E\sqrt X
\geq
\sqrt m-\frac{\operatorname{Var}(X)}{2m^{3/2}}.
\]
If, in addition, $\operatorname{Var}(X)\leq am$ for some $a>0$, then
\[
\mathbbm E\sqrt X\geq\sqrt m\left(1-\frac{a}{2m}\right).
\]
Under the same variance condition, for every $m\geq0$,
\[
\mathbbm E\sqrt{X+a}\geq\sqrt m.
\]
\end{lemma}

\begin{proof}
Set $z=\sqrt{x/m}$.  The difference between the left side of the first
inequality and its right side, divided by $\sqrt m$, equals
\[
z-1-\frac{z^2-1}{2}+\frac{(z^2-1)^2}{2}
=\frac12 z(z-1)^2(z+2)\geq0.
\]
Taking expectations proves the second display, and substituting
$\operatorname{Var}(X)\leq am$ proves the third. For the final display with
$m>0$, apply the second display to $X+a$, whose mean is $m+a>0$ and whose
variance is at most $am$:
\[
\mathbbm E\sqrt{X+a}
\geq
\sqrt{m+a}-\frac{am}{2(m+a)^{3/2}}
\geq\sqrt m,
\]
where the last step follows from
\[
\sqrt{m+a}-\sqrt m
=\frac{a}{\sqrt{m+a}+\sqrt m}
\geq\frac{a}{2\sqrt{m+a}}
\geq\frac{am}{2(m+a)^{3/2}}.
\]
When $m=0$, the hypothesis $\operatorname{Var}(X)\leq am=0$ and the
nonnegativity of $X$ imply that $X=0$ almost surely, and the conclusion reads
$\sqrt a\geq0$.
\end{proof}

\begin{lemma}[Bounded differences]\label{lem:raw-mcdiarmid}
Let $X_1,\ldots,X_L$ be independent and let $f(X_1,\ldots,X_L)$ change by at
most $c_\ell$ when only coordinate $\ell$ is changed.  Then, for every
$t>0$,
\[
\mathbbm P\{|f-\mathbbm E f|>t\}
\leq
2\exp\left(-\frac{2t^2}{\sum_{\ell=1}^Lc_\ell^2}\right).
\]
Each one-sided tail is bounded by the same exponential without the factor
$2$.
\end{lemma}

\begin{proof}
See, for example, Theorem 6.5 of \cite{boucheron2013concentration}.
\end{proof}

\subsubsection{One-sided scale adequacy of the raw plug-in}\label{sec:raw-one-sided}

\begin{lemma}[One-sided adequacy of the raw plug-in]
\label{lem:raw-one-sided}
Suppose Assumption~\ref{ass:model} holds, with $B$ fixed across $\mathcal F$.

\begin{itemize}
\item[(i)] If $s_n\to\infty$, then \eqref{eq:oneside-sigma} holds with
\[
r_{\sigma,n}^-=s_n^{-1/4}
\]
for all sufficiently large $n$.  More precisely,
\begin{align*}
&\sup_{F\in\mathcal F(n)}\mathbbm P_F\left(
 \sup_{g\in\mathcal G_c}
 \frac{\sigma(g)-\widehat\sigma^{\,0}(g)}{\sigma(g)}
 >r_{\sigma,n}^-
\right)\\
&\qquad\leq
\exp\left[
 2n\left(\log2-\frac{\sqrt{s_n}}{9B^2}\right)
\right]
\longrightarrow0.
\end{align*}

\item[(ii)--(iii)] If $a_n\to\infty$ and $\lambda_n\to0$, then
\eqref{eq:oneside-tau}--\eqref{eq:oneside-V} hold with
\[
r_{\tau V,n}^-
=\frac{\lambda_n}{2}+\lambda_n^{1/4}+a_n^{-1/4}.
\]
The common failure probability is at most
\[
\exp\left(-\frac{\sqrt{a_n}}{9B^2}\right)
+
\exp\left(-\frac{1}{2\sqrt{\lambda_n}}\right)
\longrightarrow0.
\]
\end{itemize}

If both sets of conditions hold, all three one-sided requirements hold with
the single sequence
\[
r_n^-=
\max\{r_{\sigma,n}^-,r_{\tau V,n}^-\}.
\]
Consequently, the coverage conclusions of Propositions~\ref{prop:propone} and~\ref{prop:proptwo} continue to
hold when their respective parts of Assumption~\ref{ass:scale-consistency} are replaced by the
corresponding one-sided requirements above.
\end{lemma}

\begin{proof}
Fix $n$ and $F\in\mathcal F(n)$, and work under $\mathbbm P_F$.

\medskip
\noindent\textbf{Part (i).}
For $g=(g_1,g_2)$, write
\[
\nu_g=\sum_{ij}\sigma_{ij}^2g_{i,1}g_{j,2}
=m_1m_2\sigma(g)^2
\]
and
\[
\zeta_{ij}
=(\vartheta_{ij}-Y_{ij}^2)g_{i,1}g_{j,2}.
\]
The $\zeta_{ij}$ are independent and centered under $\mathbbm P_F$,
$\zeta_{ij}\leq B^2$, and Lemma~\ref{lem:raw-var} gives
\[
\sum_{ij}\operatorname{Var}_F(\zeta_{ij})
\leq4B^2\nu_g.
\]
Moreover,
\[
m_1m_2\widehat\sigma^{\,0}(g)^2
=\sum_{ij}\vartheta_{ij}g_{i,1}g_{j,2}-\sum_{ij}\zeta_{ij}
\geq\nu_g-\sum_{ij}\zeta_{ij}.
\]
Thus the raw squared plug-in has nonnegative population bias; the next
argument controls the probability that sampling variation nevertheless
produces substantial underestimation.  For $r\in(0,1]$,
\begin{align*}
\left\{\frac{\sigma(g)-\widehat\sigma^{\,0}(g)}{\sigma(g)}>r\right\}
&\subseteq
\{\widehat\sigma^{\,0}(g)^2<(1-r)^2\sigma(g)^2\}\\
&\subseteq
\left\{\sum_{ij}\zeta_{ij}>(2r-r^2)\nu_g\right\}\\
&\subseteq
\left\{\sum_{ij}\zeta_{ij}>r\nu_g\right\}.
\end{align*}
Bernstein's inequality therefore yields
\begin{align}
\mathbbm P_F\left(\sum_{ij}\zeta_{ij}>r\nu_g\right)
&\leq
\exp\left(
 -\frac{r^2\nu_g^2}
 {2(4B^2\nu_g+B^2r\nu_g/3)}
\right)\nonumber\\
&\leq
\exp\left(-\frac{r^2\nu_g}{9B^2}\right).
\label{eq:raw-bernstein}
\end{align}
A union bound over at most $2^{N_1+N_2}$ group pairs and the inequality
$\nu_g\geq s_n(N_1+N_2)$ show, for every $F\in\mathcal F(n)$, that the
failure probability at $r=s_n^{-1/4}$ is at most
\[
\exp\left\{(N_1+N_2)
\left(\log2-\frac{\sqrt{s_n}}{9B^2}\right)\right\}.
\]
The parenthesized factor is negative eventually, and
$N_1+N_2\geq2n$, which gives the uniform bound in the statement. The
condition $s_n\to\infty$ also implies that all relevant $\nu_g$ are positive
eventually.

\medskip
\noindent\textbf{Parts (ii)--(iii), Step 1: the Frobenius component.}
The $\mathbbm P_F$-mean of $\lVert Y\rVert_F^2$ is
$\sum_{ij}\vartheta_{ij}\geq\lVert \sigma\rVert_F^2$.  Hence
\begin{align*}
\{\lVert Y\rVert_F^2<(1-r_1)\lVert \sigma\rVert_F^2\}
\subseteq
\left\{
\sum_{ij}(\vartheta_{ij}-Y_{ij}^2)
>r_1\lVert \sigma\rVert_F^2
\right\}.
\end{align*}
The same Bernstein calculation gives
\[
\mathbbm P_F\{\lVert Y\rVert_F^2<(1-r_1)\lVert \sigma\rVert_F^2\}
\leq
\exp\left(-\frac{r_1^2\lVert \sigma\rVert_F^2}{9B^2}\right).
\]
Take $r_1=a_n^{-1/4}$.  Since $\lVert \sigma\rVert_F^2\geq a_n$, the right side is no
larger than
\[
\exp\left(-\frac{\sqrt{a_n}}{9B^2}\right).
\]
On the complementary event,
\[
\lVert Y\rVert_F\geq\sqrt{1-r_1}\lVert \sigma\rVert_F
\geq(1-r_1)\lVert \sigma\rVert_F.
\]

\medskip
\noindent\textbf{Step 2: the dagger component---bias.}
For $i\in\mathcal I_+$, Lemmas~\ref{lem:raw-var} and~\ref{lem:raw-sqrt} give
\[
\mathbbm E_F[\sqrt{q_i}]
\geq
\sqrt{\bar q_i}-\frac{B^2}{2\sqrt{\bar q_i}},
\]
and analogously for $j\in\mathcal J_+$.  Inactive terms are identically zero,
so
\[
\mathbbm E_F[\lVert Y\rVert_\dagger]
\geq
S-\frac{B^2}{2}H
=S\left(1-\frac{\lambda}{2}\right).
\]
On the oracle side, Jensen's inequality gives
\begin{align}
\mathbbm E_F[\lVert \epsilon\rVert_\dagger]
&\leq
\sum_{i\in\mathcal I_+}\sqrt{v_i}
+
\sum_{j\in\mathcal J_+}\sqrt{w_j}
\leq S.
\label{eq:raw-oracle-upper}
\end{align}

\medskip
\noindent\textbf{Step 3: the dagger component---concentration.}
Only the coordinates in $\mathcal D_+$ can vary.  Let $L=|\mathcal D_+|$.  Changing one
such coordinate changes one row-energy square root and one column-energy
square root, each by at most $B$, so the bounded-differences constant is at
most $2B$.  Also,
\[
L\leq|\mathcal I_+||\mathcal J_+|.
\]
By Cauchy--Schwarz,
\[
|\mathcal I_+|
\leq
\left(\sum_{i\in\mathcal I_+}\sqrt{\bar q_i}\right)^{1/2}
\left(\sum_{i\in\mathcal I_+}\bar q_i^{-1/2}\right)^{1/2}
\leq\sqrt{SH},
\]
and the same bound holds for $|\mathcal J_+|$.  Therefore
\[
L\leq SH=\frac{\lambda S^2}{B^2}.
\]
Lemma~\ref{lem:raw-mcdiarmid} implies, for $t>0$,
\[
\mathbbm P_F\{\lVert Y\rVert_\dagger<\mathbbm E_F[\lVert Y\rVert_\dagger]-t\}
\leq
\exp\left(-\frac{t^2}{2B^2L}\right).
\]
Taking $t=r_2S$ and using $\lambda\leq\lambda_n$ gives
\[
\mathbbm P_F\{\lVert Y\rVert_\dagger<\mathbbm E_F[\lVert Y\rVert_\dagger]-r_2S\}
\leq
\exp\left(-\frac{r_2^2}{2\lambda_n}\right).
\]
With $r_2=\lambda_n^{1/4}$, the exponent is
$-1/(2\sqrt{\lambda_n})$.  Combining this event with Step~2 and
\eqref{eq:raw-oracle-upper},
\[
\lVert Y\rVert_\dagger
\geq
\left(1-\frac{\lambda_n}{2}-\lambda_n^{1/4}\right)
\mathbbm E_F[\lVert \epsilon\rVert_\dagger]
\]
with the claimed probability.

\medskip
\noindent\textbf{Step 4: assembly.}
Let
\[
r_{\tau V,n}^-
=\frac{\lambda_n}{2}+\lambda_n^{1/4}+a_n^{-1/4}.
\]
This sequence is below one eventually.  On the intersection of the good
events in Steps~1 and~3,
\[
\widehat{\bar\tau}^{\,0}
\geq
(1-r_{\tau V,n}^-)\bar\tau.
\]
Moreover, term by term,
\[
(\widehat V^{\,0})^2
\geq
(1-r_{\tau V,n}^-)V^2,
\]
so
\[
\widehat V^{\,0}
\geq
\sqrt{1-r_{\tau V,n}^-}\,V
\geq
(1-r_{\tau V,n}^-)V.
\]
The union of the two failure probabilities is the bound stated in the lemma.
\end{proof}

\begin{corollary}[Coverage with the raw plug-in]
\label{cor:raw-coverage}
Let $CI_1$ and $CI_2$ be the confidence intervals formed using  $\widehat\sigma^{\,0}(g_1,g_2)$,
$\widehat{\bar\tau}^{\,0}$, and $\widehat V^{\,0}$. Under Assumption~\ref{ass:model} and $s_n\to\infty$, $CI_1$ satisfies \eqref{simul}. Under Assumption~\ref{ass:model}, $a_n\to\infty$, and $\lambda_n\to0$,
$CI_2$ satisfies \eqref{simul}. If both sets of conditions hold, the interval $CI_\cap$ described in Section~\ref{sec:simultaneous-results} also satisfies \eqref{simul}.
\end{corollary}

\begin{proof}
The proofs of Propositions~\ref{prop:propone} and~\ref{prop:proptwo} use Assumption~\ref{ass:scale-consistency} only to rule out material
underestimation of the scale quantities. Proposition~\ref{prop:propone} uses
$\widehat\sigma(g)\geq(1-o_p(1))\sigma(g)$ uniformly in $\mathbbm P_F$-probability over
$g\in\mathcal G_c$. Proposition~\ref{prop:proptwo} uses
$\widehat{\bar\tau}\geq(1-o_p(1))\bar\tau$ and
$\widehat V\geq(1-o_p(1))V$. Lemma~\ref{lem:raw-one-sided} supplies exactly
these bounds for the raw plug-in. The remaining steps in the two proposition
proofs are unchanged.
\end{proof}

\subsubsection{Two-sided consistency under a vanishing signal share}\label{sec:raw-two-sided}

\begin{lemma}[Two-sided consistency under a vanishing signal share]
\label{lem:raw-two-sided}
Suppose the conditions of Lemma~\ref{lem:raw-one-sided} hold and
$\chi_n\to0$.

\begin{itemize}
\item[(i)] Under $s_n\to\infty$, define
\[
r_{\sigma,n}=2\chi_n+s_n^{-1/4}.
\]
For all sufficiently large $n$, $r_{\sigma,n}\leq1$ and
\begin{align*}
&\sup_{F\in\mathcal F(n)}\mathbbm P_F\left(
 \sup_{g\in\mathcal G_c}
 \frac{|\widehat\sigma^{\,0}(g)-\sigma(g)|}{\sigma(g)}
 >r_{\sigma,n}
\right)\\
&\qquad\leq
2\exp\left[
 2n\left(\log2-\frac{\sqrt{s_n}}{9B^2}\right)
\right]
\longrightarrow0.
\end{align*}

\item[(ii)--(iii)] Under $a_n\to\infty$ and $\lambda_n\to0$, define
\[
r_{\tau V,n}
=2\chi_n+15\lambda_n^{1/4}+a_n^{-1/4}.
\]
Then
\[
\sup_{F\in\mathcal F(n)}\mathbbm P_F\left(
 \frac{|\widehat{\bar\tau}^{\,0}-\bar\tau|}{\bar\tau}>r_{\tau V,n}
\right)
\longrightarrow0
\]
and
\[
\sup_{F\in\mathcal F(n)}\mathbbm P_F\left(
 \frac{|\widehat V^{\,0}-V|}{V}>r_{\tau V,n}
\right)
\longrightarrow0.
\]
For either probability, a common upper bound is
\[
2\exp\left(-\frac{\sqrt{a_n}}{9B^2}\right)
+
2\exp\left(-\frac{1}{2\sqrt{\lambda_n}}\right).
\]
\end{itemize}

If all conditions hold, the paper's two-sided Assumption~\ref{ass:scale-consistency} holds for the raw
plug-in with the single sequence
\[
r_n^*=\max\{r_{\sigma,n},r_{\tau V,n}\}.
\]
\end{lemma}

\begin{proof}
The underestimation halves follow from Lemma~\ref{lem:raw-one-sided}, because
the rates in the present lemma are larger.  It remains to control
overestimation.  For all sufficiently large $n$, assume
$\chi_n\leq1/2$ and $\lambda_n\leq1/8$.  Then, entrywise,
\[
\vartheta_{ij}
\leq
\frac{\sigma_{ij}^2}{1-\chi_n}
\leq
(1+2\chi_n)\sigma_{ij}^2.
\]

\medskip
\noindent\textbf{Part (i): cell-specific scale.}
Let
\[
\Theta_g=\sum_{ij}\vartheta_{ij}g_{i,1}g_{j,2}
\leq(1+2\chi_n)\nu_g.
\]
For $r\in[2\chi_n,1]$,
\begin{align*}
\{\widehat\sigma^{\,0}(g)>(1+r)\sigma(g)\}
&\subseteq
\left\{-\sum_{ij}\zeta_{ij}>(1+r)^2\nu_g-\Theta_g\right\}\\
&\subseteq
\left\{-\sum_{ij}\zeta_{ij}>r\nu_g\right\},
\end{align*}
where $\zeta_{ij}$ is defined in the proof of
Lemma~\ref{lem:raw-one-sided}.  The last inclusion follows because
\[
(1+r)^2-(1+2\chi_n)
=2r+r^2-2\chi_n\geq r.
\]
The summands $-\zeta_{ij}$ are centered, bounded above by $B^2$, and have
total variance at most $4B^2\nu_g$.  Applying
\eqref{eq:raw-bernstein}, taking
$r=2\chi_n+s_n^{-1/4}$, and union-bounding over both tails and all group pairs
gives the stated result.

\medskip
\noindent\textbf{Parts (ii)--(iii), Step 1: Frobenius upper tail.}
Summing the entrywise comparison gives
\[
\mathbbm E_F[\lVert Y\rVert_F^2]
=\sum_{ij}\vartheta_{ij}
\leq
(1+2\chi_n)\lVert \sigma\rVert_F^2.
\]
With $r_1=a_n^{-1/4}$,
\[
\mathbbm P_F\{\lVert Y\rVert_F^2>(1+2\chi_n+r_1)\lVert \sigma\rVert_F^2\}
\leq
\exp\left(-\frac{\sqrt{a_n}}{9B^2}\right).
\]
On the complementary event,
\[
\lVert Y\rVert_F
\leq
(1+2\chi_n+r_1)\lVert \sigma\rVert_F.
\]

\medskip
\noindent\textbf{Step 2: reverse comparison for the oracle dagger scale.}
For $i\in\mathcal I_+$, let
\[
T_i=\sum_j\epsilon_{ij}^2,
\qquad
\mathbbm E_F[T_i]=v_i.
\]
Since $|\epsilon_{ij}|\leq2B$,
\[
\operatorname{Var}_F(T_i)
\leq
\sum_j\mathbbm E_F[\epsilon_{ij}^4]
\leq4B^2v_i.
\]
Also $v_i\geq(1-\chi_n)\bar q_i>0$.  Lemma~\ref{lem:raw-sqrt} therefore
implies
\[
\mathbbm E_F[\sqrt{T_i}]
\geq
\sqrt{v_i}-\frac{2B^2}{\sqrt{v_i}},
\]
and analogously for active columns.  Summing yields
\begin{align}
\mathbbm E_F[\lVert \epsilon\rVert_\dagger]
&\geq
\sqrt{1-\chi_n}\,S
-
\frac{2B^2}{\sqrt{1-\chi_n}}H\nonumber\\
&=
\sqrt{1-\chi_n}\,S
\left(1-\frac{2\lambda}{1-\chi_n}\right).
\label{eq:raw-oracle-lower}
\end{align}
Because $\chi_n\leq1/2$ and $\lambda_n\leq1/8$,
\[
S
\leq
(1+\chi_n)(1+8\lambda_n)
\mathbbm E_F[\lVert \epsilon\rVert_\dagger]
\leq
(1+\chi_n+12\lambda_n)
\mathbbm E_F[\lVert \epsilon\rVert_\dagger].
\]

\medskip
\noindent\textbf{Step 3: dagger upper tail.}
The upper-tail form of Lemma~\ref{lem:raw-mcdiarmid}, together with the same
bound $L\leq\lambda S^2/B^2$, gives, for
$r_2=\lambda_n^{1/4}$,
\[
\mathbbm P_F\{\lVert Y\rVert_\dagger>\mathbbm E_F[\lVert Y\rVert_\dagger]+r_2S\}
\leq
\exp\left(-\frac{1}{2\sqrt{\lambda_n}}\right).
\]
On the complementary event, Jensen's inequality and the preceding display
give
\begin{align*}
\lVert Y\rVert_\dagger
&\leq(1+r_2)S\\
&\leq
(1+r_2)(1+\chi_n+12\lambda_n)
\mathbbm E_F[\lVert \epsilon\rVert_\dagger]\\
&\leq
(1+\chi_n+15\lambda_n^{1/4})
\mathbbm E_F[\lVert \epsilon\rVert_\dagger].
\end{align*}
For the last inequality, use $\chi_n\leq1/2$,
$\lambda_n\leq1/8$, and $\lambda_n\leq\lambda_n^{1/4}$.

\medskip
\noindent\textbf{Step 4: assembly.}
Let
\[
r_{\tau V,n}
=2\chi_n+15\lambda_n^{1/4}+a_n^{-1/4}.
\]
On the good events in Steps~1 and~3,
\[
\widehat{\bar\tau}^{\,0}
\leq
(1+r_{\tau V,n})\bar\tau.
\]
Likewise, term by term,
\[
(\widehat V^{\,0})^2
\leq
(1+r_{\tau V,n})V^2,
\]
which implies
\[
\widehat V^{\,0}\leq(1+r_{\tau V,n})V.
\]
Combining these upper-tail bounds with Lemma~\ref{lem:raw-one-sided} proves
the two-sided statements.
\end{proof}

\subsubsection{A shifted plug-in without a row-column energy-dispersion condition}
\label{sec:raw-shifted}

The dagger component of the unshifted raw plug-in requires the
energy-dispersion condition $\lambda_n\to0$ because rows or columns with small
expected energy make the square root severely concave. Shifting every energy
by $B^2$ removes this curvature problem at the cost of deliberate upward
bias: each zero-energy row or column contributes $B$ to the shifted norm even
though the corresponding oracle term is zero. The shifted plug-in therefore
provides a fallback when the energy-dispersion condition is considered
implausible, at the cost of additional conservatism.

\begin{corollary}[Shifted plug-in]
\label{cor:raw-shifted}
Suppose Assumption~\ref{ass:model} holds, with $B$ fixed across $\mathcal F$. Define
\[
\lVert Y\rVert_\dagger^{+}
=\sum_i\sqrt{q_i+B^2}+\sum_j\sqrt{p_j+B^2},
\]
\[
\widehat{\bar\tau}^{\,+}
=1.01\,\lVert Y\rVert_\dagger^{+}+0.25\,\lVert Y\rVert_F,
\qquad
\widehat V^{\,+}
=\left(\lVert Y\rVert_F^2+4B\,\lVert Y\rVert_\dagger^{+}
+B\,\lVert Y\rVert_F\right)^{1/2}.
\]
Suppose $a_n\to\infty$ and
\[
\underline{\kappa}_n\longrightarrow\infty.
\]
Define
\[
r_{\mathrm{sh},n}
=\underline{\kappa}_n^{-1/8}+a_n^{-1/4}.
\]
Then, for all sufficiently large $n$, $r_{\mathrm{sh},n}\leq1$ and
\[
\sup_{F\in\mathcal F(n)}\mathbbm P_F\left(
\frac{\bar\tau-\widehat{\bar\tau}^{\,+}}{\bar\tau}
>r_{\mathrm{sh},n}
\right)\longrightarrow0,
\]
\[
\sup_{F\in\mathcal F(n)}\mathbbm P_F\left(
\frac{V-\widehat V^{\,+}}{V}
>r_{\mathrm{sh},n}
\right)\longrightarrow0.
\]
More precisely, for all sufficiently large $n$,
\[
\begin{aligned}
&\sup_{F\in\mathcal F(n)}\mathbbm P_F\Big(
\widehat{\bar\tau}^{\,+}<(1-r_{\mathrm{sh},n})\bar\tau
\ \text{or}\ 
\widehat V^{\,+}<(1-r_{\mathrm{sh},n})V
\Big)\\
&\qquad\leq
\exp\left(-\frac{\sqrt{a_n}}{9B^2}\right)
+2\exp\left(
-\frac{3}{4\sqrt{2}}\,\underline{\kappa}_n^{3/8}
\right).
\end{aligned}
\]
\end{corollary}

\begin{proof}
By Lemma~\ref{lem:raw-var}(a),
$\operatorname{Var}_F(q_i)\leq B^2\bar q_i$, so the final
display of Lemma~\ref{lem:raw-sqrt}, with $a=B^2$, gives
\[
\mathbbm E_F[\sqrt{q_i+B^2}]\geq\sqrt{\bar q_i}
\]
for every row, including rows with $\bar q_i=0$, and analogously for columns.
Summing over all rows and columns and using Jensen's inequality as in
\eqref{eq:raw-oracle-upper},
\[
\mathbbm E_F[\lVert Y\rVert_\dagger^{+}]
\geq S
\geq\mathbbm E_F[\lVert \epsilon\rVert_\dagger].
\]

For the row component, write $R=\sum_iX_i$ with
$X_i=\sqrt{q_i+B^2}$. Distinct rows involve disjoint independent
coordinates, so the $X_i$ are independent under $\mathbbm P_F$. They
lie in $[B,B\sqrt{N_2+1}]$, and since
\[
X_i-\sqrt{\bar q_i+B^2}
=\frac{q_i-\bar q_i}
{\sqrt{q_i+B^2}+\sqrt{\bar q_i+B^2}},
\]
variance minimizes mean squared distance over constants, so
\[
\operatorname{Var}_F(X_i)
\leq
\frac{\operatorname{Var}_F(q_i)}{\bar q_i+B^2}
\leq\frac{B^2\bar q_i}{\bar q_i+B^2}
\leq B^2.
\]
Bernstein's inequality applied to
$\mathbbm E_F[X_i]-X_i$, with
$t=r\,\mathbbm E_F[R]$ and
$\mathbbm E_F[R]\geq N_1B$, gives
\[
\mathbbm P_F\{R<(1-r)\mathbbm E_F[R]\}
\leq
\exp\left[
 -\min\left\{
 \frac{r^2N_1}{4},
 \frac{3rN_1}{4\sqrt{N_2+1}}
 \right\}
\right],
\]
where
\[
\frac{t^2}{2(v+bt/3)}
\geq\frac14\min\left\{\frac{t^2}{v},\frac{3t}{b}\right\}
\]
with $v=B^2N_1$ and $b=B\sqrt{N_2+1}$. Take
$r=\underline{\kappa}_n^{-1/8}$. For every $F\in\mathcal F(n)$,
\[
N_1\geq\kappa(F)\geq\underline{\kappa}_n,
\qquad
\frac{N_1}{\sqrt{N_2+1}}
\geq\sqrt{\frac{\kappa(F)}{2}}
\geq\sqrt{\frac{\underline{\kappa}_n}{2}}.
\]
The two exponents are therefore at least
$\underline{\kappa}_n^{3/4}/4$ and
$\frac{3}{4\sqrt{2}}\underline{\kappa}_n^{3/8}$, respectively. For all
$\underline{\kappa}_n\geq8$, the latter is the smaller bound. The column
component satisfies the same inequality with $N_1$ and $N_2$ interchanged.
A union bound therefore gives, with $\mathbbm P_F$-probability at least
\[
1-2\exp\left(
-\frac{3}{4\sqrt{2}}\underline{\kappa}_n^{3/8}
\right),
\]
\[
\lVert Y\rVert_\dagger^{+}
\geq
(1-\underline{\kappa}_n^{-1/8})
\mathbbm E_F[\lVert \epsilon\rVert_\dagger].
\]

The Frobenius bound from Step~1 of the proof of
Lemma~\ref{lem:raw-one-sided}, with the same
$r_1=a_n^{-1/4}$, gives
\[
\lVert Y\rVert_F\geq(1-r_1)\lVert\sigma\rVert_F,
\qquad
\lVert Y\rVert_F^2\geq(1-r_1)\lVert\sigma\rVert_F^2
\]
except on an event of $\mathbbm P_F$-probability at most
$\exp\{-\sqrt{a_n}/(9B^2)\}$. On the intersection of these events, the
term-by-term assembly in Step~4 of that proof gives
\[
\widehat{\bar\tau}^{\,+}
\geq(1-r_{\mathrm{sh},n})\bar\tau,
\qquad
(\widehat V^{\,+})^2
\geq(1-r_{\mathrm{sh},n})V^2.
\]
Since $r_{\mathrm{sh},n}\leq1$ eventually,
$\widehat V^{\,+}\geq(1-r_{\mathrm{sh},n})V$. The coverage claim follows
by repeating the proof of Corollary~\ref{cor:raw-coverage}, which uses only
the one-sided lower bounds.
\end{proof}

\subsubsection{Remarks}\label{sec:raw-remarks}

\begin{remark}[Binary and undirected networks]\label{rem:raw-binary}
For a rectangular binary network, $B=1$, $q_i$ and $p_j$ are the degrees on
the two sides, and
\[
\chi_n=\sup_{F\in\mathcal F(n)}\max_{ij}\mu_{ij}.
\]
Since $\operatorname{Var}_F(Y_{ij}^2)=\sigma_{ij}^2$, the coefficient
$1/(9B^2)$ in the Bernstein exponents can also be sharpened to $3/8$.

For a loopless undirected network represented as an upper-triangular matrix,
\[
q_i=\sum_{j>i}Y_{ij}^2,
\qquad
p_i=\sum_{j<i}Y_{ji}^2.
\]
In the binary case, the ordinary degree of node $i$ is $q_i+p_i$.  A dyad $Y_{ij}$ with $i<j$ affects one row-energy term and one column-energy term, so its influence on $\lVert Y\rVert_\dagger$ is at most $2B$.
The proofs above therefore apply directly under the paper's upper-triangular
convention. The row components entering Corollary~\ref{cor:raw-shifted}
remain mutually independent under this convention, as do the column
components, because distinct rows and distinct columns of the upper-triangular
matrix involve disjoint dyads. The row and column collections need not be
independent of one another, because the proof treats them separately and
combines the two events by a union bound. When $N_1=N_2=N$,
$\kappa(F)=N$, so the uniform dimension condition in that corollary reduces
to $N\to\infty$.
\end{remark}

\begin{remark}[Interpretation of $\lambda$ and empirical diagnostics]
\label{rem:raw-lambda}
Pooling active rows and columns and assigning weights proportional to the
square root of expected energy gives
\[
\lambda
=\sum_k\omega_k\frac{B^2}{e_k},
\]
where $e_k$ denotes the relevant expected row or column energy.  Thus
$1/\lambda$ is a root-energy-weighted harmonic mean of $e_k/B^2$.  In a
homogeneous model, $\lambda=B^2/e$, so $\lambda_n\to0$ requires diverging
expected energy.  The condition is weaker than requiring the minimum active
energy to diverge because sufficiently low-energy rows or columns may carry
negligible square-root-energy mass.

A sample analogue formed from $q_i$ and $p_j$ may be reported as a descriptive
diagnostic, but the lemmas above do not establish that it consistently
estimates $\lambda$.  In particular, omitting rows or columns with zero
\emph{observed} energy can remove precisely the low-expected-energy terms that
make the population index large.  Likewise, a small maximum density among a
collection of large cells does not verify $\chi_n\to0$, because cell densities
are averages whereas $\chi_n$ is an entrywise maximum.
\end{remark}

\begin{remark}[A fixed signal-share bound]\label{rem:raw-conservatism}
Suppose $\chi_n\leq\bar\chi<1$, not necessarily with $\bar\chi\to0$.  Under
the remaining conditions of the relevant part of
Lemma~\ref{lem:raw-one-sided}, the same comparisons imply
\[
\widehat\sigma^{\,0}(g)
\leq
\{(1-\bar\chi)^{-1/2}+o_p(1)\}\sigma(g)
\]
uniformly over $g\in\mathcal G_c$, and analogously
\[
\widehat{\bar\tau}^{\,0}
\leq
\{(1-\bar\chi)^{-1/2}+o_p(1)\}\bar\tau,
\qquad
\widehat V^{\,0}
\leq
\{(1-\bar\chi)^{-1/2}+o_p(1)\}V.
\]
Thus Lemma~\ref{lem:raw-one-sided} still gives coverage, with asymptotic width
inflation bounded by the displayed constant.  When $\chi_n\to0$, this factor
converges to one and Lemma~\ref{lem:raw-two-sided} recovers full calibration.

The global nondegeneracy condition is redundant under a fixed signal-share
bound and $\lambda_n\to0$.  To see this, let
$T=\sum_{ij}\vartheta_{ij}$.  Since every active row or column energy is at
most $T$,
\[
H\geq\frac{S}{T},
\qquad
\lambda=B^2\frac{H}{S}\geq\frac{B^2}{T}.
\]
Hence $T\geq B^2/\lambda$, and
\[
\lVert \sigma\rVert_F^2\geq(1-\bar\chi)T
\geq\frac{(1-\bar\chi)B^2}{\lambda}.
\]
\end{remark}

\begin{remark}[Which results use which halves]\label{rem:raw-halves}
The proofs of Propositions~\ref{prop:propone} and~\ref{prop:proptwo} use only the underestimation halves of
Assumption~\ref{ass:scale-consistency}. Proposition~\ref{prop:propone} needs
$\widehat\sigma^{\,0}(g)\geq(1-o_p(1))\sigma(g)$ uniformly in $\mathbbm P_F$-probability, and Proposition~\ref{prop:proptwo} needs $\widehat{\bar\tau}^{\,0}\geq(1-o_p(1))\bar\tau$ and $\widehat V^{\,0}\geq(1-o_p(1))V$.  Lemma~\ref{lem:raw-one-sided} is therefore enough for simultaneous coverage.  The optimality results use all of Assumption~\ref{ass:scale-consistency}. Proposition~\ref{prop:propthree} uses two-sided consistency of the cell scale on its least-favorable subclass, and Proposition~\ref{prop:propfour} uses an upper bound on $\widehat{\bar\tau}^{\,0}$.  Those requirements are supplied by Lemma~\ref{lem:raw-two-sided}.
\end{remark}

\subsection{Residual plug-in based on singular-value thresholding}
\label{sec:usvt}
The third approach estimates the conditional mean matrix $\mu$ using the singular-value-thresholding estimator of \citet{chatterjee2015matrix}. It then forms the residuals $\hat\epsilon_{ij}=Y_{ij}-\hat\mu_{ij}$ and uses them as plug-ins into formulas for $\sigma(g_1,g_2)$, $\bar{\tau}$, and $V$. The main restriction justifying this approach is that $\mu$ is effectively low rank. Section~\ref{sec:usvt-algorithm} gives the estimator,
Section~\ref{sec:usvt-assumptions} states sufficient conditions,
Section~\ref{sec:usvt-consistency} gives the consistency results, and
Section~\ref{sec:usvt-proofs} contains their proofs. Sections
\ref{sec:sparse-thresholds} and~\ref{sec:cbar-calibration} develop a sparse
threshold and its degree-adaptive calibration.

\subsubsection{Algorithm}\label{sec:usvt-algorithm}
The following algorithm builds on Section 1.2 of \citet{chatterjee2015matrix}. Since there is no missing data in our setting, $\hat p=1$. Let
\[
N_{\min}:=\min\{N_1,N_2\},
\qquad
N_{\max}:=\max\{N_1,N_2\}.
\]
The researcher chooses a fixed $\eta\in(0,1)$ that does not depend on the data or vary with $N_1,N_2$. \citet{chatterjee2015matrix} suggests $\eta=0.01$.

\begin{itemize}
\item[1.] Define $\tilde Y:=Y/B$, so the entries of $\tilde Y$ are bounded by $1$ in absolute value.

\item[2.] Let
\[
\tilde Y=\sum_{i\in[N_{\min}]}s_i u_i v_i^T
\]
be a singular value decomposition, with the singular values in nonincreasing order.

\item[3.] Define
\[
S:=\left\{i:s_i\geq(2+\eta)\sqrt{N_{\max}}\right\},
\qquad
W:=\sum_{i\in S}s_i u_i v_i^T.
\]

\item[4.] Let $\Pi_{[-1,1]}$ denote entrywise Euclidean projection onto $[-1,1]$, and define
\[
\hat\mu:=B\,\Pi_{[-1,1]}(W).
\]

\item[5.] Define
\begin{align*}
\hat{\bar\tau}
&:=1.01\|Y-\hat\mu\|_{\dagger}+0.25\|Y-\hat\mu\|_F,\\
\hat V
&:=\left(\|Y-\hat\mu\|_F^2+B\|Y-\hat\mu\|_F
+4B\|Y-\hat\mu\|_{\dagger}\right)^{1/2},\\
\hat\sigma(g_1,g_2)
&:=\frac{\|(Y-\hat\mu)g_1g_2\|_F}{\sqrt{m_1m_2}},
\qquad (g_1,g_2)\in\mathcal G_c.
\end{align*}
Here $(Y-\hat\mu)g_1g_2$ denotes the matrix with $ij$th entry
$(Y_{ij}-\hat\mu_{ij})g_{i,1}g_{j,2}$; equivalently,
\[
(Y-\hat\mu)g_1g_2
=\operatorname{diag}(g_1)(Y-\hat\mu)\operatorname{diag}(g_2).
\]
\end{itemize}

\begin{remark}\label{rem:usvt-undirected}
For an undirected loopless network represented by the upper-triangular convention of Section~\ref{sec:model-terminology}, all spectral estimators in this section are applied to the symmetrized matrix
\[
Y^s:=Y+Y^T,
\qquad
\mu^s:=\mu+\mu^T.
\]
After thresholding, the estimated diagonal is set to zero and the upper-triangular entries are used to define $\hat\mu$ and the residual estimators in Step 5. Assumptions~\ref{ass:usvt-signal-noise} and \ref{ass:usvt-low-rank} below are imposed on the symmetrized mean matrix $\mu^s$; the corresponding Frobenius noise scale satisfies $\|\sigma^s\|_F=\sqrt2\|\sigma\|_F$. If the estimated diagonal is set to zero, then
\[
\|\hat\mu^s-\mu^s\|_F^2=2\|\hat\mu-\mu\|_F^2,
\]
so all mean-estimation bounds transfer to the upper-triangular representation up to fixed constants. This symmetrization is important: the upper-triangular representation of even a homogeneous off-diagonal mean matrix need not be effectively low rank, whereas its symmetrization has a uniformly bounded nuclear-to-Frobenius norm ratio.
\end{remark}

\begin{remark}[The rank-zero branch and the unshifted raw plug-in]
\label{rem:usvt-rankzero}
If the retained singular-value set in Step~3 is empty, then $W=0$ and
$\hat\mu=0$. On that event, the residual estimators in this section coincide
numerically with the unshifted raw plug-ins
$\widehat\sigma^{\,0}$, $\widehat{\bar\tau}^{\,0}$, and
$\widehat V^{\,0}$ from Section~\ref{sec:raw-plugin}. A rank-zero estimate in
one realized dataset does not by itself transfer the repeated-sampling results
for the unshifted raw plug-ins to the adaptive spectral procedure, because the
event $\{\hat\mu=0\}$ depends on $Y$.

There is, however, a formal high-probability connection. Let
$E_n=\{\hat\mu=0\}$. For any failure event $A_n$ appearing in the one-sided
or two-sided scale requirements, the spectral and unshifted raw estimators
agree on $E_n$. Thus, for every fixed $F$,
\[
\mathbbm P_F(A_n^{\mathrm{USVT}})
\leq
\mathbbm P_F(A_n^{0})+\mathbbm P_F(E_n^c).
\]
Consequently, if the conditions of the relevant unshifted raw-plugin result in
Section~\ref{sec:raw-plugin} hold and
\[
\sup_{F\in\mathcal F(n)}
\mathbbm P_F(\hat\mu\neq0)\longrightarrow0,
\]
then the USVT residual estimators inherit that result: the one-sided coverage
conclusion under Lemma~\ref{lem:raw-one-sided}, or
Assumption~\ref{ass:scale-consistency} under
Lemma~\ref{lem:raw-two-sided}. This numerical identity does not apply to the
shifted plug-ins in Section~\ref{sec:raw-shifted}, which replace
$\lVert Y\rVert_\dagger$ by $\lVert Y\rVert_\dagger^{+}$. When the rank-zero
event does not occur with probability approaching one, the general USVT
results below remain the relevant justification.
\end{remark}

The dense threshold in Step 3 can be conservative for sparse Bernoulli networks. Section~\ref{sec:sparse-thresholds} gives a formal alternative under additional sparse-graph regularity conditions.

\subsubsection{Assumptions}\label{sec:usvt-assumptions}
We impose four additional assumptions on the class $\mathcal F$. Assumption~\ref{ass:usvt-nondegeneracy} is a nondegeneracy condition, Assumption~\ref{ass:usvt-signal-noise} compares the magnitude of the systematic and idiosyncratic components, Assumption~\ref{ass:usvt-low-rank} is an effective low-rank condition, and Assumption~\ref{ass:usvt-rates} requires the mean-estimation and empirical residual-norm errors to be negligible at the scales used by the confidence intervals.

The first additional assumption is a lower bound on the amount of residual variation.
\begin{assumption}\label{ass:usvt-nondegeneracy}
\begin{itemize}
\item[w.] 
\[
\liminf_{F\in\mathcal F:N_1,N_2\to\infty}
\frac{\|\sigma\|_F^2}{N_1+N_2}=\infty.
\]
\item[s.]
\[
\liminf_{F\in\mathcal F:N_1,N_2\to\infty}
\frac{S_{\sigma,c}}{\sqrt{N_1+N_2}}=\infty,
\qquad
S_{\sigma,c}:=\min_{(g_1,g_2)\in\mathcal G_c}
\|\sigma g_1g_2\|_F.
\]
\end{itemize}
\end{assumption}

Assumption~\ref{ass:usvt-nondegeneracy}(w) requires the total residual variance to diverge faster than $N_1+N_2$. In a balanced Bernoulli network with typical expected degree $d_N$, it roughly requires $d_N\to\infty$. Thus it permits sparse networks with diverging average degree, but excludes bounded-degree regimes. This is stronger than the weak nondegeneracy condition needed for the oracle $CI_2$, reflecting the additional difficulty of estimating its scale quantities from one observed network. Assumption~\ref{ass:usvt-nondegeneracy}(s) imposes the analogous requirement uniformly over the large group pairs in $\mathcal G_c$.

The second additional assumption bounds the systematic component relative to the residual component.
\begin{assumption}\label{ass:usvt-signal-noise}
\[
\limsup_{F\in\mathcal F:N_1,N_2\to\infty}
\frac{\|\mu\|_F}{\|\sigma\|_F}<\infty.
\]
\end{assumption}

If this condition is badly violated, then a small error in estimating the mean matrix can be large relative to the residual scale that the plug-in estimators are intended to recover.

The third additional assumption is an effective low-rank condition. Let $\|\mu\|_*$ denote the nuclear norm of $\mu$.
\begin{assumption}\label{ass:usvt-low-rank}
\[
\limsup_{F\in\mathcal F:N_1,N_2\to\infty}
\frac{\|\mu\|_*}{\|\mu\|_F}<\infty,
\]
where the ratio is defined to be zero when $\mu=0$.
\end{assumption}

Assumption~\ref{ass:usvt-low-rank} is weaker than uniformly bounded rank because
\[
\frac{\|\mu\|_*}{\|\mu\|_F}
\leq\sqrt{\operatorname{rank}(\mu)}.
\]
It instead requires the singular values to be sufficiently concentrated. The condition is satisfied by fixed-rank interactive fixed-effects models, stochastic blockmodels with a fixed number of blocks, random dot-product graphs with fixed latent dimension, and some gravity, beta-model, or latent-space specifications under additional finite-dimensional or smoothness restrictions. See, for example, Section 3 of \citet{alidaee2020recovering}. Related nuclear-norm restrictions are used by \citet{moon2018nuclear,beyhum2019square,chernozhukov2023inference}. We regard Assumption~\ref{ass:usvt-low-rank} as a substantive and potentially restrictive condition.

For the fourth assumption, define
\begin{align*}
L_N&:=\sqrt{N_1}+\sqrt{N_2},\\
q_N&:=B\sqrt{N_{\max}},\\
T_\sigma
&:=\sum_{i\in[N_1]}\left(\sum_{j\in[N_2]}\sigma_{ij}^2\right)^{1/2}
+\sum_{j\in[N_2]}\left(\sum_{i\in[N_1]}\sigma_{ij}^2\right)^{1/2}.
\end{align*}
\begin{assumption}\label{ass:usvt-rates}
\begin{itemize}
\item[w.] There exists a deterministic sequence $r=r_{N_1,N_2}>0$ such that $r\to0$ and
\begin{align*}
&\liminf_{F\in\mathcal F:N_1,N_2\to\infty}
\frac{\mathbbm E_F\|\epsilon\|_\dagger}{T_\sigma}>0,\\
&\liminf_{F\in\mathcal F:N_1,N_2\to\infty}
\frac{rT_\sigma}{\sqrt{N_1N_2}}=\infty,\\
&\limsup_{F\in\mathcal F:N_1,N_2\to\infty}
\frac{L_N\sqrt{q_N\|\sigma\|_F}}{rT_\sigma}=0.
\end{align*}

\item[s.] There exists a deterministic sequence $r=r_{N_1,N_2}>0$ such that $r\to0$ and
\begin{align*}
&\liminf_{F\in\mathcal F:N_1,N_2\to\infty}
\frac{rS_{\sigma,c}}{\sqrt{N_1+N_2}}=\infty,\\
&\limsup_{F\in\mathcal F:N_1,N_2\to\infty}
\frac{\sqrt{q_N\|\sigma\|_F}}{rS_{\sigma,c}}=0.
\end{align*}
\end{itemize}
\end{assumption}

The first condition in Assumption~\ref{ass:usvt-rates}(w) says that $T_\sigma$ gives the correct order of $\mathbbm E_F\|\epsilon\|_\dagger$. The second guarantees concentration of the empirical row-column and Frobenius residual norms at the scale needed for $CI_2$. The final condition makes the USVT mean-estimation error negligible relative to $rT_\sigma$. Assumption~\ref{ass:usvt-rates}(s) supplies the uniform concentration and mean-estimation rates needed for $\hat\sigma(g_1,g_2)$.

\subsubsection{Consistency}\label{sec:usvt-consistency}
Our first result gives sufficient conditions for the two scale estimators used by $CI_2$.

\begin{proposition}\label{prop:usvt-bartau-V}
Suppose Assumptions~\ref{ass:model}, \ref{ass:usvt-nondegeneracy}(w), \ref{ass:usvt-signal-noise}, \ref{ass:usvt-low-rank}, and \ref{ass:usvt-rates}(w) hold, and construct $\hat{\bar\tau}$ and $\hat V$ according to Section~\ref{sec:usvt-algorithm}. Then Assumption~\ref{ass:nondegeneracy}(w) of the main text holds. Moreover, there exists a deterministic sequence $r_5\to0$ such that
\begin{align*}
\lim_{F\in\mathcal F:N_1,N_2\to\infty}
\mathbbm P_F\left(
\left|\frac{\hat{\bar\tau}-\bar\tau}{\bar\tau}\right|>r_5
\right)&=0,\\
\lim_{F\in\mathcal F:N_1,N_2\to\infty}
\mathbbm P_F\left(
\left|\frac{\hat V-V}{V}\right|>r_5
\right)&=0.
\end{align*}
Thus Assumptions~\ref{ass:scale-consistency}(ii) and \ref{ass:scale-consistency}(iii) of the main text hold.
\end{proposition}

Our second result gives sufficient conditions for the cell-specific scale estimator used by $CI_1$.

\begin{proposition}\label{prop:usvt-sigma}
Suppose Assumptions~\ref{ass:model}, \ref{ass:usvt-nondegeneracy}(s), \ref{ass:usvt-signal-noise}, \ref{ass:usvt-low-rank}, and \ref{ass:usvt-rates}(s) hold, and construct $\hat\sigma(g_1,g_2)$ according to Section~\ref{sec:usvt-algorithm}. Then Assumption~\ref{ass:nondegeneracy}(s) of the main text holds. Moreover, there exists a deterministic sequence $r_6\to0$ such that
\[
\lim_{F\in\mathcal F:N_1,N_2\to\infty}
\mathbbm P_F\left(
\max_{(g_1,g_2)\in\mathcal G_c}
\left|
\frac{\hat\sigma(g_1,g_2)-\sigma(g_1,g_2)}
{\sigma(g_1,g_2)}
\right|>r_6
\right)=0.
\]
Thus Assumption~\ref{ass:scale-consistency}(i) of the main text holds.
\end{proposition}

If the assumptions of both propositions hold, let $r_*:=\max\{r_5,r_6\}$. Then $r_*\to0$ and all three parts of Assumption~\ref{ass:scale-consistency} hold with the common sequence $r_*$. Consequently, the conditions in this section verify the nondegeneracy and scale-consistency assumptions used in Propositions~\ref{prop:propone}--\ref{prop:propfour}. The additional
hypotheses stated in each proposition remain in force.

\subsubsection{Proofs of the consistency propositions}\label{sec:usvt-proofs}
We first record a deterministic singular-value thresholding inequality that will be used for both the dense and sparse thresholds.

\begin{lemma}\label{lem:svt-oracle}
Let $Y$ and $\mu$ be $N_1\times N_2$ real matrices. Let $\mathcal I\subset\mathbb R$ be a closed interval containing every entry of $\mu$, and let
\[
Y=\sum_{i\in[N_{\min}]}s_i u_i v_i^T
\]
be a singular value decomposition. Fix $t>0$ and $\delta>0$, and suppose
\[
(1+\delta)\|Y-\mu\|_{\mathrm{op}}\leq t.
\]
Define
\[
W:=\sum_{i:s_i\geq t}s_i u_i v_i^T,
\]
and let $\hat\mu$ be the entrywise Euclidean projection of $W$ onto $\mathcal I$. Then
\[
\|\hat\mu-\mu\|_F
\leq
\sqrt{2\left(2+\delta^{-1}\right)}
\left(t\|\mu\|_*\right)^{1/2}.
\]
\end{lemma}

\begin{proof}
Let $S:=\{i:s_i(Y)\geq t\}$. Weyl's inequality gives
\[
|s_i(Y)-s_i(\mu)|\leq\|Y-\mu\|_{\mathrm{op}}
\leq\frac{t}{1+\delta}.
\]
First,
\[
\|Y-W\|_{\mathrm{op}}<t,
\qquad
\|W-\mu\|_{\mathrm{op}}
\leq t+\frac{t}{1+\delta}
\leq2t.
\]
Second, for each $i\in S$,
\[
s_i(\mu)\geq t-\frac{t}{1+\delta}
=\frac{\delta t}{1+\delta},
\]
so
\[
|S|\leq\frac{(1+\delta)\|\mu\|_*}{\delta t}.
\]
Third,
\begin{align*}
\|W\|_*
&=\sum_{i\in S}s_i(Y)\\
&\leq\sum_{i\in S}s_i(\mu)
+|S|\frac{t}{1+\delta}\\
&\leq\left(1+\delta^{-1}\right)\|\mu\|_*.
\end{align*}
Therefore,
\begin{align*}
\|W-\mu\|_F^2
&\leq\|W-\mu\|_{\mathrm{op}}
\|W-\mu\|_*\\
&\leq2t\left(\|W\|_*+\|\mu\|_*\right)\\
&\leq2t\left(2+\delta^{-1}\right)\|\mu\|_*.
\end{align*}
Projection onto $\mathcal I$ is coordinatewise nonexpansive relative to every point in $\mathcal I$, so $\|\hat\mu-\mu\|_F\leq\|W-\mu\|_F$.
\end{proof}

\begin{remark}\label{rem:usvt-uniform}
All stochastic-order statements below are uniform over $F\in\mathcal F$ as $N_1,N_2\to\infty$. For a family of random variables $Z=Z_F$ and positive, possibly $F$-dependent, deterministic scales $a=a_F$, we write $Z=O_p(a)$ if
\[
\lim_{M\to\infty}
\limsup_{F\in\mathcal F:N_1,N_2\to\infty}
\mathbbm P_F(|Z|>Ma)=0,
\]
and $Z=o_p(a)$ if, for every $\varepsilon>0$,
\[
\limsup_{F\in\mathcal F:N_1,N_2\to\infty}
\mathbbm P_F(|Z|>\varepsilon a)=0.
\]
The probabilistic inputs are Bernstein's inequality, the bounded-differences inequality, and the fixed-threshold spectral-norm bound in Theorem 3.4 of \citet{chatterjee2015matrix}. Their constants depend only on the fixed entry bound and tuning constants, not on other features of $F$. Lemma~\ref{lem:svt-oracle} is deterministic.
\end{remark}

\begin{proof}[Proof of Proposition~\ref{prop:usvt-bartau-V}]
Assumption~\ref{ass:usvt-nondegeneracy}(w) immediately implies $\|\sigma\|_F\to\infty$, so Assumption~\ref{ass:nondegeneracy}(w) of the main text holds. Let $r$ be the sequence in Assumption~\ref{ass:usvt-rates}(w), let
\[
\epsilon:=Y-\mu,
\qquad
\Delta:=\hat\mu-\mu,
\qquad
L_N:=\sqrt{N_1}+\sqrt{N_2},
\]
and interpret all stochastic orders in the sense of Remark~\ref{rem:usvt-uniform}.

We first control the error from estimating $\mu$. Write
\[
\tilde\mu:=\mu/B,
\qquad
\tilde\epsilon:=(Y-\mu)/B,
\qquad
t_N:=(2+\eta)\sqrt{N_{\max}}.
\]
Under $\mathbbm P_F$, the entries of $\tilde\epsilon/2$ are independent, mean zero, bounded by one, and have variances bounded by $1/4$. Applying Theorem 3.4 of \citet{chatterjee2015matrix} with any fixed $\eta_0\in(0,\eta)$ gives
\[
\mathbbm P_F\left(
\|\tilde\epsilon\|_{\mathrm{op}}
>(2+\eta_0)\sqrt{N_{\max}}
\right)\to0
\]
uniformly over $\mathcal F$. Fix $\delta_0>0$ such that
\[
(1+\delta_0)(2+\eta_0)\leq2+\eta.
\]
On the resulting high-probability event, Lemma~\ref{lem:svt-oracle}, applied to $\tilde Y$, $\tilde\mu$, $t_N$, and $\mathcal I=[-1,1]$, gives
\begin{align}\label{eq:dense-delta-rate}
\|\Delta\|_F
&\leq
B\sqrt{2(2+\delta_0^{-1})}
\left(t_N\|\tilde\mu\|_*\right)^{1/2}\notag\\
&=O_p\left(\sqrt{q_N\|\mu\|_*}\right)
=O_p\left(\sqrt{q_N\|\sigma\|_F}\right),
\end{align}
where the final equality uses Assumptions~\ref{ass:usvt-signal-noise} and \ref{ass:usvt-low-rank}. Since
\[
\|\Delta\|_\dagger
\leq L_N\|\Delta\|_F,
\]
Assumption~\ref{ass:usvt-rates}(w) yields
\begin{align}
\|\Delta\|_\dagger&=o_p(rT_\sigma),\label{eq:dense-delta-dagger}\\
\|\Delta\|_F&=o_p(r\|\sigma\|_F).\label{eq:dense-delta-frob}
\end{align}
For the second display, use $T_\sigma\leq L_N\|\sigma\|_F$ to obtain
\[
\frac{\sqrt{q_N\|\sigma\|_F}}{r\|\sigma\|_F}
\leq
\frac{L_N\sqrt{q_N\|\sigma\|_F}}{rT_\sigma}
\to0.
\]
Assumption~\ref{ass:usvt-rates}(w) also gives a constant $c_\dagger>0$ such that
\[
\mathbbm E_F\|\epsilon\|_\dagger\geq c_\dagger T_\sigma
\]
for all sufficiently large dimensions. Consequently,
\[
\bar\tau
=1.01\mathbbm E_F\|\epsilon\|_\dagger+0.25\|\sigma\|_F
\geq1.01c_\dagger T_\sigma,
\]
and \eqref{eq:dense-delta-dagger} implies
\[
\|\Delta\|_\dagger=o_p(r\bar\tau).
\]
Because $V\geq\|\sigma\|_F$, \eqref{eq:dense-delta-frob} also gives
\[
\|\Delta\|_F=o_p(rV).
\]

Next, changing one entry of $\epsilon$ changes $\|\epsilon\|_\dagger$ by at most a fixed multiple of $B$. Lemma~\ref{lem:raw-mcdiarmid} therefore gives
\[
\|\epsilon\|_\dagger-\mathbbm E_F\|\epsilon\|_\dagger
=O_p\left(B\sqrt{N_1N_2}\right).
\]
Since $rT_\sigma/\sqrt{N_1N_2}\to\infty$ and $\bar\tau\gtrsim T_\sigma$,
\begin{align}\label{eq:dagger-concentration}
\left|
\|\epsilon\|_\dagger-\mathbbm E_F\|\epsilon\|_\dagger
\right|
=o_p(r\bar\tau).
\end{align}

Bernstein's inequality applied to
$\sum_{i,j}(\epsilon_{ij}^2-\sigma_{ij}^2)$ gives
\begin{align}\label{eq:frob-square-concentration}
\left|
\|\epsilon\|_F^2-\|\sigma\|_F^2
\right|
=O_p(B\|\sigma\|_F).
\end{align}
Assumption~\ref{ass:usvt-rates}(w), together with $T_\sigma\leq L_N\|\sigma\|_F$ and $\sqrt{N_1N_2}/L_N\to\infty$, implies $r\|\sigma\|_F\to\infty$. Hence
\begin{align}
\left|
\|\epsilon\|_F^2-\|\sigma\|_F^2
\right|&=o_p(r\|\sigma\|_F^2)=o_p(rV^2),\label{eq:frob-square-small}\\
\frac{\|\epsilon\|_F}{\|\sigma\|_F}&=1+o_p(1),\label{eq:frob-ratio}
\end{align}
and therefore
\begin{align}\label{eq:frob-norm-small}
\left|
\|\epsilon\|_F-\|\sigma\|_F
\right|
=O_p(B)=o_p(r\bar\tau).
\end{align}
Moreover,
\[
B\left|
\|\epsilon\|_F-\|\sigma\|_F
\right|
=o_p(rV^2),
\]
because $V^2\geq4Bc_\dagger T_\sigma$ and $rT_\sigma\to\infty$.

Using the reverse triangle inequality,
\begin{align*}
|\hat{\bar\tau}-\bar\tau|
&\leq1.01\|\Delta\|_\dagger
+1.01\left|
\|\epsilon\|_\dagger-
\mathbbm E_F\|\epsilon\|_\dagger
\right|\\
&\quad+0.25\|\Delta\|_F
+0.25\left|
\|\epsilon\|_F-\|\sigma\|_F
\right|\\
&=o_p(r\bar\tau).
\end{align*}
Thus the relative error of $\hat{\bar\tau}$ is $o_p(r)$.

For $\hat V$, write
\begin{align*}
|\hat V^2-V^2|
&\leq
\left|
\|\epsilon-\Delta\|_F^2-\|\sigma\|_F^2
\right|\\
&\quad+B\left|
\|\epsilon-\Delta\|_F-\|\sigma\|_F
\right|\\
&\quad+4B\left|
\|\epsilon-\Delta\|_\dagger-
\mathbbm E_F\|\epsilon\|_\dagger
\right|.
\end{align*}
For the first term,
\begin{align*}
\left|
\|\epsilon-\Delta\|_F^2-\|\sigma\|_F^2
\right|
&\leq
\left|
\|\epsilon\|_F^2-\|\sigma\|_F^2
\right|
+2\|\epsilon\|_F\|\Delta\|_F
+\|\Delta\|_F^2.
\end{align*}
By \eqref{eq:frob-ratio}, $\|\epsilon\|_F=O_p(\|\sigma\|_F)=O_p(V)$. Together with $\|\Delta\|_F=o_p(rV)$ and $r\to0$, the last two terms are respectively $o_p(rV^2)$ and $o_p(r^2V^2)=o_p(rV^2)$. The first term is $o_p(rV^2)$ by \eqref{eq:frob-square-small}. The remaining two summands are $o_p(rV^2)$ by \eqref{eq:dense-delta-frob}, \eqref{eq:frob-norm-small}, \eqref{eq:dense-delta-dagger}, and \eqref{eq:dagger-concentration}, using $B\bar\tau\lesssim V^2$. Hence
\[
|\hat V^2-V^2|=o_p(rV^2).
\]
Since $|\sqrt a-\sqrt b|\leq\sqrt{|a-b|}$ for nonnegative $a,b$,
\[
\frac{|\hat V-V|}{V}=o_p(\sqrt r).
\]
Taking $r_5:=3\sqrt r$ proves the proposition.
\end{proof}

\begin{proof}[Proof of Proposition~\ref{prop:usvt-sigma}]
Assumption~\ref{ass:usvt-nondegeneracy}(s) implies Assumption~\ref{ass:nondegeneracy}(s) of the main text. Indeed, for every $(g_1,g_2)\in\mathcal G_c$,
\begin{align*}
\sqrt{\min(N_1,N_2)}\,\sigma(g_1,g_2)
&=\frac{\sqrt{\min(N_1,N_2)}}{\sqrt{m_1m_2}}
\|\sigma g_1g_2\|_F\\
&\geq\frac{\|\sigma g_1g_2\|_F}{\sqrt{N_{\max}}}\\
&\geq\frac{\|\sigma g_1g_2\|_F}{\sqrt{N_1+N_2}},
\end{align*}
which diverges uniformly by Assumption~\ref{ass:usvt-nondegeneracy}(s).

Let $r$ be the sequence in Assumption~\ref{ass:usvt-rates}(s), and let $\epsilon:=Y-\mu$ and $\Delta:=\hat\mu-\mu$. The same dense spectral event and Lemma~\ref{lem:svt-oracle} used in the preceding proof give
\[
\|\Delta\|_F
=O_p\left(\sqrt{q_N\|\sigma\|_F}\right).
\]
Therefore,
\begin{align}\label{eq:prop6-mean-error}
\max_{(g_1,g_2)\in\mathcal G_c}
\frac{\|\Delta g_1g_2\|_F}
{\|\sigma g_1g_2\|_F}
\leq
\frac{\|\Delta\|_F}{S_{\sigma,c}}
=o_p(r)
\end{align}
by Assumption~\ref{ass:usvt-rates}(s).

For the empirical residual norms, define
\begin{align*}
A(g_1,g_2)
&:=\|\sigma g_1g_2\|_F^2,\\
S(g_1,g_2)
&:=\sum_{i,j}
(\epsilon_{ij}^2-\sigma_{ij}^2)g_{i,1}g_{j,2}.
\end{align*}
For any fixed $(g_1,g_2)$, Bernstein's inequality gives, for all sufficiently small $r$,
\[
\mathbbm P_F\left(
|S(g_1,g_2)|>rA(g_1,g_2)
\right)
\leq
2\exp\left(
-C\frac{r^2A(g_1,g_2)}{B^2}
\right)
\]
for a universal constant $C>0$. Since $|\mathcal G_c|\leq2^{N_1+N_2}$ and $A(g_1,g_2)\geq S_{\sigma,c}^2$, the union bound yields
\[
\mathbbm P_F\left(
\max_{(g_1,g_2)\in\mathcal G_c}
\frac{|S(g_1,g_2)|}{A(g_1,g_2)}>r
\right)
\leq
2^{N_1+N_2+1}
\exp\left(-C\frac{r^2S_{\sigma,c}^2}{B^2}\right)
\to0
\]
by Assumption~\ref{ass:usvt-rates}(s). Hence
\begin{align}\label{eq:prop6-residual-error}
\max_{(g_1,g_2)\in\mathcal G_c}
\frac{
\left|
\|\epsilon g_1g_2\|_F-
\|\sigma g_1g_2\|_F
\right|}
{\|\sigma g_1g_2\|_F}
=o_p(r),
\end{align}
because $|\sqrt a-\sqrt b|\leq|a-b|/\sqrt b$ for $a,b\geq0$.

Finally, the reverse triangle inequality gives
\begin{align*}
\left|
\frac{\hat\sigma(g_1,g_2)-\sigma(g_1,g_2)}
{\sigma(g_1,g_2)}
\right|
&\leq
\frac{\|\Delta g_1g_2\|_F}
{\|\sigma g_1g_2\|_F}\\
&\quad+
\frac{
\left|
\|\epsilon g_1g_2\|_F-
\|\sigma g_1g_2\|_F
\right|}
{\|\sigma g_1g_2\|_F}.
\end{align*}
Combining \eqref{eq:prop6-mean-error} and \eqref{eq:prop6-residual-error} proves the result with $r_6:=2r$.
\end{proof}

\subsubsection{Sparse threshold}\label{sec:sparse-thresholds}
The dense threshold in Section~\ref{sec:usvt-algorithm} is calibrated to the
worst case permitted by Assumption~\ref{ass:model}, in which residual variances may be of
constant order. For a sparse Bernoulli network, the relevant row and column
variance scales can be much smaller than $N_{\max}$, and the operator norm of
$Y-\mu$ can therefore be much smaller than $\sqrt{N_{\max}}$. Here we provide a
sparse threshold under additional regularity conditions. Section
\ref{sec:cbar-calibration} then gives a data-driven calibration based on
observed degrees, together with a dense fallback.

Throughout this subsection and the next, the entries of $Y$ are Bernoulli and
$B=1$. For a directed or bipartite network, define
\[
\rho:=\frac{1}{N_1N_2}\sum_{i\in[N_1],j\in[N_2]}\mu_{ij},
\qquad
\hat\rho:=\frac{1}{N_1N_2}\sum_{i\in[N_1],j\in[N_2]}Y_{ij},
\qquad
\varrho_N:=\sqrt{N_{\max}\rho},
\]
and let
\[
d_N:=N_1+N_2.
\]
Define the expected row and column degrees and their observed analogues by
\[
e_{i,1}:=\sum_{j\in[N_2]}\mu_{ij},
\qquad
e_{j,2}:=\sum_{i\in[N_1]}\mu_{ij},
\qquad
d_{i,1}:=\sum_{j\in[N_2]}Y_{ij},
\qquad
d_{j,2}:=\sum_{i\in[N_1]}Y_{ij}.
\]
Thus $N_2\rho$ and $N_1\rho$ are the average expected row and column
degrees. Define
\begin{align}
C_N^*
&:=\max\left\{
\max_{i\in[N_1]}
\frac{\sum_{j\in[N_2]}\sigma_{ij}^2}{N_2\rho},
\;
\max_{j\in[N_2]}
\frac{\sum_{i\in[N_1]}\sigma_{ij}^2}{N_1\rho}
\right\},\label{eq:Cstar-directed}\\
D_N^*
&:=\max\left\{
\max_{i\in[N_1]}\frac{e_{i,1}}{N_2\rho},
\;
\max_{j\in[N_2]}\frac{e_{j,2}}{N_1\rho}
\right\}.\label{eq:Dstar-directed}
\end{align}

For an undirected loopless network on $N$ nodes, compute $\rho$ and $\hat\rho$
over the $\binom N2$ admissible unordered dyads, set
\[
\varrho_N:=\sqrt{N\rho},
\qquad
d_N:=N,
\]
and define
\begin{align}
C_N^*
&:=\max_{i\in[N]}
\frac{\sum_{j\neq i}\sigma_{ij}^2}{(N-1)\rho},
\qquad
D_N^*
:=\max_{i\in[N]}
\frac{\sum_{j\neq i}\mu_{ij}}{(N-1)\rho}.\label{eq:C-D-undirected}
\end{align}
The spectral estimator is applied to the symmetrized matrix as in Remark
\ref{rem:usvt-undirected}.

These ratios are well-defined under the assumption that $\rho>0$, which holds eventually under
the conditions imposed below. Because
$\sigma_{ij}^2=\mu_{ij}(1-\mu_{ij})\leq\mu_{ij}$ for Bernoulli entries,
\[
0\leq C_N^*\leq D_N^*,
\qquad
D_N^*\geq1.
\]
Thus $D_N^*$ is a conservative upper proxy for $C_N^*$. The two ratios need
not be close: high expected degrees may arise from nearly deterministic links,
which contribute little conditional variance. Under the additional sparse
probability condition $\max_{i,j}\mu_{ij}\to0$, however, each row and column
variance sum is uniformly asymptotic to its corresponding expected degree, and
therefore $C_N^*/D_N^*\to1$.

Let $\bar C>0$ be fixed. The sparse algorithm is identical to Section
\ref{sec:usvt-algorithm}, with $B=1$ and projection onto $[0,1]$, except that
the threshold is
\begin{align}\label{eq:sparse-threshold}
\hat t(\bar C)
:=\left(2\sqrt{2\bar C}+\eta\right)
\sqrt{N_{\max}\hat\rho}.
\end{align}
For an undirected network, $N_{\max}=N$. On $\{\hat\rho=0\}$, set
$\hat\mu=0$. The factor $\sqrt2$ is inherited from the probabilistic
symmetrization argument in Remark 3.13 of \citet{bandeira2016sharp}, which
allows the centered Bernoulli entries to have asymmetric distributions. 

\begin{assumption}\label{ass:sparse}
\begin{itemize}
\item[i.] Every $F\in\mathcal F$ has support $\{0,1\}$, and $B=1$.

\item[ii.]
\[
\limsup_{F\in\mathcal F:N_1,N_2\to\infty}C_N^*\leq\bar C,
\]
where $\bar C$ is the fixed constant in \eqref{eq:sparse-threshold}.

\item[iii.]
\[
\liminf_{F\in\mathcal F:N_1,N_2\to\infty}
\frac{N_{\max}\rho}{\log d_N}=\infty.
\]
\end{itemize}
\end{assumption}

Assumption~\ref{ass:sparse}(ii) is a row- and column-level variance-envelope
condition. It is weaker than the entrywise condition
$\max_{ij}\sigma_{ij}^2\leq\bar C\rho$, which implies it with the same
constant. It requires that no row or column carry a conditional variance load
more than a fixed multiple of the average expected degree on the corresponding
side of the network. Assumption~\ref{ass:sparse}(iii) is a logarithmic-degree
condition. Below this scale, an unregularized sparse adjacency matrix need not
concentrate around its mean in operator norm at the scale $\varrho_N$; see
\citet{le2017concentration}.

\begin{lemma}\label{lem:sparse-spectral}
Suppose Assumptions~\ref{ass:model} and \ref{ass:sparse} hold. Then, uniformly over
$F\in\mathcal F$ in the sense of Remark~\ref{rem:usvt-uniform},
\begin{itemize}
\item[i.] $\hat\rho/\rho\to_p1$;

\item[ii.] there exists a constant $\delta=\delta(\eta,\bar C)>0$ such that
\[
\mathbbm P_F\left(
(1+\delta)\|Y-\mu\|_{\mathrm{op}}
\leq\hat t(\bar C)
\text{ and }
\hat t(\bar C)>0
\right)\to1;
\]

\item[iii.]
\[
\hat t(\bar C)=O_p(\varrho_N).
\]
More specifically, with probability approaching one,
\[
\hat t(\bar C)
\leq
\sqrt2\left(2\sqrt{2\bar C}+\eta\right)\varrho_N.
\]
\end{itemize}
\end{lemma}

\begin{proof}
We give the proof for a directed or bipartite network. The undirected proof is
identical after restricting sums to admissible unordered dyads and applying
the same argument to the symmetric adjacency matrix.

\textbf{Part i.}
Write
\[
\hat\rho-\rho
=\frac{1}{N_1N_2}\sum_{i,j}\epsilon_{ij}.
\]
Under $\mathbbm P_F$, the summands are independent, mean zero, bounded by one, and satisfy
\[
\sum_{i,j}\mathbbm E_F[\epsilon_{ij}^2]
=\|\sigma\|_F^2
\leq\sum_{i,j}\mu_{ij}
=N_1N_2\rho.
\]
For each fixed $u\in(0,1]$, Bernstein's inequality therefore gives
\[
\mathbbm P_F\left(
\left|\frac{\hat\rho}{\rho}-1\right|>u
\right)
\leq
2\exp\left(
-\frac{u^2\rho N_1N_2}{2+2u/3}
\right).
\]
Assumption~\ref{ass:sparse}(iii), together with
$\min(N_1,N_2)\to\infty$, implies
$\rho N_1N_2=\min(N_1,N_2)N_{\max}\rho\to\infty$. Hence the right-hand
side converges to zero uniformly. In particular, $\hat\rho>0$ with
probability approaching one.

\textbf{Part ii.}
Let
\[
\nu_N
:=\max\left\{
\max_i\left(\sum_j\sigma_{ij}^2\right)^{1/2},
\max_j\left(\sum_i\sigma_{ij}^2\right)^{1/2}
\right\}.
\]
Choose $a>0$ sufficiently small that
\[
2\sqrt{2(\bar C+a)}
<
2\sqrt{2\bar C}+\frac{\eta}{8}.
\]
Then choose $\varepsilon\in(0,1/2]$ sufficiently small that
\[
(1+\varepsilon)2\sqrt{2(\bar C+a)}
\leq
2\sqrt{2\bar C}+\frac{\eta}{4}.
\]
By Assumption~\ref{ass:sparse}(ii), for all sufficiently large dimensions,
$C_N^*\leq\bar C+a$, and hence
\[
\nu_N
\leq
\sqrt{(\bar C+a)N_{\max}\rho}
=\sqrt{\bar C+a}\,\varrho_N.
\]

Apply Remark 3.13 of \citet{bandeira2016sharp} to the self-adjoint dilation
\[
\mathcal D(Y-\mu)
:=
\begin{pmatrix}
0&Y-\mu\\
(Y-\mu)^T&0
\end{pmatrix}.
\]
The entries on and above the diagonal are independent and centered, their
absolute values are bounded by one, the maximum row variance scale of the
dilation is $\nu_N$, and the uniform entry bound in Remark 3.13 satisfies
$\widetilde\sigma_*\leq1$. Moreover,
$\|\mathcal D(Y-\mu)\|_{\mathrm{op}}=\|Y-\mu\|_{\mathrm{op}}$. Thus there
is a finite constant $c_\varepsilon$ such that, for every $t\geq0$,
\begin{align}\label{eq:centered-bvh-tail}
\mathbbm P_F\left(
\|Y-\mu\|_{\mathrm{op}}
\geq
(1+\varepsilon)2\sqrt2\,\nu_N+t
\right)
\leq
d_N\exp\left(-\frac{t^2}{c_\varepsilon}\right).
\end{align}
Choose
\[
t_N:=\sqrt{2c_\varepsilon\log d_N}.
\]
The right-hand side of \eqref{eq:centered-bvh-tail} is then at most
$d_N^{-1}$. Moreover, Assumption~\ref{ass:sparse}(iii) implies
$t_N=o(\varrho_N)$. It follows that, with probability approaching one
uniformly over $\mathcal F$,
\begin{align}\label{eq:sparse-noise-bound}
\|Y-\mu\|_{\mathrm{op}}
\leq
\left(2\sqrt{2\bar C}+\frac{\eta}{2}\right)\varrho_N.
\end{align}

Fix $u_0=u_0(\eta,\bar C)\in(0,1)$ and $\delta=\delta(\eta,\bar C)>0$
satisfying the deterministic inequality
\[
(1-u_0)\left(2\sqrt{2\bar C}+\eta\right)
\geq
(1+\delta)\left(2\sqrt{2\bar C}+\frac{\eta}{2}\right);
\]
such constants exist because the ratio of the two parenthesized factors
exceeds one. By Part i, with probability approaching one,
$\sqrt{\hat\rho/\rho}\geq1-u_0$, and on that event
\[
\hat t(\bar C)
=\left(2\sqrt{2\bar C}+\eta\right)\sqrt{N_{\max}\hat\rho}
\geq
(1-u_0)\left(2\sqrt{2\bar C}+\eta\right)\varrho_N
\geq
(1+\delta)\left(2\sqrt{2\bar C}+\frac{\eta}{2}\right)\varrho_N.
\]
Combining this display with \eqref{eq:sparse-noise-bound} yields
\[
\hat t(\bar C)
\geq
(1+\delta)\|Y-\mu\|_{\mathrm{op}}
\]
on the intersection of the two high-probability events. Positivity follows
from Part i.

\textbf{Part iii.}
Part i implies that $\hat\rho\leq2\rho$ with probability approaching one,
which gives the stated upper bound.

All probability inequalities used above are finite-sample bounds whose
constants depend only on the fixed quantities $\eta$, $\bar C$, and
$\varepsilon$; the convergence is therefore uniform in the sense of Remark
\ref{rem:usvt-uniform}.
\end{proof}

\begin{corollary}\label{cor:sparse}
Suppose the hypotheses of Proposition~\ref{prop:usvt-bartau-V} hold, add
Assumption~\ref{ass:sparse}, and replace the final condition of Assumption~\ref{ass:usvt-rates}(w) by
\[
\limsup_{F\in\mathcal F:N_1,N_2\to\infty}
\frac{L_N\sqrt{\varrho_N\|\sigma\|_F}}
{rT_\sigma}=0.
\]
If $\hat{\bar\tau}$ and $\hat V$ are constructed using the threshold
\eqref{eq:sparse-threshold}, then the conclusion of Proposition
\ref{prop:usvt-bartau-V} holds.

Similarly, suppose the hypotheses of Proposition~\ref{prop:usvt-sigma} hold,
add Assumption~\ref{ass:sparse}, and replace the final condition of Assumption~\ref{ass:usvt-rates}(s) by
\[
\limsup_{F\in\mathcal F:N_1,N_2\to\infty}
\frac{\sqrt{\varrho_N\|\sigma\|_F}}
{rS_{\sigma,c}}=0.
\]
If $\hat\sigma(g_1,g_2)$ is constructed using the same sparse threshold, then
the conclusion of Proposition~\ref{prop:usvt-sigma} holds.
\end{corollary}

\begin{proof}
By Lemma~\ref{lem:sparse-spectral}, there is an event whose probability
approaches one, uniformly over $\mathcal F$, on which
\[
(1+\delta)\|Y-\mu\|_{\mathrm{op}}\leq\hat t(\bar C),
\qquad
\hat t(\bar C)>0,
\qquad
\hat t(\bar C)\leq\sqrt2\left(2\sqrt{2\bar C}+\eta\right)\varrho_N.
\]
On this event, Lemma~\ref{lem:svt-oracle} applies pathwise with the
realized threshold and projection interval $[0,1]$, giving
\[
\|\hat\mu-\mu\|_F
\leq\sqrt{2\left(2+\delta^{-1}\right)}
\left(\hat t(\bar C)\,\|\mu\|_*\right)^{1/2}
=O_p\left(\sqrt{\varrho_N\|\sigma\|_F}\right),
\]
where the final bound uses the displayed upper bound on $\hat t(\bar C)$ and
Assumptions~\ref{ass:usvt-signal-noise} and \ref{ass:usvt-low-rank}. The modified rate conditions then give exactly the
mean-estimation bounds used in the proofs of Propositions
\ref{prop:usvt-bartau-V} and \ref{prop:usvt-sigma}; all remaining
residual-concentration steps are unchanged.
\end{proof}

\begin{remark}[Relation to concentrated variation]\label{rem:sparse-heterogeneity}
Assumption~\ref{ass:sparse}(ii) does not impose homogeneous variances across
dyads, but it does restrict how strongly stochastic variation can be
concentrated across nodes. In sequences where an asymptotically vanishing
fraction of rows or columns carries most of the residual variation, $C_N^*$
will generally diverge, and no fixed-$\bar C$ sparse result applies. This does
not exclude such sequences from the model class or from the main confidence-
interval theory in Section~\ref{sec:main-results}. It only means that the sparse residual estimator
need not attain the improved mean-estimation rate on those sequences. The
adaptive rule below then uses the dense threshold. For example, the
deterministic core--periphery designs of Appendix Section~\ref{sec:simulation-models} have
$D_N^*\asymp\sqrt N$: the inflated degree ratio eventually exceeds any fixed
grid, and the adaptive rule of Section~\ref{sec:cbar-calibration} falls back
to the dense threshold. (Those designs also have bounded average degree, so
they lie outside the formal scope of the sparse results in any case.) Because
$D_N^*$ may be much larger than $C_N^*$ when high-degree links are nearly
deterministic, the degree-based rule can also revert to the dense threshold
in some models for which the variance envelope itself is benign; this is the
price of using an observable conservative proxy.
\end{remark}

\subsubsection{Data-driven calibration of the threshold constant}
\label{sec:cbar-calibration}

The fixed sparse result requires a constant $\bar C$ that bounds the
unobserved variance-load ratio $C_N^*$. Although $C_N^*$ is not generally
estimable from one network, the degree ratio $D_N^*$ is observable up to
sampling error and satisfies $C_N^*\leq D_N^*$. For a directed or bipartite
network, define
\begin{align}\label{eq:Dhat-directed}
\hat D_N
:=\max\left\{
\max_{i\in[N_1]}\frac{d_{i,1}}{N_2\hat\rho},
\;
\max_{j\in[N_2]}\frac{d_{j,2}}{N_1\hat\rho}
\right\}.
\end{align}
For an undirected loopless network, define instead
\begin{align}\label{eq:Dhat-undirected}
\hat D_N
:=\max_{i\in[N]}
\frac{d_i}{(N-1)\hat\rho},
\qquad
d_i:=\sum_{j\neq i}Y_{ij}.
\end{align}
Thus $\hat D_N$ estimates $D_N^*$, not $C_N^*$; it is used because
$D_N^*$ is a conservative upper bound for the variance-load ratio.

\begin{lemma}\label{lem:degree-ratio}
Suppose Assumptions~\ref{ass:model} and \ref{ass:sparse}(i) hold. For a directed or
bipartite network, suppose
\begin{align}\label{eq:two-sided-degrees}
\liminf_{F\in\mathcal F:N_1,N_2\to\infty}
\frac{\min(N_1,N_2)\rho}{\log(N_1+N_2)}
=\infty.
\end{align}
For an undirected loopless network, replace this condition by
\begin{align}\label{eq:undirected-degrees}
\liminf_{F\in\mathcal F:N\to\infty}
\frac{(N-1)\rho}{\log N}
=\infty.
\end{align}
Then, uniformly over $F\in\mathcal F$,
\[
\frac{\hat D_N}{D_N^*}\to_p1.
\]
Consequently, for every fixed $\gamma>0$,
\[
\mathbbm P_F\left(
C_N^*\leq D_N^*\leq(1+\gamma)\hat D_N
\right)\to1.
\]
\end{lemma}

The average-degree conditions in the lemma are convenient, transparent
sufficient conditions. The proof only requires the maximum expected row and
column degrees to dominate the corresponding logarithmic dimensions, together
with consistency of $\hat\rho$; we retain the stated conditions because they
avoid introducing additional notation.

\begin{proof}
We give the directed or bipartite argument. Let
\[
e_{1,\max}:=\max_i e_{i,1},
\qquad
e_{2,\max}:=\max_j e_{j,2}.
\]
For every fixed $\varepsilon\in(0,1)$, Bernstein's inequality and a union
bound give
\begin{align*}
\mathbbm P_F\left(
\max_i|d_{i,1}-e_{i,1}|>\varepsilon e_{1,\max}
\right)
&\leq
2N_1\exp\left(-c_\varepsilon e_{1,\max}\right),\\
\mathbbm P_F\left(
\max_j|d_{j,2}-e_{j,2}|>\varepsilon e_{2,\max}
\right)
&\leq
2N_2\exp\left(-c_\varepsilon e_{2,\max}\right),
\end{align*}
for a constant $c_\varepsilon>0$. The bound does not require every row or
column to have a large mean; it uses only
$e_{i,1}\leq e_{1,\max}$ and $e_{j,2}\leq e_{2,\max}$. Since
\[
e_{1,\max}\geq N_2\rho,
\qquad
e_{2,\max}\geq N_1\rho,
\]
condition \eqref{eq:two-sided-degrees} makes both right-hand sides converge
to zero uniformly. It follows that
\[
\frac{\max_i d_{i,1}}{e_{1,\max}}\to_p1,
\qquad
\frac{\max_j d_{j,2}}{e_{2,\max}}\to_p1.
\]
The same condition and Bernstein's inequality give $\hat\rho/\rho\to_p1$.
Therefore each of the two normalized observed maximum degrees is ratio-
consistent for its population counterpart, and taking their maximum gives
$\hat D_N/D_N^*\to_p1$. The undirected proof is identical, using the $N-1$
independent Bernoulli variables incident to each fixed node; dependence
between different node degrees does not affect the union bound. The final
claim follows from $C_N^*\leq D_N^*$ and ratio consistency.
\end{proof}

Fix before examining the data a constant $\gamma>0$, an integer $K$, and a
finite grid
\[
\mathcal C:=\{C_1,\ldots,C_K\},
\qquad
1\leq C_1<\cdots<C_K.
\]
The quantities $\gamma$, $K$, and $\mathcal C$ do not vary along the
asymptotic sequence. Requiring $C_1\geq1$ is without loss, because
$D_N^*\geq1$ deterministically. Define
\[
\hat C
:=\min\left\{C\in\mathcal C:
C\geq(1+\gamma)\hat D_N\right\},
\]
when this set is nonempty. Let
\[
t_D:=(2+\eta)\sqrt{N_{\max}}
\]
be the dense threshold from Section~\ref{sec:usvt-algorithm}, with
$N_{\max}=N$ in the undirected case, and define
\begin{align}\label{eq:adaptive-threshold}
\hat t_A
:=
\begin{cases}
\min\{\hat t(\hat C),t_D\},&\text{if $\hat\rho>0$ and $\hat C$ exists},\\[2pt]
t_D,&\text{otherwise}.
\end{cases}
\end{align}
On $\{\hat\rho=0\}$ the ratio $\hat D_N$ is undefined and the rule uses the
dense threshold. Thus the rule inflates the observed degree ratio, rounds
upward to a fixed grid, and uses the smaller of the resulting sparse
threshold and the dense threshold. If the inflated degree ratio lies above
the grid, the procedure uses the dense threshold. A possible default is
$\gamma=1$ and a fixed grid such as $\{1,2,4,8,16,32,64\}$.

\begin{corollary}\label{cor:sparse-adaptive}
Suppose Assumptions~\ref{ass:model} and \ref{ass:sparse}(i) hold, together with
\eqref{eq:two-sided-degrees} in the directed or bipartite case or
\eqref{eq:undirected-degrees} in the undirected case.
\begin{itemize}
\item[i.] \emph{Spectral domination.}
There exists a constant $\delta_{\min}>0$, depending only on the fixed
quantities $\eta$ and $\mathcal C$, such that
\[
\mathbbm P_F\left(
(1+\delta_{\min})\|Y-\mu\|_{\mathrm{op}}
\leq\hat t_A
\right)\to1
\]
uniformly over $F\in\mathcal F$, without imposing a bounded variance or
degree ratio.

\item[ii.] \emph{Sparse order.}
If, in addition,
\begin{align}\label{eq:grid-rate-condition}
\limsup_{F\in\mathcal F:N_1,N_2\to\infty}D_N^*
<\frac{C_K}{1+\gamma},
\end{align}
then $\hat C$ exists with probability approaching one and
\[
\hat t_A=O_p(\varrho_N).
\]

\item[iii.] \emph{Feasible scale estimation, sparse branch.}
Under the condition in Part ii, suppose the hypotheses of Proposition
\ref{prop:usvt-bartau-V} also hold and replace the final condition of
Assumption~\ref{ass:usvt-rates}(w) by
\[
\limsup_{F\in\mathcal F:N_1,N_2\to\infty}
\frac{L_N\sqrt{\varrho_N\|\sigma\|_F}}
{rT_\sigma}=0.
\]
Then the conclusion of Proposition~\ref{prop:usvt-bartau-V} holds for the
estimators constructed from $\hat t_A$. Likewise, under the hypotheses of
Proposition~\ref{prop:usvt-sigma}, replace the final condition of Assumption~\ref{ass:usvt-rates}(s) by
\[
\limsup_{F\in\mathcal F:N_1,N_2\to\infty}
\frac{\sqrt{\varrho_N\|\sigma\|_F}}
{rS_{\sigma,c}}=0.
\]
Then the conclusion of Proposition~\ref{prop:usvt-sigma} holds for the
estimator constructed from $\hat t_A$.

\item[iv.] \emph{Baseline dense-rate guarantee.}
Suppose, in addition to the hypotheses of Part i, that the hypotheses of
Proposition~\ref{prop:usvt-bartau-V} hold with the original final condition
of Assumption~\ref{ass:usvt-rates}(w). Then the conclusion of Proposition
\ref{prop:usvt-bartau-V} holds for the estimators constructed from
$\hat t_A$, whether or not \eqref{eq:grid-rate-condition} holds. The
analogous statement holds for Proposition~\ref{prop:usvt-sigma} under
Assumption~\ref{ass:usvt-rates}(s) with its original final condition.
\end{itemize}
\end{corollary}

\begin{proof}
\textbf{Part i.}
Let $E_D$ be the gapped noise-domination event for the dense threshold used in
the proofs of Propositions~\ref{prop:usvt-bartau-V} and
\ref{prop:usvt-sigma}; there is a fixed $\delta_D>0$ such that
$\mathbbm P_F(E_D)\to1$ uniformly under Assumption~\ref{ass:model}.

The degree condition in the corollary implies
$N_{\max}\rho/\log d_N\to\infty$. For each $C_k\in\mathcal C$, apply the
finite-sample argument in the proof of Lemma~\ref{lem:sparse-spectral} with
$\bar C=C_k$ and $a=0$. This yields a constant $\delta_k>0$ and an event $E_k$
whose failure probability tends uniformly to zero over the finite-sample
subclass of models satisfying $C_N^*\leq C_k$. Define
\[
G_k
:=
\{C_N^*>C_k\}\cup E_k.
\]
Because $C_N^*$ is a population quantity, the first event is deterministic for
a fixed model. Hence $\mathbbm P_F(G_k)\to1$ uniformly for every $k$. Since the
grid is finite, $E_D\cap\bigcap_{k=1}^K G_k$ has probability approaching one.
Define
\[
\delta_{\min}
:=\min\{\delta_D,\delta_1,\ldots,\delta_K\}>0.
\]
By Lemma~\ref{lem:degree-ratio}, with probability approaching one,
\[
C_N^*\leq(1+\gamma)\hat D_N.
\]
On this event, if $\hat C$ exists then $\hat C\geq C_N^*$, so both the sparse
threshold $\hat t(\hat C)$ and the dense threshold $t_D$ exceed
$(1+\delta_{\min})\|Y-\mu\|_{\mathrm{op}}$. Their minimum therefore does as
well. If $\hat C$ does not exist, or if $\hat\rho=0$, then $\hat t_A=t_D$ and
the dense event applies.

\textbf{Part ii.}
By Lemma~\ref{lem:degree-ratio}, $\hat D_N/D_N^*\to_p1$. The strict margin in
\eqref{eq:grid-rate-condition} therefore implies
\[
(1+\gamma)\hat D_N<C_K
\]
with probability approaching one. Hence $\hat C$ exists and
$\hat C\leq C_K$. Since $\hat\rho/\rho\to_p1$,
\[
\hat t_A
\leq\hat t(C_K)
=O_p(\varrho_N).
\]

\textbf{Part iii.}
Condition \eqref{eq:grid-rate-condition}, together with
$C_N^*\leq D_N^*$, implies Assumption~\ref{ass:sparse}(ii) with
$\bar C=C_K$, while the degree condition in the corollary implies
Assumption~\ref{ass:sparse}(iii). By Parts i and ii, with probability
approaching one,
$(1+\delta_{\min})\|Y-\mu\|_{\mathrm{op}}\leq\hat t_A$ and
$\hat t_A=O_p(\varrho_N)$. The adaptive threshold is strictly positive by
definition. Lemma~\ref{lem:svt-oracle} applies pathwise with the realized threshold and
projection interval $[0,1]$, giving
\[
\|\hat\mu-\mu\|_F
=O_p\left(\sqrt{\varrho_N\|\mu\|_*}\right)
=O_p\left(\sqrt{\varrho_N\|\sigma\|_F}\right)
\]
under Assumptions~\ref{ass:usvt-signal-noise} and \ref{ass:usvt-low-rank}. The modified rate conditions give the same
mean-estimation bounds used in Propositions~\ref{prop:usvt-bartau-V} and
\ref{prop:usvt-sigma}; their remaining residual-concentration arguments are
unchanged.

\textbf{Part iv.}
By construction, $0<\hat t_A\leq t_D$ deterministically. On the event in
Part i, Lemma~\ref{lem:svt-oracle} applies pathwise with
$t=\hat t_A$, $\delta=\delta_{\min}$, and projection interval $[0,1]$, giving
\[
\|\hat\mu-\mu\|_F
\leq
\sqrt{2\left(2+\delta_{\min}^{-1}\right)}
\left(\hat t_A\,\|\mu\|_*\right)^{1/2}
\leq
\sqrt{2\left(2+\delta_{\min}^{-1}\right)}
\left(t_D\,\|\mu\|_*\right)^{1/2}
=O_p\left(\sqrt{q_N\|\sigma\|_F}\right),
\]
where the final bound uses $t_D=(2+\eta)\sqrt{N_{\max}}$, $q_N=\sqrt{N_{\max}}$
since $B=1$, and Assumptions~\ref{ass:usvt-signal-noise} and \ref{ass:usvt-low-rank}. This is the mean-estimation bound used
in the proofs of Propositions~\ref{prop:usvt-bartau-V} and
\ref{prop:usvt-sigma} under the original final conditions of Assumption~\ref{ass:usvt-rates};
their remaining residual-concentration arguments are unchanged.
\end{proof}

\begin{remark}[Direction of threshold error]\label{rem:direction}
Choosing a threshold below the operator-norm scale is a particular concern:
idiosyncratic singular components may enter $\hat\mu$, potentially causing
the residual-based estimators of $\bar\tau$, $V$, and
$\sigma(g_1,g_2)$ to understate uncertainty. Rounding the degree calibration
upward and using the dense fallback reduce this risk. A larger threshold often
leaves more variation in the residual, but we do not assert general
monotonicity for the entrywise-clipped estimator, the $\dagger$-norm, or the
group-specific residual norms. Nor does a larger threshold automatically make
the resulting feasible confidence intervals conservative without the
corresponding consistency or one-sided conditions.
\end{remark}

\begin{remark}[Scope and reporting]\label{rem:sparse-scope}
Part i of Corollary~\ref{cor:sparse-adaptive} establishes a spectral-
domination event, not by itself consistency of the residual scale estimators
or coverage of the feasible intervals. Under the original dense rate
conditions, Part iv shows that the adaptive estimator satisfies the baseline
dense consistency conclusions regardless of which branch of the rule is
selected. Under condition \eqref{eq:grid-rate-condition} and the modified
sparse rate conditions, Part iii yields the sharper sparse mean-estimation
rate. The stated degree calibration requires diverging average degree on both
sides of a rectangular network. This is a convenient sufficient condition,
not a necessary condition for degree-ratio estimation. In bounded or slowly
growing degree regimes, the present unregularized spectral theory does not
apply, and the deterministic estimators of Section~\ref{sec:deterministic-scale-bounds} remain available.
\end{remark}

\begin{remark}[Bounded-degree networks]\label{rem:sparse-bounded-degree}
When expected degrees are bounded, an unregularized adjacency matrix need not
concentrate around its mean at the sparse operator-norm scale. The trimming and
regularization methods of \citet{le2017concentration} suggest a possible
extension. Incorporating such an estimator into the present residual plug-in
procedure would additionally require controlling regularization bias relative
to the original mean matrix and establishing suitable concentration for the
row-column residual norms. We leave this extension to future work.
\end{remark}

\section{Simulation evidence}\label{sec:simulations}
The simulations are designed to address four questions. First, how severely can
data-dependent group selection distort a conventional confidence interval that
treats the selected groups as fixed? Second, do the proposed confidence
intervals retain useful power against a selected network
structure? Third, can the shorter width of $CI_2$ under concentrated variation
translate into a power advantage over $CI_1$? Fourth, how do the estimates in Appendix~\ref{app:variance-estimation} compare with their oracle
counterparts?  \looseness=-1

\subsection{Designs and implementation}\label{sec:simulation-models}
All simulations use undirected, unweighted, loopless networks on $N$ agents,
represented by their upper-triangular adjacency matrices. For $i<j$,
\[
Y_{ij}\sim\operatorname{Bernoulli}(\mu_{ij})
\]
independently across dyads. We use
\[
N\in\{200,500,1000,2000\},
\qquad
\alpha=0.05,
\]
and perform $2{,}000$ Monte Carlo replications for each design and
network size.

The designs compare a core $C$ containing $m=N/4$ agents with its complement
$P=C^c$. Following Remark~\ref{rem:alternative-normalization}, we report the
natural densities among eligible unordered dyads,
\[
\widehat\theta_{CC}
=
\binom{m}{2}^{-1}
\sum_{i<j}Y_{ij}\mathbbm 1\{i,j\in C\}
\]
and
\[
\widehat\theta_{PP}
=
\binom{N-m}{2}^{-1}
\sum_{i<j}Y_{ij}\mathbbm 1\{i,j\in P\}.
\]
Our main inferential object is the difference
\[
\Delta(C)
=
\theta_{CC}(C)-\theta_{PP}(C),
\qquad
\widehat\Delta(C)
=
\widehat\theta_{CC}(C)-\widehat\theta_{PP}(C).
\]
For any interval construction $t\in\{0,1,2\}$, let
$CI_t^{CC}=[L_t^{CC},U_t^{CC}]$ and
$CI_t^{PP}=[L_t^{PP},U_t^{PP}]$ denote the corresponding intervals for the
two densities. We construct the difference interval as
\[
CI_{t,\Delta}
=
\left[
L_t^{CC}-U_t^{PP},
\;
U_t^{CC}-L_t^{PP}
\right].
\]
We say that the interval detects a positive core--periphery difference when
the lower endpoint of $CI_{t,\Delta}$ is positive.

We use four designs. The first is a homogeneous selection-null design. Let
\[
\rho_N=\frac{2\log N}{N-1},
\qquad
\mu_{ij}=\rho_N.
\]
For the controlled-search rule, we generate $K$ candidate cores of size
$N/4$ independently of the network and hold this collection fixed across
Monte Carlo replications. We then select
\[
\widehat C_K
=
\arg\max_{1\leq k\leq K}
\widehat\Delta(C_k),
\qquad
K\in\{1,10,100,1000\}.
\]
For the spectral rule, we select the $N/4$ agents with the largest absolute
entries of the leading eigenvector of the symmetric adjacency matrix. Under
this design, $\Delta(C)=0$ for every possible group $C$, so rejection is a
false positive.

The second design studies power against a core--periphery alternative. Let
$C^*$ contain the first $N/4$ agents and define
\[
a_i
=
\begin{cases}
2.75,&i\in C^*,\\[2pt]
5/12,&i\notin C^*.
\end{cases}
\]
Let
\[
\bar a_N
=
\binom{N}{2}^{-1}
\sum_{i<j}a_i a_j,
\qquad
\mu_{ij}
=
\rho_N\frac{a_i a_j}{\bar a_N}.
\]
The normalization preserves the same average connection probability as in the
homogeneous design while producing substantial degree heterogeneity. We use
the spectral rule above to select the reported core.

The third design concentrates all stochastic variation in a sublinear set of
nodes. Let
\[
r_N=\left\lceil1.5\sqrt N\right\rceil
\]
and let the fixed core $C^*$ contain the $r_N$ active nodes together with
$N/4-r_N$ deterministic filler nodes. Define
\[
p_N=\min\{0.50,\;2.25N^{-1/4}\}
\]
and
\[
\mu_{ij}
=
p_N\mathbbm 1\{i,j\leq r_N\}.
\]
Consequently,
\[
\Delta(C^*)
=
p_N
\frac{\binom{r_N}{2}}{\binom{N/4}{2}},
\]
while the periphery density and its variance are zero. This design satisfies
the weak global nondegeneracy condition in Assumption~\ref{ass:nondegeneracy}(w)
but not the strong condition in Assumption~\ref{ass:nondegeneracy}(s).
Accordingly, $CI_2$ is the theoretically justified interval in this design;
we include $CI_1$ to illustrate the cost of its global
$\sqrt{N_1+N_2}$ inflation.

The fourth design is used to compare feasible scale estimators when the mean
matrix contains a detectable low-rank component. It is a two-block model with
block shares $1/4$ and $3/4$ and probability template
\[
B_0
=
\begin{pmatrix}
3&0.5\\
0.5&4/3
\end{pmatrix}.
\]
We divide $B_0$ by its pair-weighted average and multiply it by
$3N^{-1/2}$. This produces approximately equal expected degrees in the two
blocks while giving the mean matrix a singular component that exceeds the
degree-adaptive USVT threshold.

We report both oracle and feasible intervals. The oracle intervals use the
known conditional variances, $V$, and $\bar\tau$. The conditional expectation
in $\bar\tau$ is approximated using $10{,}000$ auxiliary draws from the exact
blockwise row- and column-energy distributions. The principal feasible
results use the unshifted raw second-moment plug-in from
Section~\ref{sec:raw-plugin}. We also evaluate the deterministic bounds from
Section~\ref{sec:deterministic-scale-bounds} and the degree-adaptive USVT
estimator from Section~\ref{sec:usvt}; because of its greater computational
cost, USVT is evaluated in $250$ replications for $N\leq1{,}000$. The shifted
raw plug-in in Section~\ref{sec:raw-shifted} is not included because it is
designed as a conservative coverage fallback rather than as a calibrated
scale estimator. With $2{,}000$ replications, the Monte Carlo standard error
of a reported rejection or coverage probability is at most $0.011$.
\looseness=-1

\subsection{Conventional inference after group selection}
\label{sec:simulation-selection}

Table~\ref{tab:simulation-selection} reports the probability that the
conventional difference interval $CI_{0,\Delta}$ falsely excludes zero under
the homogeneous design. To isolate the effect of group selection from scale
estimation, the table gives $CI_0$ the oracle conditional variance.

\begin{table}[!htbp]
\centering
\small
\setlength{\tabcolsep}{7pt}
\renewcommand{\arraystretch}{1.10}
\caption{False-positive rates after group selection under a homogeneous network}
\label{tab:simulation-selection}
\begin{tabular}{rccccc}
\toprule
& \multicolumn{4}{c}{Controlled search} & Spectral\\
\cmidrule(lr){2-5}
$N$ & $K=1$ & $K=10$ & $K=100$ & $K=1{,}000$ & rule\\
\midrule
200   & 0.011 & 0.080 & 0.567 & 1.000 & 1.000\\
500   & 0.012 & 0.073 & 0.533 & 1.000 & 1.000\\
1,000 & 0.013 & 0.065 & 0.521 & 0.999 & 1.000\\
2,000 & 0.015 & 0.067 & 0.512 & 1.000 & 1.000\\
\bottomrule
\end{tabular}
\begin{minipage}{0.94\textwidth}
\footnotesize
\emph{Notes:} Each entry is the fraction of $2{,}000$ replications in which
the oracle conventional interval for the selected core--periphery density
difference excludes zero. The true difference is zero for every candidate
group. The controlled rule selects the largest observed difference among
$K$ prespecified candidate cores. The spectral rule selects the $N/4$ nodes
with the largest absolute entries of the leading adjacency-matrix
eigenvector. In every corresponding cell, the oracle versions of $CI_1$ and
$CI_2$ covered the selected difference in all replications and never excluded
zero.
\end{minipage}
\end{table}

With one prespecified candidate, the false-positive probability is small. It
increases rapidly with the breadth of the search. When $K=100$, the
conventional interval falsely detects a core--periphery difference in roughly
one-half of the replications, and when $K=1{,}000$, it does so essentially
every time. The failure is even more pronounced under the spectral rule,
where $CI_{0,\Delta}$ excludes zero in every replication for every network
size. Because the data-generating process is homogeneous and the conventional
interval uses the oracle variance, this failure is caused by treating the
selected groups as fixed rather than by model misspecification or inaccurate
scale estimation.

The simultaneous intervals are conservative in this experiment. We observe
no noncoverage for either $CI_{1,\Delta}$ or $CI_{2,\Delta}$ in the reported
replications. The simulations therefore do not suggest exact finite-sample
calibration at the nominal level; instead, they illustrate the contrast
between severe post-selection failure of $CI_0$ and the protection provided
by simultaneous coverage. \looseness=-1

\subsection{Power against a selected core--periphery alternative}
\label{sec:simulation-ci1-power}

Panel A of Table~\ref{tab:simulation-power} studies the stronger
core--periphery alternative. The spectral rule recovers the core
almost perfectly. Its average overlap with $C^*$ exceeds $0.9996$ at every
network size. The table reports power for the selected density difference, so
the inferential target is the group actually returned by the spectral rule
rather than the latent group.

\begin{table}[!htbp]
\centering
\scriptsize
\setlength{\tabcolsep}{4.2pt}
\renewcommand{\arraystretch}{1.10}
\caption{Power under core--periphery and concentrated-variation alternatives}
\label{tab:simulation-power}
\begin{tabular}{rcccccc}
\toprule
$N$ & Gap & Overlap
& Oracle MoE & Oracle power
& Raw MoE & Raw power\\
\midrule
\multicolumn{7}{l}{\emph{Panel A: $CI_1$ after spectral core selection}}\\
\midrule
200   & 0.3955 & 0.9998 & 0.3544 & 1.000 & 0.4532 & 0.000\\
500   & 0.1844 & 0.9997 & 0.1754 & 0.982 & 0.1938 & 0.001\\
1,000 & 0.1023 & 0.9998 & 0.0965 & 1.000 & 0.1017 & 0.719\\
2,000 & 0.0562 & 0.9998 & 0.0517 & 1.000 & 0.0532 & 1.000\\
\midrule
\multicolumn{7}{l}{\emph{Panel B: feasible comparison under concentrated variation}}\\
\midrule
$N$ & Gap & $\mathrm{MoE}_1$ & Power$_1$
& $\mathrm{MoE}_2$ & Power$_2$
& $\mathrm{MoE}_2/\mathrm{MoE}_1$\\
\midrule
200   & 0.09429 & 0.20808 & 0.000 & 0.13545 & 0.000 & 0.651\\
500   & 0.03444 & 0.07875 & 0.000 & 0.03695 & 0.001 & 0.469\\
1,000 & 0.01450 & 0.03601 & 0.000 & 0.01343 & 0.999 & 0.373\\
2,000 & 0.00614 & 0.01656 & 0.000 & 0.00495 & 1.000 & 0.299\\
\bottomrule
\end{tabular}
\begin{minipage}{0.96\textwidth}
\footnotesize
\emph{Notes:} ``Gap'' is the core--periphery density difference in
Panel A and the fixed-group population difference in Panel B. ``Overlap'' is
the fraction of core nodes contained in the spectrally selected core.
MoE denotes the average half-width of the interval for the density
difference, and power is the fraction of replications in which its lower
endpoint is positive. ``Raw'' uses the unshifted raw second-moment plug-in of
Section~\ref{sec:raw-plugin}. Panel B uses the raw plug-in for both intervals, and the reported
margin-of-error ratio is averaged replication by replication. The proposed
difference intervals covered their respective targets in every
replication in both panels.
\end{minipage}
\end{table}

At $N=200$ and $N=500$, the oracle interval detects the structure with high
probability, but the feasible raw plug-in is sufficiently conservative that
its margin still exceeds the gap. By $N=1{,}000$, the feasible margin and the
population gap are nearly equal, and the power of $CI_1$ rises to $0.719$.
At $N=2{,}000$, feasible power is one. Thus the simultaneous correction does
not mechanically eliminate power: a genuinely pronounced structure selected
from the data is detected with high probability once the feasible scale
estimator becomes sufficiently close to its oracle target.

In every USVT replication through $N=1{,}000$, the degree-adaptive
estimator retains rank zero in this particular core--periphery design. It
therefore coincides numerically with the unshifted raw plug-in on those
replications, as described in
Remark~\ref{rem:usvt-rankzero}. We use a separate low-rank design below to
evaluate the nonzero-rank branch. \looseness=-1

\subsection{Concentrated variation and the power of $CI_2$}
\label{sec:simulation-ci2-power}

Panel B of Table~\ref{tab:simulation-power} reports the fixed-group
concentrated-variation design. Although the population density difference
shrinks from $0.0943$ at $N=200$ to $0.00614$ at $N=2{,}000$, the feasible
margin of error of $CI_2$ contracts more quickly. The ratio
$\mathrm{MoE}_2/\mathrm{MoE}_1$ falls from $0.651$ to $0.299$. As a result,
the feasible power of $CI_2$ rises from essentially zero at $N=500$ to
$0.999$ at $N=1{,}000$ and one at $N=2{,}000$. In contrast, $CI_1$ has zero
power at every network size.

The result is not driven by undercoverage. Both feasible intervals cover the
true density difference in all $2{,}000$ replications at every $N$. Moreover,
the raw plug-in remains conservative. At $N=2{,}000$, the average ratios of
the raw estimates to the oracle quantities are
\[
\frac{\widehat V^{\,0}}{V}=1.137,
\qquad
\frac{\widehat{\bar\tau}^{\,0}}{\bar\tau}=1.215,
\qquad
\frac{\lVert Y\rVert_F}{\lVert\sigma\rVert_F}=1.228.
\]
Thus $CI_2$ has complete power even though its feasible scale estimates
continue to exceed the oracle scales materially.

This design also illustrates the different nondegeneracy requirements of the
two intervals. Since the periphery contains no stochastic dyads,
Assumption~\ref{ass:nondegeneracy}(s) fails, and $CI_1$ is included only as a
comparison. The total residual variation nevertheless diverges, so
Assumption~\ref{ass:nondegeneracy}(w) holds and $CI_2$ is the relevant
procedure. Its adaptation to the number of rows and columns carrying
variation is what produces the substantial width and power advantage.
\looseness=-1

\subsection{Feasible scale estimation}\label{sec:simulation-scale-estimation}

Table~\ref{tab:simulation-scale} compares the three principal implementation
strategies of Appendix~\ref{app:variance-estimation} in the USVT-detectable
two-block design at $N=1{,}000$. The table reports the average estimated scale
relative to its oracle counterpart. Let $\widehat S_F$ denote the empirical
Frobenius residual scale, equal to $\lVert Y\rVert_F$ for the raw plug-in and
$\lVert Y-\widehat\mu\rVert_F$ for USVT.

\begin{table}[!htbp]
\centering
\small
\setlength{\tabcolsep}{7pt}
\renewcommand{\arraystretch}{1.10}
\caption{Scale estimation in a USVT-detectable low-rank design}
\label{tab:simulation-scale}
\begin{tabular}{lcccc}
\toprule
Strategy
& $\widehat V/V$
& $\widehat{\bar\tau}/\bar\tau$
& $\widehat S_F/\lVert\sigma\rVert_F$
& Retained rank\\
\midrule
Deterministic bounds & 7.047 & 10.308 & --- & ---\\
Raw plug-in          & 1.050 & 1.068  & 1.069 & ---\\
Degree-adaptive USVT & 1.011 & 1.016  & 1.015 & $1$ or $2$\\
\bottomrule
\end{tabular}
\begin{minipage}{0.94\textwidth}
\footnotesize
\emph{Notes:} The deterministic and raw rows use $2{,}000$
replications. The USVT row uses $250$ replications. USVT retained one
component in $247$ replications and two components in three replications.
The deterministic construction directly bounds the cell-specific scale rather
than estimating $\lVert\sigma\rVert_F$, so no Frobenius ratio is reported.
\end{minipage}
\end{table}

The deterministic bounds guarantee coverage with minimal structure, but are
too conservative to be competitive in this design. The raw plug-in overstates
the oracle scales by roughly five to seven percent. Estimating and removing
the low-rank mean component reduces the discrepancy to roughly one to two
percent. This comparison gives the residual-USVT strategy a distinct role:
when the systematic component is nonnegligible and spectrally detectable, it
can materially improve calibration relative to setting the conditional mean
identically equal to zero. In the sparser logarithmic-degree designs, by
contrast, USVT generally retains rank zero and therefore reduces to the
unshifted raw plug-in in the realized sample. \looseness=-1

\subsection{Summary}\label{sec:simulation-summary}

The simulations yield four main findings. First, conventional fixed-group
inference can fail severely after data-dependent selection even when the
network model is correctly specified and the oracle variance is known.
Second, the simultaneous intervals are conservative in the designs studied,
but feasible $CI_1$ nevertheless has substantial power against a pronounced
core--periphery structure selected by a spectral rule. Third, when residual
variation is concentrated in a relatively small number of rows and columns,
feasible $CI_2$ can be substantially shorter and more powerful than $CI_1$.
Finally, the raw and USVT scale estimators behave as suggested by
Appendix~\ref{app:variance-estimation}: the raw plug-in becomes accurate as
the signal share falls, while USVT improves calibration when a detectable
low-rank mean component is present. \looseness=-1

\section{Empirical implementation details}
\label{app:empirical-details} 
This appendix gives the implementation details for
Tables~\ref{tab:facebook100-type-rest-summary},
\ref{tab:baci-contrast-family-summary}, and
\ref{tab:psid-market-structure-summary} in the main text. It first describes the normalization used in the Facebook100 example (see Remark~\ref{rem:alternative-normalization}). It then considers the three tables in order.   \looseness=-1

\subsection{Undirected unipartite networks}
\label{app:undirected-partition-cells}
\label{app:undirected-normalization}
The main results are stated for an $N_1\times N_2$ matrix whose entries are
indexed by ordered pairs and whose density is normalized by $m_1m_2$, the number
of ordered pairs between the two groups.  In the case of an undirected unipartite network with no loops, the density of connections is measured in a slightly different way. First, a connection is
associated with an unordered dyad rather than an ordered pair.  Second, pairs
of the form $ii$ are not eligible links.  Under our convention in
Section~\ref{sec:model-terminology}, the network is nevertheless stored in an
$N\times N$ matrix by placing the observation for dyad $\{i,j\}$ in the
$ij$th entry when $i<j$ and setting every entry on or below the diagonal to
zero. The following normalization is meant to correct for the fact that the ordering of the nodes is arbitrary and that there are no connections on the main diagonal. \looseness=-1

\subsubsection{Density and estimator}
\label{sec:undirected-normalization}

Let the network contain $N$ nodes.  Under the upper-triangular convention,
$Y_{ij}$ may be non-degenerate only when $i<j$, and $Y_{ij}=0$ when $i\geq j$.  Write
\[
\mu_{ij}=\mathbbm E[Y_{ij}\mid\mathcal H],
\qquad
\epsilon_{ij}=Y_{ij}-\mu_{ij},
\qquad
\sigma_{ij}^2=\mathbbm E[\epsilon_{ij}^2\mid\mathcal H].
\]
Define the symmetrized matrices
\[
Y^s:=Y+Y^T,
\qquad
\mu^s:=\mu+\mu^T.
\]
Under this symmetrization, $Y^s_{ij}=Y^s_{ji}$ records the same undirected connection in both
orientations, while $Y^s_{ii}=0$. For two groups $g,h\in\{0,1\}^N$, let
\[
m_g:=\sum_{i=1}^N g_i,
\qquad
m_h:=\sum_{i=1}^N h_i,
\qquad
r(g,h):=\sum_{i=1}^N g_i h_i.
\]
The number of eligible ordered incidences from group $g$ to group $h$ is
\begin{align}
D^u(g,h)
&:=\sum_{i\neq j}g_i h_j
 =m_gm_h-r(g,h).
\label{eq:undirected-normalization-denominator}
\end{align}
The subtraction of $r(g,h)$ removes the ineligible self-pairs contributed by
nodes that belong to both groups.  Define the normalized population density
and its empirical analogue by
\begin{align}
\theta^u(g,h)
&:=
\frac{1}{D^u(g,h)}
\sum_{i,j=1}^N\mu^s_{ij}g_i h_j,
\label{eq:undirected-normalized-parameter}\\
\widehat\theta^u(g,h)
&:=
\frac{1}{D^u(g,h)}
\sum_{i,j=1}^N Y^s_{ij}g_i h_j.
\label{eq:undirected-normalized-estimator}
\end{align}
Equivalently, for $i<j$, define
\[
\omega_{ij}(g,h):=g_i h_j+g_j h_i\in\{0,1,2\}.
\]
Then
\begin{align}
D^u(g,h)
&=\sum_{i<j}\omega_{ij}(g,h),
\label{eq:undirected-weight-sum}\\
\theta^u(g,h)
&=
\frac{1}{D^u(g,h)}
\sum_{i<j}\mu_{ij}\omega_{ij}(g,h),
\qquad
\widehat\theta^u(g,h)
=
\frac{1}{D^u(g,h)}
\sum_{i<j}Y_{ij}\omega_{ij}(g,h).
\label{eq:undirected-upper-triangle-density}
\end{align}

The normalized parameter is related to the ordered-matrix parameter in
\eqref{poi} by
\begin{align}
\theta^u(g,h)
=
\frac{m_gm_h}{D^u(g,h)}
\left\{
\theta(g,h)+\theta(h,g)
\right\},
\label{eq:undirected-relation-ordered}
\end{align}
with the same identity for the empirical estimators.  Thus an undirected
density generally requires both orientations of the ordered parameter, as
well as the normalization that removes the diagonal. 


\subsubsection{Confidence intervals}
\label{sec:undirected-normalized-intervals}

Fix the same constant $c\in(0,1/2]$ used in
Section~\ref{sec:index-set}, and let
\[
\mathcal G_c^u(N)
:=
\left\{
(g,h)\in\{0,1\}^N\times\{0,1\}^N:
 m_g\geq cN,\ m_h\geq cN
\right\}.
\]
This is the specialization of $\mathcal G_c$ to $N_1=N_2=N$.  Since
\[
D^u(g,h)=m_gm_h-r(g,h)\geq c^2N^2-N,
\]
the denominator is positive and of order $N^2$ uniformly over
$\mathcal G_c^u(N)$ for all sufficiently large $N$.

Define 
\begin{align}
\sigma^u(g,h)^2
&:=
\frac{1}{D^u(g,h)}
\sum_{i<j}\sigma_{ij}^2\omega_{ij}(g,h)^2.
\label{eq:undirected-normalized-scale}
\end{align}
Then the conditional standard deviation of
$\widehat\theta^u(g,h)$ is
$\sigma^u(g,h)/\sqrt{D^u(g,h)}$.  Let
$\widehat\sigma^u(g,h)$ denote a nonnegative estimator of $\sigma^u(g,h)$.

For the first interval, the natural analogues of Assumptions 2(s) and 3(i) are 
\begin{align}
\lim_{F\in\mathcal F:\,N\to\infty}
\sqrt N
\min_{(g,h)\in\mathcal G_c^u(N)}
\sigma^u(g,h)
&=\infty,
\label{eq:undirected-normalized-strong-nondegeneracy}\\
\lim_{F\in\mathcal F:\,N\to\infty}
\mathbbm P_F\left(
\max_{(g,h)\in\mathcal G_c^u(N)}
\left|
\frac{\widehat\sigma^u(g,h)-\sigma^u(g,h)}
{\sigma^u(g,h)}
\right|>r
\right)
&=0
\label{eq:undirected-normalized-scale-consistency}
\end{align}
for a deterministic sequence $r\to0$.  As elsewhere in the paper, the ratio
is used only when its denominator is positive, which
\eqref{eq:undirected-normalized-strong-nondegeneracy} guarantees eventually.
For coverage, the two-sided condition in
\eqref{eq:undirected-normalized-scale-consistency} may be replaced by the
corresponding one-sided condition that rules out underestimation, as in Appendix Section~\ref{sec:raw-one-sided-definition} above.

Let
\[
K_1^u(\alpha)
:=
\sqrt{1.39(2N)-2\ln(\alpha/2)}
\qquad
K_2(\alpha):=\sqrt{-2\ln(\alpha)}
\]
and define
\begin{align}
CI_1^u(g,h;\alpha)
&:=
\widehat\theta^u(g,h)
\pm
K_1^u(\alpha)
\frac{\widehat\sigma^u(g,h)}
{\sqrt{D^u(g,h)}}
\label{eq:undirected-normalized-CI1}\\
CI_2^u(g,h;\alpha)
&:=
\widehat\theta^u(g,h)
\pm
\frac{2\left\{\widehat{\bar\tau}
+K_2(\alpha)\widehat V\right\}}
{D^u(g,h)}.
\label{eq:undirected-normalized-CI2}
\end{align}
For comparison, the conventional fixed-group benchmark is
\[
CI_0^u(g,h;\alpha)
:=
\widehat\theta^u(g,h)
\pm
K_0(\alpha)
\frac{\widehat\sigma^u(g,h)}
{\sqrt{D^u(g,h)}},
\]
where $K_0(\alpha)$ is the $1-\alpha/2$ standard-normal quantile. 

\begin{proposition}
\label{prop:undirected-normalized-density}
Let $\mathcal F$ be a deterministic class of loopless undirected
unipartite random graph models represented by their upper triangles as in
Section~\ref{sec:model-terminology}. Suppose Assumption~\ref{ass:model}
holds on $\mathcal F$.

\begin{itemize}
\item[(i)] If
\eqref{eq:undirected-normalized-strong-nondegeneracy} and
\eqref{eq:undirected-normalized-scale-consistency} hold, then, for every fixed
$\alpha\in(0,1)$,
\[
\liminf_{F\in\mathcal F:\,N\to\infty}
\mathbbm P_F\left(
\bigcap_{(g,h)\in\mathcal G_c^u(N)}
\left\{
\theta^u(g,h)\in CI_1^u(g,h;\alpha)
\right\}
\right)
\geq1-\alpha.
\]

\item[(ii)] If Assumptions~\ref{ass:nondegeneracy}(w) and
\ref{ass:scale-consistency}(ii)--(iii) hold, then, for every fixed
$\alpha\in(0,1)$,
\[
\liminf_{F\in\mathcal F:\,N\to\infty}
\mathbbm P_F\left(
\bigcap_{(g,h)\in\mathcal G_c^u(N)}
\left\{
\theta^u(g,h)\in CI_2^u(g,h;\alpha)
\right\}
\right)
\geq1-\alpha.
\]
\end{itemize}
\end{proposition}

\begin{proof}
Fix $F\in\mathcal F$ and work under $\mathbbm P_F$.

\medskip
\noindent\textbf{Part (i):}
For $(g,h)\in\mathcal G_c^u(N)$, write
\begin{align*}
S^u(g,h)
&:=
\sum_{i<j}\epsilon_{ij}\omega_{ij}(g,h),\\
\nu^u(g,h)
&:=
\sum_{i<j}\sigma_{ij}^2\omega_{ij}(g,h)^2
=
D^u(g,h)\sigma^u(g,h)^2.
\end{align*}
Then
\[
\widehat\theta^u(g,h)-\theta^u(g,h)
=
\frac{S^u(g,h)}{D^u(g,h)}.
\]
The summands in $S^u(g,h)$ are independent and centered.  Because
$|\epsilon_{ij}|\leq2B$ and $\omega_{ij}(g,h)\leq2$, each summand is bounded
in absolute value by $4B$.  Bernstein's inequality therefore gives, for every
$x>0$,
\[
\mathbbm P_F\left(
|S^u(g,h)|\geq x
\right)
\leq
2\exp\left\{
-
\frac{x^2}
{2\left(\nu^u(g,h)+4Bx/3\right)}
\right\}.
\]
There are at most $4^N=2^{2N}$ pairs $(g,h)$.  Moreover,
$D^u(g,h)\asymp N^2$ uniformly over $\mathcal G_c^u(N)$, and
\eqref{eq:undirected-normalized-strong-nondegeneracy} makes the linear term in
the Bernstein denominator negligible at
$x=K_1^u(\alpha)\sigma^u(g,h)\sqrt{D^u(g,h)}$.  The remaining union-bound
calculation is the same as in the proof of
Proposition~\ref{prop:propone}. Finally, \eqref{eq:undirected-normalized-scale-consistency} permits replacing
the oracle scale by $\widehat\sigma^u(g,h)$.  This proves part (i).

\medskip
\noindent\textbf{Part (ii):}
For indicator vectors $g,h\in\{0,1\}^N$, define the ordered cut sum
\[
T(g,h):=\sum_{i,j}\epsilon_{ij}g_i h_j.
\]
The simultaneous event used in Proposition~\ref{prop:proptwo} satisfies
\[
|T(g,h)|
\leq
\widehat{\bar\tau}+K_2(\alpha)\widehat V
\]
for every $(g,h)\in\mathcal G_c^u(N)$ with limiting probability at least
$1-\alpha$ under the assumptions in part (ii).  By the upper-triangular
convention,
\[
S^u(g,h)
=
T(g,h)+T(h,g).
\]
Hence, on the same event,
\[
|S^u(g,h)|
\leq
2\left\{
\widehat{\bar\tau}+K_2(\alpha)\widehat V
\right\}
\]
simultaneously over all $(g,h)\in\mathcal G_c^u(N)$.  Dividing by
$D^u(g,h)$ proves part (ii).
\end{proof}

\begin{remark}[Why normalization affects the two intervals differently]
\label{rem:undirected-normalization-factor-two}
Equation~\eqref{eq:undirected-relation-ordered} shows that the normalized
undirected numerator combines two ordered cuts.  The event used for $CI_2$
controls each cut separately, so the triangle inequality produces the factor
two in \eqref{eq:undirected-normalized-CI2}.  The numerator is also a single
sum over independent upper-triangular dyads with known weights
$\omega_{ij}(g,h)$.  Applying Bernstein's inequality directly to that sum
uses its combined variance and therefore does not require an additional
factor of two in \eqref{eq:undirected-normalized-CI1}.
\end{remark}

\begin{remark}[Selection and differences]
\label{rem:undirected-normalized-selection}
Because Proposition~\ref{prop:undirected-normalized-density} is simultaneous
over $\mathcal G_c^u(N)$, either interval remains valid after data-dependent
selection of the two groups, provided their relative sizes do not vanish with
probability approaching one.  Intervals for differences between two
normalized densities are obtained by the endpoint arithmetic in
Section~\ref{sec:universal-validity}; no additional multiplicity correction
is required.
\end{remark}

\subsubsection{Partition cells and empirical implementation}
\label{sec:undirected-partition-implementation}
The Facebook100 application in Section~\ref{sec:facebook100} of the main text uses a partition of the node set into
disjoint categories or communities. The homophily structure of a network is then characterized using the density of connections within the group and between the group and its complement. Confidence intervals for the between-group densities can be constructed as in the main text. Confidence intervals for the within-group densities can be constructed as in Appendix Section~\ref{sec:undirected-normalized-intervals} above. 

Specifically, let $A_a$ denote the set of nodes with label $a$ and let $n_a=|A_a|$.  For an unordered pair of labels $a,b$, define
\[
D_{ab}
:=
\begin{cases}
\binom{n_a}{2},&a=b,\\[4pt]
n_an_b,&a\neq b,
\end{cases}
\qquad
q_{ab}
:=
\begin{cases}
1,&a=b,\\
2,&a\neq b.
\end{cases}
\]
Let $\mathcal D_{ab}$ be the corresponding set of eligible unordered dyads, $\hat{\theta}_{ab} = \frac{1}{D_{ab}}\sum_{\{i,j\}\in\mathcal D_{ab}}Y_{ij}$, and $\widehat\sigma_{ab}^{\,2}$ an estimator for $\sigma_{ab}$. For a within-group cell, the general normalization has
$D^u(\mathbbm 1_{A_a},\mathbbm 1_{A_a})=2D_{aa}$ and gives each eligible dyad a
weight of two.  For a between-group cell,
$D^u(\mathbbm 1_{A_a},\mathbbm 1_{A_b})=D_{ab}$ and gives each eligible dyad a
weight of one.  Substituting these identities into
\eqref{eq:undirected-normalized-CI1}--\eqref{eq:undirected-normalized-CI2}
yields the implementation formulas
\begin{align}
CI_1^u(a,b;\alpha)
&=
\widehat\theta_{ab}
\pm
K_1^u(\alpha)
\frac{\widehat\sigma_{ab}}{\sqrt{D_{ab}}},
\label{eq:undirected-partition-CI1}\\
CI_2^u(a,b;\alpha)
&=
\widehat\theta_{ab}
\pm
q_{ab}
\frac{\widehat{\bar\tau}+K_2(\alpha)\widehat V}
{D_{ab}}.
\label{eq:undirected-partition-CI2}
\end{align}
$CI_1^u$ has the same formula for within- and between-group cells.  For
$CI_2^u$, the between-group margin is twice the within-group margin after each
is expressed using the number of eligible unordered dyads. The global quantities $\widehat{\bar\tau}$ and $\widehat V$ are unchanged as they are computed from the complete upper-triangular adjacency matrix. 

\subsection{Facebook100 homophily network}
\label{app:facebook100-details}

\subsubsection{Data}
We use the Facebook100 campus-network files introduced by
\citet{traud2012social} and distributed through \cite{nr}.\footnote{The replication script downloads an archived copy of the MATLAB files from \url{https://archive.org/download/oxford-2005-facebook-matrix/facebook100.zip}.}
We retain every campus network containing between 500 and
10,000 users. We find 50 campuses satisfy this rule. We analyze each campus individually, rather than combine the campuses into one all-campus network. The raw data are undirected Facebook friendship networks. The networks are undirected and unweighted with no loops. 

The raw attribute file contains numeric category codes for student/faculty
status, gender, major, second major or minor, dormitory or house, class year,
and high school. We call these characteristics \emph{attributes} and the individual codes \emph{categories}. So, for example, the attribute gender has two categories: gender 1 and gender 2. Missing and nonpositive codes are treated as missing. The primary analysis considers a category eligible when it contains at least 50
users and at least 5 percent of the full campus network. The latter
restriction is a finite-sample analogue of the nonvanishing-group condition
used in the theory. We retain at most the 20 most frequent eligible
categories for any attribute.

Table~\ref{tab:facebook100-type-rest-summary} in the main text reports results for five principal
single attributes: gender, class year, student/faculty status, major, and
residence or dormitory. 

\subsubsection{Density differences}
We characterize the homophily structure of the Facebook100 networks using density differences. Fix
a campus, an attribute, and an eligible category $A$. Let $\mathcal V$ denote
the set of users with a nonmissing value of that attribute and define its complement $A^{-}:=\mathcal V\setminus A$. We only consider differences where $A$ and $A^{-}$ both contain at least 50
users and at least 5 percent of the full campus network. Let
$n_A=|A|$ and $n_{A^-}=|A^-|$. The numbers of eligible unordered dyads are
\[
D_{A,A}=\binom{n_A}{2},
\qquad
D_{A,A^-}=n_A n_{A^-}.
\]
We additionally require that each denominator is at least 1,000.

Let $E_{A,A}$ be the number of observed friendship edges with both endpoints
in $A$, and let $E_{A,A^-}$ be the number with one endpoint in $A$ and the
other in $A^-$. The corresponding empirical densities are
\[
\widehat\theta^u(A,A)
=
\frac{E_{A,A}}{D_{A,A}},
\qquad
\widehat\theta^u(A,A^-)
=
\frac{E_{A,A^-}}{D_{A,A^-}},
\]
and the reported empirical density difference is
\[
\widehat\Delta_A
=
\widehat\theta^u(A,A)
-
\widehat\theta^u(A,A^-).
\]
When $\widehat\Delta_A>0$, the members of category $A$ are more densely
connected to one another than to other users in other non-missing categories of that attribute. This comparison is category-specific. For an attribute with multiple categories, each eligible category generates a separate difference.

\subsubsection{Estimating the variance parameters}
We estimate the variance parameters using the singular value thresholding procedure described in Appendix Section~\ref{sec:usvt}. Let $A$ denote the symmetric $n\times n$ adjacency matrix and define
\[
\widehat\rho
=
\binom{n}{2}^{-1}\sum_{i<j}A_{ij}.
\]
Let
\[
d_i:=\sum_{j\neq i}A_{ij},
\qquad
\widehat d_{\max}:=\max_i d_i,
\qquad
\widehat d_{\mathrm{avg}}:=(n-1)\widehat\rho.
\]
Following Appendix Section~\ref{sec:cbar-calibration}, we calculate the
maximum-to-average observed degree ratio
\[
\widehat D_n
=
\frac{\widehat d_{\max}}{\widehat d_{\mathrm{avg}}}.
\]
We set $\eta=0.01$ and $\gamma=1$, and fix the grid
\[
\mathcal C=\{1,2,4,8,16,32,64\}.
\]
When $\widehat\rho>0$, we define
\[
\widehat C
=
\min\left\{
C\in\mathcal C:
C\geq(1+\gamma)\widehat D_n
\right\}
\]
when this set is nonempty. The sparse candidate and dense thresholds are
\[
\widehat t_S
=
\left(2\sqrt{2\widehat C}+\eta\right)
\sqrt{n\widehat\rho},
\qquad
\widehat t_D
=
(2+\eta)\sqrt n.
\]
The final threshold is
\[
\widehat t_A
=
\begin{cases}
\min\{\widehat t_S,\widehat t_D\},
&\text{if $\widehat C$ exists},\\[2pt]
\widehat t_D,
&\text{otherwise}.
\end{cases}
\]

We compute singular components of $A$ adaptively, beginning with 100 singular
triplets and increasing this number as necessary, with no user-specified rank
cap. We retain components whose singular values are at least
$\widehat t_A$, reconstruct the resulting matrix, project its entries onto
$[0,1]$, symmetrize it, and set its diagonal to zero. We denote the resulting
estimate $\widehat\mu$. If no component is retained, then
$\widehat\mu=0$ and the residual scales below coincide numerically with the
unshifted raw plug-ins in Section~\ref{sec:raw-plugin}. 

We then define the upper-triangular residuals
\[
\widehat\epsilon_{ij}
=
A_{ij}-\widehat\mu_{ij},
\qquad i<j.
\]
The full-network residual scales are
\[
\widehat S
=
\left(
\sum_{i<j}\widehat\epsilon_{ij}^{\,2}
\right)^{1/2}
\]
and
\[
\widehat T
=
\sum_{i=1}^n
\left(
\sum_{j>i}\widehat\epsilon_{ij}^{\,2}
\right)^{1/2}
+
\sum_{j=1}^n
\left(
\sum_{i<j}\widehat\epsilon_{ij}^{\,2}
\right)^{1/2}.
\]
Because the network is binary, $B=1$, and the global quantities entering
$CI_2$ are
\[
\widehat{\bar\tau}
=
1.01\widehat T+0.25\widehat S
\]
and
\[
\widehat V
=
\left(
\widehat S^2+\widehat S+4\widehat T
\right)^{1/2}.
\]
They are calculated using the entire campus network and are common to every
attribute cell and difference for that campus.

For an unordered cell with eligible dyad set $\mathcal D$, let
$D=|\mathcal D|$. The cell-specific scale used by $CI_0$ and $CI_1$ is
\[
\widehat\sigma_{\mathcal D}^{\,2}
=
\frac{1}{D}
\sum_{\{i,j\}\in\mathcal D}
\widehat\epsilon_{ij}^{\,2}.
\]

\subsubsection{Confidence intervals for densities}

Throughout this subsection, we suppress the superscript $u$ used for the
undirected normalization in Appendix
Section~\ref{app:undirected-normalization}. Let $C=50$ denote the number of
campus networks in the reported analysis and define
\[
\alpha_{\mathrm{FB}}
=
\frac{0.05}{C}
=
0.001.
\]
The conventional fixed-cell benchmark uses
\[
K_0=\Phi^{-1}(0.975).
\]
For the simultaneous intervals, define
\[
K_1(n,\alpha_{\mathrm{FB}})
=
\sqrt{
1.39(2n)-2\log(\alpha_{\mathrm{FB}}/2)
},
\qquad
K_2(\alpha_{\mathrm{FB}})
=
\sqrt{-2\log(\alpha_{\mathrm{FB}})}.
\]
The value $n$ entering $K_1$ is the full campus size rather than the number of
users with a nonmissing value of the attribute.

For a cell with point estimate $\widehat\theta$, denominator $D$, and
cell-specific scale $\widehat\sigma$, define
\[
q=
\begin{cases}
1,&\text{for a within-category cell},\\
2,&\text{for a between-category cell}.
\end{cases}
\]
Following Appendix Section~\ref{app:undirected-normalization} above, the confidence intervals we consider are
\[
CI_0
=
\widehat\theta
\pm
K_0\frac{\widehat\sigma}{\sqrt D},
\]
\[
CI_1
=
\widehat\theta
\pm
K_1(n,\alpha_{\mathrm{FB}})
\frac{\widehat\sigma}{\sqrt D},
\]
\[
CI_2
=
\widehat\theta
\pm
q\,
\frac{
\widehat{\bar\tau}
+
K_2(\alpha_{\mathrm{FB}})\widehat V
}{D},
\]
and
\[
CI_{\cap}
=
CI_1(\alpha_{\mathrm{FB}}/2)
\cap
CI_2(\alpha_{\mathrm{FB}}/2).
\]

\subsubsection{Confidence intervals for the density differences}
For any one of  $t\in\{0,1,2,\cap\}$, we denote the associated within-density confidence interval 
\[
CI_t(A,A)=[L_t^W(A),U_t^W(A)]
\]
and between-density confidence interval
\[
CI_t(A,A^-)=[L_t^B(A),U_t^B(A)]
\]
as described in the previous subsection.  Using these intervals, we construct the difference interval for the within-between density difference $\theta(A,A) - \theta(A,A^-)$
\[
CI_{t,\Delta}(A)
=
\left[
L_t^W(A)-U_t^B(A),
\;
U_t^W(A)-L_t^B(A)
\right].
\]
We say that the difference \emph{survives simultaneous correction} when the lower endpoints of both $CI_{0,\Delta}$ and $CI_{\cap,\Delta}$ are positive or the upper endpoints of both intervals are negative. These confidence intervals for the within-between density differences are what we report in TTable~\ref{tab:facebook100-type-rest-summary}. 

\subsection{BACI supplier-market network}
\label{app:baci-details}

\subsubsection{Data}
We use the 2023 HS22 file from the CEPII BACI database \citep{gaulier2010baci} revision
V202601. BACI reports bilateral trade values in thousands of current U.S.
dollars. Let $X_{ijk}$ denote the value exported by supplier country $i$ of
HS6 product $k$ to destination country $j$, after summing duplicate source
records. We remove observations with $i=j$, so the network describes
international rather than domestic supply.

Let $\mathcal C$ contain every country and $\mathcal K$ contain every product appearing in the 2023 file. A row node is an exporter $i \in \mathcal{C}$.  A column node is an importer--product market
$m=(j,k)\in\mathcal M:=\mathcal C\times\mathcal K$. The node universe has
\[
N_1=|\mathcal C|=226,
\qquad
N_2=|\mathcal M|
=226\cdot5{,}606
=1{,}266{,}956.
\]
The binary network used for inference is
\[
Y_{i,(j,k)}
=
\mathbbm 1\{X_{ijk}>50{,}000\},
\qquad i\neq j.
\]
Because BACI values are measured in thousands of dollars, a link indicates
that supplier $i$ exports more than \$50 million of product $k$ to
destination $j$. The network contains 57,831 links among 285,065,100
structurally eligible supplier--market pairs, for an observed density of
0.0002029.

We also calculate the supplier share
\[
W_{i,(j,k)}
=
\begin{cases}
X_{ijk}/X_{\cdot jk},&X_{\cdot jk}>0,\\
0,&X_{\cdot jk}=0,
\end{cases}
\qquad
X_{\cdot jk}:=\sum_{h\neq j}X_{hjk}
\]
which we use to construct weighted centrality
rankings.
\looseness=-1

\subsubsection{Density differences}

A supplier--market pair is structurally unavailable when the supplier is
also the destination. The code stores a zero in these entries but does not
count them as observed nonlinks. For a supplier group
$G\subseteq\mathcal C$ and market group $R\subseteq\mathcal M$, define
\[
D_0(G,R)=|G||R|,
\qquad
Z(G,R)=\sum_{(j,k)\in R}\mathbbm 1\{j\in G\},
\qquad
D(G,R)=D_0(G,R)-Z(G,R).
\]
If $E(G,R)$ is the number of observed links in the cell, the associated
density is
\[
\widehat\theta(G,R)=\frac{E(G,R)}{D(G,R)}.
\]
This is a rescaling along the lines of Remark~\ref{rem:alternative-normalization}. An
interval formed under the $D_0(G,R)$ normalization is simply multiplied by
$D_0(G,R)/D(G,R)$. \looseness=-1

To form the density differences that make up Table~\ref{tab:baci-contrast-family-summary}, we rank the suppliers using three network centrality measures. Binary supplier degree is
\[
d_i^A=\sum_mY_{im},
\]
weighted supplier outdegree is
\[
d_i^W=\sum_mW_{im},
\]
and bipartite Katz centrality is calculated from
\[
\mathcal W=
\begin{pmatrix}
0&W\\
W^T&0
\end{pmatrix},
\qquad
d^{K} =
\left(
I-0.5\,\frac{\mathcal W}{\rho(\mathcal W)}
\right)^{-1}\iota.
\]
The supplier-side entries of $d^{K}$ give the country ranking. For each ranking
and $q\in\{0.05,0.10,0.20\}$, the top $\lceil qN_1\rceil$ countries form
the high-centrality tier $H_q$, containing 12, 23, or 46 countries, and the
remaining countries form $O_q$. Ties are resolved deterministically by
country code.

The market groups have three forms. First, we use all destination--product
markets. Second, we divide products into seven mutually exclusive baskets:
HS 01--24; 25--40; 41--49; 50--67; 68--83; 84--85; and 86--97. Each
basket contains every destination for the relevant products and at least
5 percent of the column universe. Third, we form top 5-, 10-, and
20-percent market groups by total imports and market-side bipartite Katz
centrality. Every displayed supplier and market group has relative size at least
0.05.
\looseness=-1

We use the above density, centrality, and market group definitions to form the density differences that show up in Table~\ref{tab:baci-contrast-family-summary}. For a high-centrality supplier group $H$, its complement $O$, and market
group $R$, the empirical density difference is
\[
\widehat\Delta(H,R)
=
\widehat\theta(H,R)-\widehat\theta(O,R).
\]
A positive difference means that the selected high-centrality suppliers have
a larger directed link density to $R$ than the remaining suppliers. The
complete candidate collection contains every combination of the three
supplier rankings, three supplier tiers, and eligible market groups.
\looseness=-1

\subsubsection{Estimating the variance parameters}
Following the strategy outlined in Appendix Section~\ref{sec:raw-plugin}, we set $\widehat\epsilon=Y$ and define 
\[
d_i=\sum_mY_{im},
\qquad
u_m=\sum_iY_{im},
\qquad
L=\sum_{i,m}Y_{im}.
\]
Because $Y$ is binary,
\[
\lVert Y\rVert_F=\sqrt L,
\qquad
\lVert Y\rVert_\dagger
=
\sum_i\sqrt{d_i}+\sum_m\sqrt{u_m}.
\]
With $B=1$, 
\[
\widehat{\bar\tau}^{\,0}
=
1.01\lVert Y\rVert_\dagger+0.25\lVert Y\rVert_F,
\qquad
\widehat V^{\,0}
=
\left(
\lVert Y\rVert_F^2+\lVert Y\rVert_F
+4\lVert Y\rVert_\dagger
\right)^{1/2}.
\]
For the \$50 million threshold this gives,
\[
\widehat{\bar\tau}^{\,0}=38{,}298.863,
\qquad
\widehat V^{\,0}=457.725.
\]

\subsubsection{Confidence intervals}

All intervals use $\alpha=0.05$. Under the full market universe,
\[
K_0=\Phi^{-1}(0.975),
\qquad
K_1
=
\sqrt{
1.39(N_1+N_2)-2\log(0.05/2)
}
=
1327.1738,
\]
\[
K_2=\sqrt{-2\log(0.05)}=2.4477.
\]
For a cell $(G,R)$, we use the intervals
\[
CI_0(G,R)
=
\widehat\theta(G,R)
\pm
K_0\frac{\sqrt{E(G,R)}}{D(G,R)},
\]
\[
CI_1(G,R)
=
\widehat\theta(G,R)
\pm
K_1\frac{\sqrt{E(G,R)}}{D(G,R)},
\]
\[
CI_2^{0}(G,R)
=
\widehat\theta(G,R)
\pm
\frac{
\widehat{\bar\tau}^{\,0}+K_2\widehat V^{\,0}
}{D(G,R)}
\]
and 
\[
CI_{\cap}(G,R;\alpha)
=
CI_1(G,R;\alpha/2)
\cap
CI_2^{0}(G,R;\alpha/2).
\]
Since the network is bipartite, no adjustment along the lines of Appendix Section~\ref{app:undirected-normalization} is needed. Writing $CI_t(H,R)=[L_t^H(R),U_t^H(R)]$ and
$CI_t(O,R)=[L_t^O(R),U_t^O(R)]$ for $t \in \{0,1,2,\cap\}$, the associated confidence intervals for the density difference $\theta(H,R)-\theta(O,R)$ are
\[
CI_{t,\Delta}(H,R)
=
\left[
L_t^H(R)-U_t^O(R),
\;
U_t^H(R)-L_t^O(R)
\right].
\]
These are the confidence intervals that we report in Table~\ref{tab:baci-contrast-family-summary} of the main text. 

\looseness=-1

\subsection{PSID pseudo-employer mobility network}
\label{app:psid-details}

\subsubsection{Data}
Our data construction follows \cite{schmutte2014free}. We use annual industry and occupation information for individuals in the Panel Study of Income Dynamics (PSID) years  1987--1993. An individual identifier combines the 1968 family interview number and person number. We retain person-year observations in which the individual is the household head or spouse and has positive, nonmissing three-digit
industry and occupation codes. Workers must have at least two valid
person-year observations and must be at least 23 years old in 1987. For
worker $w$ in year $t$, the pseudo-employer is denoted 
\[
J_{wt}
=
\left(
\operatorname{industry}_{wt},
\operatorname{occupation}_{wt}
\right).
\]
Our sample contains 45,905 person-year observations on 8,119
workers and 6,346 distinct pseudo-employers. \looseness=-1

For each worker, retained observations are ordered by year. A directed
transition from pseudo-employer $a$ to pseudo-employer $b$ is recorded when
the worker's next valid observation is at a different pseudo-employer. Two successive valid observations are linked even when an intervening year has missing industry or occupation information. Let
\[
M_{ab}
=
\#\{
\text{observed transitions from $a$ to $b$}
\}.
\]
The full sample contains 19,921 transitions. To select communities, we
use the undirected weighted network
\[
W_{ab}=M_{ab}+M_{ba},
\qquad a<b.
\]
Among the 6,346 observed pseudo-employers, 6,304 are part of at least one
non-self transition. The symmetrized graph contains 13,901 distinct
undirected edges. Its largest connected component has
\[
n=6{,}039
\]
pseudo-employers and 13,714 distinct undirected edges.

Table~\ref{tab:psid-market-structure-summary} treats this largest connected
component as the empirical analysis universe, which is designed to match Figure 1 of \cite{schmutte2014free}. Within it there are 19,666 weighted transition events and 16,123 distinct directed links. Our density parameters concern the directed network 
\[
Y_{ab}=\mathbbm 1\{M_{ab}>0\},
\qquad a\neq b,
\]
with $Y_{aa}=0$. \looseness=-1

\subsubsection{Market segments}
Following \cite{schmutte2014free}, we apply weighted Louvain community detection to the 6,039-node undirected network using $W_{ab}$ as the edge weight. The reported run contains 34
communities and has modularity 0.6889. We order the communities by size and
retain the four largest, denoted $g_1,\ldots,g_4$, whose sizes are
\[
n_1=445,\qquad n_2=445,\qquad n_3=365,\qquad n_4=349.
\]
We restrict attention to the four largest communities following Figure 1 of \cite{schmutte2014free}.

We also define market segments using major 1970 Census industry, major 1970
Census occupation, and their interaction. A market segment is counted
only when both it and its complement contain at least 5 percent of the
6,039-node universe. The minimum admissible size segment is therefore
\[
\left\lceil0.05(6{,}039)\right\rceil=302.
\]
The four Louvain communities, six major-industry groups, seven
major-occupation groups, and four industry-by-occupation groups satisfy this
restriction.
\looseness=-1

\subsubsection{Density differences}

For an admissible market segment $g$, let $m_g=|g|$ and let $g^c$ denote its
complement in the 6,039-node analysis universe. The within-segment density is
\[
\widehat\theta(g,g)
=
\frac{
\sum_{a\in g}\sum_{\substack{b\in g\\b\neq a}}Y_{ab}
}{
m_g(m_g-1)
}.
\]
The outward and inward densities are
\[
\widehat\theta(g,g^c)
=
\frac{
\sum_{a\in g}\sum_{b\in g^c}Y_{ab}
}{
m_g(n-m_g)
},
\qquad
\widehat\theta(g^c,g)
=
\frac{
\sum_{a\in g^c}\sum_{b\in g}Y_{ab}
}{
m_g(n-m_g)
}.
\]
The reported segmentation difference is
\[
\widehat\Delta(g)
=
\widehat\theta(g,g)
-
\frac{
\widehat\theta(g,g^c)+\widehat\theta(g^c,g)
}{2}.
\]
A positive difference means that directed worker-flow links are denser within
the segment than across its boundary, averaging the outward and inward
directions. This is a reduced-form mobility-segmentation measure rather than
a structural parameter of a labor-supply or wage-setting model.
\looseness=-1

\subsubsection{Estimating the variance parameters}
Following the strategy outlined in Appendix Section~\ref{sec:raw-plugin} we use $\widehat\epsilon=Y$ which leads to the estimates 
\[
\widehat{\bar\tau}^{\,0}=16{,}441.042,
\qquad
\widehat V^{\,0}=285.022.
\]

All intervals use $\alpha=0.05$. Since $N_1=N_2=6{,}039$,
\[
K_0=\Phi^{-1}(0.975),
\qquad
K_1(0.05)=129.5986,
\qquad
K_2(0.05)=2.4477.
\]
For a directed density cell with $E$ links and eligible denominator $D$,
\[
\left(\widehat\sigma^{\,0}\right)^2
=
\frac{E}{D}
=
\widehat\theta.
\]
These estimates lead to the following confidence intervals for the individual density parameters
\[
CI_0
=
\widehat\theta
\pm
K_0\frac{\sqrt E}{D},
\qquad
CI_1
=
\widehat\theta
\pm
K_1\frac{\sqrt E}{D},
\]
\[
CI_2^0
=
\widehat\theta
\pm
\frac{
\widehat{\bar\tau}^{\,0}+K_2\widehat V^{\,0}
}{D}
\]
and 
\[
CI_{\cap}(g_1,g_2;\alpha)
=
CI_1(g_1,g_2;\alpha/2)
\cap
CI_2^0(g_1,g_2;\alpha/2).
\]
where  $K_1(\alpha/2)=129.6040$ and $K_2(\alpha/2)=2.7162$.

The difference intervals reported in Table~\ref{tab:psid-market-structure-summary} are then constructed in the following way. For $t \in \{0,1,2,\cap\}$, where $t=2$ denotes $CI_2^{0}$, let the within, outward, and inward intervals be
$[L_t^{W},U_t^{W}]$, $[L_t^{O},U_t^{O}]$, and
$[L_t^{I},U_t^{I}]$. The induced interval for the segmentation difference is
\[
CI_{t,\Delta}(g)
=
\left[
L_t^{W}-\frac{U_t^{O}+U_t^{I}}{2},
\;
U_t^{W}-\frac{L_t^{O}+L_t^{I}}{2}
\right].
\]

\end{document}